\documentclass{article}
\usepackage[letterpaper,margin=1in]{geometry}
\usepackage[intlimits]{amsmath}
\usepackage{amsfonts}
\usepackage{authblk}
\usepackage[bf]{caption}       
\usepackage{fancyhdr}
\usepackage{fancybox}
\usepackage{ifthen}
\usepackage{url}
\usepackage{lscape,afterpage}
\usepackage{xspace}
\usepackage{epstopdf} 
\usepackage{subfig}
\usepackage{graphicx}
\usepackage{appendix}
\usepackage[utf8]{inputenc}
\usepackage[T1]{fontenc}
\usepackage{microtype}
\usepackage{braket}
\usepackage{qcircuit}
\usepackage{amsthm}
\usepackage{amsmath,amssymb}
\usepackage{mathrsfs}
\usepackage{eufrak}
\usepackage{soul}
\usepackage{array}
\usepackage{makecell}
\usepackage{mdframed}
\usepackage{stmaryrd}
\usepackage{hyperref}
\usepackage[bottom]{footmisc}
\usepackage{scrextend}
\usepackage{tikz}
\usepackage{tikz-cd}
\usetikzlibrary{calc,arrows,decorations.pathreplacing,decorations.pathmorphing,intersections}
\usetikzlibrary{positioning}
\usetikzlibrary{fit}
\theoremstyle{plain}
\newtheorem{thm}{Theorem}[section]
\newtheorem{lem}[thm]{Lemma}

\newtheorem{prtl}{Protocol}
\newtheorem{conv}{Convention}
\newtheorem{toyprtl}{Toy Protocol}

\newtheorem{cor}[thm]{Corollary}
\newtheorem{assump}{Assumption}
\newtheorem{claim}[thm]{Claim}
\newtheorem{fact}{Fact}

\newtheorem{setup}{Set-up}

\theoremstyle{definition}
\newtheorem{defn}{Definition}[section]

\newtheorem{exmp}{Example}[section]
\newtheorem{nota}{Notation}[section]
\theoremstyle{remark}

\newcommand{\nRe}{{\operatorname{Re}}}
\newcommand{\nIm}{{\operatorname{Im}}}
\newcommand{\cR}{{\mathcal{R}}}
\newcommand{\cC}{{\mathcal{C}}}

\newcommand{\cF}{{\mathcal{F}}}

\newcommand{\cU}{{\mathcal{U}}}
\newcommand{\cP}{{\mathcal{P}}}

\newcommand{\mi}{{\mathrm{i}}}
\newcommand{\btheta}{{\boldsymbol{\theta}}}
\newcommand{\bTheta}{{\boldsymbol{\Theta}}}
\newcommand{\bDelta}{{\boldsymbol{\Delta}}}
\newcommand{\bI}{{\boldsymbol{I}}}
\newcommand{\bS}{{\boldsymbol{S}}}
\newcommand{\bD}{{\boldsymbol{D}}}
\newcommand{\bA}{{\boldsymbol{A}}}

\newcommand{\bx}{{\boldsymbol{x}}}
\newcommand{\bK}{{\boldsymbol{K}}}
\newcommand{\bH}{{\boldsymbol{H}}}
\newcommand{\bR}{{\boldsymbol{R}}}
\newcommand{\bY}{{\boldsymbol{Y}}}
\newcommand{\bd}{{\boldsymbol{d}}}
\newcommand{\bindic}{{\boldsymbol{indic}}}
\newcommand{\bct}{{\boldsymbol{ct}}}
\newcommand{\bflag}{{\boldsymbol{flag}}}
\newcommand{\btype}{{\boldsymbol{type}}}
\newcommand{\bscore}{{\boldsymbol{score}}}

\newcommand{\brec}{{\boldsymbol{rec}}}
\newcommand{\bpads}{{\boldsymbol{pads}}}
\newcommand{\bpad}{{\boldsymbol{pad}}}
\newcommand{\btrans}{{\boldsymbol{trans}}}
\newcommand{\bbC}{{\boldsymbol{C}}}
\DeclareMathOperator{\bE}{\mathbb{E}}
\DeclareMathOperator{\bC}{\mathbb{C}}
\DeclareMathOperator{\bbR}{\mathbb{R}}
\DeclareMathOperator{\bN}{\mathbb{N}}
\DeclareMathOperator{\bZ}{\mathbb{Z}}
\newcommand{\bbI}{\mathbb{I}}

\newcommand{\fH}{{\sf H}}
\newcommand{\fT}{{\sf T}}
\newcommand{\fP}{{\sf P}}

\newcommand{\fX}{{\sf X}}
\newcommand{\fY}{{\sf Y}}
\newcommand{\fZ}{{\sf Z}}
\newcommand{\fM}{{\sf M}}

\newcommand{\fToffoli}{{\sf Toffoli}}

\newcommand{\fCN}{{\sf CNOT}}
\newcommand{\fSWAP}{{\sf SWAP}}
\newcommand{\fSet}{{\sf Set}}
\newcommand{\fAdd}{{\sf Add}}
\newcommand{\fCOPY}{{\sf COPY}}
\newcommand{\fCPhase}{{\sf CPhase}}
\newcommand{\fReviseRO}{{\sf ReviseRO}}
\newcommand{\fServersim}{{\sf Serversim}}

\newcommand{\fDisgard}{{\sf Disgard}}
\newcommand{\fCalcRel}{{\sf CalcRel}}
\newcommand{\fCalc}{{\sf Calc}}

\newcommand{\Domain}{{\text{Domain}}}
\newcommand{\basishonest}{{\text{basishonest}}}
\newcommand{\tSUM}{{\text{SUM}}}

\newcommand{\fEn}{{\mathsf{Enc}}}

\newcommand{\fPadHadamard}{{\mathsf{PadHadamard}}}

\newcommand{\fGAUVBQC}{{\mathsf{GAUVBQC}}}

\newcommand{\fLT}{{\mathsf{LT}}}

\newcommand{\fPrtl}{{\mathsf{Prtl}}}

\newcommand{\fCombine}{{\mathsf{Combine}}}
\newcommand{\fAuxInf}{{\mathsf{AuxInf}}}

\newcommand{\fDc}{{\mathsf{Dec}}}
\newcommand{\fEv}{{\mathsf{Eval}}}
\newcommand{\fKg}{{\mathsf{KeyGen}}}

\newcommand{\fneg}{{\mathsf{negl}}}
\newcommand{\fpoly}{{\mathsf{poly}}}

\newcommand{\fgadget}{{\mathsf{gadget}}}

\newcommand{\fAdv}{{\mathsf{Adv}}}

\mdfdefinestyle{figstyle}{ 
  linecolor=black!7, 
  backgroundcolor=black!7
}

\mdfdefinestyle{figstyle}{ 
  linecolor=black!7, 
  backgroundcolor=black!7
}
\usepackage{float}
\newcommand{\fSim}{{\mathsf{Sim}}}

\newcommand{\fNTCF}{{\mathsf{NTCF}}}
\newcommand{\fCHK}{{\mathsf{CHK}}}
\newcommand{\ftrue}{{\mathsf{true}}}
\newcommand{\ffalse}{{\mathsf{false}}}

\newcommand{\fpass}{{\mathsf{pass}}}
\newcommand{\ffail}{{\mathsf{fail}}}

\newcommand{\tflag}{{\text{flag}}}
\newcommand{\ttype}{{\text{type}}}
\newcommand{\score}{{\text{score}}}
\newcommand{\ftest}{{\mathsf{test}}}
\newcommand{\fquiz}{{\mathsf{quiz}}}
\newcommand{\fcomp}{{\mathsf{comp}}}
\newcommand{\fwin}{{\mathsf{win}}}
\newcommand{\flose}{{\mathsf{lose}}}

\newcommand{\fAddPhaseWithswitch}{{\mathsf{SwPhaseUpdate}}}
\newcommand{\fCollapse}{{\mathsf{Collapse}}}

\newcommand{\fBNTest}{{\mathsf{BUTest}}}

\newcommand{\fCoPhTest}{{\mathsf{CoPhTest}}}
\newcommand{\fStdBTest}{{\mathsf{StdBTest}}}
\newcommand{\fInPhTest}{{\mathsf{InPhTest}}}

\newcommand{\fHadamardTest}{{\mathsf{HadamardTest}}}
\newcommand{\fRSPV}{{\mathsf{RSPV}}}
\newcommand{\fpreRSPV}{{\mathsf{preRSPV}}}
\newcommand{\fpreRSPVTemp}{{\mathsf{preRSPVTemp}}}

\newcommand{\tspan}{{\text{span}}}
\newcommand{\sk}{{\text{sk}}}
\newcommand{\OPT}{{\text{OPT}}}
\newcommand{\SETUP}{{\text{SETUP}}}

\newcommand{\pad}{{\text{pad}}}
\newcommand{\pk}{{\text{pk}}}

\newcommand{\switch}{{\text{switch}}}
\newcommand{\fResponse}{{\mathsf{Response}}}

\begin{document}
\title{Classical Verification of Quantum Computations in Linear Time}
\author[1]{Jiayu Zhang\thanks{Supported by the IQIM, an NSF Physics Frontiers Center 
(NSF Grant PHY-1125565) with support of the Gordon and Betty Moore 
Foundation (GBMF-12500028).}}
\affil[1]{California Institute of Technology}
\affil[ ]{jiayu@caltech.edu}
\maketitle
\begin{abstract}
	In the quantum computation verification problem, a quantum server wants to convince a client that the output of evaluating a quantum circuit $C$ is some result that it claims. This problem is considered very important both theoretically and practically in quantum computation \cite{GKK,ABEM,RUV}. The client is considered to be limited in computational power, and one desirable property is that the client can be completely classical, which leads to the classical verification of quantum computation (CVQC) problem. In terms of the time complexity of server-side quantum computations (which typically dominate the total time complexity of both the client and the server), the fastest single-server CVQC protocol so far has complexity $O(\fpoly(\kappa)|C|^3)$ where $|C|$ is the size of the circuit to be verified and $\kappa$ is the security parameter, given by Mahadev \cite{MahadevVerification}. This leads to a similar cubic time blowup in many existing protocols including multiparty quantum computation, zero knowledge and obfuscation \cite{B21,VidickZ19,BM21,CCT,CLLW20,ACGH}. Considering the preciousness of quantum computation resources, this cubic complexity barrier could be a big obstacle for theoretical and practical development of protocols for these problems. \par
	 In this work, by developing new techniques, we give a new CVQC protocol with complexity $O(\fpoly(\kappa)|C|)$ (in terms of the total time complexity of both the client and the server), which is significantly faster than existing protocols. Our protocol is secure in the quantum random oracle model \cite{QRO} assuming the existence of noisy trapdoor claw-free functions \cite{BCMVV}, which are both extensively used assumptions in quantum cryptography. Along the way, we also give a new classical channel remote state preparation protocol for states in $\{\ket{+_\theta}=\frac{1}{\sqrt{2}}(\ket{0}+e^{\mi\theta\pi/4}\ket{1}):\theta\in \{0,1\cdots 7\}\}$, another basic primitive in quantum cryptography. Our protocol allows for parallel verifiable preparation of $L$ independently random states in this form (up to a constant overall error and a possibly unbounded server-side simulator), and runs in only $O(\fpoly(\kappa)L)$ time and constant rounds; for comparison, existing works (even for possibly simpler state families) all require very large or unestimated time and round complexities  \cite{GVRSP,qfactory,GMP,FWZ}.\par
\end{abstract}
\newpage\tableofcontents
\section{Introduction}\label{sec:1}
\subsection{Background}
Verification of computations is one of the most basic questions that one could ask about computations. In this problem, an untrusted server claims the output of running a circuit $C$ is $o$, and the client would like to check its validity without doing all the computations from scratch. 
The study of this type of problems in different settings has a very long history.  In the setting where the server has unlimited computation resources, many important complexity classes (like NP, IP, MIP \cite{AroraBarak}) and famous results (like the PCP theorem \cite{AroraBarak} or IP=PSPACE \cite{IPPSPACE}) can be understood as characterizations of power of verification protocols in different settings.\par
Consider a practical setting where a user wants to outsource a large scale computation to an untrusted cloud server, we need to additionally assume the server runs in polynomial time. The computation verification problem in this setting is widely-studied and widely-used in cryptography. For example, computation verification has been studied in various different settings \cite{GGP,PHGR} and is the foundation of various cryptographic problems (for example, zero-knowledge \cite{ZK20}).\par

Today quantum computations are gradually coming into reality. \cite{supremacy} Naturally, we would like to know whether quantum computations are also verifiable, and how it could be executed in practice. Formally speaking, a quantum computation verification protocol is defined as:
\begin{defn}[Quantum computation verification, review of \cite{ABEM}]\label{defn:qcv1}
A quantum computation verification protocol takes a quantum circuit $C$ and an output string $o$ as the inputs. It has completeness $c$ and soundness $s$ if:
\begin{itemize}
\item (Completeness) For $(C,o)$ such that $\Pr[C\ket{0}=o]\geq \frac{99}{100}$, the verifier accepts with probability $\geq c$.
\item (Soundness) For any malicious quantum server, for $(C,o)$ such that $\Pr[C\ket{0}=o]\leq \frac{1}{100}$, the verifier rejects with probability $\geq 1-s$.
\end{itemize}
	In addition to that, we want the protocol to be efficient, that is, both the client and the server should be in polynomial time.
\end{defn}
Quantum computation verification is also very important both in theory and in practice:
\begin{itemize}
\item The motivations of classical computation verification generally also hold in the quantum world. Historically, the study of quantum computation verification has led to a series of important works: For example, quantum computation verification protocols are the basis of many other quantum cryptographic protocols like multiparty computation and zero-knowledge \cite{VidickZ19,B21}; and the study of quantum computation verification in the multi-prover setting leads to one of the most striking results in quantum complexity theory \cite{MIPstar}. 
\item There is a potentially strong practical motivation for quantum computation verification: In foreseeable future, it is possible that large scale quantum computers will be used as cloud services instead of personal computers due to its extreme running conditions \cite{supremacy}, which makes the trust issue between the client and the server(s) more problematic.
\end{itemize}
On the other hand, quantum computation verification faces new difficulties that do not exist in the classical world:
\begin{itemize}
\item In quantum world, measurements are generally destructive, which is very different from the classical world. This means we can't trace and see what is happening during executions of quantum algorithms, which forbids an intuitive way of verifying computations.
\item In a classical world, a user that holds a small computation device could always simulate a slightly larger scale computation by using real-world storage devices to enlarge its memory. Storing classical information in real life is generally cheap. Studies of computation verification in the classical world generally aims at verification of very large scale computations or more advanced functionalities. \cite{GGP} However, quantum memory is not necessarily cheap. That implies, a user that already holds a quantum computer will still need to worry about the validity of outputs of larger scale quantum computations claimed by other untrusted parties.\par
\end{itemize}
We review some verification methods that are (possibly) practically useful but do not follow Definition \ref{defn:qcv1}, and discuss their restrictions.
\begin{itemize}
	\item Cross-check of different quantum devices. This method relies on the assumption that either at least one quantum computer is reliable, or they will not maliciously deviate in a similar way. However, based on the development of classical computation technologies, it's possible that in the long run only a small number of nations or companies could be able to build large scale quantum computers. In this situation, this technique may not be sufficient for building trust in a large scale. 
		\item Verification by solving problems in $NP\cap BQP$, like the factorization problem \cite{Shorsalg}. However, good choices for this class of problems are limited and (as far as we know) do not contain many important quantum computation algorithms like Hamiltonian simulation \cite{Hamiltonian}. With only this verification method, a malicious server could choose to behave honestly only on these specific problems and deviate on all the other problems.
\end{itemize}
Besides the basic conditions given in Definition \ref{defn:qcv1}, there are various additional factors that people pay attention to. These include the assumptions used, complexities, etc. One very desirable property is the client could be completely classical. This is called the classical verification of quantum computation (CVQC) problem, which is the focus of this paper.\par

\subsection{Existing Works}
\subsubsection{Verification of quantum computations}
There are various approaches for the verification of quantum computations. \cite{GKK} 
\begin{itemize}\item One approach that has a very long history is verification with a single quantum server and a client with a bounded quantum memory. These protocols include the Clifford-authentication-based protocol \cite{ABEM}, polynomial-code-based protocol \cite{ABEM}, protocols based on measurement-based quantum computation and trap qubits \cite{UVBQC}, the receive-and-measure protocol based on Hamiltonians \cite{FHM}, verification by randomly selected round types \cite{B15}, etc. These protocols generally require a small (for example, single-qubit) quantum device on the client side, and are information-theoretically (IT-) secure; however, the client still needs to do quantum computations, and for all the existing verification protocols that achieve IT-security in this setting, the client side quantum computations (thus also the total complexity) are at least linear (or even more) in the circuit size \cite{ABEM,FKD,UVBQC,FHM,B15}.
\item There is also a long history of verification with multiple entangled quantum servers and a completely classical client. \cite{RUV,leash,svqc} It is assumed that there is no communication among servers, thus the client can make use of the joint behavior of these servers to test each other. This class of protocols generally achieves information-theoretical security; but the requirement of multiple non-communicating entangled servers might be costly to guarantee at a scale in practice.
\item A relatively new approach is to base the protocols on computational assumptions. Early stage works like \cite{ADSS} do not achieve classical verification. Mahadev constructed the first classical verification of quantum computation (CVQC) protocol in \cite{MahadevVerification}. This protocol is based on a new primitive called \emph{noisy trapdoor claw-free functions}, which can be constructed from the Learning-With-Errors assumption. Based on this work, a series of new CVQC protocols are developed \cite{CCT,ACGH,CLLW20} which improve \cite{MahadevVerification} in different ways.\end{itemize}
 As said before, single-server cryptography-based CVQC is possible by \cite{MahadevVerification}. Considering the preciousness of quantum computation resources, the next factor to consider after proving the existence might be to find a protocol with lower complexity. However, the currently fastest existing works run in $O(|C|^3)$ time complexity for verifying a circuit of size $|C|$, for a fixed security parameter. In more detail:
\begin{itemize}
\item The original Mahadev's protocol builds on the Hamiltonian-based approach \cite{FHM} which leads to a cubic complexity. This complexity is inherited by the series of works built on it \cite{CCT,ACGH,CLLW20}.
\item There are also works that take the approach of remote state preparation like \cite{GVRSP,qfactory}. Their complexities are polynomial but are unestimated; a back-of-envelope calculation shows their complexities might be very large.\footnote{We point out that the $O(1/\epsilon^3)$ and $O(T^4)$ complexities claimed in Section 1.1 of \cite{GVRSP} underestimate the real complexities of their protocols. A calculation following its security proofs gives a much higher complexity. We thank the author(s) of \cite{GVRSP} for confirming it.}
\item Even if we allow the usage of multiple quantum servers, the problem still exists in a sense: although there exists a quasi-linear time protocol \cite{leash}, the no-communication requirement is an a-priori assumption, and it's not known how to base it on relativity-based space separation---which means guaranteeing the separation of different servers will be hard to achieve in a practical quantum network environment. If we focus on multi-server protocols where the no-communication condition can be based on space-like (relativity-based) separation, the fastest protocol known is still in cubic time \cite{svqc}.	
\end{itemize}
For an intermediate-size problem, cubic complexity might already be too large to run in practice, especially for quantum computations. (For all the protocols listed above, the total time complexities are equal to the complexities of the server's quantum computations; we will simply use ``complexity'' to mean both.) This leads to the following question:
\begin{center}
\emph{Could classical verification of quantum computations be faster?}	
\end{center}
\subsubsection{Related problem: remote state preparation}\label{sec:1.2.2}
A very basic notion in quantum cryptography that our work will be closely related to is \emph{remote state preparation}, raised in \cite{firstrsp}. There are different security notions for remote state preparation, including blindness and verifiability \cite{DK16,GVRSP}; We focus on  remote state preparation with verifiability (RSPV). In this problem, ideally, the client wants to send a uniformly random state from a state family. The client wants to use a protocol to interact with the server, so that when the protocol completes, if the server is not caught cheating, the server should hold the ideal state (approximately), as if the client prepares and sends it directly. This property is called the verifiability of remote state preparation. Necessarily, this notion of verifiability is defined up to a server-side isometry\footnote{There is a subtle difference between \emph{isometry-based RSPV} and \emph{simulation-based RSPV}; in the formal proof we use simulation-based notion (Definition \ref{defn:rspvv}) but for the informal discussion we blur the differences between them.}: the server could choose the basis freely and use this basis for all of its own operations, and there is no way to detect this change-of-basis from the outside.\par
This notion is basic and very useful in quantum cryptography. To demonstrate its applications, we note that many existing quantum cryptographic protocols have the following structure \cite{UBQC,UVBQC}: \begin{enumerate}\item The client first prepares some quantum gadgets (small size secret states) and sends them to the server; \item Both parties interact classically to achieve some tasks. We will call this step the \emph{gadget-assisted protocol}.

 \end{enumerate}
An undesirable property of these protocols is the client still needs to prepare and send (possibly many) quantum gadgets in the first step. RSPV could be used to replace the first step above approximately; if the RSPV protocol only relies on classical channels, we get a compiler that compiles a quantum channel protocol to a classical channel protocol with a similar functionality.\par
The authors of \cite{DK16} consider whether it's possible to design a classical-channel RSPV protocol for the gadgets used in \cite{UBQC,UVBQC} etc. In \cite{UBQC} the set of possible gadgets is $\{\ket{+_\theta}=\frac{1}{\sqrt{2}}(\ket{0}+e^{\theta\mi\pi/4}\ket{1}):\theta\in \{0,1\cdots 7\}\}$; in \cite{UVBQC} computational basis states $\{\ket{0},\ket{1}\}$ are added to the state family to support more advanced functionalities. The success of Mahadev's technique \cite{BCMVV,MahadevVerification} leads to a series of works on the possibility of constructing classical channel RSPV protocols for these state families \cite{GVRSP,qfactory}.\par
 Usually RSPV is defined on a small state family (for example, $\ket{+_{\theta}}$ discussed above), while gadget-assisted protocol in general requires a large number of such gadgets. For convenience we define a variant of RSPV that takes the gadget number $L$ as the input:
 \begin{defn}[Informal]\label{defn:infrspv}
 	RSPV for $L$ gadgets in the form of $\{\ket{+_\theta}=\frac{1}{\sqrt{2}}(\ket{0}+e^{\theta\mi\pi/4}\ket{1}):\theta\in \{0,1\cdots 7\}\}$ is defined to be a protocol that takes $1^L$ as inputs and satisfies:
 	\begin{itemize}
 	\item (Completeness) If the server is honest, it gets $\ket{+_{\theta^{(1)}}}\ket{+_{\theta^{(2)}}}\cdots \ket{+_{\theta^{(L)}}}$ in the end, where each of $\theta^{(1)}\cdots \theta^{(L)}$ are uniformly independently random in $\{0,1\cdots 7\}$. The client gets $\theta^{(1)}\cdots \theta^{(L)}$.
 	\item (Verifiability) For any (efficient) malicious server, if it could pass the protocol with significant probability, the joint state of the client and the server on the passing space is approximately indistinguishable to the honest state up to a server-side isometry.
 	\end{itemize}

 \end{defn}

 The authors of \cite{GVRSP} provide a positive answer to this problem using only classical channel; independently \cite{qfactory} also provides a candidate protocol whose security is shown against a restricted form of adversaries.\par
 However, the time complexities of protocols in \cite{qfactory,GVRSP} are not clear. Although both protocols are in polynomial time, the complexities are either completely implicit \cite{qfactory} or not fully calculated \cite{GVRSP} (a back-of-envolope calculation shows the order of the polynomial is tens or hundreds.) Considering the wide applications of classical channel RSPV, we ask the following question:
 \begin{center}
 \emph{Could classical channel RSPV for $L$ gadgets in the form of  $\{\ket{+_\theta},\theta\in \{0,1\cdots 7\}\}$ be faster?}	
 \end{center}
 An answer to this question could also open the road to RSPV protocols for more general state families. 
\subsubsection{Related works: a review of existing applications of CVQC and RSPV}\label{sec:1.2.3}
Since \cite{MahadevVerification}, there have been a series of works that built on the protocol or its techniques.
\begin{itemize}
	\item CVQC protocols with improvements over \cite{MahadevVerification}: the authors of \cite{CCT}, \cite{ACGH} construct non-interactive (2-rounds) CVQC protocols; \cite{CCT} further construct a protocol where the client-side classical computation is in only $\fpoly(\kappa)$ time. The authors of \cite{CLLW20} constructs blind, constant-round CVQC protocols for sampling problems.
	\item Multiparty quantum computation: The author of \cite{B21} constructs multiparty quantum computation protocols over classical channel, and constructs a composable bind CVQC protocol along the way from \cite{MahadevVerification, Mahadev2017ClassicalHE,CLLW20}.
	\item Zero-knowledge: \cite{VidickZ19} constructs classical zero-knowledge arguments for QMA. \cite{ACGH} constructs a non-interactive zero-knowledge protocol for QMA.
	\item Obfuscation: \cite{BM21} constructs an indistinguishability obfuscation scheme for null quantum circuits, based on non-interactive CVQC protocol with special properties. This results implies a series of fancy functionalities including publicly-verifiable NIZK, k-SNARG, ZAPR for QMA, attribute-based encryption for BQP etc, as discussed in \cite{BM21}.
\end{itemize}
All of these protocols have a cubic time complexity blowup inherited from \cite{MahadevVerification}.\par 
Since \cite{BCMVV}, there are also a series of works on classical-channel RSPV and related problems.
\begin{itemize}
	\item \cite{GVRSP,qfactory} construct classical-channel RSPV protocols (with or without proofs) for non-trivial single qubit state families, and  show these protocols could be useful for important problems like composable CVQC.
	\item \cite{MV21} constructs a single-server self-testing protocol; this leads to a new protocol for device-independent quantum key distribution \cite{MDCA}.
	\item \cite{GMP} construct a new RSPV protocol that allows for preparation of a large number of BB84 states; in \cite{FWZ} the authors construct a parallel single-server self-testing protocol. As shown in \cite{GMP}, these type of RSPV protocols could lead to a series of classical-channel protocols for problems including unclonable quantum encryption, quantum copy-protection, and more. 
\end{itemize}
These protocols, although polynomial-time, have very large or unestimated time complexity based on current analysis.\footnote{We note the settings of these protocols are not the same, and we consider the following setting for a fair comparison: if the protocol only considers the preparation of a single state, we consider its $L$-fold repetition and require the total error to be a constant; for protocols with small soundness error like \cite{FWZ}, we consider its repetition-based amplification that takes it to constant soundness error.}

\subsection{Our Results}\label{sec:1.3}
In this paper we make significant progress for the problems above. We work in the quantum random oracle model (QROM) \cite{QRO}, the ideal model for symmetric key encryptions or hash functions in the quantum world. (See Section \ref{sec:3.3} for a review.)\par
 As our central result, we prove the following:
\begin{thm}\label{thm:main}
Assuming the existence of post-quantum noisy trapdoor claw-free functions, there exists a single server CVQC protocol in QROM such that:
\begin{itemize}
\item The protocol has completeness $\frac{2}{3}$.
	\item For verifying a circuit of size $|C|$, the total time complexity is $O(\fpoly(\kappa)|C|)$, where $\kappa$ is the security parameter.
	\item The protocol has soundness $\frac{1}{3}$ against BQP adversaries. 
\end{itemize}
\end{thm}
This means we construct a CVQC protocol that runs in time only linear in the circuit size $|C|$, which is optimal in terms of dependence on $|C|$. The noisy trapdoor claw-free functions (NTCF) \cite{BCMVV} in this theorem could be constructed from the Learning-With-Errors assumption, as given in \cite{BCMVV}. (See Section \ref{sec:ntcf} for a review.) The random oracle could be heuristically instantiated by a symmetric key encryption or hash function in practice, which is called the random oracle methodology \cite{roreliable}. Both are widely-used assumptions in cryptography.
\par
The $\fpoly(\kappa)$ in Theorem \ref{thm:main} is only linear in the time complexity of the noisy trapdoor claw-free functions (and the hash functions if we instantiate the random oracle).\par
Along the way, we construct a classical channel RSPV protocol for $\{\ket{+_\theta},\theta\in \{0,1\cdots 7\}\}$ that runs in linear time and constant rounds:
\begin{thm}[Informal]\label{thm:infrspv}There exists a classical channel RSPV protocol  for $L$ gadgets in the form of $\{\ket{+_\theta},\theta\in \{0,1\cdots 7\}\}$ in QROM that runs in time $O(\fpoly(\kappa)L)$ and constant rounds.
\end{thm}
The construction of the RSPV protocol is completed in Protocol \ref{prtl:rspvr} and the construction of the CVQC protocol is completed in Protocol \ref{prtl:cvqc}. Their proofs are completed in the corresponding sections. 
\paragraph{A quick summary of technical innovations} We develop a set of techniques that are very different from existing CVQC or RSPV protocols. At a high level, we give up the Hamiltonian approach used in many existing works and seek for a fast RSPV protocol as an intermediate step towards a fast CVQC protocol. As discussed before, our RSPV protocol aims at preparing states in the form of $\ket{+_{\theta}}$.\par
This problem is nontrivial even without considering the complexity. Existing works that are powerful enough to handle this type of states, like \cite{GVRSP}, work as follows at a high level: the client instructs the server to do an NTCF \cite{BCMVV} evaluation followed by a partial measurement, which creates $\ket{+_{\theta}}$ (or similar states) for a random $\theta$ on the server-side; the client could calculate $\theta$ from the server's response and its secret information (trapdoor etc). Then a series of tests are probabilistically executed on this state, where the client asks the server to measure the qubit in a basis (either related to $\theta$ or unrelated to $\theta$) and uses the server's feedback to check it has really prepared $\ket{+_{\theta}}$ as expected. Importantly, \cite{GVRSP} designed a test based on quantum random access code \cite{ALMO}.\par
We give a very brief overview of our protocol as follows. In Section \ref{sec:2} we give a detailed technical overview.
\begin{enumerate}
	\item To allow the honest server to get the state, different from existing works, we make use of the \emph{phase table} construction \cite{revGT1}, coming from a work on a different quantum delegation problem, that generalizes garbled tables \cite{YaoGCOrigin} construction into the quantum world. This technique could not directly create the single-qubit state $\ket{+_{\theta}}$; instead, it creates an encoded form of this state, which is $e^{\theta_0\mi\pi/4}\ket{x_0}+e^{\theta_1\mi\pi/4}\ket{x_1}$, where $x_0,x_1$ are long keys held secretly by the client, and $\theta_0,\theta_1$ are sampled randomly such that $\theta_1-\theta_0=\theta$.\par
	The advantage of doing this is it allows us to design a series of new tests that are not possible on $\ket{+_\theta}$ states. But it results in a serious cost: The client needs to reveal keys to allow the server to decode the state and get $\ket{+_{\theta}}$. But doing this directly turns out to be insecure since the revealed keys together with the phase tables allow the adversary to break the protocol. To address this problem, we develop the \emph{switch gadget technique} (see Section \ref{sec:2.3} for a detailed review of this technique). This technique, in a sense, allows the client to directly reveal the keys without sacrificing the necessary secrecy of the phase tables. We consider this part as the central step of our protocol.
	\item Then we design a series of new tests that allow the client to test the server's states. (See Section \ref{sec:2.4} for a further review.) The difficulty is we not only want these tests to verify the server's state is indeed honest, but also want the whole tests to be in linear time when applied on $L$ states. Existing works like \cite{GVRSP} design tests that work on each state separately; unluckily, as discussed in Section \ref{sec:2.7}, there are inherent barriers to get linear time protocols if we want the overall error of all these states to be within a constant.\par
	In our design of tests, there are tests that work on all these gadgets collectively, which allows us to bypass this barrier. 
\end{enumerate}
We put a series of diagrams that illustrate the differences of previous RSPV protocols and our protocols. We first show how our protocol goes in the single-gadget setting, then show the multi-gadget setting. The states below are server-side states in the honest setting, and the client knows all the information.
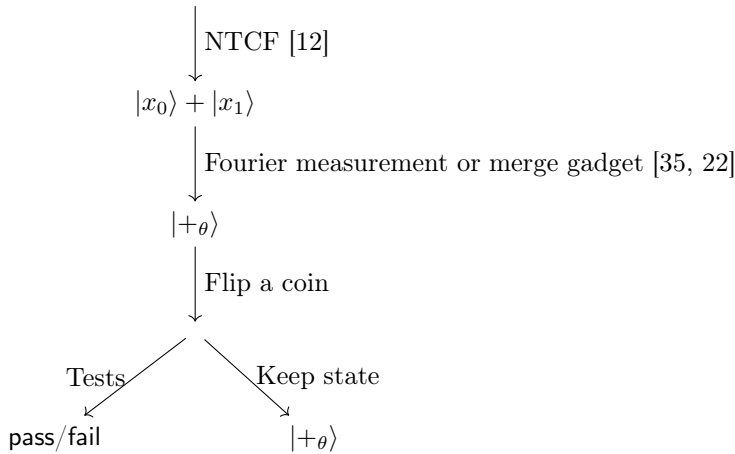
\begin{figure}
\centering
	\begin{tikzpicture}
		\node (a00) at (0,0) {};
		\node[right=of a00] (a1) {};
		\node[right=of a1] (a2) {};
		\node[right=of a2] (a3) {};
		\node[right=of a3] (a4) {};
		\node[right=of a4] (a5) {};
		\node[right=of a5] (a6) {};
		\node[left=of a00] (b1) {};
		\node[left=of b1] (b2) {};
		\node[left=of b2] (b3) {};
		\node[left=of b3] (b4) {};
			\node[left=of b4] (b5) {};
		\node[below=of a00] (a01) {$\ket{x_0}+\ket{x_1}$};
		\node[below=of a01] (a02) {$\ket{+_\theta}$};
		\node[below=of a02] (a03) {};
		\node[below right=of a03] (a04) {$\ket{+_\theta}$};
		\node[below left=of a03] (a05) {$\fpass$/$\ffail$};
		\draw[->] (a00)--(a01) node[midway,right] {NTCF \cite{BCMVV}};
		\draw[->] (a01)--(a02) node[midway,right] {Fourier measurement or merge gadget \cite{GVRSP,qfactory}};
		\draw[->] (a02)--(a03) node[midway,right] {Flip a coin};
		\draw[->] (a03)--(a04) node[midway,right] {Keep state};
		\draw[->] (a03)--(a05) node[midway,left] {Tests};
	\end{tikzpicture}
	\caption{Simplified gadget creation and testing outline in previous works\cite{GVRSP,qfactory}}
\end{figure}
\begin{figure}
\centering
	\begin{tikzpicture}
		\node (a00) at (0,0) {};
		\node[right=of a00] (a1) {};
		\node[right=of a1] (a2) {};
		\node[right=of a2] (a3) {};
		\node[right=of a3] (a4) {};
		\node[right=of a4] (a5) {};
		\node[right=of a5] (a6) {};
		\node[left=of a00] (b1) {};
		\node[left=of b1] (b2) {};
		\node[left=of b2] (b3) {};
		\node[left=of b3] (b4) {};
			\node[left=of b4] (b5) {};
		\node[below=of a00] (a01) {$\ket{x_0}+\ket{x_1}$};
		\node[below=of a01] (a02) {$e^{\theta_0\mi\pi/4}\ket{x_0}+e^{\theta_1\mi\pi/4}\ket{x_1}$};
		\node[below=of a02] (a03) {};
		\node[below right=of a03] (a04) {$\ket{+_\theta}$};
		\node[below left=of a03] (a05) {$\fpass$/$\ffail$};
		\draw[->] (a00)--(a01) node[midway,right] {NTCF \cite{BCMVV}};
		\draw[->] (a01)--(a02) node[midway,right] {Phase table \cite{revgt}};
		 \draw[->] (a02)--(a03) node[midway,right] {Flip a coin};
		\draw[->] (a03)--(a04) node[midway,right] {\begin{tabular}{l}Reveal keys to allow the server to decode\\Security provided by switch gadget technique\end{tabular}};
		\draw[->] (a03)--(a05) node[midway,left] {Different tests};
	\end{tikzpicture}
	\caption{Simplified gadget creation and testing outline in our protocol, for a single gadget}
\end{figure}
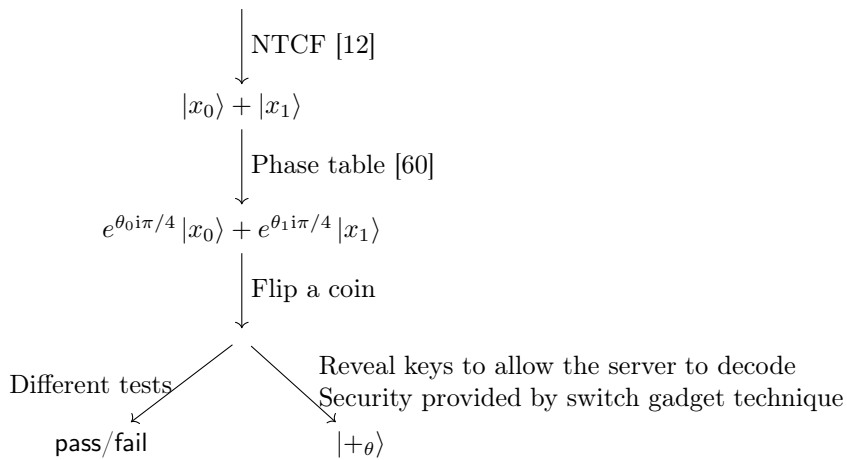
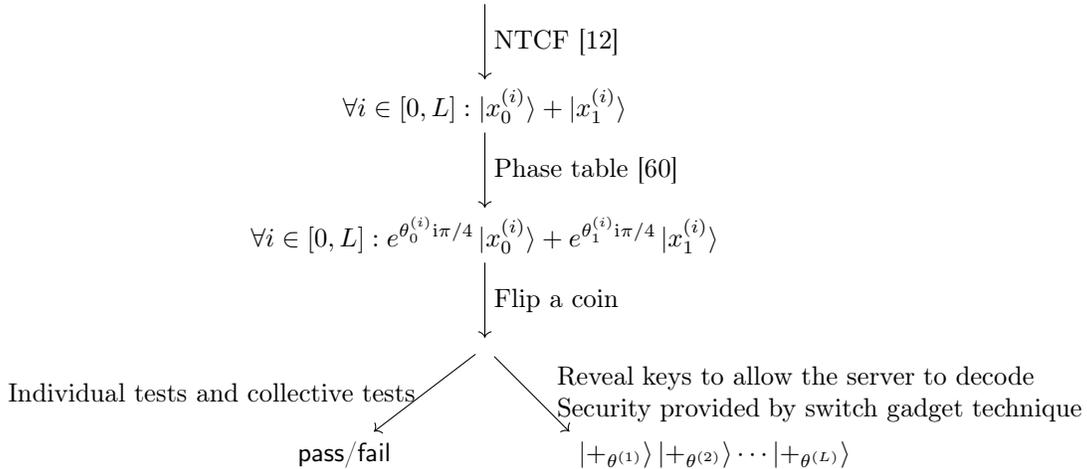
\begin{figure}
\centering
	\begin{tikzpicture}
		\node (a00) at (0,0) {};
		\node[right=of a00] (a1) {};
		\node[right=of a1] (a2) {};
		\node[right=of a2] (a3) {};
		\node[right=of a3] (a4) {};
		\node[right=of a4] (a5) {};
		\node[right=of a5] (a6) {};
		\node[left=of a00] (b1) {};
		\node[left=of b1] (b2) {};
		\node[left=of b2] (b3) {};
		\node[left=of b3] (b4) {};
			\node[left=of b4] (b5) {};
		\node[below=of a00] (a01) {$\forall i\in [0,L]:\ket{x_0^{(i)}}+\ket{x_1^{(i)}}$};
		\node[below=of a01] (a02) {$\forall i\in [0,L]:e^{\theta_0^{(i)}\mi\pi/4}\ket{x_0^{(i)}}+e^{\theta_1^{(i)}\mi\pi/4}\ket{x_1^{(i)}}$};
		\node[below=of a02] (a03) {};
		\node[below right=of a03] (a04) {   };
		\node[below right=of a03] (a04b) {$\ket{+_{\theta^{(1)}}}\ket{+_{\theta^{(2)}}}\cdots \ket{+_{\theta^{(L)}}}$};
		\node[below left=of a03] (a05) {$\fpass$/$\ffail$};
		\draw[->] (a00)--(a01) node[midway,right] {NTCF \cite{BCMVV}};
		\draw[->] (a01)--(a02) node[midway,right] {Phase table \cite{revgt}};
		 \draw[->] (a02)--(a03) node[midway,right] {Flip a coin};
		\draw[->] (a03)--(a04) node[midway,right] {\begin{tabular}{l}Reveal keys to allow the server to decode\\Security provided by switch gadget technique\end{tabular}};
		\draw[->] (a03)--(a05) node[midway,left] {Individual tests and collective tests};
	\end{tikzpicture}
	\caption{Simplified gadget creation and testing outline in our protocol, for multiple gadgets}
\end{figure}
\subsection{Discussion}
We first remark several limitations of our results that we do not aim to solve in this work.
\begin{itemize}
	\item Although our remote state preparation protocol runs in constant rounds, we do not aim to construct a constant round CVQC protocol. The underlying gadget-assisted verification protocol runs in linear rounds which implies the overall round complexity of our CVQC protocol is also linear.
	\item We do not focus on optimizing the hidden constant in the big-O notation. The current constant blow-up is far from being practical, but it's largely from the very loose security analysis (even elementary calculations in our work could be far from being tight). Despite being loose, in this work we still give explicit bounds for most of the constants to set a record for future improvements, which could be important for the studies of practicality of these protocols.
	\item The security of our RSPV protocol is defined up to constant error and possibly unbounded isometry, and does not take composability into consideration. (See Section \ref{sec:4.3} for details.) This is sufficient for constructing our CVQC protocol but will restrict the application scenarios of our RSPV.
	\item The circuit description and complexity are in the MBQC model \cite{MBQC}, since one existing work that our work relies on \cite{FKD} is in the MBQC model. In general there is an additional blowup of circuit width when a circuit in the usual circuit model is transformed into MBQC model, due to the fact that in the usual circuit model gates can be applied onto arbitrary wires while in the MBQC model only interactions between neighbors are natively supported and long-range gates come with a cost. It's debatable which model is more suitable for modeling quantum computations; nevertheless it's better to have protocols for both models and the usual circuit model version remains to be resolved.\footnote{We thank anonymous reviewers and Simons reunion program attendees for pointing this out.}
\end{itemize}
It will be desirable to see the resolutions of these problems.\par
Besides the problems above, our results also naturally give rise to the following questions.
\begin{itemize}
	\item Could we use either CVQC or RSPV in our work to construct new faster protocols for other problems?\par
	As discussed in Section \ref{sec:1.2.3}, \cite{MahadevVerification,BCMVV} leads to a series of protocols for various different problems. It is promising to explore the possibility of using our protocols or techniques to construct fast protocols for these problems.
	\item Could we replace the random oracle by standard model assumptions, and prove its security formally?\par
	We feel our protocols have a relatively clean and clear structure, and the usage of the random oracle is not involved. The applications of the random oracle in our protocol construction are as follows:
	\begin{itemize}
	\item We use the random oracle to construct the underlying symmetric key encryption scheme used in lookup-tables (see Section \ref{sec:2.2} for a discussion).
	\item We use the random oracle in RO-padded Hadamard tests \cite{jiayu20}. (We will mention it in Section \ref{sec:2.3}.)
	\end{itemize}
	It is an intriguing question to instantiate the random oracle (or reducing its usage) with a formal proof.
\end{itemize}
We believe our work, together with answers to these questions, will be important for theoretical and practical development of secure quantum computations.
\subsection{Paper Organizations}
This paper is organized as follows.
\begin{enumerate}
\item Section \ref{sec:1} is the introduction of the background and our results. Then in Section \ref{sec:2} we give a technical overview of our construction.
\item In Section \ref{sec:3} we give a review to the preliminaries. In Section \ref{sec:4r} we formalize the notion of CVQC and RSPV, and introduce the notion of pre-RSPV as an intermediate step.
\item In Section \ref{sec:4} we formalize our pre-RSPV protocol.
\item In Section \ref{sec:5}  to Section \ref{sec:11} we analyze each idea or subprotocol in our construction.
\item In Section \ref{sec:12} we combine all these stand-alone analysis of subprotocols together to prove the verifiability of the overall pre-RSPV protocol.
\item In Section \ref{sec:13} we use our pre-RSPV protocol to complete the construction of our RSPV and CVQC protocol.
\end{enumerate}

\section*{Acknowledgement}
We thank Thomas Vidick for helpful discussions. We also thank Alexander Poremba, Dominik Leichtle, Simons Quantum Reunion attendees, and annoymous reviewers for discussions.
%
\section{Technical Overview}\label{sec:2}
Let's give an overview of the construction of our protocols.
\subsection{Fast Parallel RSPV for 8-basis Qfactory and Gadget-assisted Quantum Computation Verification}\label{sec:2.1}
At a high level, to construct the fast CVQC protocol, we take the approach of constructing RSPV protocol as an intermediate step.\par
As informally defined in Section \ref{sec:1.2.2}, we consider the parallel version of RSPV for state family $\{\ket{+_\theta}=\frac{1}{\sqrt{2}}(\ket{0}+e^{\mi\theta\pi/4}\ket{1}),\theta\in \{0,1\cdots 7\}\}$. The inputs of the RSPV protocol are the gadget number $1^L$ and security parameter $1^\kappa$. Equivalently we could consider this protocol as an RSPV protocol for a single uniformly random state from a large state family
  \begin{equation}\label{eq:rr1}\{\ket{+_{\theta^{(1)}}}\otimes \ket{+_{\theta^{(2)}}}\otimes \cdots \otimes \ket{+_{\theta^{(L)}}}:\forall i,\theta^{(i)}\in \{0,1\cdots 7\}\}\end{equation}
  In the end the server should get a random element from \eqref{eq:rr1} and the client should get $(\theta^{(i)})_{i\in [L]}$. Equivalently we could express the joint cq-state of the client and server in this ideal functionality as\footnote{Below we express the cq-state by a mixture of density operators and pure states. This notation is not fully standard but is convenient and is indeed used in some places; we do not rely on any operational property of it and use it solely as a notation.}
  \begin{equation}\label{eq:cqtarget0}\sum_{\forall i, \theta^{(i)}\in \{0,1\cdots 7\}}\frac{1}{8^L}\underbrace{\ket{\theta^{(1)}}\bra{\theta^{(1)}}\ket{\theta^{(2)}}\bra{\theta^{(2)}}\cdots \ket{\theta^{(L)}}\bra{\theta^{(L)}}}_{\text{client}}\otimes\underbrace{\ket{+_{\theta^{(1)}}}\otimes \ket{+_{\theta^{(2)}}}\otimes \cdots \otimes \ket{+_{\theta^{(L)}}}}_{\text{server}} \end{equation}
 The road to prove Theorem \ref{thm:main} is as follows.
  \begin{enumerate}
  	\item Construct an RSPV protocol for target state  \eqref{eq:cqtarget0} that runs in time $O(\fpoly(\kappa)L)$.\par
  	This is achieved by Protocol \ref{prtl:rspvr} assuming NTCF and QROM and it also has the desirable property that it only has constant rounds. This proves Theorem \ref{thm:infrspv}.
  	\item Given a circuit $C$ to be verified, find a  gadget-assisted quantum computation verification protocol that uses \eqref{eq:rr1} as the initial gadgets, where the  gadget number $L$ needed is linear in the circuit size $|C|$. This is achieved by existing work \cite{FKD}.
  \end{enumerate}
  \paragraph{Remark} Classical-channel RSPV defined above for arbitrary state families are generally impossible due to the existence of complex-conjugate attack, discussed in \cite{RUV,leash,AGthesis}: the malicious server could choose to execute the complex conjugate of the honest behaviors, and the output state will be the complex conjugate of the target state, which is not isometric to the honest state in general. The client has no way to detect it over a classical channel. However, the state \eqref{eq:cqtarget0} that we aim at  is indeed invariant under complex conjugate: a complex conjugate of \eqref{eq:cqtarget0} is isometric to \eqref{eq:cqtarget0} up to a global phase by a sequence of $\fX$ flips:
  $$\fX\ket{+_{-\theta}}=e^{-\theta\mi\pi/4}\ket{+_\theta}$$
  $$\Rightarrow \fX^{\otimes L}(\ket{+_{-\theta^{(1)}}}\otimes \ket{+_{-\theta^{(2)}}}\otimes \cdots \otimes \ket{+_{-\theta^{(L)}}})=e^{-(\theta^{(1)}+\theta^{(2)}+\cdots \theta^{(L)})\mi\pi/4}\ket{+_{\theta^{(1)}}}\otimes \ket{+_{\theta^{(2)}}}\otimes \cdots \otimes \ket{+_{\theta^{(L)}}}$$
  which means a malicious server executing the complex-conjugate attack is equivalent to the honest server up to a server-side isometry and an undetectable global phase.
\subsection{Lookup-table-based Techniques for State Generation in the Honest Setting}\label{sec:2.2}
So how could we generate such states? Let's use the preparation of one gadget as an example. Suppose we need to prepare the state \begin{equation}\label{eq:rr2}\frac{1}{\sqrt{2}}(\ket{0}+e^{\mi\theta\pi/4}\ket{1})\end{equation} The first idea is to first use  the noisy trapdoor claw-free function (NTCF) \cite{BCMVV} technique to prepare the \emph{key-pair-superposition state}, then use \emph{phase tables} \cite{revgt} to add the phases. In more detail:
\begin{enumerate}
\item As shown in \cite{BCMVV} (see Section \ref{sec:ntcf} for a review), evaluating an NTCF function could result in a state of the following form: the client gets a key pair $K=(x_0,x_1)$, $x_0,x_1\in \{0,1\}^\kappa,x_0\neq x_1$; the server holds the state \begin{equation}\label{eq:1}\frac{1}{\sqrt{2}}(\ket{x_0}+\ket{x_1})\end{equation}
In this paper later we will frequently use the word ``keys'' to denote the $K=(x_0,x_1)$ coming from the NTCF evaluation. (In the construction of NTCF there are also ``secret keys'' and ``public keys'', which are different.)
\item We will use the lookup-table-based techniques \cite{revgt} for adding phases. These lookup-tables have a similar structure to garbled tables \cite{YaoGCOrigin} but do not carry computations directly.\par
For constructing look-up tables, we need an underlying symmetric key encryption scheme $\fEn$ with a key authentication part.\par
 Recall that we are working in the quantum random oracle model (QROM). We use $H$ to denote the random oracle.  A natural construction of $\fEn$ on encryption key $k$ and plaintext $p$ is
\begin{equation}\label{eq:4rps}\fEn_k(p):=(\underbrace{(R,H(R||k)+ p)}_{\text{ciphertext}}, \underbrace{(R^\prime,H(R^\prime||k))}_{\text{key authentication}}); R,R^\prime\leftarrow_r \{0,1\}^\kappa,\end{equation}	

where the addition is over some specific group. Then we define the lookup table $\fLT(x_1\rightarrow r_1,x_2\rightarrow r_2,\cdots x_D\rightarrow r_D)$, or simply $$(x_1\rightarrow r_1,x_2\rightarrow r_2,\cdots x_D\rightarrow r_D),$$as the tuple \begin{equation}\label{eq:r4}(\fEn_{x_1}(r_1),\fEn_{x_2}(r_2),\cdots, \fEn_{x_D}(r_D)).\end{equation}
Each $\fEn_{x_u}(r_u)$ in \eqref{eq:r4} is called a row of this table. Given the table, and one key $x_u$ used in some row, the server could decrypt the corresponding $r_u$ as follows: it first uses the key authentication part of $\fEn$ to find out the index $u$, then it decrypts the ciphertext part of $\fEn_{x_u}(r_u)$ and gets $r_u$.\par

 This type of lookup table technique is very useful in manipulating states in the form of \eqref{eq:1}. One nice property is the table decoding process above could be applied with superpositions of keys, for example, \eqref{eq:1}. Below we show how to use the \emph{phase table} technique \cite{revgt,jiayu20} for adding phases to \eqref{eq:1}.\par
  If the client wants to add a phase of $e^{\theta\mi\pi/4}$ to the $x_1$ basis of $\ket{x_0}+\ket{x_1}$, it can prepare the lookup table that encodes the following classical mapping:
 \begin{equation}\label{eq:r2}(x_0\rightarrow \theta_0, x_1\rightarrow \theta_1),\text{ where $\theta_0,\theta_1\in \{0,1\cdots 7\}$ are sampled randomly such that }\theta_1-\theta_0=\theta,\end{equation}
where we use $\bZ_8$ as the group of addition. Then the adding of phase is achieved with the following mappings honestly:
 \begin{align}
 	&\ket{x_0}+\ket{x_1}\label{eq:r6}\\
 	\text{(Decrypt table) }\rightarrow &\ket{x_0}\ket{\theta_0}+\ket{x_1}\ket{\theta_1}\label{eq:9}\\
 	\text{(Controlled phase) }\rightarrow &e^{\mi\theta_0\pi/4}\ket{x_0}\ket{\theta_0}+e^{\mi\theta_1\pi/4}\ket{x_1}\ket{\theta_1}\\
 	\text{(Decrypt table again) }\rightarrow &e^{\mi\theta_0\pi/4}\ket{x_0}+e^{\mi\theta_1\pi/4}\ket{x_1}\label{eq:r7}
 \end{align}
 where in the last step the decryption outcome is written into the same register that is introduced in \eqref{eq:9}, and the values in this register will be erased, as discussed in \cite{revgt}.
 \item If the server really holds the state in the form of \eqref{eq:r7}, the client could reveal the key pair $K$ and the server could decode the keys from \eqref{eq:r7} and get \eqref{eq:rr2}.
 \end{enumerate}
 This completes the construction in the honest setting. And we can naturally generalize it to prepare $L$ states:
 \begin{enumerate}
 	\item Both parties execute $L$ evaluations of NTCF functions and prepare the following states on the server side:
 	$$\frac{1}{\sqrt{2^L}}(\ket{x_0^{(1)}}+\ket{x_1^{(1)}})\otimes(\ket{x_0^{(2)}}+\ket{x_1^{(2)}})\otimes\cdots \otimes (\ket{x_0^{(L)}}+\ket{x_1^{(L)}})$$
 	while the client knows all the keys.
 	\item The client samples random phase pairs $\Theta=(\Theta^{(i)})_{i\in [L]}$, $\Theta^{(i)}=(\theta^{(i)}_0,\theta^{(i)}_1),\theta_0^{(i)},\theta_1^{(i)}\in_r \{0,1\cdots 7\}^2$. ($\in_r$ means uniformly random sampling.) Then it prepares the following table for each $i\in [L]$:
 	$$(x^{(i)}_0\rightarrow \theta^{(i)}_0,x^{(i)}_1\rightarrow \theta^{(i)}_1)$$
 	and sends all of them to the server.\par
 	The server is able to evaluate the phase tables and get the following state:
 	\begin{equation}\label{eq:gadget}\small\frac{1}{\sqrt{2^L}}(e^{\theta_0^{(1)}\mi\pi/4}\ket{x_0^{(1)}}+e^{\theta_1^{(1)}\mi\pi/4}\ket{x_1^{(1)}})\otimes(e^{\theta_0^{(2)}\mi\pi/4}\ket{x_0^{(2)}}+e^{\theta_1^{(2)}\mi\pi/4}\ket{x_1^{(2)}})\otimes\cdots \otimes (e^{\theta_0^{(L)}\mi\pi/4}\ket{x_0^{(L)}}+e^{\theta_1^{(L)}\mi\pi/4}\ket{x_1^{(L)}})\end{equation}
 	The client calculates $\theta^{(i)}=\theta_1^{(i)}-\theta_0^{(i)}$, $\forall i\in [L]$.
 	\item The client reveals all the keys and the server gets \eqref{eq:rr1} up to a global phase.
 \end{enumerate}
 
Have we got an RSPV protocol for \eqref{eq:cqtarget0}? The protocol above guarantees the honest behavior, but does not provide any security or verifiability. So where does this protocol violate the verifiability property of RSPV?\par
Recall the informal definition of RSPV in Definition \ref{defn:infrspv}. The verifiability property of RSPV requires the server to only hold the target state, it should not be either too little or too much---where ``too little'' means the server's state does not contain a subsystem that is isometric to the target state, and ``too much'' means the server gets additional information about the state descriptions that could not be simulated from the ideal state on the server side.\par
 The problem here is the lookup table will contain ciphertexts that encode phases $\theta^{(i)}$, which could be decrypted after the client reveals $K$ in the final step---which means the server knows too much.\par
  To address this problem we introduce the \emph{switch gadget technique}, which is one of the key ingredients of this work.
\subsection{Switch Gadget Technique, and Phase Update under This Technique ($\fAddPhaseWithswitch$)}\label{sec:2.3}
One important technique that we will use is the \emph{switch gadget technique}. This technique was also used in \cite{jiayu20} (under the name of \emph{helper gadget}), with an early stage  version of analysis techniques; in general this technique does not solve concrete problems on its own, and we need to make smart usage of it and combine it with other techniques.\par
As discussed in the last section, the client needs a way to introduce new phases in a completely secret way. The idea is, instead of using a simple phase table, the protocol will make use of one additional \emph{switch gadget}: 
$$\text{server holds $\frac{1}{\sqrt{2}}(\ket{x_0^{(\switch)}}+\ket{x_1^{(\switch)}})$, client holds }K^{(\switch)}=(x_0^{(\switch)},x_1^{(\switch)}).$$
 Importantly, when the mapping \eqref{eq:r2} is encoded, $K^{(\switch)}$ will also be used as a part of the encryption keys. (We say ``encryption keys'' to mean $k$ in $\fEn_k(p)$ appeared in the construction of tables \eqref{eq:r4}.) In more detail, each encryption in the table has the following form, where the encryption key is the concatenation of two keys:\begin{equation}\label{eq:9r}(\forall b^{(\switch)}\in \{0,1\},b\in \{0,1\}, x^{(\switch)}_{b^{(\switch)}}||x_b\rightarrow \theta_b)\end{equation}
The table \eqref{eq:9r} contains four rows coming from different values of $b^{(\switch)}$ and $b$. Different choices of the switch gadget key correspond to the same encrypted phases.\par
With this table, the same mapping could still be implemented. The server could use similar operations as \eqref{eq:r6} to \eqref{eq:r7}; it could decrypt the table \eqref{eq:9r} with the key superposition it holds, and the switch gadget remains in a product form from the other parts in each step analogous to \eqref{eq:r6} to \eqref{eq:r7}:
\begin{align}
&(\ket{x_0^{(\switch)}}+\ket{x_1^{(\switch)}})\otimes (\ket{x_0}+\ket{x_1})\\
\rightarrow & (\ket{x_0^{(\switch)}}+\ket{x_1^{(\switch)}})\otimes (\ket{x_0}\ket{\theta_0}+\ket{x_1}\ket{\theta_1})\\
\rightarrow & (\ket{x_0^{(\switch)}}+\ket{x_1^{(\switch)}})\otimes (e^{\theta_0\mi\pi/4}\ket{x_0}\ket{\theta_0}+e^{\theta_1\mi\pi/4}\ket{x_1}\ket{\theta_1})\\
\rightarrow & (\ket{x_0^{(\switch)}}+\ket{x_1^{(\switch)}})\otimes (e^{\theta_0\mi\pi/4}\ket{x_0}+e^{\theta_1\mi\pi/4}\ket{x_1})\label{eq:14r}
\end{align}

So why do we need the switch gadget? We will see, starting from \eqref{eq:14r}, the switch gadget will go through an \emph{(RO-padded) Hadamard test} \cite{jiayu20}, which we define below.  The (unpadded) Hadamard test is raised in \cite{BCMVV} and in \cite{jiayu20} it is observed that a random-oracle-padded version of this test certifies the server has to forget these keys; this observation will be formalized in our work in a nicer way\footnote{The fact that quantum techniques allow us to do deletion or revocation is not new in this paper. For example, a series of papers study security of message encryption against key leakage \cite{onewaytohiding,BI,HMNY}, which is now called \emph{certified deletion} (or \emph{proof of deletion}). What's different in our techniques is, the switch gadget will first \emph{allow} the server to decrypt, and revoke the encoded mapping after that. For comparison, in certified deletion the plaintext is protected at all time. This different goal leads to very different constructions from works in certified deletion.}.  For simplicity of the introduction we review the unpadded version, given in \cite{BCMVV}:
\begin{toyprtl}[Hadamard test]\label{toyprtl:1}
Suppose the client holds $K=(x_0,x_1)$, $x_0,x_1\in \{0,1\}^\kappa$, and the server holds $\frac{1}{\sqrt{2}}(\ket{x_0}+\ket{x_1})$.\par
The client asks for a non-zero $d$ such that $d\cdot x_0=d\cdot x_1\mod 2$. The server does a bit-wise Hadamard measurement on $\frac{1}{\sqrt{2}}(\ket{x_0}+\ket{x_1})$ and the output will satisfy the client's testing equation, which comes from the identity $\fH^{\otimes \kappa}(\frac{1}{\sqrt{2}}(\ket{x_0}+\ket{x_1}))=\frac{1}{\sqrt{2^{\kappa-1}}}\sum_{d:d\cdot x_0=d\cdot x_1}\ket{d}$.
\end{toyprtl}
As discussed above, the (RO-padded) Hadamard test satisfies the following informal property: 
\begin{claim}[Successful Hadamard test destroys keys]Suppose the server holds an initial state that satisfies some property (say, the honest initial state $\frac{1}{\sqrt{2}}(\ket{x_0}+\ket{x_1})$). Both parties execute an (RO-padded) Hadamard test. If the server passes the protocol with high probability, the probability that it could predict one of $x_0$ or $x_1$ from the post-test state is small.
	
\end{claim}
Since the keys of the switch gadget are part of the encryption keys for the mapping \eqref{eq:9r}, if the server loses the predictability of keys in $K^{(\switch)}$, intuitively the server loses the ability to make use of the mapping encoded by the table.\par
In more detail, the switch gadget technique works like a switch: before the test the server is able to evaluate the mapping, while after the test the information in the table will be hidden from efficient malicious servers.\par
With this in mind, our RSPV protocol roughly has the following structure from the viewpoint of the switch gadget technique:
\begin{toyprtl}\label{toyprtl:2r}
\begin{enumerate}
\item For each $i\in [L]$, the client sends the extended phase table for the $i$-th gadget
\begin{equation}(\forall b^{(\switch)}\in \{0,1\},b^{(i)}\in \{0,1\}, x^{(\switch)}_{b^{(\switch)}}||x_{b^{(i)}}^{(i)}\rightarrow \theta_b^{(i)})\end{equation}
\item ``Turn off'' the switch of the switch gadget technique (that is, to execute an RO-padded Hadamard test on the switch gadget).
\item The client flips a coin and does one of the following:
\begin{itemize}
\item Verify that the state is really in the form \eqref{eq:gadget}. (This step might be destructive.)
\item Reveal all the keys and allow the server to decrypt and get the states \eqref{eq:rr1}.
\end{itemize}
\end{enumerate}
\end{toyprtl}
Generally speaking, we call this type of protocol design technique as the \emph{switch gadget technique}. The switch gadget technique gives us a protocol of the following structure:
\begin{enumerate}
\item Encode a mapping on a switch gadget. The switch gadget keys will be an encryption key for the mapping, and the honest server could use either branch of the switch gadget (and thus their superpositions) to evaluate the mapping.	
\item Both parties do a Hadamard test on the switch gadget.\par
\end{enumerate}

Besides the switch gadget technique, the next non-trivial step in Toy Protocol \ref{toyprtl:2r} is the design of sub-tests, the first bullet of the third step. Below we give an overview of the important techniques that we develop for it.
\subsection{Overall Protocol Structure with Full Verification Procedures}\label{sec:2.4}
We will design various types of tests for verifying the states. These tests include \emph{standard basis test} ($\fStdBTest$), \emph{individual phase test} ($\fInPhTest$), \emph{collective phase test} ($\fCoPhTest$) and \emph{basis uniformity test} ($\fBNTest$).\par
To support these tests, in the very beginning when both parties use $\fNTCF$ to generate key-pair superpositions, they will generate $2+L$ gadgets and keys. These keys are denoted by $K^{(\switch)}$, $K^{(0)}$, $\cdots$, $K^{(L)}$. The $K^{(\switch)}$ corresponds to the switch gadget in the $\fAddPhaseWithswitch$ step; and for the remaining keys we note that there is one additional key pair and gadget with index $^{(0)}$ which is solely used for verification and will not appear in the output states \eqref{eq:cqtarget0}.\par
Overall speaking, our protocol goes as follows. Note that the standard basis test is probabilistically executed both before and after the $\fAddPhaseWithswitch$ step, while the other tests are executed after that.
\begin{toyprtl}\label{toyprtl:3}
\begin{enumerate}
	\item Both parties use $\fNTCF$ to create $2+L$ key-pair superpositions; the client gets key pairs $K^{(\switch)}=(x_0^{(\switch)},x_1^{(\switch)})$, $K^{(i)}=(x_0^{(i)},x_1^{(i)})$ for all $i\in [0,L]$ and the server gets
	 	\begin{equation}\label{eq:18mq}\frac{1}{\sqrt{2^{2+L}}}(\ket{x_0^{(\switch)}}+\ket{x_1^{(\switch)}})\otimes(\ket{x_0^{(0)}}+\ket{x_1^{(0)}})\otimes(\ket{x_0^{(1)}}+\ket{x_1^{(1)}})\otimes(\ket{x_0^{(2)}}+\ket{x_1^{(2)}})\otimes\cdots \otimes (\ket{x_0^{(L)}}+\ket{x_1^{(L)}})\end{equation}
	\item The client chooses one of the following two branches randomly:
	\begin{itemize}
		\item Execute the $\fStdBTest$;
		\item Execute $\fAddPhaseWithswitch$. This will consume the switch gadget and allow the honest server to prepare gadgets with phases in the form of
		 	\begin{equation}\label{eq:20}\small\frac{1}{\sqrt{2^{1+L}}}(e^{\theta_0^{(0)}\mi\pi/4}\ket{x_0^{(0)}}+e^{\theta_1^{(0)}\mi\pi/4}\ket{x_1^{(0)}})\otimes(e^{\theta_0^{(1)}\mi\pi/4}\ket{x_0^{(1)}}+e^{\theta_1^{(1)}\mi\pi/4}\ket{x_1^{(1)}})\otimes\cdots \otimes (e^{\theta_0^{(L)}\mi\pi/4}\ket{x_0^{(L)}}+e^{\theta_1^{(L)}\mi\pi/4}\ket{x_1^{(L)}})\end{equation}
		Then the client chooses one of the following five branches uniformly randomly:
		\begin{itemize}
		\item Execute the $\fStdBTest$;
		\item Execute the $\fInPhTest$;
		\item Execute the $\fCoPhTest$;
		\item Execute the $\fBNTest$;
		\item Reveal the keys and allow the server to output the state.
		\end{itemize}
	\end{itemize}
\end{enumerate}
\end{toyprtl}
The gadgets that these tests applied on could be illustrated as follows.
$$\overbrace{\underbrace{K^{(\switch)}}_{\text{switch gadget for $\fAddPhaseWithswitch$}},\underbrace{\underbrace{K^{(0)}}_{\fInPhTest},\underbrace{K^{(1)},K^{(2)},\cdots K^{(L)}}_{\fBNTest,\text{ output states}}}_{\text{2nd $\fStdBTest$, }\fCoPhTest}}^{\text{1st $\fStdBTest$}}$$
This overall protocol could only prepare the target state \eqref{eq:cqtarget} when the last case (``allow the server to output the state'') in Toy Protocol \ref{toyprtl:3} is reached; in addition to that, an honest server could not always win in each test of Toy Protocol \ref{toyprtl:3} due to an issue that will be discussed when we introduce the $\fInPhTest$. But a suitable repetition-based amplification of Toy Protocol \ref{toyprtl:3} (or formally, Protocol \ref{prtl:prerspv}) will lead to the formal RSPV protocol (Protocol \ref{prtl:rspvr}) that we want. (In Section \ref{sec:4} we will give the notion of pre-RSPV which captures these construction details.)
\begin{conv}
When we work on multiple key pairs, we use $K$ to denote $(K^{(i)})_{i\in [0,L]}$ and use $\tilde K$ to denote $(K^{(i)})_{i\in [L]}$. ($[L]=\{1,\cdots L\}$ and $[0,L]=\{0,1\cdots L\}$.)
\end{conv}

\subsection{Standard Basis Test ($\fStdBTest$)}\label{sec:2.5}
As said before, an intermediate target state of the protocol is of the form of \eqref{eq:gadget}. Expanding all of these states in the standard basis, each component in the expansion will have the form of key-vectors $x^{(1)}_{b^{(1)}}x^{(2)}_{b^{(2)}}\cdots x^{(L)}_{b^{(L)}}$, for some $b^{(1)}\cdots b^{(L)}\in \{0,1\}^L$. (We omit the symbols that separate these keys; we never multiply keys in this work.) This inspires us to define a test that verifies the state is of the following form up to a server-side isometry:
\begin{defn}[Basis-honest form]
Suppose the client holds a tuple of key pairs $(K^{(i)})_{i\in [L]}$ where each $K^{(i)}=(x_0^{(i)},x_0^{(i)})$, define the basis-honest form to be the form of state (where we omit the concatenation notation for simplicity):
\begin{equation}\label{eq:19ox}\sum_{b^{(1)}b^{(2)}\cdots b^{(L)}: \forall i\in [L],b^{(i)}\in \{0,1\}}\underbrace{\ket{x^{(1)}_{b^{(1)}}x^{(2)}_{b^{(2)}}\cdots x^{(L)}_{b^{(L)}}}}_{\text{some server side registers}}\underbrace{\ket{\varphi_{b^{(1)}b^{(2)}\cdots b^{(L)}}}}_{\text{other part}}\end{equation}
We call $\ket{x^{(1)}_{b^{(1)}}x^{(2)}_{b^{(2)}}\cdots x^{(L)}_{b^{(L)}}}\ket{\varphi_{b^{(1)}b^{(2)}\cdots b^{(L)}}}$ as the $\vec{x}_{\vec{b}}$-branch of this basis-honest state.\par
It could be naturally generalized to the $2+L$ key pairs appeared in Toy Protocol \ref{toyprtl:3}.
\end{defn}
Let's consider the tests on \eqref{eq:18mq}. The testing of basis-honest form can be achieved by the client simply asking the server to make a standard basis measurement to provide a key vector classically:
\begin{toyprtl}[Standard basis test ($\fStdBTest$)]
\begin{enumerate}
	\item The client asks the server to provide a key vector in the form of $x^{(\switch)}_{b^{(\switch)}}x^{(0)}_{b^{(0)}}x^{(1)}_{b^{(1)}}x^{(2)}_{b^{(2)}}\cdots x^{(L)}_{b^{(L)}}$, $b^{(\switch)}\in \{0,1\},b^{(i)}\in \{0,1\}$ for all $i\in [0,L]$.\par
	The honest server could pass the protocol by measuring \eqref{eq:18mq} in the standard basis.
\end{enumerate}	
\end{toyprtl}
The description above is the 1st standard basis test shown in Toy Protocol \ref{toyprtl:3}; the 2nd standard basis test is similarly defined on the remaining keys and gadgets.\par
Now we continue to discuss the other tests that could possibly be applied on the state \eqref{eq:20}. Importantly, we need a way to test whether the phases are really introduced by the server.

\subsection{Individual Phase Test ($\fInPhTest$)}\label{sec:2.6}
Let's start with the single state case and see how we could design a protocol that verifies a single state. Suppose:
\begin{itemize}
\item The client holds key pair $K=(x_0,x_1)$ and phase pair $\Theta=(\theta_0,\theta_1)$. $\theta_0,\theta_1\in \{0,1\cdots 7\}$.
\item In the honest setting the server should hold the state 
\begin{equation}\label{eq:7}e^{\theta_0\mi\pi/4}\ket{x_0}+e^{\theta_1\mi\pi/4}\ket{x_1}.\end{equation}
\end{itemize}
 The server wants to cheat. Let's first consider a restricted form of attack for simplicity of the introduction. Suppose the server's attack is just to add some different phases to the gadget. Instead of holding \eqref{eq:7}, it might hold
\begin{equation}\label{eq:8}e^{f(\theta_0)\mi\pi/4}\ket{x_0}+e^{g(\theta_1)\mi\pi/4}\ket{x_1}\end{equation}
for some arbitrary functions $f$, $g$. We want to design a test that verifies \eqref{eq:8} is isometric to \eqref{eq:7} (under an isometry that does not depend on the phases). This means, we want to design a test such that the server could pass the test from \eqref{eq:8} if and only if $f$, $g$ have the form:
\begin{equation}\label{eq:13n}f(\theta_0)\approx\theta_0+c_0,g(\theta_1)\approx\theta_1+c_1;c_0,c_1\text{ are constants.}\end{equation}
 Certainly \eqref{eq:8} does not capture all the possible attacks that the adversary can make; but it captures a non-trivial class of attacks that will be helpful for illustrating our ideas.\par
\paragraph{Remark} To make this type of simplification make sense, the switch gadget technique plays an important role here. After the switch gadget is measured and destroyed, the adversary could not decrypt the phase table any more and it is not able to do any $\theta$-related operation on state \eqref{eq:8}. Without the switch gadget technique, the adversary can change the phases in \eqref{eq:8} to $f^\prime(\theta_0),g^\prime(\theta_1)$ arbitrarily.\par
We first note that (without loss of generality) there is a simple way to verify the relation of the following two states, which correspond to the cases where the client side phase pair is $(\tilde\theta_0,\tilde\theta_1)$ and $(\tilde\theta_0,\tilde\theta_1+4)$:
\begin{equation}\label{eq:22pw}e^{f(\tilde\theta_0)\mi\pi/4}\ket{x_0}+e^{g(\tilde\theta_1)\mi\pi/4}\ket{x_1}\quad e^{f(\tilde\theta_0)\mi\pi/4}\ket{x_0}+e^{g(\tilde\theta_1+4)\mi\pi/4}\ket{x_1}\end{equation}
for which the honest states are
\begin{equation}\label{eq:12n}e^{\tilde\theta_0\mi\pi/4}\ket{x_0}+e^{\tilde\theta_1\mi\pi/4}\ket{x_1};\quad e^{\tilde\theta_0\mi\pi/4}\ket{x_0}-e^{\tilde\theta_1\mi\pi/4}\ket{x_1}(=e^{\tilde\theta_0\mi\pi/4}\ket{x_0}+e^{(\tilde\theta_1+4)\mi\pi/4}\ket{x_1})\end{equation}
Note that these two states in \eqref{eq:12n} are orthogonal.\par
Consider the following protocol, which aims at verifying the relation between states in \eqref{eq:22pw}:
 \begin{toyprtl}
 Suppose the client holds phase pair $(\tilde\theta_0,\tilde\theta_1)$ or $(\tilde\theta_0,\tilde\theta_1+4)$ in its $\Theta$ register with equal probability. The honest server holds \eqref{eq:12n} while the malicious server holds \eqref{eq:22pw}.
\begin{enumerate}
\item The client could simply reveal $\tilde\theta_1-\tilde\theta_0$ and the honest server could remove the phases from \eqref{eq:12n} and get $\ket{x_0}+\ket{x_1}$ or $\ket{x_0}-\ket{x_1}$ correspondingly; 
\item Then both parties do a Hadamard test. Suppose the server's response is $d$, the client will calculate $d\cdot (x_0+x_1)\mod 2$, whose result is deterministically $0$ for $\ket{x_0}+\ket{x_1}$ and deterministically $1$ for $\ket{x_0}-\ket{x_1}$. The client rejects if $d=0$ or  $d\cdot (x_0+x_1)\mod 2$ does not have the correct value.
\end{enumerate}
\end{toyprtl}
This test could be translated to a test on \eqref{eq:8} by a change of variables: the client will randomly choose $\tilde\theta_1=\theta_1$ or $\tilde\theta_1=\theta_1-4$. For describing this test (and tests later) we introduce the notion of $\delta$-bias Hadamard test as follows\footnote{As before, the formal version of this test will also have a random oracle padding; here we omit this part.}:
 \begin{toyprtl}[Hadamard test with extra bias]\label{toyprtl:hteb} The Hadamard test with $\delta$-extra-bias is defined as follows.\par
 Suppose the client holds key pair $(x_0,x_1)$ and phase pair $(\theta_0,\theta_1)$. Honestly the server should hold $e^{\theta_0\mi\pi/4}\ket{x_0}+e^{\theta_1\mi\pi/4}\ket{x_1}$. We call $\theta_1-\theta_0$ the \emph{relative phase} and $\delta$ the \emph{extra phase bias}.
 \begin{enumerate}\item The client reveals $\theta_1-\theta_0-\delta$ to the server. 
  \item The honest server adds a phase of $e^{-(\theta_1-\theta_0-\delta)}$ controlled by subscripts of keys\footnote{This is possible assuming there is some authentication information about the keys (for example, the $\fEn$ used in phase tables).} and prepares the following state up to a global phase:
  \begin{equation}\label{eq:24pz} \ket{x_0}+e^{\delta\mi\pi/4}\ket{x_1}\end{equation}
	\item Then both parties run the normal Hadamard test: the server measures all the bits in Hadamard basis and sends out the outcome $d$; the client outputs $\ffail$ if $d=0$ and otherwise could calculate $d\cdot (x_0+x_1)\mod 2$. 
	\end{enumerate}
	The client outputs the test results as follows:
	\begin{itemize}
		\item If $\delta=0$ the client outputs $\fpass$ to the flag register if $d\cdot (x_0+x_1)\mod 2=0$ and $\ffail$ otherwise.
		\item If $\delta=4$ the client outputs $\fpass$ to the flag register if $d\cdot (x_0+x_1)\mod 2=1$ and $\ffail$ otherwise.
	\end{itemize}
	The client's action for the other $\delta$ remains to be defined later.
\end{toyprtl}
 For a malicious server to pass this $\delta$-biased Hadamard test from \eqref{eq:8}, where $\delta\in \{0,4\}$, there has to be, on average of $ \theta_0,\theta_1\in  \{0,1\cdots 7\}^2$,
\begin{equation}\label{eq:13int}f(\theta_0+4)\approx f(\theta_0)+4,g(\theta_1+4)\approx g(\theta_1)+4\end{equation}
One important property of this test is that an honest server could pass deterministically, which implies, once the server fails in this test, the verifier will catch it cheating immediately.\par

 But this does not simply work generally for verifying the relations between states on different values of $\theta_0,\theta_1\in \{0,\cdots 7\}$. One obstacle is the Hadamard test does not give a deterministic answer (in the sense of  $d\cdot (x_0+x_1)\mod 2$) for a general state in the form of \eqref{eq:24pz} (for general $\delta$). Finding a test with one-sided error (which means the honest server could always pass) is also impossible since \eqref{eq:7} for different $\theta_0,\theta_1$ are not orthogonal in general.\par
  Here we generalize an idea from \cite{GVRSP,qfactory}: we do not restrict ourselves on verification processes with one-sided error; instead we turn to use a game where the optimal winning strategy is allowed to lose with some probability. \cite{GVRSP,qfactory} designed tests under this idea to verify single-qubit states; here we adapt their ideas to our setting and handle technical differences.\par
  In more detail, besides the $\fpass$/$\ffail$ flag, where a $\ffail$ result directly catches the server cheating, the $\fInPhTest$ will also (possibly) produce a $\fwin$/$\flose$ score. Then if both parties repeat such a game for many times (a large constant is sufficient to verify it to constant error tolerance), the client can calculate the winning ratio statistically and see whether the server's winning ratio is close to optimal. The test is designed to have a self-testing property, which says, any strategy that has close-to-optimal winning probability should also be close to the optimal strategy up to an isometry.\par
  Let's introduce the idea in more detail. 
  To summarize, our individual phase test goes as follows:
  \begin{toyprtl}
  The setup is the same as Toy Protocol \ref{toyprtl:hteb}.\par
  The client samples $\delta\leftarrow \{0,4,1\}$ and runs the protocol as given in Toy Protocol \ref{toyprtl:hteb}. The client's output for $\delta\in \{0,4\}$ is the same as Toy Protoocl \ref{toyprtl:hteb}. For $\delta=1$ case, the client outputs $\fwin$ to the score register if $d\cdot (x_0+x_1)\mod 2=0$ and $\flose$ otherwise. 
  \end{toyprtl}
  We could show the optimal winning probability (conditioned on a $\fwin/\flose$ score is generated) is $\cos^2(\pi/8)$, achieved by the honest initial state and the honest behavior. What's more, in the malicious setting, as said before, this test has a self-testing property:
  \begin{claim}
  Starting from \eqref{eq:8}, suppose the server does not fail in the protocol.\footnote{This condition is mainly on the $\delta\in \{0,4\}$ case; and for $\delta=1$ case it is required that $d\neq 0$.} Then:
  \begin{itemize} 
  \item The optimal winning probability conditioned on $\delta=1$ is $\cos^2(\pi/8)$.
  \item If the adversary could win in the $\delta=1$ case with probability $\approx\cos^2(\pi/8)$, then 
  $$\text{either }f(\theta_0)\approx\theta_0+c_0,g(\theta_1)\approx\theta_1+c_1$$
  $$\text{or }f(\theta_0)\approx-\theta_0+c_0,g(\theta_1)\approx-\theta_1+c_1$$
  	
  \end{itemize}
\end{claim}
Thus the test could only verify \eqref{eq:13int} up to a possible negation. This is as expected: as discussed in Section \ref{sec:2.1}, no classical channel protocol could rule out the complex conjugate attack. This is where the negation comes from. (After the keys $K$ are revealed, the complex-conjugate term is isometric to the honest output and the two terms could be merged together.)\par
  Finally we note our protocol could not only handle the simplified attack \eqref{eq:8} in the example above; it could also verify the initial state is close to a specific form in general. We give the following theorem which characterize the verifiability property of the $\fInPhTest$ protocol:
\begin{thm}[Properties of $\fInPhTest$, informal]\label{thm:2.3}
Suppose the client holds key pair $K^{(0)}=(x_0^{(0)},x_1^{(0)})$, phase pair $\Theta^{(0)}=(\theta_0^{(0)},\theta_1^{(0)})$. Suppose the client and server's purified joint state has necessary security properties and has the following form (here we make the $\Theta^{(0)}$ register explicit and make the client-side key register implicit):
$$\sum_{\theta_0^{(0)},\theta^{(0)}_1\in \{0,1\cdots 7\}^2}\underbrace{\ket{\theta_0^{(0)}}\ket{\theta_1^{(0)}}}_{\text{$\Theta^{(0)}$}}\otimes (\underbrace{\ket{x_0^{(0)}}}_{\text{server-side register required in the basis-honest form}}\ket{\varphi_{0,\theta_0^{(0)}}}+\ket{x_1^{(0)}}\ket{\varphi_{1,\theta_1^{(0)}}})$$
 Suppose $\fInPhTest$ with this initial state against an efficient adversary could pass (the client outputs $\fpass$ as flag) with probability close to $1$ and win (the client outputs $\fwin$ as score) with probability close to $\cos^2(\pi/8)$ conditioned on a $\fwin/\flose$ score is generated. Then there exist four states $\ket{\varphi_{0,+}},\ket{\varphi_{0,-}}$, $\ket{\varphi_{1,+}},\ket{\varphi_{1,-}}$ such that:
\begin{equation}\label{eq:22to}\text{on average over } \theta_0^{(0)}\in \{0,1\cdots 7\}:\ket{\varphi_{0,\theta_0^{(0)}}}\approx e^{\theta_0^{(0)}\mi\pi/4}\ket{\varphi_{0,+}}+e^{-\theta_0^{(0)}\mi\pi/4}\ket{\varphi_{0,-}}\end{equation}
\begin{equation}\label{eq:22to2}\text{on average over } \theta_1^{(0)}\in \{0,1\cdots 7\}:\ket{\varphi_{1,\theta_1^{(0)}}}\approx e^{\theta_1^{(0)}\mi\pi/4}\ket{\varphi_{1,+}}+e^{-\theta_1^{(0)}\mi\pi/4}\ket{\varphi_{1,-}}\end{equation}
\end{thm}
We will discuss its formalization in Section \ref{sec:2.8} and \ref{sec:2.9}.
\subsection{Collective Phase Test ($\fCoPhTest$)}\label{sec:2.7}
So far, we are only focusing on the simplified setting where only one gadget is considered. However a large part of the difficulties of this problem is to create a large number of such states verifiably, and guarantee the total complexity is still linear in the output number.\par
We highlight two limitations of existing works that focus on the verifiability of individual gadgets \cite{GVRSP,qfactory}:
\begin{itemize}
\item One frequently used technique for this and similar problems is the cut-and-choose technique, which is also the technique used in \cite{GVRSP,qfactory}. In this technique, the tests are all locally applied on single gadgets, and both parties repeat the single-gadget protocol for many rounds and choose a random subset from all the output gadgets. However, as far as we know, there seems to be an obstacle to make such type of protocols linear-time. The reason is, under this technique, to control the total error of $L$ gadgets down to a constant, the error tolerance of each gadget on average is no more than $O(1/L)$. This implies at least $O(L^2)$ repetitions are needed for a single gadget since the probability of detecting an $O(1/L)$ error from a single state scales with the square of the error norm. 
\item The performance becomes worse if we take the two-sided error issue appeared in Section \ref{sec:2.6}. The $\fInPhTest$ needs to be applied for many rounds to estimate the winning probability. For constant error tolerance this blowup is constant, but for $O(1/L)$ error tolerance this leads to further complexity blowup in the high-level protocol.
\end{itemize}
Here we develop a central sub-protocol for resolving these problems, which is called the collective phase test ($\fCoPhTest$). The structure of this protocol is called combine-and-test. With this test, we get rid of the obstacles in the following way:
\begin{itemize}
\item Before this test $O(|C|)$ number of gadgets are prepared in parallel. Then both parties combine these gadgets into a single big gadget, and test the combined gadget. Thus this test is not local on each individual gadget and not suffered from the first obstacle.
\item $\fCoPhTest$ has only one-sided error, which means, once the client sees a wrong answer, it knows the server is cheating right away.\par
Although $\fCoPhTest$ does not help us fully verify the phases, we will see it together with the $\fInPhTest$ achieves full verification on the phases of all the $L$ gadgets. In this overall phase testing protocol, with the help of the $\fCoPhTest$, the $\fInPhTest$ only needs to be applied on a single gadget to a constant error tolerance.\par
\end{itemize}

As an example of the combine-and-test technique, let's start from $2$ gadgets. Consider the initial state which honestly should be in the state
\begin{equation}\label{eq:29la}(e^{\mi\theta_0 \pi/4}\ket{x_0}+e^{\mi\theta_1 \pi/4}\ket{x_1})\otimes (e^{\mi\theta_0^\prime \pi/4}\ket{x_0^\prime}+e^{\mi\theta_1^\prime \pi/4}\ket{x_1^\prime}))\end{equation}
while a malicious server might deviate and prepare some other states. As what we did in the last section, for explaining the intuition, we will consider a specific attack where the server only tries to add different phases. That means, the state might be\footnote{From now on we interchangeably make the concatenation notation either explicit or implicit.}
\begin{equation}\label{eq:9rr}e^{\mi f_{00}(\theta_0,\theta^\prime_0) \pi/4}\ket{x_0}\ket{x_0^\prime}+e^{\mi f_{01}(\theta_0,\theta^\prime_1) \pi/4}\ket{x_0}\ket{x_1^\prime}+e^{\mi f_{10}(\theta_1,\theta^\prime_0) \pi/4}\ket{x_1}\ket{x_0^\prime}+e^{\mi f_{11}(\theta_1,\theta^\prime_1) \pi/4}\ket{x_1}\ket{x_1^\prime}\end{equation}
Instead of testing these two gadgets independently, the client will first combine these two gadgets into a single big gadget. This is achieved by sending a lookup table to instruct the server to decrypt and measure. Then both parties run Hadamard test on the combined gadget. In more detail:
\begin{toyprtl}\label{toyprtl:6}
\begin{enumerate}
\item The client samples $r_0,r_1\leftarrow \{0,1\}^\kappa$ and prepares the table
$$(x_0||x_0^\prime\rightarrow r_0,x_1||x_1^\prime\rightarrow r_0,$$
$$x_0||x_1^\prime\rightarrow r_1,x_1||x_0^\prime\rightarrow r_1)$$
The server decrypts the table with \eqref{eq:29la} measures the $r$ register and collapses the state into a superposition of two combined keys. Note the phases are also combined. This means in the honest setting the post-measurement states are:
$$\text{output }r_0:e^{\mi(\theta_0+\theta_0^\prime)\pi/4}\ket{x_0||x_0^\prime}+e^{\mi(\theta_1+\theta_1^\prime)\pi/4}\ket{x_1||x_1^\prime}$$
$$\text{output }r_1:e^{\mi(\theta_0+\theta_1^\prime)\pi/4}\ket{x_0||x_1^\prime}+e^{\mi(\theta_1+\theta_0^\prime)\pi/4}\ket{x_1||x_0^\prime}$$
Then the server will send back the measurement result $r$ to the client, and the client will check the validity of server's response (check it's in $\{r_0,r_1\}$) and calculate the keys and phases:
$$r=r_0:K^{(combined)}=(x_0||x^\prime_0,x_1||x^\prime_1);\quad r=r_1:K^{(combined)}=(x_0||x^\prime_1,x_1||x^\prime_0)$$
$$r=r_0:\Theta^{(combined)}=(\theta_0+\theta^\prime_0,\theta_1+\theta^\prime_1);\quad r=r_1:\Theta^{(combined)}=(\theta_0+\theta^\prime_1,\theta_1+\theta^\prime_0)$$
\end{enumerate}
Then the client can use a Hadamard test to test the combined gadget in the next step:
\begin{enumerate}

\item[2.] The client will reveal the relative phase (defined in Toy Protocol \ref{toyprtl:hteb}) of $\Theta^{(combined)}$ and the server could remove the joint phase of the combined gadget. Then the Hadamard test could be applied on the combined gadget.
 	 
\end{enumerate}
\end{toyprtl}
 Let's consider a malicious server. Starting from \eqref{eq:9rr}, the malicious server will end up in states:
\begin{equation}\label{eq:10}\text{output }r_0:e^{\mi f_{00}(\theta_0,\theta_0^\prime)\pi/4}\ket{x_0||x_0^\prime}+e^{\mi f_{11}(\theta_1,\theta_1^\prime)\pi/4}\ket{x_1||x_1^\prime}\end{equation}
\begin{equation}\text{output }r_1:e^{\mi f_{01}(\theta_0,\theta_1^\prime)\pi/4}\ket{x_0||x_1^\prime}+e^{\mi f_{10}(\theta_1,\theta_0^\prime)\pi/4}\ket{x_1||x_0^\prime}\end{equation}
  Without loss of generality let's assume the output is $r_0$ and the state is collapsed to \eqref{eq:10}. What's counter-intuitive here is to understand why the Hadamard test could test the joint phases. The observation is, the phase of $x_0x_0^\prime$ branch does not depend on the values of $\theta_1,\theta_1^\prime$, and the phase of $x_1x_1^\prime$ branch does not depend on the values of $\theta_0,\theta_0^\prime$. A more detailed calculation is as follows. To pass the test from \eqref{eq:10}, by the property of Hadamard test, there has to be
  \begin{equation}\label{eq:32lm}f_{11}(\theta_1,\theta_1^\prime)-f_{00}(\theta_0,\theta_0^\prime)\approx \text{the relative phase  in the honest setting} \end{equation}
  Recall the relative phase in the honest setting when $r=r_0$ is $(\theta_1+\theta^\prime_1)-(\theta_0+\theta^\prime_0)$. This together with \eqref{eq:32lm} implies\footnote{The detail is as follows. Fixing $\theta_0+\theta_0^\prime$, the right hand side of \eqref{eq:32lm} is fixed which implies the left hand side of \eqref{eq:32lm} is also fixed.}
  \begin{equation}\label{eq:36e}\forall \Delta\in \{0,1\cdots 7\}, f_{00}(\theta_0,\theta_0^\prime)\approx f_{00}(\theta_0-\Delta,\theta_0^\prime+\Delta)\end{equation}
A similar statement holds for all these four terms of \eqref{eq:9rr}, thus holds on average. This could be understood as follows: on average on each branch (term) in \eqref{eq:9rr}, the form of this branch is only a function of the honest joint phase (where the honest joint phase for branch $x_bx_{b^\prime}^\prime$ is $\theta_b+\theta_{b^\prime}^\prime$).\par More generally, we will see, when we consider the attack that does not follow the restricted form \eqref{eq:9rr}, $\fCoPhTest$ could still guarantee this property. 
Generalizing it to a combine-and-test process on all the $1+L$ gadgets in \eqref{eq:20} leads to a linear time phase sub-test, which is the $\fCoPhTest$:
\begin{toyprtl}
\begin{enumerate}
\item Both parties combine all the gadgets (with index from $0$ to $L$) to a single gadget;
\item The client computes the honest joint phase pair of the combined gadget. The client reveals the relative phase of the combined phase pair and both parties run the Hadamard test.	
\end{enumerate}	
\end{toyprtl}
Informally, $\fCoPhTest$ has the following properties, which could be seen as a generalization of \eqref{eq:36e}.
\begin{thm}[Properties of $\fCoPhTest$, informal]\label{thm:2.4}
	Suppose the client holds a tuple of key pairs $K=(K^{(i)})_{i\in [0,L]}$, $K^{(i)}=(x_0^{(i)},x_1^{(i)})$ and holds a tuple of phase pairs $\Theta=(\Theta^{(i)})_{i\in [0,L]}$, $\Theta^{(i)}=(\theta^{(i)}_0,\theta^{(i)}_1)$. Suppose the client and server's purified joint state has necessary security properties and has the following form:
\begin{equation}\label{eq:36}\sum_{\text{All valid values of } \Theta}\underbrace{\ket{\Theta}}_{client-side}\otimes \sum_{\vec{b}\in \{0,1\}^{1+L}}\underbrace{\ket{x_{b^{(0)}}^{(0)}x_{b^{(1)}}^{(1)}\cdots x_{b^{(L)}}^{(L)}}}_{\text{server-side register required in the basis-honest form}}\otimes\ket{\varphi_{\vec{b},\theta_{b^{(0)}}^{(0)}\theta_{b^{(1)}}^{(1)}\cdots \theta_{b^{(L)}}^{(L)}}}\end{equation}
where $b^{(0)}b^{(1)}\cdots b^{(L)}$ is the coordinate expansion of $\vec{b}$. Suppose in $\fInPhTest$ an efficient adversary could pass (make the client outputs $\fpass$ as flag) with probability close to $1$. Then on average over all the possible $\vec{b}$ in \eqref{eq:36}, consider the branch
 $$\sum_{\text{All valid values of } \Theta}\underbrace{\ket{\Theta}}_{client-side}\otimes\ket{x_{b^{(0)}}^{(0)}x_{b^{(1)}}^{(1)}\cdots x_{b^{(L)}}^{(L)}}\otimes\ket{\varphi_{\vec{b},\theta_{b^{(0)}}^{(0)}\theta_{b^{(1)}}^{(1)}\cdots \theta_{b^{(L)}}^{(L)}}}$$ 
there is, informally, the $\ket{\varphi_{\vec{b},\theta_{b^{(0)}}^{(0)}\theta_{b^{(1)}}^{(1)}\cdots \theta_{b^{(L)}}^{(L)}}}$ part is close to a state that only depends on the honest joint phase (instead of depending on all the phases here). Here the honest joint phase is $\theta_{b^{(0)}}^{(0)}+\theta_{b^{(1)}}^{(1)}+\cdots +\theta_{b^{(L)}}^{(L)}$, the phase information for this branch in the honest setting when the client-side phase tuple is $\Theta$.
\end{thm}
We will discuss its formalization in Section \ref{sec:2.8} and \ref{sec:2.9}.
\subsection{State Forms, and the Overall Implication of $\fCoPhTest$ and $\fInPhTest$ Applied on Multiple Gadgets}\label{sec:2.8}
To analyze the protocol formally, we define a series of \emph{state forms}, which are classes of states that satisfy some specific structures. We have already seen the \emph{basis-honest form} in Section \ref{sec:2.5}; below we will further define the \emph{basis-phase correspondence form}, \emph{pre-phase-honest form} and the \emph{phase-honest form}.
\subsubsection{Basis-phase correspondence form}
Recall by the end of $\fAddPhaseWithswitch$ in Toy Protocol \ref{toyprtl:3} the client holds a tuple of key pairs $K=(K^{(i)})_{i\in [0,L]}$, $K^{(i)}=(x_0^{(i)},x_1^{(i)})$ and holds a tuple of phase pairs $\Theta=(\Theta^{(i)})_{i\in [0,L]}$, $\Theta^{(i)}=(\theta^{(i)}_0,\theta^{(i)}_1)$. Honestly the server is instructed to hold the state \eqref{eq:20}, 
 while maliciously we assume the attacker's state $\ket{\varphi}$ is in the basis-honest form  (verified by the standard basis test). Expanding the state we can write
 \begin{equation}\label{eq:19ox2}\sum_{b^{(0)}b^{(1)}b^{(2)}\cdots b^{(L)}: \forall i\in [0,L],b^{(i)}\in \{0,1\}}\underbrace{\ket{x^{(0)}_{b^{(0)}}x^{(1)}_{b^{(1)}}x^{(2)}_{b^{(2)}}\cdots x^{(L)}_{b^{(L)}}}}_{\text{some server side registers}}\underbrace{\ket{\varphi_{b^{(0)}b^{(1)}b^{(2)}\cdots b^{(L)}}}}_{\text{other part}}\end{equation}
 For simplicity we make the client side register $K$ implicit.\par
  We define the \emph{basis-phase correspondence form} as follows, which characterize an intuitively property of output states of $\fAddPhaseWithswitch$. We assume the state in \eqref{eq:19ox2} corresponding to branch $$\vec{x}_{\vec{b}}:=x^{(0)}_{b^{(0)}}x^{(1)}_{b^{(1)}}x^{(2)}_{b^{(2)}}\cdots x^{(L)}_{b^{(L)}}$$ could depend on the values of $$\vec{\Theta}_{\vec{b}}:=\theta^{(0)}_{b^{(0)}}\theta^{(1)}_{b^{(1)}}\theta^{(2)}_{b^{(2)}}\cdots \theta^{(L)}_{b^{(L)}}$$ but independent of the values of $\theta^{(0)}_{1-b^{(0)}}\theta^{(1)}_{1-b^{(1)}}\theta^{(2)}_{1-b^{(2)}}\cdots \theta^{(L)}_{1-b^{(L)}}$. Recall that in $\fAddPhaseWithswitch$ the client sends many look-up tables and the server  could decrypt some rows of them with the keys it holds, and this property intuitively says the adversary could not decrypt the rows where it does not has the corresponding keys. Thus we can express \eqref{eq:19ox2} as
\begin{equation}\label{eq:25to}\eqref{eq:19ox2}=\sum_{\Theta\in \{0,1\cdots 7\}^{2(1+L)}}\underbrace{\ket{\Theta}}_{\text{client-side register that stores $\Theta$}}\otimes \ket{\varphi_\Theta},\qquad\ket{\varphi_{\Theta}}=\sum_{\vec{b}\in \{0,1\}^{1+L}}\ket{\vec{x}_{\vec{b}}}\ket{\varphi_{\vec{b},\vec{\Theta}_{\vec{b}}}}, \end{equation}
Note that in Theorem \ref{thm:2.3}, \ref{thm:2.4} we have already implicitly assume it.
 \subsubsection{Pre-phase-honest form and phase-honest form}\label{sec:2.8.2}
We have designed tests that aim at allowing the client to verify \eqref{eq:25to} actually has the form of \eqref{eq:20} (if we only consider the phases); but there is quite a big gap between them. To understand how our tests bridge the gap between \eqref{eq:20} and \eqref{eq:25to}, we define the following two forms of states.\par
 The first is the \emph{pre-phase-honest form}, which is a basis-honest form, and additionally, for each branch indexed by $\vec{b}$ in \eqref{eq:25to}, the state should be determined only by the honest joint phase. As before we assume the client holds a tuple of key pairs $K$ and a tuple of phase pairs $\Theta$.

\begin{defn}[Pre-Phase-honest form] We say a state $\ket{\varphi}$ is in the pre-phase-honest form if there exists a class of states $\ket{\varphi_{\vec{b},\theta}}$ for each $\vec{b}\in \{0,1\}^{1+L},\theta\in \{0,1\cdots 7\}$ such that in \eqref{eq:25to},
	\begin{equation}\label{eq:38islx}\ket{\varphi_{\Theta}}=\sum_{\vec{b}\in \{0,1\}^{1+L}}\ket{\vec{x}_{\vec{b}}}\otimes \ket{\varphi_{\vec{b},\tSUM(\vec{\Theta}_{\vec{b}})}},\text{ where }\tSUM(\vec{\Theta}_{\vec{b}})=\theta^{(0)}_{b^{(0)}}+\theta^{(1)}_{b^{(1)}}+\theta^{(2)}_{b^{(2)}}+\cdots +\theta^{(L)}_{b^{(L)}}\end{equation}
\end{defn}
Then we define the phase-honest form, which is a pre-phase-honest form, and for each branch, the phases on the server-side state is determined by the client-side phase information in a way similar to \eqref{eq:20}:
\begin{defn}[Phase-honest form]  We say a state $\ket{\varphi}$ is in the phase-honest form if there exists a class of states $\ket{\varphi_{\vec{b},+}},\ket{\varphi_{\vec{b},1}}$ for each $\vec{b}\in \{0,1\}^{1+L}$ such that 
	\begin{equation}\label{eq:39islx}\ket{\varphi_\Theta}=\sum_{\vec{b}\in \{0,1\}^{1+L}}\ket{\vec{x}_{\vec{b}}}\otimes (e^{\tSUM(\vec{\Theta}_{\vec{b}})\mi \pi/4}\ket{\varphi_{\vec{b},+}}+e^{-\tSUM(\vec{\Theta}_{\vec{b}})\mi \pi/4}\ket{\varphi_{\vec{b},-}})\end{equation}
\end{defn}
The second term comes from the fact that we could not rule out the complex-conjugate attack.\par
 
%

\subsubsection{A summary}
The relation of these forms of states are
\begin{center}
arbitrary states $\supseteq$ basis-honest form \eqref{eq:19ox} $\supseteq$ basis-phase correspondence form \eqref{eq:25to} $\supseteq$ pre-phase-honest form \eqref{eq:38islx} $\supseteq$ phase-honest-form \eqref{eq:39islx} $\supseteq$ \eqref{eq:20}
\end{center}
For the two state forms described in Section \ref{sec:2.8.2}, intuitively, the $\fCoPhTest$ aims at testing a basis-phase correspondence form is a pre-phase honest form, and $\fInPhTest$ aims at testing a pre-phase-honest form is a phase-honest form. Once the overall state is known to be in the form of \eqref{eq:39islx}, the verification of phase information of the server-side states has been completed up to a complex-conjugate ambiguity.\par
In the next section we discuss how we formally analyze our protocols to bridge these gap step-by-step.
%
\subsection{Security Proofs Structure}\label{sec:2.9}
The security proofs go as follows at a high level. Below $\ket{\varphi}$ stands for the output state of $\fAddPhaseWithswitch$. The goal is roughly to show this state, after the client reveals the key $K$, is approximately isometric to the honest state \eqref{eq:cqtarget0}.
\begin{align}
&\text{An unverified state }\ket{\varphi}\\
\text{($\fStdBTest$) }\Rightarrow &\text{Basis-honest form}\label{eq:38po}\\
\text{(Properties of $\fAddPhaseWithswitch$) }\Rightarrow &\text{Basis-phase correspondence form}\label{eq:39po}\\
(\fCoPhTest)\Rightarrow &\text{Pre-phase-honest form}\label{eq:40po}\\
(\fInPhTest)\Rightarrow &\text{Phase-honest form}\label{eq:41po}\\	
 (\fBNTest)\Rightarrow &\text{Form \eqref{eq:20} up to a complex-conjugate ambiguity}\label{eq:42po}\\
 (\text{Client reveals $K$})\Rightarrow &\text{Form \eqref{eq:cqtarget0}}\label{eq:43po}
\end{align}

where each arrow in \eqref{eq:38po}\eqref{eq:40po}\eqref{eq:41po}\eqref{eq:42po} means we make use of the fact that the adversary could pass these tests with high probability to derive that the initial state has a specific form, and each arrow in \eqref{eq:39po}\eqref{eq:43po} means the design of the protocol implies the initial state $\ket{\varphi}$ has the corresponding forms regardless of the adversary's passing probability.\par
However, there is a tricky problem during the proof of the arrows of \eqref{eq:39po}\eqref{eq:40po}\eqref{eq:41po}. In these arrows we implicitly assume the previous steps perfectly verifies the forms of states; but this is not the case, all the steps in  \eqref{eq:39po}\eqref{eq:40po}\eqref{eq:41po} are approximate, which leads to a composability issue between the analysis of each subprotocol. In more detail, for example, \eqref{eq:41po} says ``if the initial state is in a pre-phase-honest form, and it could pass the $\fInPhTest$ with high probability, then it's close to a phase honest form''; but it is not necessarily the case that if the initial state is only approximately in a pre-phase-honest form, the same statement still works! Note that for each arrow in \eqref{eq:39po}\eqref{eq:40po}\eqref{eq:41po}, the condition before the arrow are typically not the only conditions that we need when we want to formally prove the result after the arrow; many properties are needed, for example, efficiently preparable property or security of keys. These properties are typically not preserved by a general approximation on the state, which makes it sophisticated to work on server-side approximation directly.\par 
To address this problem, the first step is to work on the \emph{purified joint state} of both the client and the server. In the real execution the client is classical while the server holds quantum states; in the security analysis the purified joint state is defined to be the state where all the classical randomness are replaced by quantum superpositions (on which a collapsing measurement gives the same classical randomness; note that for this purification we do not introduce environment or reference system). \footnote{This treatment is slightly different from the usual notion of purification, where a classical register is purified by entangling it with the environment; here classical registers are replaced by quantum superpositions directly. However, we will see, the registers that hold these classical randomness, once initialized, will not be revised by any operation during the protocol including the final distinguisher; these registers are \emph{read-only} once initialized. In this setting two purifications look completely the same.}\footnote{This purification treatment of classical information is not new in our protocol. It is also used in several existing works like \cite{vidicknotes,compressedoracle}. The treatments after the purification are different.}\par
Then the observation is as follows. The purified joint state contains a large entanglement between the client and the server. Then a duality between the client-side and the server-side emerges:
\begin{align*}
	&\text{ Server-side state approximately has a form}\\
	\rightleftharpoons &\text{ Joint state is approximately invariant under an operation that revises the client-side registers}
\end{align*}
 In more detail, in our security proofs, we will design a series of \emph{randomization operators}. Corresponding to \eqref{eq:39po}\eqref{eq:40po}\eqref{eq:41po}, these operators are denoted as $\cR_1$, $\cR_2$, $\cP$. These operators are defined on the joint state of the client and the server, and revise the client-side registers (possibly controlled by registers in other parties). 
 These operators have the following properties: \begin{itemize}\item The honest state is invariant under these operators; \item The execution of the protocol or the ability to pass a test in each of \eqref{eq:39po}\eqref{eq:40po}\eqref{eq:41po} implies approximate invariance of the purified joint state under the corresponding operator; \footnote{The formal theorems corresponding to Theorem \ref{thm:2.3} and Theorem \ref{thm:2.4} will also be described in this way.} \item The output of a randomization perfectly has form that we aim at in each of \eqref{eq:39po}\eqref{eq:40po}\eqref{eq:41po}.\end{itemize}
 With these tool, when we analyze our subprotocols, the theorem statement will be ``if this test could be passed with high probability, the state will be approximately invariant under the corresponding randomization operator''. Approximate invariance under randomization operators composes with each other naturally and turns out to have much nicer properties than simply saying the server's state is close to a state that has a specific property: for example, randomization operators are efficient operators that operate on some specific registers, which allow us to prove some security properties that we need on the state are preserved.\par
To give the reader a feeling of our technique, we give a minimum example, which only contains one gadget, that illustrates the first property (the honest state is invariant) of $\cR_1$:
\begin{exmp}\label{exmp:r1}Expanding the honest state by writing down all the possible client side phases:
	$$\sum_{\theta_0,\theta_1\in \{0,1\cdots 7\}^2}\frac{1}{8}\underbrace{\ket{\theta_0}\ket{\theta_1}}_{\text{client}}\otimes \underbrace{\frac{1}{\sqrt{2}}(e^{\theta_0\mi\pi/4}\ket{x_0}+e^{\theta_1\mi\pi/4}\ket{x_1})}_{\text{server}}$$
	This is the purified joint state of the client and the server since the client-side phase registers are explicit and entangled with the server-side system. (Note that, the purified joint states have actually already been used in the previous subsections of this technical overview.) We make the client-side phase registers explicit since we will need to work on them, and omit the client-side key registers.\par
	 Introduce randomness $\Delta_0,\Delta_1\in_r \{0,1\cdots 7\}^2$, and write out their registers explicitly (after purification):
	\begin{equation}\label{eq:28to}\sum_{\Delta_0,\Delta_1\in \{0,1\cdots 7\}^2}\frac{1}{8}\ket{\Delta_0}\ket{\Delta_1}\otimes\sum_{\theta_0,\theta_1\in \{0,1\cdots 7\}^2}\frac{1}{8}\underbrace{\ket{\theta_0}\ket{\theta_1}}_{\text{client}}\otimes \underbrace{\frac{1}{\sqrt{2}}(e^{\theta_0\mi\pi/4}\ket{x_0}+e^{\theta_1\mi\pi/4}\ket{x_1})}_{\text{server}}\end{equation}
	Then \eqref{eq:28to} is invariant under the following controlled-swap operations controlled by the server-side branch subscripts:
	\begin{equation}\label{eq:45v}x_0\text{-branch: }\ket{\Delta_0}\ket{\Delta_1}\ket{\theta_0}\ket{\theta_1}\ket{x_0}\rightarrow \ket{\Delta_0}\ket{\theta_1}\ket{\theta_0}\ket{\Delta_1}\ket{x_0}\end{equation}
		\begin{equation}\label{eq:45v2}x_1\text{-branch: }\ket{\Delta_0}\ket{\Delta_1}\ket{\theta_0}\ket{\theta_1}\ket{x_1}\rightarrow \ket{\theta_0}\ket{\Delta_1}\ket{\Delta_0}\ket{\theta_1}\ket{x_1}\end{equation}
		which means in \eqref{eq:45v} the client-side value of $\theta_1$ is randomized by $\Delta_1$ and in \eqref{eq:45v2} the client-side value of $\theta_0$ is randomied by $\Delta_0$.\par
		What's more, we can also show this randomization operations takes an arbitrary basis-honest form to a basis-phase correspondence form.
\end{exmp}

\subsection{Basis Uniformity Test ($\fBNTest$)}\label{sec:2.10}
We have developed a set of tools for verifying the phases, which could verify the client and server's joint state is approximately in the form of a phase-honest state. Compare to the target state \eqref{eq:gadget}, what remains to be verified is the norm of each branch is close to each other.\par
Let's again start with the simple single-gadget case to explain the initial intuition. Suppose the client holds keys $K=(x_0,x_1)$ and the server-side state is in the form of
\begin{equation}\label{eq:11}\alpha_0\ket{x_0}+\alpha_1\ket{x_1}\end{equation}
If the client wants to verify $\alpha_0\approx \alpha_1$, the protocol used here is still the (RO-padded) Hadamard test:
\begin{claim}[Informal]\label{claim:2.7} If the server could pass the Hadamard test with initial state in the form of \eqref{eq:11}, there has to be $\alpha_0\approx \alpha_1$.\end{claim}
The difficulty is still in the multi-gadget case. Suppose the client holds $L$ pairs of keys $K^{(i)}=(x_0^{(i)},x_1^{(i)}),i\in [L]$ and the server-side state is already verified to have the form:
\begin{equation}\label{eq:bn18}\sum_{b^{(1)}b^{(2)}\cdots b^{(L)}\in \{0,1\}^L}\alpha_{b^{(1)}b^{(2)}\cdots b^{(L)}}\ket{x^{(1)}_{b^{(1)}}x^{(2)}_{b^{(2)}}\cdots x^{(L)}_{b^{(L)}}}\end{equation}
where the coefficients are non-negative real numbers. The client wants to verify the state is close to
\begin{equation}\label{eq:12r}\sum_{b^{(1)}b^{(2)}\cdots b^{(L)}\in \{0,1\}^L}\frac{1}{\sqrt{2^L}}\ket{x^{(1)}_{b^{(1)}}x^{(2)}_{b^{(2)}}\cdots x^{(L)}_{b^{(L)}}}\end{equation}
in linear time.\par
Note that here we assume the honest state is \eqref{eq:12r}, while in Toy Protocol \ref{toyprtl:3} the phases have already been added when $\fBNTest$ is executed, which seem incompatible; in real protocol in the $\fBNTest$ the client will first simply reveal all the phase information on these gadgets to allow the honest server to remove the phases.\par
Again, we will use the global combine-and-test method to achieve this goal. Informally, the \emph{basis uniformity test} ($\fBNTest$) is as follows:
\begin{toyprtl}\label{toyprtl:7}
\begin{enumerate}
\item The client chooses a random subset of index $I\subseteq [L]$;
\item The client instructs the server to combine the gadgets with index in $I$ into a single gadget (using the lookup tables discussed in $\fCoPhTest$); for the gadgets with index outside $I$, the client instructs the server to make a standard basis measurement and check the results.
\item Both parties execute a Hadamard test on the combined gadget.	
\end{enumerate}
\end{toyprtl}
Informally we have the following claim that captures the power of basis uniformity test.
\begin{claim}
	If a protocol of form \eqref{eq:bn18} could pass the basis uniformity test with high probability, the state is close to \eqref{eq:12r}.
\end{claim}
Let's first compare the basis uniformity test with the collective phase test, and discuss its intuitions.
\paragraph{Comparison to the Collective Phase Test} We note that there is an important difference of this test and the collective phase test constructed previously, even if both tests have the combine-and-test structure. In the collective phase test all the gadgets are combined together; while in the basis uniformity test the client samples a random subset of gadgets. The importance of this difference is illustrated by the following example, in which the combine-all test could not detect the deviation, while the combine-a-subset could detect.
\begin{exmp}
Consider the state
\begin{equation}\label{eq:12}\frac{1}{\sqrt{2}}(\ket{x_0^{(1)}x_0^{(2)}\cdots x_0^{(L)}}+\ket{x_1^{(1)}x_1^{(2)}\cdots x_1^{(L)}})\end{equation}
We can see the state \eqref{eq:12} is far from the target state \eqref{eq:12r}. And we have:
\begin{itemize}
\item It passes the combine-all protocol (that is, choose $I=[L]$ in Toy Protocol \ref{toyprtl:7}).
\item It could not pass the test in Toy Protocol \ref{toyprtl:7}: intuitively, if the client chooses some subset $I$ of all the indices, and instructs the server to make a standard basis measurement on the remaining registers, the state in registers with indices in $I$ will also collapse and the server will not be able to pass the Hadamard test using the remaining state. 
\end{itemize}
 \end{exmp}
 Thus we can see the random selection of $I$ is necessary for the basis uniformity test. An intuition for the basis uniformity test is as follows. We note that, just before the Hadamard test step, the server-side state is expected to be in the form of
 \begin{equation}\label{eq:38r}\alpha_{\vec{b}_0}\ket{\vec{x}_{\vec{b}_0}}+\alpha_{\vec{b}_1}\ket{\vec{x}_{\vec{b}_1}},\quad \vec{b}_0,\vec{b}_1\in \{0,1\}^L\end{equation}
where $\vec{x}_{\vec{b}_0}$, $\vec{x}_{\vec{b}_1}$ represent two branches of \eqref{eq:bn18}, and the randomness of $\vec{b}_0$, $\vec{b}_1$ come from the random choice of $I$ and the random collapsing in step 2 of Toy Protocol \ref{toyprtl:7}. \par
By Claim \ref{claim:2.7} intuitively we know  
 \begin{equation}\label{eq:50ll}\alpha_{\vec{b}_0}\approx \alpha_{\vec{b}_1}\end{equation}
 However we note \eqref{eq:50ll} only holds on average. We further note the probability that $\vec{b}_0,\vec{b}_1$ appear are in turn determined by the values of $\alpha_{\vec{b}_0}$, $\alpha_{\vec{b}_1}$ themselves. What's more, each of these probabilities is only exponentially small, which leads to additional obstacles in the security proof. (A re-normalized state of an exponentially-small state does not necessarily follow the formal version of Claim \ref{claim:2.7} since the state might not even be efficiently-preparable.) Thus we need a careful analysis of the protocol that addresses these problems. Finally we could prove, the high passing probability of the $\fBNTest$ implies:
 \begin{equation}\label{eq:51ll}\frac{1}{2^L}\sum_{\vec{b}_0\in \{0,1\}^L,\vec{b}_1\in \{0,1\}^L}|\alpha_{\vec{b}_0}- \alpha_{\vec{b}_1}|^2\leq O(1),\end{equation}
 which could be understood as a suitable average version of \eqref{eq:50ll}. Then by linear algebra \eqref{eq:51ll} implies
  $$\sum_{\vec{b}\in \{0,1\}^L}|\alpha_{\vec{b}}-\frac{1}{\sqrt{2^L}}|^2\leq O(1)$$
  which completes the proof.
 \subsection{Amplification to RSPV}
 We do not construct RSPV protocol directly; instead, we define an intermediate notion called \emph{pre-RSPV}. We will use the techniques described so far to design a pre-RSPV protocol, and then use a repetition-based amplification procedure to get an RSPV protocol.\par
 In more detail, our techniques so far have the following limitations, which do not fit into the RSPV notion, but are allowed in the pre-RSPV notion:
 \begin{itemize}
 	\item As said in Section \ref{sec:2.6}, the honest server does not necessarily win the test; what our protocol can verify is, if the server passes and wins with close-to-optimal probability, the output state should have the verifiability property we want.
 	\item The protocol construction framework in Section \ref{sec:2.4} does not necessarily generate an output state; it only generates output state in the last case of Toy Protocol \ref{toyprtl:3}. 
 \end{itemize}
We formalize the notion of pre-RSPV in Section \ref{sec:4.5}. Informally:
\begin{itemize}
	\item In pre-RSPV in the beginning of the protocol the client will randomly choose a round type in $\{\ftest,\fquiz,\fcomp\}$. 
	\item In the end of the protocol the client will output a $flag\in \{\fpass,\ffail\}$ and a $score\in \{\fwin,\flose,\perp\}$. 
	\item In all round types, a $flag$ will be generated. The honest server will not lead to a $\ffail$ flag (except with negligible probability). Thus if the client outputs $\ffail$ as the flag, it directly shows the server is cheating.\par
	In $\fquiz$ round, the client will possibly write $\fwin$ or $\flose$ as the score. (In the other round types the score is $\perp$ by default.) The honest server should $\fwin$ with probability $\OPT$ conditioned on a $\fwin/\flose$ score is generated.\par
	In $\fcomp$ round, the client will get output keys and the server will get output states.
\end{itemize}
Then the amplification of pre-RSPV to RSPV is achieved in Section \ref{sec:13.1}. This is achieved in the following way:
\begin{enumerate}
	\item Run many rounds of pre-RSPV and the client calculates the total score (the number of $\fwin$). The client outputs $\ffail$ if any pre-RSPV subprotocol returns $\ffail$ or the total score is significantly smaller than the expected value of honest behavior.
	\item The client chooses a random round and if it's a $\fcomp$ round, use the output keys and output state of this round as the output keys and output state of the RSPV protocol. If it's not a $\fcomp$ round, go back to step $1$.
\end{enumerate}

\section{Preliminaries}\label{sec:3}
\subsection{Basic Notations and Facts}
We refer to \cite{NielsenChuangs} for basics of quantum computation. Here we clarify some notations that will be used in our works.
\subsubsection{Quantum gates}
We choose the elementary gate set to be $\{\fX,\fY,\fZ,\fH,\fP,\fT,\fCN,\fToffoli\}$ (which is a typical choice). \footnote{To avoid ambiguity we note $\fP=\begin{bmatrix}
	1&\\&\mi
\end{bmatrix}$.}
\subsubsection{Basic notations}
\begin{nota}\label{nota:l}
We use $[N]$ to denote set $\{1,2\cdots N\}$. Use $[0,N]$ to denote set $\{0,1\cdots N\}$. When an algorithm iterates through all the elements in these sets it iterates from the smaller to the bigger.
\end{nota}
\begin{nota}
A normalized vector is defined to be a vector with norm $1$; a sub-normalized vector is defined to be a vector with norm $\leq 1$. A sub-normalized probability distribution is defined to be a non-negative real vector whose sum of coordinates is $\leq 1$. 	
\end{nota}

\begin{nota}\label{nota:register}
	We use bold font like $\bK,\bTheta$ to denote registers. And we use normal font like $K,\Theta$ to denote values of the corresponding registers.
For a register $\bR$, we use $\Domain(\bR)$ to denote the set of its valid values.\par
When we write equations on registers, for example, $\bS=\bx$, there are two possible meanings: (1) the subspace where the values of $\bS$ is equal to the values of $\bx$, or (2) a statement which says the state that we are studying falls completely in the space where the values of $\bS$ is equal to the values of $\bx$. The choice of meanings is determined by the context.
\end{nota}
\begin{nota}
	We use $\tSUM(\vec{e})$ to denote the sum of all the terms of $\vec{e}$. Note that in this work this notation is only applied on phase information and the addition is in $\bZ_8$.
\end{nota}
\begin{nota}We use $\Pi$ to denote projections. We use the superscript to denote the registers that the projection applied on and the subscript to denote the space that it projects on. For example, $\Pi^{\bS}_0$ projects onto the space that the register $\bS$ is in value $0$.
\end{nota}

We use $\bbI$ to denote the identity. Thus $(\bbI-\Pi)$ is the complementary projection of $\Pi$.

\begin{nota}
We use $|\cdot|$ to denote the norm of a state, length of a string, size of a set, and the number of random oracle queries of an operator (in the quantum random oracle model). 	
\end{nota}
\begin{nota}
Use $\Pi_E$ to denote the projection onto some space $E$. We call $|\Pi_{E}\ket{\varphi}|^2$ the probability that event $E$ happens when the state of the system is described by $\ket{\varphi}$, or simply the probability that $E$ happens. Note that we only require $\ket{\varphi}$ to be sub-normalized to make this notion well-defined.
\end{nota}
\begin{nota}
Suppose $A,B$ are two quantum operations. We use $AB$ or $A\circ B$ to denote the composition (matrix multiplication when they are represented as matrices) of these two operators. (We make $\circ$ explicit or implicit interchangeably.)
\end{nota}

The fonts used in this paper for describing operations could be normal (like $A,B,U$), sans-serif (like $\fH,\fX$, or $\fPrtl,\fAdv$)  or calligraphic (like $\cP$). Typically sans-serif fonts are used for elementary gates and protocol execution steps, calligraphic fonts are used for abstract operations, and normal fonts are mainly used as intermediate symbols, but we do not put strict rule for their usage.\par

\begin{nota}
In cryptographic protocols there is often a parameter $\kappa$ called security parameter. Then we say an operator parameterized by $\kappa$ (denoted by $(O_\kappa)_{\kappa\in \bN}$, or $O$ if we make the security parameter implicit), is \emph{efficient} if there exists a polynomial time Turing machine that takes $1^\kappa$ as input and outputs the description of $O_\kappa$.\par
A state family $(\ket{\varphi_\kappa})_{\kappa\in \bN}$ is efficiently-preparable if there exists a efficient family of polynomial time operators (which could include projections) $(O_\kappa)_{\kappa\in \bN}$ such that $\ket{\varphi_\kappa}=O_\kappa\ket{0}$. Similarly in later proofs we make the security parameter implicit.\par
Negligible function $\fneg(\kappa)$ means a function that decreases to $0$ faster than any polynomial when $\kappa\rightarrow +\infty$.
\end{nota}
\begin{nota}
As seen in the introduction, in this work we need to work on many key pairs and phase pairs; we use superscript with parentheses to index them: for example, $K^{(1)}$, $K^{(2)}$, etc. Other types of information like the time step counter could also appear in the superscript position, but they do not have parentheses.
\end{nota}
\begin{nota}
We use $\in_r$ or $\leftarrow_r$ to mean an element is randomly sampled from a domain.	
\end{nota}

\subsubsection{Indistinguishability notations}
The following indistinguishability notations are used in our work.
\begin{nota}
We write
$\ket{\varphi}\approx_\epsilon\ket{\phi}$	
if $|\ket{\varphi}-\ket{\phi}|\leq \epsilon$.	
\end{nota}
Note the $\approx_\epsilon$ notation could also be used for two real numbers.


\begin{nota}
Let $\cF$ be a set of operators. We write $\ket{\varphi}\approx^{ind:\cF}_\epsilon\ket{\phi}$ if for any $\fAdv\in \cF$ that outputs a bit in a fixed register (denoted by $\bS$), there is
$$|\Pi_0^\bS\fAdv\ket{\varphi}|\approx_{\epsilon}|\Pi_0^\bS\fAdv\ket{\phi}|$$
\end{nota}
\begin{nota} The states and operators below are implicitly parameterized by the security parameter $\kappa$.\par
We write $\ket{\varphi}\approx^{ind}_\epsilon\ket{\phi}$ if $\ket{\varphi}\approx^{ind:\cF}_\epsilon\ket{\phi}$ where $\cF$ contains all the efficient operations on some registers (the choices of registers should be from the context).
\end{nota}
\subsubsection{Approximate invariance}
\begin{defn}
	If $O\ket{\varphi}\approx_\epsilon\ket{\varphi}$, we say $\ket{\varphi}$ is $\epsilon$-invariant under $O$.
\end{defn}
\subsubsection{CQ-states and purified joint states}\label{sec:2.1.1}
In this work since the client and the random oracle are classical and the server is quantum, the overall states of all the parties are generally described by CQ-states. However, CQ-states could be unnatural to work on; for security proofs in this work, we will mainly work on their purifications. In more detail, we introduce the following (which is similar to \cite{jiayu20}). 
\begin{nota}\label{nota:cq}
Consider a cq-state where the set of possible values for the classical part is $\cC$, the classical register, denoted by $\bbC$, is in value $c\in \cC$ with probability $p_c$, and the quantum part is in state $\ket{\varphi_c}$ correspondingly. Then the overall cq-state is denoted as
$$\sum_{c\in \cC} p_c\ket{c}\bra{c}\otimes \ket{\varphi_c}$$	
Generally the corresponding purified state is defined to be
\begin{equation}\label{eq:54nb}\sum_c\sqrt{p_c}\underbrace{\ket{c}}_{\bbC}\otimes \ket{\varphi_c}\underbrace{\otimes \ket{c}}_{\text{environment}}\end{equation}
In this work we consider its purified state to be
\begin{equation}\label{eq:55nb}\sum_c\sqrt{p_c}\underbrace{\ket{c}}_{\bbC}\otimes \ket{\varphi_c}\end{equation}
\end{nota}
In general these two purifications are not equivalent; but we could see they are equivalent under a specific class of operators:
\begin{defn}[Read-only]
We say an operator $O$ operates on a register $\bbC$ in a read-only way if the elementary gates in $O$ that are applied on $\bbC$ are solely $\fCN$ and $\fToffoli$, and $\bbC$ is only used as the control wire.
\end{defn}

\begin{fact}\label{fact:rra}
Define $\cF$ as the set of operators that operate on $\bbC$ in a read-only  way. Then $\eqref{eq:54nb}\approx^{ind:\cF}\eqref{eq:55nb}$.
\end{fact}
We also have the following fact.
\begin{fact}\label{fact:overalltool}
Define $\cF$ as the set of operators that operate on $\bbC$ in a read-only  way. Then for any set of real values $(\alpha_{c})_{c\in \cC}$, $\sum_{c\in \cC}\underbrace{\ket{c}}_{\bbC}\otimes \ket{\varphi_c}\approx^{ind:\cF}\sum_{c\in \cC}\underbrace{\ket{c}}_{\bbC}\otimes e^{\alpha_c\mi\pi}\ket{\varphi_c}$.
\end{fact}
We introduce the following notion for simplicity of later discussions.
\begin{nota}\label{nota:compo}
	In the setting of Notation \ref{nota:cq}, for state \eqref{eq:55nb}, we call $\sqrt{p_c}\ket{c}\otimes \ket{\varphi_c}$ the component of $\ket{\varphi}$ when the register $\bbC$ is in value $c$.
\end{nota}
Then we give the notion of a state does not depend on (or is independent to) the value of some registers, as follows.
\begin{nota}\label{nota:3.16}
We say a purified joint state $\ket{\varphi}$ does not depend on the value of register $\bbC$ if it can be written as
$$\ket{\varphi}=\sum_{c\in \cC}\underbrace{\ket{c}}_{\bbC}\otimes \ket{\psi}$$	
\end{nota}


\subsubsection{Basic facts from linear algebra}\label{sec:factla}
The following facts from linear algebra will be used in the later proofs. These facts will be proved in Appendix \ref{app:a}.
\begin{fact}\label{fact:3}
If $|\ket{\varphi}|^2+|\ket{\phi}|^2\leq \frac{1}{2}$,  and $1-\epsilon\leq|\ket{\varphi}+\ket{\phi}|\leq 1$, then
$$\ket{\varphi}\approx_{\sqrt{\epsilon}}\ket{\phi}$$	
\end{fact}
The following two lemmas come from linear algebra and will be used in Section \ref{sec:11}.
\begin{fact}\label{fact:bntesttool1}
	Suppose $\vec{c}$, $\vec{d}$ are two vectors of non-negative real numbers of the same dimensions. Suppose there exists a vector $\vec{c^\prime}$ such that each coordinate of it is no bigger than the corresponding coordinate of $\vec{c}$, and $\vec{c^\prime}\approx_{\epsilon_1}\vec{d}$; and there exists a vector $\vec{d^\prime}$ such that each coordinate of it is no bigger than the corresponding coordinate of $\vec{d}$, and $\vec{d^\prime}\approx_{\epsilon_2}\vec{c}$. Then there is $\vec{c}\approx_{\sqrt{\epsilon_1^2+\epsilon_2^2}}\vec{d}$
\end{fact}

\begin{fact}\label{fact:bntesttool}
If non-negative real numbers $(c_d)_{d\in D}$ (where $D$ is the set of valid values of $d$) satisfy $$\sum_{d\in D}|c_d|^2\leq 1$$
	$$\frac{1}{{|D|}}\sum_{d_1\in D}\sum_{d_2\in D}|c_{d_1}-c_{d_2}|^2\approx_\epsilon 0$$ where $|D|$ is the size of set $D$. Then
	$$\sum_{d\in D} |c_{d}-\frac{1}{\sqrt{|D|}}c|^2\approx_{\min\{4c,4\epsilon/c+\frac{2}{D}\}} 0,\text{where }c:=\sqrt{\sum_{d\in D}|c_d|^2}$$
\end{fact}
The following lemma roughly says the normalization of the output of an efficiently preparable operator is also efficiently preparable. This lemma will be used in Section \ref{sec:13}.
\begin{fact}\label{fact:simfa}
(The states and operators in this fact are implicitly parameterized by $\kappa$.) Suppose an efficient quantum operation $O$ satisfies
$$O\ket{0}=\sum_{c\in \cC}\ket{c}\otimes \ket{\varphi_c}$$
where $\cC$ is a fixed set. Define $p_{c}=|\ket{\varphi_c}|^2$. Then for any $c\in \cC$, there exists an efficient quantum operation $\fSim^{O,c}$, a state $\ket{aux}$ such that 
$$\sqrt{p_c}\fSim^{O,c}\ket{0}\approx_{\fneg(\kappa)}\ket{\varphi_c}\otimes \ket{aux}$$	
\end{fact}

\subsubsection{Basic facts from probability theory}
\begin{lem}[Chernoff's bounds for streaming samples]\label{lem:chernoff}
	Consider a stream of samples $(s_i)_{i\in [N]}$ that are sampled sequentially. Each $s_i$ is sampled from $\{0,1\}$, and the probability of getting $1$ when the previous samples are $s_{<i}=s_1s_2\cdots s_{i-1}$ is $p_{i,s_{<i}}$. Suppose there exists a constant $p<1$ such that for each $i\in [N]$, each possible history of samples $s_{<i}$, there is $p_{i,s_{<i}}\leq p$. 
	Then
	$$\Pr[|\{i:s_i=1\}|\geq (1+\delta)pN]\leq e^{-\delta^2N/4}$$
\end{lem}
\begin{cor}[Chernoff's bounds for streaming samples, with noise]\label{cor:chernoff}
	Similarly consider a stream of samples $(s_i)_{i\in [N]}$ and define $p_{i,s_{<i}}$ similarly. Suppose there exists a constant $p<1$ such that for each $i\in [N]$, there exists a set of possible sample histories $S_i$, and: 
	\begin{itemize}
	\item The total probability that the sample sequences in $S_i$ appear is $\leq \epsilon$;	\item For each sample history $s_{<i}\not\in S_i$, there is $p_{i,s_{<i}}\leq p$.
	\end{itemize}
	Then
	$$\Pr[|\{i:s_i=1\}|\geq (1+\delta)pN]\leq e^{-\delta^2N/4}+N\epsilon$$
\end{cor}

\subsection{Noisy Trapdoor Claw-free Functions}\label{sec:ntcf}
We need to use the noisy trapdoor claw-free functions raised in \cite{BCMVV}. Note that we do not need the adaptive hardcore-bit property. Let's review the definition of $\fNTCF$ below.\footnote{Our formalization has a slightly different form from various existing works \cite{BCMVV,BKVV}; our formalism is no stronger than existing formalisms.}
\begin{defn}[$\fNTCF$]\label{defn:ntcf}
We define trapdoor claw-free function family $\fNTCF$ with post-quantum security as follows. It is parameterized by security parameter $\kappa$ and is defined to be a class of polynomial time algorithms as below. $\fNTCF.\fKg$ is a sampling algorithm. $\fNTCF.\fDc$, $\fNTCF.\fCHK$ are deterministic algorithms. $\fNTCF.\fEv$ is allowed to be a sampling algorithm. $ \fpoly^\prime$ is a polynomial that determines the the range size. $$\fNTCF.\fKg(1^\kappa)\rightarrow (\sk,\pk),$$ $$ \fNTCF.\fEv_\pk: \{0,1\}\times \{0,1\}^{\kappa}\rightarrow \{0,1\}^{\fpoly^\prime(\kappa)},$$ $$\fNTCF.\fDc_\sk: \{0,1\}\times \{0,1\}^{\fpoly^\prime(\kappa)}\rightarrow \{0,1\}^{\kappa}\cup \{\bot\},$$ $$\fNTCF.\fCHK_{\pk}: \{0,1\}\times \{0,1\}^{\kappa}\times \{0,1\}^{\fpoly^\prime(\kappa)}\rightarrow \{\ftrue ,\ffalse\}$$ And they satisfy the following properties:\par
\begin{itemize}
\item (Correctness) 
\begin{itemize}
\item (Noisy 2-to-1) For all possible $(\sk,\pk)$ in the range of $\fNTCF.\fKg(1^\kappa)$ there exists a sub-normalized probability distribution $(p_y)_{y\in \{0,1\}^{\fpoly^\prime(\kappa)}}$ that satisfies: for any $y$ such that $p_y\neq 0$, $\forall b\in \{0,1\}$, there is $\fNTCF.\fDc_\sk(b,y)\neq \bot$, and
$$\fNTCF.\fEv_\pk(\ket{+}^{\otimes \kappa})\approx_{\fneg(\kappa)}\sum_{y:p_y\neq 0}\frac{1}{\sqrt{2}}(\ket{\fNTCF.\fDc_\sk(0,y)}+\ket{\fNTCF.\fDc_\sk(1,y)})\otimes \sqrt{p_y}\ket{y}$$
\item (Correctness of $\fCHK$) For all possible $(\sk,\pk)$ in the range of $\fNTCF.\fKg(1^\kappa)$, $\forall x\in \{0,1\}^{\kappa}, \forall b\in \{0,1\}$: 	
$$\fNTCF.\fCHK_\pk(b,x,y)=\ftrue\Leftrightarrow\fNTCF.\fDc_{\sk}(b,y)=x$$

\end{itemize}
\item (Claw-free) For any BQP adversary $\fAdv$,
\begin{equation}\Pr\left[\begin{aligned}&(\sk,\pk)\leftarrow \fNTCF.\fKg(1^\kappa),\\&\fAdv(\pk,1^\kappa)\rightarrow (x_0,x_1,y):\quad x_0\neq \bot,x_1\neq \bot, x_0\neq x_1\\&\fNTCF.\fDc_\sk(0,y)=x_0,\fNTCF.\fDc_\sk(1,y)=x_1\end{aligned}\right]\leq \fneg(\kappa)\end{equation}
	\end{itemize}
\end{defn}
Then we have the following assumption about the existence of $\fNTCF$.
\begin{assump}\label{assump:1}
	There exists an efficient post-quantum $\fNTCF$ family.
\end{assump}
$\fNTCF$ can be instantiated using the Learning-with-Errors assumption: \cite{BCMVV}
\begin{thm}[Review of \cite{BKVV}]
	Assuming QLWE (post-quantum hardness of the Learing-with-Errors assumption) with suitable parameters, Assumption \ref{assump:1} holds.
\end{thm}
And we further note that, based on the construction in \cite{BCMVV,BKVV}, assuming a suitable version of hardness of Ring-LWE, the running time of $\fNTCF$ could be only $\tilde O(\kappa)$.
\subsubsection{Evaluation of $\fNTCF$ functions}
A typical protocol for $\fNTCF$ evaluation is as follows.
\begin{prtl}[$\fNTCF$ evaluation, review of subprotocols in \cite{BCMVV}]\label{prtl:ntcf}
	Suppose the security parameter is $\kappa$.
	\begin{enumerate}
		\item The client runs $\fNTCF.\fKg(1^\kappa)$ and gets $\sk,\pk$. Send $\pk$ to the server.
		\item The server evaluates $\fNTCF.\fEv_\pk(\ket{+}^{\otimes \kappa})$ and measures to get $y$ and $\frac{1}{\sqrt{2}}(\ket{x_0}+\ket{x_1})$ where $x_0\neq x_1$, $\fNTCF.\fCHK_\pk(0,x_0,y)=\fNTCF.\fCHK_\pk(1,x_1,y)=\ftrue$. Send $y$ back to the client.
		\item The client uses $\sk$ to decrypt and gets $\fNTCF.\fDc_\sk(0,y)=x_0,\fNTCF.\fDc_\sk(1,y)=x_1$.
	\end{enumerate}
\end{prtl}
\subsection{Random Oracle Model}\label{sec:3.3}
In this work we will use the quantum random oracle model (QROM). We give a simple review here.\par
The random oracle model is an ideal cryptographic model for symmetric encryption schemes or hash functions. \cite{roreliable} In this model there is a global oracle that encodes a random function in $\{0,1\}^*\rightarrow \{0,1\}^{\infty}$. Security notions in this model are typically defined with respect to attackers that could only query the oracle for polynomial (or subexponential) times. In practice the random oracle is instantiated by a symmetric encryption scheme or hash function, and the security of the protocol can be conjectured heuristically by the \emph{random oracle methodology}. Albeit there exist artificial uninstantiable constructions \cite{rorevisited,ES20}, this methodology turns out to be very successful in practice \cite{roreliable}: it has been used extensively in cryptography, and becomes the foundation of many famous protocols \cite{roreliable,GW11,DF13,CP18}.\par
The quantum random oracle model is raised as the quantum analog of the classical random oracle model. \cite{QRO} This model allows quantum access to the random oracle, which captures a natural analog of the random oracle model in quantum world. This model is also used in a series of works from post-quantum security of classical protocols \cite{cms,lz,dfms} to design of quantum protocols \cite{BKVV,CCT,ACGH}.
\paragraph{Input and output length of the random oracle} We do a cut-off on the input and output length of the random oracle to make its description finite. Parameterized by $\kappa$, we assume the maximal input length of the random oracle is $2^{10\kappa}$. Then we assume the maximal allowed output length for input $x$ is the square of the length of $x$. These are sufficient for our work.
\paragraph{Description of the random oracle} We consider the random oracle as a stand-alone party that holds a tuple of random strings. When the security parameter is $\kappa$, the content of the random oracle after the cut-off above could be expressed as a tuple $(H(x))_{x\in \{0,1\}^{10\kappa}}$ where each $H(x)\leftarrow_r \{0,1\}^{|x|^2}$ where $|x|$ is the length of $x$. With this tuple-description, we could explicitly say the output value of the random oracle for input $x$ is stored in register $\bH(x)$. And the tuple of all the random oracle content registers is denoted by $\bH$.

Especially, we do not need on-the-fly simulation techniques of the random oracle like \cite{compressedoracle}.

\paragraph{Register-oriented formalism of purified joint states in the quantum random oracle model} Recall that output values of $H$ are stored in a tuple of registers, whose purified state is:
$$\frac{1}{\sqrt{|\Domain(\bH)|}}\sum_{H\in \Domain(\bH)}\underbrace{\ket{H(1)}\ket{H(2)}\cdots \ket{H(d)}\cdots \ket{H(2^\kappa)}\cdots}_{\text{ random oracle outputs}}$$
where $H$ denotes the tuple of all the random oracle outputs, $H(x)$ denotes the values of the $x$-th coordinates of this tuple, and $\Domain(\bH)$ denotes the set of all the possible values of $\bH$. Note this is compatible with the usual formalism of the random oracle.\par
One property of these random oracle registers is that initially they are all set to hold uniformly distributed random values. Thus a purified joint state that can be prepared in the quantum random oracle model should satisfy the following property (Recall Notation \ref{nota:compo} for ``component''):
\begin{defn}\label{defn:validityro}
	We say a sub-normalized purified joint state is valid in the quantum random oracle model if the norm of any component when $\bH=H$ is no more than $\frac{1}{\sqrt{|\Domain(\bH)|}}$. 
\end{defn}
\begin{fact}
	If a sub-normalized purified joint state $\ket{\varphi}$ is valid in the quantum random oracle model, then $O\ket{\varphi}$ is valid in the quantum random oracle model where $O$ could contain any operation of the client and the server and random oracle queries.
\end{fact}
\paragraph{Notations for multiple registers}
\begin{nota}\label{nota:3.18}
We use $\bH(D)$ as a simplified notation for the tuple of all the registers $\bH(x),x\in D$. 
\end{nota}
\begin{nota}
Additionally, we introduce the following notation to denote a subset of the domain: as an example, $\{0,1\}^\kappa||K||\cdots$ where $K$ contains a tuple of keys, is defined to be the set of input entries in the following form: the first $\kappa$ bits are arbitrary; then the remaining part has a prefix equal to one of the keys in $K$.	
\end{nota}
\paragraph{Remark} In Notation \ref{nota:3.18} $D$ is fixed. But in proofs later $D$ might come out of the values of some registers $\bD$. But as long as $\bD$ is read-only we could still use $\bH(D)$ without problems, by interpreting ``$\bH(D)$ satisfies some properties'' as ``for each value $D$ of $\bD$, $\bH(D)$ satisfies some properties''. 
\subsubsection{Blinded oracle}
One tool that we need in this work is the \emph{blinded oracle}. The blinded oracle replaces the output values on some entries of the original oracle by freshly new values.
\begin{defn}[Blinded oracle]\label{defn:blind}
	Suppose $D$ is a subset of the input domain of the random oracle. We define \emph{the blinded oracle where entries in $D$ are blinded} as follows:\par
	Denote this blinded oracle by $\bH^{blind}$ (which could be understood as a tuple of registers in the random oracle party), for each query input $x$, take the output value (or, more precisely here, the register in the random oracle party that stores the output value) to be:
	\begin{itemize}
	\item If $x\in D$, $\bH^{blind}(x)$ is a new register that stores a freshly new random string (that is, after purification, a uniform superposition over all the possible outputs).
	\item If $x\not\in D$, $\bH^{blind}(x)$ is the same as $\bH(x)$.	
	\end{itemize}
\end{defn}
\subsubsection{Freshly-new oracle and approximate freshly-new oracle by random padding}
We need a way to say some part of the random oracle contains freshly new random strings. This will happen if this part of the random oracle is not queried. Following Notation \ref{nota:3.16}, we can say the overall state does not depend on the values of registers $\bH(D)$.

The following property will be very useful. Starting from an efficiently-preparable state, sampling a long enough padding could make the random oracle approximately freshly new on the salted inputs (which mean inputs with these paddings as prefixes):
\begin{lem}\label{lem:padro}
Suppose a sub-normalized purified joint state $\ket{\varphi}=O\ket{0}$ where $O$ is a polynomial-time operator (implicitly parameterized by security parameter $\kappa$). Then a tuple of random paddings, stored in register $\bpads$ that is read-only once initialized, is sampled as follows: $\bpads$ has $\fpoly(\kappa)$ sub-registers, and the value of each sub-register is sampled from $\{0,1\}^\kappa$ independently uniformly randomly. Thus the overall state is
\begin{equation}\label{eq:61tsl}\sum_{pads\in \Domain(\bpads)}\frac{1}{\sqrt{|\Domain(\bpads)|}}\underbrace{\ket{pads}}_{\bpads}\otimes\ket{\varphi}\end{equation}
 Then there exists an efficiently preparable state $\ket{\tilde\varphi}$ independent to $\bH(pads||\cdots)$ and $$\ket{\tilde\varphi}\approx_{\fneg(\kappa)}\eqref{eq:61tsl}$$
\end{lem}
We put a proof in Appendix \ref{app:a}.\par
\subsection{Lookup Tables and Phase Tables}
In this work we will need to use some simple lookup table notations.\par
Let's first formalize the symmetric encryption scheme that we will use in this work.
\begin{defn}\label{defn:enc} We formally define $\fEn$ in the quantum random oracle model as follows:
\begin{equation}\label{eq:enc}\fEn_k(p;1^\kappa):=(R,H(R||k)+ p); (R^\prime,H(R^\prime||k)),\quad R\leftarrow_r \{0,1\}^\kappa, R^\prime\leftarrow_r \{0,1\}^\kappa\end{equation}	
where the output length of $H(R||k)$ is the same as length of $p$ and the output length of $H(R^\prime||k)$ is $\kappa$.
\end{defn}
Note that we sometimes need to use multi-key encryption; we simply use the concatenated key as the encryption key.\par

Then we introduce the notation for look up tables:

\begin{defn}[Lookup tables]\label{defn:lt}
 $\fLT(x_1\rightarrow r_1,x_2\rightarrow r_2,\cdots x_D\rightarrow r_D;1^\kappa)$ is defined as the tuple $$(\fEn_{x_1}(r_1;1^\kappa),\fEn_{x_2}(r_2;1^\kappa),\cdots, \fEn_{x_D}(r_D;1^\kappa)).$$
\end{defn}
\section{Quantum Computation Verification: Problem Set-up}\label{sec:4r}
\subsection{Models of Protocol Formalizations}\label{sec:4.1}
\paragraph{Parties} The set-up of our protocol contains the following parties.
\begin{itemize}
\item Client (also called verifier);
\item Server (considered as the attacker in the malicious setting);
\item Random oracle: \par
As described in the introduction and Section \ref{sec:3.3}.
\item Transcript registers: holds the transcripts of the protocol. Both the client and the server could read all the transcript registers; both the client and the server use it to transmit messages by copying (bit-wise $\fCN$) their own registers into empty transcript registers; each of the transcript registers could only be written once and becomes read-only (for both the client and the server) after that.
\item Environment.
\end{itemize}
\paragraph{States} We use the following notations to describe the joint states of these parties.
\begin{itemize}
\item For describing the honest setting behavior we use the natural notation: for example, after one application of $\fNTCF$ (Protocol \ref{prtl:ntcf}) in the honest setting we say the client gets key pair $(x_0,x_1)$ and the server gets state $\frac{1}{\sqrt{2}}(\ket{x_0}+\ket{x_1})$.
\item To study the malicious setting we turn to use \emph{purified joint states} of all the parties to describe the overall states. In this notation everything including the client side keys is purified to an entangled state. (Recall the purification is in the sense of Notation \ref{nota:cq}.) 
\begin{exmp}[Purified joint states after an evaluation of $\fNTCF$] After an application of $\fNTCF$, we say a client holds a key pair in register $\bK=(\bx_0,\bx_1)$. (Recall that we use bold font for registers.) Then the purified joint state in the malicious setting could be expanded into
\begin{equation}\label{eq:53m}\sum_{(x_0,x_1)\in \Domain(\bK)}\underbrace{\underbrace{\ket{x_0}}_{\bx_0}\underbrace{\ket{x_1}}_{\bx_1}}_{\bK\text{(client side)}}\ket{\varphi_{x_0,x_1}}\end{equation}
where we explicitly write out the values of the client side register $\bK$, and the $\ket{\varphi_{x_0,x_1}}$ part contains states in all the other parts (all the other parties including the server, transcripts, environment, the random oracle, and also the registers of the client that are not $\bx_0,\bx_1$).
\end{exmp}
\end{itemize}
In the very beginning of our protocol, the client and server both are in all-zero states, and the randomness of the random oracle have been sampled out. When we describe this initial situation we simply use $\ket{0}$ and make the random oracle party implicit.\par
In general, if both parties execute a protocol $\fPrtl$ where the inputs are keys in client-side register $\bK$, has parameter $1^\kappa$ and the initial state is a purified joint state $\ket{\varphi}$ (for example, $\ket{0}$ described in the last paragraph), the final post-execution state $\ket{\varphi^\prime}$ could be denoted as:
$$\ket{\varphi^\prime}=\fPrtl^\fAdv(\bK;1^\kappa)\ket{\varphi}$$

\paragraph{Executions} In general, protocols in our work have the following structure. In each time step, one of the followings happen:
\begin{itemize}
	\item The server does some server-side operations. Here \emph{server-side operations} contain operations that have full access to the registers in the server party, have read-only access to the transcript party, and could query the random oracle.
	\item The server sends back a response to the client. Denote this operation as $\fResponse$, which copies a specific server-side register to a specific transcript register. 
	\item The client does some client-side operations on its own registers. Similarly \emph{client-side operations} are defined to be operations that have full access to the registers in the client party, have read-only access to the transcript party, and could query the random oracle.\par
		\item The client sends a message to the server. If the client sends the content of register $\bK$ to the server, we use $\odot \bK$ to denote this operation. (For example, if the initial state is $\ket{\varphi}$, the state after this sending message operation is $\ket{\varphi}\odot \bK$.)\par
		 If $\fAuxInf$ is an algorithm that takes client-side registers as its inputs, $\llbracket\fAuxInf\rrbracket$ denotes the output of this algorithm. Then for a protocol $\fPrtl$, we use $\llbracket\fPrtl\rrbracket$ to denote the tuple of for all the client-side messages generated by client-side operations in $\fPrtl$. Thus $\ket{\varphi}\odot \llbracket\fPrtl\rrbracket$ denotes the state after the client sends out all of its messages to the server. 
\end{itemize}
\paragraph{Outputs} In the end of the protocol, the outputs are:
\begin{itemize}
	\item The client outputs a \emph{flag} $\in\{\fpass,\ffail\}$ (into the transcript). The flag registers are denoted by symbol $\bflag$. The default value is $\fpass$.\par
	\item We will see, in some protocols the client could also output a \emph{score} $\in \{\fwin,\flose,\perp\}$ (into the transcript). The score registers are denoted by symbol $\bscore$. The default value of score registers is $\perp$ and when the client writes a score into it the value becomes either $\fwin$ or $\flose$.
	\item In the protocols designed in our work the client typically gets a tuple of keys or phases and the honest server gets some states.
\end{itemize}
There could be many subprotocol calls in our work  and each subprotocol call could output its own flag or score. We use subscripts (for example, $\bflag_1,\bflag_2,\bscore_1,\bscore_2$) to distinguish these different registers. Then we use, for example, $\Pi_{\fpass}^{\bflag_1}$ to denote the projection onto the space that the value of $\bflag_1$ is $\fpass$. We omit the superscripts here when the register that  $\Pi_{\fpass}$ is applied on could be determined from the context: for example, when we write $\Pi_{\fpass}\fPrtl\ket{\varphi}$, $\Pi_{\fpass}$ is applied onto the overall flag register of $\fPrtl$.
\subsection{Quantum Computation Verification}
In this subsection we formalize the quantum computation verification problem in our setting.\par
The following definition reviews the definition of quantum computation verification given in \cite{ABEM}, with simple adaptation to our setting:
\begin{defn}[Classical Verification of Quantum Computation, adapted from \cite{ABEM}]
We say a protocol that takes (1) a quantum circuit $C$ and a proposed output $o$, (2) a security parameter $1^\kappa$ as inputs, is a CVQC protocol with completeness $c$ and soundness $s$ against BQP adversaries in QROM if:
\begin{itemize}
\item (Completeness) For $(C,o)$ such that $\Pr[C\ket{0}=o]\geq \frac{99}{100}$, the verifier accepts with probability $\geq c-\fneg(\kappa)$.
\item (Soundness) For any malicious BQP server that makes at most polynomial query to the random oracle, for $(C,o)$ such that $\Pr[C\ket{0}=o]\leq \frac{1}{100}$, the verifier rejects with probability $\geq 1-s-\fneg(\kappa)$.
\item (Efficiency) Honestly both parties run in time polynomial in the size of the inputs.
\end{itemize}

\end{defn}
As said in the introduction, we construct a linear-time CVQC protocol. We repeat the main theorem here:
\begin{thm}[Repeat of Theorem \ref{thm:main} in Section \ref{sec:1}]\label{thm:mainf}
Assuming the existence of noisy trapdoor claw-free functions \cite{BCMVV}, there exists a single server CVQC protocol in QROM such that:
\begin{itemize}
\item The protocol has completeness $\frac{2}{3}$.
	\item For verifying a circuit of size $|C|$, the total time complexity is $O(\fpoly(\kappa)|C|)$, where $\kappa$ is the security parameter.
	\item The protocol has soundness $\frac{1}{3}$ against BQP adversaries in QROM. 
\end{itemize}
\end{thm}
\subsection{Existing Gadget-assisted Verification Protocol}
As reviewed in the introduction, to prove Theorem \ref{thm:mainf}, we make use of an existing quantum computation verification protocol \cite{FKD} where the server initially holds states
  \begin{equation}\label{eq:24}\ket{+_{\theta^{(1)}}}\otimes \ket{+_{\theta^{(2)}}}\otimes \cdots \otimes \ket{+_{\theta^{(L)}}} \end{equation}
  where each of $\theta^{(1)}\cdots \theta^{(L)}$ is uniformly independently random from $\{0,1\cdots 7\}$, and is known by the client. We formalize it as a theorem.
  \begin{thm}[\cite{FKD}]\label{thm:gauvbqc}There exists a quantum computation verification protocol such that, for any quantum circuit $C$, take $L=O(|C|)$, initially  the server holds \eqref{eq:24} where $\theta^{(1)}\cdots \theta^{(L)}$ are all independently random from $\{0,1\cdots 7\}$ and known by the client, and the protocol only uses classical interactions later, and it satisfies:
  \begin{itemize}
  	\item It has completeness $\frac{9}{10}$;
  	\item It has soundness $\frac{1}{10}$;
  	\item The time complexity is $O(|C|)$.
  \end{itemize}
  	
  \end{thm}
  The parameters needed could be obtained by choosing suitable parameters in the constructions and statements in \cite{FKD}\footnote{In fact, the construction of \cite{FKD} has perfect completeness; we do not need it here.}. Based on this protocol, what we need to do is to construct a protocol for remote preparation of gadgets \eqref{eq:24}. We formalize the notion of remote state preparation with verifiability (RSPV) for \eqref{eq:24}, as follows.
\subsection{Our Notion of RSPV}\label{sec:4.3}
We define our RSPV as follows. The target state, where the client holds phase tuple $\theta^{(1)}\theta^{(2)}\cdots \theta^{(L)}$, and server holds \eqref{eq:24}, could be written jointly as
\begin{equation}\label{eq:cqtarget}\sum_{\theta^{(1)}\theta^{(2)}\cdots \theta^{(L)}\in \{0,1\cdots 7\}^L}\frac{1}{8^L}\underbrace{\ket{\theta^{(1)}}\bra{\theta^{(1)}}\ket{\theta^{(2)}}\bra{\theta^{(2)}}\cdots \ket{\theta^{(L)}}\bra{\theta^{(L)}}}_{\text{client}}\otimes\underbrace{\ket{+_{\theta^{(1)}}}\otimes \ket{+_{\theta^{(2)}}}\otimes \cdots \otimes \ket{+_{\theta^{(L)}}}}_{\text{server}} \end{equation}
The RSPV protocol for \eqref{eq:cqtarget} takes a gadget number $1^L$ and a security parameter $1^\kappa$ as inputs. 
\begin{defn}[Correctness of RSPV]\label{defn:rspv}
We say an RSPV for a target state defined in \eqref{eq:cqtarget} has correctness in QROM if in the honest setting:
\begin{itemize}
\item The server could make the client outputs $\fpass$ with probability $\geq \frac{9}{10}-\fneg(\kappa)$;
\item In the honest setting, conditioned on the client outputs $\fpass$, with probability $\geq 1-\fneg(\kappa)$ the joint state of the client and the server is \eqref{eq:cqtarget}.	
\end{itemize}
\end{defn}
To define the verifiability, we note the purified target state could be written as
\begin{equation}\label{eq:target}\sum_{\theta^{(1)}\theta^{(2)}\cdots \theta^{(L)}\in \{0,1\cdots 7\}^L}\frac{1}{\sqrt{8^L}}\underbrace{\ket{\theta^{(1)}\theta^{(2)}\cdots \theta^{(L)}}}_{\text{client}}\otimes\underbrace{\ket{+_{\theta^{(1)}}}\otimes \ket{+_{\theta^{(2)}}}\otimes \cdots \otimes \ket{+_{\theta^{(L)}}}}_{\text{server}}\end{equation}
\begin{defn}[Verifiability of RSPV]\label{defn:rspvv}
The operators in this definition are implicitly parameterized by the security parameter $\kappa$.\par
We say an RSPV protocol $\fRSPV$ with target state \eqref{eq:target} has verifiability in QROM if for any polynomial time adversary $\fAdv$, 
there exists a server-side operation $\fSim^{\fAdv}$ such that:
$$\Pi_{\fpass}\fRSPV^\fAdv(1^L,1^\kappa)\ket{0}\approx^{ind}_{\frac{1}{5}+\fneg(\kappa)}\Pi_{\fpass}\fSim^{\fAdv} \ket{\text{Equation \eqref{eq:target}}}$$ 
where the distinguisher in the indistinguishability symbol could operate on the client-side registers\footnote{More naturally we can define the distinguisher to only has access to the $\btheta^{(1)}\cdots\btheta^{(L)}$ registers shown in \eqref{eq:target} for the client-side access. But we do not put this condition here for simplicity, and this is fine: we could always assume in the end of the protocol $\fRSPV$ the client disgards all the temporary registers (that are used during the protocol execution) into the environment, and only keeps the phase tuple registers in \eqref{eq:target}.}, the transcript and the server-side systems, only has read-only access to the client and the transcript, and is of polynomial time with polynomial random oracle queries.

\end{defn}
We clarify that $\fSim$ in the definition above could query the random oracle, could write into the transcript registers and could disgard registers to the environment, and is not required to be in polynomial time. (It is desirable to have this property, but we do not aim at it in this work.)
\paragraph{Remark}  We will make $\kappa$ implicit in the remaining definitions and theorems and will not repeat it every time.\par
\subsection{Pre-RSPV}\label{sec:4.5}
Aiming at constructing the RSPV protocol, it will be very convenient to define a relaxed notion that captures low-level details arose from the construction. We will introduce the notion of \emph{pre-RSPV}. We will see (1) it's easier to construct a pre-RSPV (compared to a direct construction of RSPV); (2) a repetition-based amplification of a pre-RSPV protocol leads to an RSPV protocol.\par 
 Compared to the RSPV protocol, pre-RSPV further relaxes the correctness and verifiability in the definition in the following way:
 \begin{itemize}
 \item It works under the following protocol design framework: in our protocol the client will secretly sample a round type in $\{\ftest,\fquiz,\fcomp\}$ in the beginning, with some fixed probability. These round types are revealed in the end of the protocol but not during the protocol. We use $\btype$ to denote the transcript register that holds the round type, $\Pi^{\btype}_\fcomp$ (for example) to denote the projection onto the space that $\btype$ register has value $\fcomp$, and $\Pi_{\fcomp}$ if there is no ambiguity on the register. (For example, when we write $\Pi_\fcomp\fpreRSPV$.)\par
 Only in the $\fcomp$ round the correct target state is generated (in the honest setting). In all rounds the client will produce a flag of $\fpass$ or $\ffail$; and only in the $\fquiz$ round the client could additionally generate a \emph{score} whose value is $\fwin$ or $\flose$. They have the following difference: if the server behaves honestly, it will $\fpass$ with probability $1-\fneg(\kappa)$. However, an honest server does not always $\fwin$. Instead, it $\fwin$ with probability $\OPT$ which is a fixed upper bound. Additionally, no malicious attacker could $\fwin$ with probability bigger than that. Thus when we construct the $\fRSPV$ protocol from repetitions of a $\fpreRSPV$ protocol, once the client sees some subprotocol $\ffail$, it could output $\ffail$ directly; however both parties have to repeat the protocol for many times so that the client can calculate the ratio of $\fwin$ and $\flose$ statistically.
 \item  Correspondingly, we need to adapt the definition of verifiability to take the probability of generating a score of $\fwin$ in the $\fquiz$ round into account.
 \end{itemize}
 Formally speaking, a pre-RSPV protocol is defined as follows. It takes a security parameter $1^\kappa$ and the output number $1^L$ as the inputs.
 \begin{defn}[Correctness of pre-RSPV]\label{defn:prerspv}
We say a pre-RSPV for the target state defined as Definition \ref{defn:rspv} has correctness in the quantum random oracle model if:\par
 In the honest setting, a round type $\in \{\ftest,\fquiz,\fcomp\}$ is sampled with fixed probabilities $\{p_\ftest,p_\fquiz,p_\fcomp\}$ and:
\begin{itemize}
\item In all round types, the probability that the client outputs $\fpass$ is $\geq 1-\fneg(\kappa)$.
\item In $\fquiz$ round, the probability that the client outputs $\fwin$ as the score is $\geq \OPT-\fneg(\kappa)$ where $\OPT$ is a constant.
\item In $\fcomp$ round, with probability $\geq 1-\fneg(\kappa)$ the target state \eqref{eq:cqtarget} is generated.
\end{itemize}

 \end{defn}
The constants $p_\ftest$, $p_{\fquiz}$, $p_{\fcomp}$ and $\OPT$ will be explicit when we formalize the correctness property of our concrete pre-RSPV protocol.
\begin{defn}[Verifiability of pre-RSPV]\label{defn:prerspvv}
	We say a protocol $\fpreRSPV$ is a pre-RSPV protocol for target state \eqref{eq:target} in QROM with error tolerance $(\epsilon_1,\epsilon_2)$ if:\par
	 For any polynomial time adversary $\fAdv$, any initial state $O\ket{0}$ where $O$ is efficient, at least one of the following three cases is true:
	\begin{itemize}
\item (Small passing probability) 
\begin{equation}\label{eq:59q}|\Pi_{\fpass}\fpreRSPV^\fAdv(1^L,1^\kappa)\circ O\ket{0}|^2\leq 1-\epsilon_1\end{equation}
\item (Small winning probability) 
\begin{equation}\label{eq:60q}|\Pi_{\fwin}\fpreRSPV^\fAdv(1^L,1^\kappa)\circ O\ket{0}|^2\leq p_\fquiz\cdot (\OPT-\epsilon_2)\end{equation}
where $\OPT$ is the same as the constant in the correctness (Definition \ref{defn:prerspv})
\item (Verifiability) There exists a server-side operation $\fSim^{\fAdv,O}$ such that:
\begin{equation}\label{eq:61q}\Pi_{\fcomp}\fpreRSPV^\fAdv(1^L,1^\kappa) \circ O\ket{0}\approx^{ind}_{\frac{1}{10}\sqrt{p_\fcomp}+\fneg(\kappa)}\sqrt{p_\fcomp}\fSim^{\fAdv,O} \ket{\text{Equation \eqref{eq:target}}}\end{equation}
where the distinguisher has read-only access to the client side registers and full access to the server-side registers and is polynomial time.
\end{itemize}
\end{defn}
Additionally, for the amplification from pre-RSPV to RSPV, we also need to require the constant $\OPT$ is indeed optimal:
\begin{defn}[Optimality of $\OPT$ in pre-RSPV]\label{defn:opt}
	We say a protocol $\fpreRSPV$ has optimal winning probability $\OPT$ with error tolerance $(\epsilon_1,\epsilon_2)$ if:\par
	 For any polynomial time adversary $\fAdv$, any initial state $O\ket{0}$ where $O$ is efficient, at least one of the following two cases is true:
	\begin{itemize}
\item (Small passing probability) 
\begin{equation}\label{eq:62q}|\Pi_{\fpass}\fpreRSPV^\fAdv(1^L,1^\kappa)\circ O\ket{0}|^2\leq 1-\epsilon_1\end{equation}
\item (Bounded winning probability) 
\begin{equation}\label{eq:63q}|\Pi_{\fwin}\fpreRSPV^\fAdv(1^L,1^\kappa)\circ O\ket{0}|^2\leq p_\fquiz\cdot (\OPT+\epsilon_2+\fneg(\kappa))\end{equation}

\end{itemize}
\end{defn}
\section{Formalization of Our Pre-RSPV Protocol}\label{sec:4}
In this subsection we formalize our construction of pre-RSPV protocol.
\subsection{High Level Construction and $\fAddPhaseWithswitch$}
 As said in the introduction, we will use $\fNTCF$-based techniques to create gadgets in the form of key pair superpositions, and use lookup-table-tabled techniques to add phases to them. Let's give a notation for such form of gadgets.
\begin{nota} For a key pair $K=(x_0,x_1)$, phase pair $\Theta=(\theta_0,\theta_1)$, define
$$\fgadget(K,\Theta)=\frac{1}{\sqrt{2}}(e^{\theta_0\mi\pi/4}\ket{x_0}+e^{\theta_1\mi\pi/4}\ket{x_1})$$
We call $\theta_1-\theta_0$ the \emph{relative phase}.\par
And define
$$\fgadget(K)=\frac{1}{\sqrt{2}}(\ket{x_0}+\ket{x_1})$$
\end{nota}

Now we formalize our pre-RSPV protocol below.
 \paragraph{Overall Protocol} This protocol is a pre-RSPV protocol with target state in the form of \eqref{eq:cqtarget}.
\begin{mdframed}
\begin{prtl}[Pre-RSPV]\label{prtl:prerspv}
	Suppose the security parameter is $\kappa$. Output number is $L$.\par
	\begin{enumerate}
	\item (Generation of key-pair-superpositions) 
	 As discussed in Section \ref{sec:ntcf}, both parties run $2+L$ blocks of $\fNTCF$ evaluations in parallel.\par
	 In the end the server gets $2+L$ key-pair-superpositions, and the client gets the corresponding key pairs. 
	The client names these key pairs as follows:
	\begin{itemize}
	\item The first key pair is denoted as $K^{(\switch)}=(x^{(\switch)}_0,x^{(\switch)}_1)$, $x^{(\switch)}_0\neq x^{(\switch)}_1$.
	\item The remaining $(1+L)$ key pairs are denoted as $K=(K^{(i)})_{i\in [0,L]}$; for each $i\in [0,L]$, $ K^{(i)}=( x_0^{(i)}, x_1^{(i)})$, $ x_0^{(i)}\neq  x_1^{(i)}$.
	\end{itemize}
	The honest server holds
	\begin{equation}\label{eq:25prep}\fgadget(K^{(\switch)})\otimes \fgadget(K^{(0)})\otimes  \fgadget(K^{(1)})\otimes  \fgadget(K^{(2)})\otimes \cdots\fgadget(K^{(L)}) \end{equation}

	\item The client randomly chooses to run one of the following two with the server:
	\begin{itemize}
	\item (Standard basis test) Both parties execute $\fStdBTest((K^{(i)})_{i\in \{\switch\}\cup [0,L]})$. 

	\item (Verifiable state preparation) 
	
	\begin{enumerate}
	\item (State transformation with the switch gadget) For each $i\in [0,L]$, the client samples $\Theta^{(i)}=(\theta^{(i)}_0,\theta_1^{(i)})\leftarrow_r \{0,1\cdots 7\}^2$. Denote $\Theta =(\Theta^{(i)})_{i\in [0,L]}$.\par
	Both parties execute $\fAddPhaseWithswitch((K^{(\switch)},K),\Theta;1^\kappa)$, which consumes the switch gadget.\par
	After this step the honest server's state is: \begin{equation}\label{eq:27new}\fgadget(K^{(0)},\Theta^{(0)})\otimes\fgadget(K^{(1)},\Theta^{(1)})\otimes \fgadget(K^{(2)},\Theta^{(2)})\otimes\cdots\otimes \fgadget(K^{(L)},\Theta^{(L)})\end{equation}
	\item The client randomly chooses to run one of the following five subprotocols with the server:
	\begin{itemize}
	\item (Standard basis test) Both parties execute $\fStdBTest(K)$.
		\item (Collective phase test) 
	 Both parties execute $\fCoPhTest(K,\Theta;1^\kappa)$.
		\item (Individual phase test)  Both parties execute $\fInPhTest(K,\Theta;1^\kappa)$. This is considered as the $\fquiz$ round and the client will possibly write $\fwin$ or $\flose$ in the score register.
		\item (basis uniformity test) 
		Both parties execute $\fBNTest((K^{(i)})_{i\in [L]},(\Theta^{(i)})_{i\in [L]};1^\kappa)$ on these states.
\item (Output states) The client reveals $(K^{(i)})_{i\in [L]}$ and the honest server could decode the gadgets and get the following state up to a global phase: 
$$\ket{+_{\theta^{(1)}}}\otimes \ket{+_{\theta^{(2)}}}\otimes \cdots \otimes \ket{+_{\theta^{(L)}}} $$
where $\theta^{(i)}=\theta^{(i)}_1-\theta^{(i)}_0$ for each $i\in [L]$. And the client could also calculate these phases from $\Theta$. All the other client-side registers are disgarded.\par

This is considered as the $\fcomp$ round.

\end{itemize}
 \end{enumerate}
 \end{itemize}
	\end{enumerate}

\end{prtl}
\end{mdframed}
The subprotocols used in this protocol are formalized below. We first formalize the standard basis test and the $\fAddPhaseWithswitch$ step.\par
\paragraph{Standard basis test} The $\fStdBTest$ is formalized below. Note in this protocol we use $D$ to denote the set of indices; when we use this protocol it could be $\{\switch\}\cup [0,L]$ (as in the first usage in Protocol \ref{prtl:prerspv}) or $[0,L]$ (as in the first case in step 2.b in Protocol \ref{prtl:prerspv}).
\begin{mdframed}
\begin{prtl}[$\fStdBTest$]
	The client holds a tuple of key pairs $(K^{(i)})_{i\in D}$ where $D$ is a set of indices. For each $i\in D$, the honest server holds state $\fgadget(K^{(i)},\cdots)$ for some phase pair omitted in ``$\cdots$''  (where the values of these phase pairs do not have influence in this protocol).
	\begin{enumerate}
		\item The client asks the server to measure all the gadgets in the standard basis. For each $i\in D$, the server measures $\fgadget(K^{(i)},\cdots)$ and sends back the response $r^{(i)}$.
		\item The client checks for each $i\in D$, $r^{(i)}\in K^{(i)}$.
	\end{enumerate}
\end{prtl}
	
\end{mdframed}

\paragraph{Add phases under the switch gadget technique} Below we formalize the state transformation subprotocol, which guarantees the honest server could add phases to \eqref{eq:25prep}. This step will use the switch gadget technique.
\begin{mdframed}
\begin{prtl}[$\fAddPhaseWithswitch$]\label{prtl:stwh} 
Suppose the security parameter is $\kappa$. Gadget number is controlled by $L$. \par
The client holds a tuple of key pairs $K^{(\switch)}=(x^{(\switch)}_0,x_1^{(\switch)})$, $K=(K^{(i)})_{i\in [0,L]}$, $K^{(i)}=(x_0^{(i)},x_1^{(i)})$ and a tuple of phase pairs $\Theta=(\Theta^{(i)})_{i\in [0,L]}$, $\Theta^{(i)}=(\theta_0^{(i)},\theta_1^{(i)})$.\par
Honest server holds: 		 	$$\fgadget(K^{(\switch)})\otimes \fgadget(K^{(0)})\otimes  \fgadget(K^{(1)})\otimes  \fgadget(K^{(2)})\otimes \cdots\fgadget(K^{(L)}) $$
\begin{enumerate}
\item (Add phases)  For each $i\in [0,L]$, the client prepares the following table and sends it to the server:
$$\fLT(x_{0}^{(\switch)}x^{(i)}_{0}\rightarrow \theta^{(i)}_{0},x_{1}^{(\switch)}x^{(i)}_{0}\rightarrow \theta^{(i)}_{0},x_{0}^{(\switch)}x^{(i)}_{1}\rightarrow \theta^{(i)}_{1},x_{1}^{(\switch)}x^{(i)}_{1}\rightarrow \theta^{(i)}_{1};1^\kappa)$$
The honest server should do the following mapping for each $i\in [0,L]$ to add the phases:
\begin{align}
&(\ket{x_0^{(\switch)}}+\ket{x_1^{(\switch)}})\otimes (\ket{x_0^{(i)}}+\ket{x_1^{(i)}})\\
\rightarrow & (\ket{x_0^{(\switch)}}+\ket{x_1^{(\switch)}})\otimes (\ket{x_0^{(i)}}\ket{\theta_0^{(i)}}+\ket{x_1^{(i)}}\ket{\theta_1^{(i)}})\\
\rightarrow & (\ket{x_0^{(\switch)}}+\ket{x_1^{(\switch)}})\otimes (e^{\theta_0^{(i)}\mi\pi/4}\ket{x_0^{(i)}}\ket{\theta_0^{(i)}}+e^{\theta^{(i)}_1\mi\pi/4}\ket{x_1^{(i)}}\ket{\theta_1^{(i)}})\\
\rightarrow & (\ket{x_0^{(\switch)}}+\ket{x_1^{(\switch)}})\otimes (e^{\theta_0^{(i)}\mi\pi/4}\ket{x_0^{(i)}}+e^{\theta^{(i)}_1\mi\pi/4}\ket{x_1^{(i)}})
\end{align}
\item Both parties execute $\fHadamardTest(K^{(\switch)};1^\kappa)$ (formalized below) on the switch gadget. 
\end{enumerate}
An honest server holds the following state in the end:
$$\fgadget(K^{(0)},\Theta^{(0)})\otimes\fgadget(K^{(1)},\Theta^{(1)})\otimes \fgadget(K^{(2)},\Theta^{(2)})\otimes\cdots\otimes \fgadget(K^{(L)},\Theta^{(L)})$$
\end{prtl}
	
\end{mdframed}

%
%
The Hadamard test with random oracle padding is defined as follows. 
\subsection{Subprotocols: Hadamard Tests and Gadget Combination}
\paragraph{RO-padded Hadamard test} The Hadamard tests used in our protocols are defined below.  For this work, we need to define different versions of RO-padded Hadamard tests, which deal with the extra phases in different ways. In each of these tests, the client asks the server to add an extra random oracle paddings before the Hadamard operation. As discussed in the introduction and \cite{jiayu20}, this allows the client to control the server-side states in a way that the un-padded version could not give. 
\begin{itemize}
\item (Unphased Hadamard test) $\fHadamardTest(K;1^\kappa)$, where $K$ is a key pair: this is the most basic form of Hadamard test used in our protocol. Here the phases are trivial: In this protocol the honest server is suppose to hold $\fgadget(K)$ in the beginning. This test has only one-sided error in the sense that if the honest server could pass with probability $1$ thus if the server fails the cheating behavior will be caught immediately. This version will be used in the $\fAddPhaseWithswitch$ and the $\fBNTest$. 
\item (Phased Hadamard test) $\fHadamardTest(K,\Theta;1^\kappa)$, where $K$ is a key pair, $\Theta$ is a phase pair: in this version, the honest server is suppose to hold $\fgadget(K,\Theta)$ in the beginning. The client will first reveal the relative phase to allow the server to dephase the gadget and run the unphased Hadamard test in the previous bullet. This version will be used in the $\fCoPhTest$.
\item (Extra-phase-biased Hadamard test) $\fHadamardTest(K,\Theta,\delta;1^\kappa)$, where $K$ is a key pair, $\Theta$ is a phase pair, $\delta\in \{0,4,1\}$: similar to the previous bullet, the honest server is suppose to hold $\fgadget(K,\Theta)$. Different from the previous bullet, when the client reveals the phase, it adds an \emph{extra-phase-bias} $\delta$ ($\delta$ itself is hidden from the server): suppose $\Theta=(\theta_0,\theta_1)$, it will reveal $\theta_1-\theta_0-\delta$. Then corresponding to different $\delta$, the client will produce results in different ways:
\begin{itemize}
	\item If $\delta=0$, the behavior of the protocol is the same as the following protocol. The honest server will first de-phase the gadget as the last version and both parties do unphased Hadamard test on $\fgadget(K)$. 
	\item If $\delta=4$, the dephasing will give $\ket{x_0}-\ket{x_1}$  (up to a global phase, where $K=(x_0,x_1)$). Then the unphased Hadamard test on this state gives opposite $\fpass/\ffail$ flag to $\ket{x_0}+\ket{x_1}$. Thus for this case the client will reverse the $\fpass/\ffail$ from the $\delta=0$ case.
	\item If $\delta=1$:  the dephasing will give $\ket{x_0}+e^{\mi\pi/4}\ket{x_1}$. Then the server will go through an unphased Hadamard test from this state (note that $\delta$ is hidden from the server so the server does not know the state description it holds). The client will  use the equation deciding the $\fpass/\ffail$ flag from the unphased version, but it will only record $\fwin$ or $\flose$ in this case. \par
	 The honest server could also $\flose$ in this test with some probability, and what we want to guarantee is the winning probability is not far from the optimal value.
\end{itemize}

 This version will be used in the $\fInPhTest$.
\end{itemize}

\begin{mdframed}
\begin{prtl}[$\fHadamardTest(K;1^\kappa)$, unphased]\label{prtl:phadamard}
Suppose the security parameter is $\kappa$.\par
Client holds a pair of keys $K=(x_0,x_1)$, where the length of each key is $|x|$.\par
Honest server should hold state $\frac{1}{\sqrt{2}}(\ket{x_0}+\ket{x_1})$.\par
\begin{enumerate}
\item The client samples $\pad\leftarrow_r \{0,1\}^\kappa$. Send it to the server.
\item The server is suppose to use the padding to map the state into
$$\frac{1}{\sqrt{2}}(\ket{x_0}\underbrace{\ket{H(\pad || x_0)}}_{\kappa\text{ qubits}}+\ket{x_1}\ket{H(\pad || x_1)})$$
and do Hadamard measurements on all the qubits above.\par
Suppose the measurement result is $d\in \{0,1\}^{|x|+\kappa}$. The server sends back $d$.
\item The client sets $\bflag=\ffail$ if the last $\kappa$ bits of $d$ are all-zero.\par
Otherwise the client calculates
 \begin{equation}\label{eq:htt}d\cdot (x_0||H(pad||x_0))+d\cdot (x_1||H(pad||x_1))\mod 2\end{equation} If \eqref{eq:htt} is $0$, set $\bflag=\fpass$; if \eqref{eq:htt} is $1$, set $\bflag=\ffail$.
\end{enumerate}

\end{prtl}
\begin{prtl}[$\fHadamardTest(K,\Theta;1^\kappa)$, phased]
Suppose the security parameter is $\kappa$.\par
Client holds a pair of keys $K=(x_0,x_1),\Theta=(\theta_0,\theta_1)$.\par
Honest server should hold state $\fgadget(K,\Theta)$.\par

\begin{enumerate}
	\item The client reveals the relative phase $\theta_1-\theta_0$ to the server. The server could de-phase and get the state $\fgadget(K)$ up to a global phase.
	\item Both parties run $\fHadamardTest(K;1^\kappa)$ on $\fgadget(K)$.
\end{enumerate}
	
\end{prtl}

	\begin{prtl}[$\fHadamardTest(K,\Theta,\delta;1^\kappa)$, the Hadamard test with extra phase bias]\label{prtl:hadamard}
		Suppose the security parameter is $\kappa$.\par
Client holds a pair of keys $K=(x_0,x_1)$, a pair of phases $\Theta=(\theta_0,\theta_1)$, and an extra phase bias $\delta\in \{0,4,1\}$.\par
Honest server should hold state $\fgadget(K,\Theta)$.
\begin{enumerate}
\item The client reveals $\theta_1-\theta_0-\delta$ to the server. The honest server could use it to transform $\fgadget(K,\Theta)$ to the following state up to a global phase:
\begin{equation}\label{eq:30ht}\frac{1}{\sqrt{2}}(\ket{x_0}+e^{\delta\mi\pi/4}\ket{x_1})\end{equation}
\item Both parties execute the non-phase-bias Hadamard test $\fHadamardTest(K;1^\kappa)$ above, where the server uses \eqref{eq:30ht}. This subprotocol call returns a measurement result $d$ (see \eqref{eq:htt}) to the client. The client sets $\bflag=\ffail$ if the last $\kappa$ bits of $d$ are all-zero. Otherwise, depending on the value of $\delta$, the client determines the output flag and score as follows:
\begin{itemize}
	\item If $\delta=0$, the client sets $\bflag=\fpass$ if \eqref{eq:htt} is $0$ and $\bflag=\ffail$ if \eqref{eq:htt} is $1$.
	\item If $\delta=4$, the client sets $\bflag=\ffail$ if \eqref{eq:htt} is $0$ and $\bflag=\fpass$ if \eqref{eq:htt} is $1$.
	\item If $\delta=1$, the client sets $\bscore=\fwin$ if \eqref{eq:htt} is $0$ and $\bscore=\flose$ if \eqref{eq:htt} is $1$.
\end{itemize}
The unspecified flag is $\fpass$ and the unspecified score is $\perp$ by default.
\end{enumerate}

	\end{prtl}

	
\end{mdframed}
To formalize the collective phase test and basis uniformity test, we will use subprotocols for combining multiple gadgets into one single gadget. This is defined as follows.
\paragraph{$\fCombine$: subprotocols for collective phase tests and basis uniformity tests} $\fCombine$ subprotocols combines multiple gadgets to a single gadget. These subprotocols will be used in $\fCoPhTest$ and $\fBNTest$:
\begin{itemize}
\item In $\fCoPhTest$, the indices of gadgets to be combined are $[0,L]$. Correspondingly, $K=(K^{(i)})_{i\in [0,L]}$, $\Theta=(\Theta^{(i)})_{i\in [0,L]}$. The protocol is denoted by $\fCombine(K,\Theta;1^\kappa)$.
\item In $\fBNTest$, $\Theta$ are taken to be all-zero, and the indices of gadgets to be combined are $[L]$. The protocol is denoted by $\fCombine(\tilde K,I;1^\kappa)$, where $\tilde K=(K^{(i)})_{i\in [L]}$. Here $I$ is a subset of $[L]$.

\end{itemize}
\begin{mdframed}
\begin{prtl}[$\fCombine$ for $\fCoPhTest$]\label{prtl:combine}
	Suppose the security parameter is $\kappa$.\par
	Client holds a tuple of key pairs $K=(K^{(i)})_{i\in [0,L]};K^{(i)}=(x_0^{(i)},x_1^{(i)})$. Correspondingly it also holds a tuple of phase pairs $\Theta=(\Theta^{(i)})_{i\in [0,L]};\Theta^{(i)}=(\theta_0^{(i)},\theta_1^{(i)})$. Each phase is in $\{0,1\cdots 7\}$.\par
	Honest server should hold
	$$\otimes_{i\in [0,L]}\fgadget(K^{(i)},\Theta^{(i)})$$
	\begin{enumerate}
		\item For each $i\in [L]$:\par
		The client samples different $r_0^{(i)},r_1^{(i)}\leftarrow \{0,1\}^\kappa$, prepares the table
		$$\fLT^{(i)}:=\fLT(x_0^{(0)}x^{(i)}_{0}\rightarrow r_0^{(i)},x_{1}^{(0)}x^{(i)}_{1}\rightarrow r_0^{(i)},x_{0}^{(0)}x^{(i)}_{1}\rightarrow r_1^{(i)},x_{1}^{(0)}x^{(i)}_{0}\rightarrow r_1^{(i)};1^\kappa)$$
		and sends the table to the server.
		\item The client defines $K^{(combined)}=(x_0^{(combined)},x_1^{(combined)})$, sets it to be $K^{(0)}$ in the beginning. And it defines $\Theta^{(combined)}=(\theta_0^{(combined)},\theta_1^{(combined)})$, sets it to be $\Theta^{(0)}$ in the beginning.\par
		For each $i\in [L]$:
	\begin{enumerate}
	\item If the server is honest, by this time it should hold $\fgadget(K^{(combined)},\Theta^{(combined)})$ and $\fgadget(K^{(i)},\Theta^{(i)})$. The keys in $K^{(combined)}$ have prefix in $K^{(0)}$. It will combine these two gadgets into a single gadget by decrypting $\fLT^{(i)}$ and measures to get $r^{(i)}\in \{r_0^{(i)},r_1^{(i)}\}$. In more detail, the following operations are applied by the honest server:
		\begin{align}
		&\fgadget(K^{(combined)},\Theta^{(combined)})\otimes\fgadget(K^{(i)},\Theta^{(i)})\\
		=&\frac{1}{2}(e^{\theta_0^{(combined)}\mi\pi/4}\ket{x_0^{(combined)}}+e^{\theta_1^{(combined)}\mi\pi/4}\ket{x_1^{(combined)}})\otimes (e^{\theta_0^{(i)}\mi\pi/4}\ket{x_0^{(i)}}+e^{\theta_1^{(i)}\mi\pi/4}\ket{x_1^{(i)}})\\
		=&\frac{1}{2}\sum_{b^{(combined)}b^{(i)}\in \{0,1\}^2}e^{(\theta_{b^{(combined)}}^{(combined)}+\theta_{b^{(i)}}^{(i)})\mi\pi/4}\ket{x_{b^{(combined)}}^{(combined)}x_{b^{(i)}}^{(i)}}\\
		&\text{(Decrypt $\fLT^{(i)}$ with $K^{(0)}$ (in the prefix of $K^{(combined)}$) and $K^{(i)}$):}\\
		\rightarrow &\frac{1}{2}\sum_{b^{(combined)}b^{(i)}\in \{0,1\}^2}e^{(\theta_{b^{(combined)}}^{(combined)}+\theta_{b^{(i)}}^{(i)})\mi\pi/4}\ket{x_{b^{(combined)}}^{(combined)}x_{b^{(i)}}^{(i)}}\ket{r^{(i)}_{b^{(combined)}+b^{(i)}}}\\
		\rightarrow &\text{measure and get $r^{(i)}\in \{r_0^{(i)},r_1^{(i)}\}$}:\\
		&(r^{(i)}=r_0^{(i)}): e^{(\theta_0^{(combined)}+\theta_0^{(i)})\mi\pi/4}\ket{x_0^{(combined)}x_0^{(i)}}+e^{(\theta_1^{(combined)}+\theta_1^{(i)})\mi\pi/4}\ket{x_1^{(combined)}x_1^{(i)}}\label{eq:35cb}\\
		&(r^{(i)}=r_1^{(i)}): e^{(\theta_0^{(combined)}+\theta_1^{(i)})\mi\pi/4}\ket{x_0^{(combined)}x_1^{(i)}}+e^{(\theta_1^{(combined)}+\theta_0^{(i)})\mi\pi/4}\ket{x_1^{(combined)}x_0^{(i)}}\label{eq:36cb}\\
		\end{align}
The server sends back $r^{(i)}$ to the client.\par
\item	The client checks the server's response $r^{(i)}$ is in $\{r_0^{(i)},r_1^{(i)}\}$ and stores
\begin{itemize}
	\item If $r^{(i)}=r_0^{(i)}$:
	$$K^{(combined)}=(x_0^{(combined)}x_0^{(i)},x_1^{(combined)}x_1^{(i)}),$$
	$$\Theta^{(combined)}=(\theta_0^{(combined)}+\theta_0^{(i)},\theta_1^{(combined)}+\theta_1^{(i)})$$
	\item If $r^{(i)}=r_1^{(i)}$:
	$$K^{(combined)}=(x_0^{(combined)}x_1^{(i)},x_1^{(combined)}x_0^{(i)}),$$
	$$\Theta^{(combined)}=(\theta_0^{(combined)}+\theta_1^{(i)},\theta_1^{(combined)}+\theta_0^{(i)})$$
	Correspondingly the honest server's state is
		$\fgadget(K^{(combined)},\Theta^{(combined)})$.
\end{itemize}

		
		\end{enumerate}
In the end all these gadgets are combined together; the client holds $K^{(combined)}$, $\Theta^{(combined)}$ and the server holds $\fgadget(K^{(combined)},\Theta^{(combined)})$. 
 
	\end{enumerate}

\end{prtl}

\end{mdframed}
\begin{mdframed}
\begin{prtl}[$\fCombine$ for $\fBNTest$]\label{prtl:combine2}
	Suppose the security parameter is $\kappa$. $I$ is a tuple of indices $i_1i_2\cdots i_{|I|}$ which is a subset of $[L]$ arranged in increasing order.\par
	Client holds a tuple of key pairs $\tilde K=(K^{(i)})_{i\in [L]};K^{(i)}=(x_0^{(i)},x_1^{(i)})$. 
	Honest server should hold
	$$\otimes_{i\in [L]}\fgadget(K^{(i)})$$
	\begin{enumerate}
		\item For each $i\in i_2,i_3\cdots i_{|I|}$:\par
		The client samples different $r_0^{(i)},r_1^{(i)}\leftarrow \{0,1\}^\kappa$, prepares the table
		$$\fLT^{(i)}:=\fLT(x_0^{(i_1)}x^{(i)}_{0}\rightarrow r_0^{(i)},x_{1}^{(i_1)}x^{(i)}_{1}\rightarrow r_0^{(i)},x_{0}^{(i_1)}x^{(i)}_{1}\rightarrow r_1^{(i)},x_{1}^{(i_1)}x^{(i)}_{0}\rightarrow r_1^{(i)};1^\kappa)$$
		and sends the table to the server.
	\item The client defines $K^{(combined)}=(x_0^{(combined)},x_1^{(combined)})$, sets it to be $K^{(i_1)}$ in the beginning.\par	
	For each $i\in i_2,i_3\cdots i_{|I|}$:
	\begin{enumerate}
	\item If the server is honest, by this time it should hold $\fgadget(K^{(combined)})$ and $\fgadget(K^{(i)})$. The keys in $K^{(combined)}$ have prefix in $K^{(i_1)}$. It will combine these two gadgets into a single gadget by decrypting $\fLT^{(i)}$ and measures to get $r^{(i)}\in \{r_0^{(i)},r_1^{(i)}\}$. In more detail, the following operations are applied by the honest server:

				\begin{align}
		&\fgadget(K^{(combined)})\otimes\fgadget(K^{(i)})\\
		=&\frac{1}{2}(\ket{x_0^{(combined)}}+\ket{x_1^{(combined)}})\otimes (\ket{x_0^{(i)}}+\ket{x_1^{(i)}})\\
		=&\frac{1}{2}\sum_{b^{(combined)}b^{(i)}\in \{0,1\}^2}\ket{x_{b^{(combined)}}^{(combined)}x_{b^{(i)}}^{(i)}}\\
		&\text{(Decrypt $\fLT^{(i)}$ with $K^{(i_1)}$ (in the prefix of $K^{(combined)}$) and $K^{(i)}$):}\\
		\rightarrow &\frac{1}{2}\sum_{b^{(combined)}b^{(i)}\in \{0,1\}^2}\ket{x_{b^{(combined)}}^{(combined)}x_{b^{(i)}}^{(i)}}\ket{r^{(i)}_{b^{(combined)}+b^{(i)}}}\\
		\rightarrow &\text{measure and get $r^{(i)}\in \{r_0^{(i)},r_1^{(i)}\}$}:\\
		&(r^{(i)}=r_0^{(i)}): \ket{x_0^{(combined)}x_0^{(i)}}+\ket{x_1^{(combined)}x_1^{(i)}}\label{eq:35cb}\\
		&(r^{(i)}=r_1^{(i)}): \ket{x_0^{(combined)}x_1^{(i)}}+\ket{x_1^{(combined)}x_0^{(i)}}\label{eq:36cb}\\
		\end{align}
The server sends back $r^{(i)}$ to the client.\par
\item	The client checks the server's response $r^{(i)}$ is in $\{r_0^{(i)},r_1^{(i)}\}$ and stores
\begin{itemize}
	\item If $r^{(i)}=r_0^{(i)}$:
	$$K^{(combined)}=(x_0^{(combined)}x_0^{(i)},x_1^{(combined)}x_1^{(i)}),$$
	\item If $r^{(i)}=r_1^{(i)}$:
	$$K^{(combined)}=(x_0^{(combined)}x_1^{(i)},x_1^{(combined)}x_0^{(i)}),$$
\end{itemize}
Correspondingly, the honest server's state is
		$\fgadget(K^{(combined)})$.
\end{enumerate}

	In the end both parties combine the gadgets with indices in $I$ into a single gadget, the client holds $K^{(combined)}$ and the server holds $\fgadget(K^{(combined)})$. The gadgets with indices outside $I$ remain unchanged.
	\end{enumerate}
\end{prtl}
	
\end{mdframed}

\subsection{Sub-tests}

\paragraph{Collective Phase Test} The collective phase test is formalized as follows. The input gadgets are indexed by $[0,L]$. The client first runs $\fCombine$ to combine all these $1+L$ gadgets into a single gadget. Then one of the following two is randomly selected:
\begin{itemize}
	\item A standard basis test of the combined keys;
	\item As the main step of this test, both parties do an (RO-padded) Hadamard test on the combined gadget. 
\end{itemize}
\begin{mdframed}
	\begin{prtl}[$\fCoPhTest$]
		Suppose the security parameter is $\kappa$ and the gadget number is controlled by $L$.\par
Client holds a tuple of key pairs $K=(K^{(i)})_{i\in [0,L]}$ a tuple of phase pairs $\Theta=(\Theta^{(i)})_{i\in [0,L]}$.\par
Honest server holds:
$$\otimes_{i\in [0,L]}\fgadget(K^{(i)},\Theta^{(i)})$$
\begin{enumerate}
\item 
Both parties run $\fCombine(K,\Theta;1^\kappa)$. The client gets $K^{(combined)}$ and $\Theta^{(combined)}$ and the server gets $\fgadget(K^{(combined)},\Theta^{(combined)})$.
%
\item The client chooses to run one of the following two randomly, and the honest server could use $\fgadget(K^{(combined)},\Theta^{(combined)})$ to pass the tests:
\begin{itemize}
\item $\fStdBTest(K^{(combined)})$.
\item  $\fHadamardTest(K^{(combined)},\Theta^{(combined)};1^\kappa)$.
\end{itemize}
\end{enumerate}

	\end{prtl}
\end{mdframed}

\paragraph{Individual Phase Test} The individual phase test is defined as follows.
\begin{mdframed}
\begin{prtl}[$\fInPhTest$]\label{prtl:rephtest}
			Suppose the security parameter is $\kappa$ and the gadget number is $L$.\par
Client holds a tuple of key pairs $K=(K^{(i)})_{i\in [0,L]}$ a tuple of phase pairs $\Theta=(\Theta^{(i)})_{i\in [0,L]}$.\par
Honest server holds:
$$\otimes_{i\in [0,L]}\fgadget(K^{(i)},\Theta^{(i)})$$
This protocol will only use the first gadget, which corresponds to $\fgadget(K^{(0)},\Theta^{(0)})$.
\begin{enumerate}
\item With probability $\frac{1}{3}$ each, the client executes the following with the server without telling the server which is the case:
\begin{itemize}
\item Both parties execute $\fHadamardTest(K^{(0)},\Theta^{(0)},0;1^\kappa)$;
\item Both parties execute $\fHadamardTest(K^{(0)},\Theta^{(0)},4;1^\kappa)$;
\item Both parties execute $\fHadamardTest(K^{(0)},\Theta^{(0)},1;1^\kappa)$; Note ``$\fwin$'' or ``$\flose$'' are recorded as the score corresponding to the client's checking result.
	
\end{itemize}

\end{enumerate}

\end{prtl}
	
\end{mdframed}
\paragraph{basis uniformity test} Finally we formalize the basis uniformity test as follows. The input gadgets are indexed by $1$ to $L$. The client selects a random subset of index $I$ from $[L]$, which represents the indices of the gadgets that will be used for the combine-and-test process; then it uses $\fCombine$ to combine these gadgets into a single gadget. Then:
\begin{enumerate}
\item For the gadgets outside $I$, the client will ask the server to measure them in the standard basis.
\item For the combined part, one of the following two is randomly selected:
\begin{itemize}
\item A standard basis test on the combined gadget.	
\item As the main step of this test, both parties execute a Hadamard test on this combined gadget.\par
\end{itemize}

\end{enumerate}
Note that the main body of $\fBNTest$ is on the gadgets without phases ($\fgadget(K^{(\cdots)})$). But when we use the $\fBNTest$ in Protocol \ref{prtl:prerspv} the phases are already added. Thus we first formalize a version of $\fBNTest$ that additionally takes a tuple of phase information as parameters that does the following: it simply reveals the phases and allows the server to de-phase the gadgets and then calls the main body of $\fBNTest$.
\begin{mdframed}
\begin{prtl}[$\fBNTest$ starting from gadgets with phases]
Suppose the security parameter is $\kappa$ and the gadget number is $L$.\par
Client holds a tuple of key pairs $\tilde K=(K^{(i)})_{i\in [L]}$ and a tuple of phase pairs $(\Theta^{(i)})_{i\in [L]}$.\par
Honest server holds:
\begin{equation}\label{eq:99}\fgadget(K^{(1)},\Theta^{(1)})\otimes \fgadget(K^{(2)},\Theta^{(2)})\otimes\cdots\otimes \fgadget(K^{(L)},\Theta^{(L)})\end{equation}
\begin{enumerate}
\item The client reveals $\Theta^{(1)},\Theta^{(2)}\cdots \Theta^{(L)}$ and the server could remove the phases from \eqref{eq:99} and get
		 $$\fgadget(K^{(1)})\otimes  \fgadget(K^{(2)})\otimes \cdots\fgadget(K^{(L)}) $$
		 \item Both parties execute $\fBNTest(\tilde K^{(i)};1^\kappa)$ defined below.
		 \end{enumerate}	
\end{prtl}

\begin{prtl}[$\fBNTest$]\label{prtl:bntest}
Suppose the security parameter is $\kappa$ and the gadget number is $L$.\par
Client holds a tuple of key pairs $\tilde K=(K^{(i)})_{i\in [L]}$.\par
Honest server holds:
\begin{equation}\fgadget(K^{(1)})\otimes \fgadget(K^{(2)})\otimes\cdots\otimes \fgadget(K^{(L)})\end{equation}
\begin{enumerate}
\item The client samples a random subset of index $I\subseteq [L]$.\par 
 Both parties execute $\fCombine(K,I;1^\kappa)$. This combines the gadgets with superscripts in $I$ into a single gadget. The client gets $K^{(combined)}$ and the server gets $\fgadget(K^{(combined)})$. The gadgets with superscripts in $[L]-I$ do not change.
\item The client asks the server to measure all the gadgets excluding $\fgadget(K^{(combined)})$ in the standard basis and send back the results (denoted as $rp^{(j)}$ for each $j\in [L]-I$). The client checks $rp^{(j)} \in K^{(j)}$ for each $j\in [L]-I$, and rejects if it's not satisfied.
\item The client chooses one of the following two randomly:
\begin{itemize}
\item The client asks the server to measure  $K^{(combined)}$ in the computational basis and send back the result. The client checks the server's response is within $K^{(combined)}$.
\item Both parties execute $\fHadamardTest(K^{(combined)};1^\kappa)$ on the combined gadget.	
\end{itemize}

\end{enumerate}

\end{prtl}
	
\end{mdframed}
So far we have completed the formalization of our pre-RSPV protocol. Below we give its correctness, efficiency and verifiability, which corresponds to Definition \ref{defn:prerspv}, \ref{defn:prerspvv}.
\subsection{Properties of Our $\fpreRSPV$ Protocol}\label{sec:5.4r}
Protocol \ref{prtl:prerspv} has the following properties.
\paragraph{Round type probability} The round type is chosen from $(\ftest,\fquiz,\fcomp)$ with probability $(\frac{4}{5},\frac{1}{10},\frac{1}{10})$ correspondingly. Denote them as $p_{\ftest},p_{\fquiz},p_{\fcomp}$.
\paragraph{Correctness} The protocol in the honest settings prepares the target state
$$\ket{+_{\theta^{(1)}}}\otimes \ket{+_{\theta^{(2)}}}\otimes \cdots \otimes \ket{+_{\theta^{(L)}}} $$
in the $\fcomp$ round. The client gets $\theta^{(1)}\theta^{(2)}\cdots \theta^{(L)}\in_r \{0,1\cdots 7\}^L$. This is as required in Definition \ref{defn:prerspv}.
\paragraph{Winning Probability} The winning probability in the quiz round in the honest setting is
$$\OPT=\frac{1}{3}|\frac{1}{2}+\frac{1}{2}e^{\mi\pi/4}|^2=\frac{1}{3}\cos^2(\pi/8)= 0.28451779686\cdots$$
where the first $1/3$ comes from the fact that in the extra-biased Hadamard test, if the extra phase bias $\delta\in \{0,4\}$ the protocol does not generate $\fwin/\flose$ output.
\paragraph{Efficiency} The complexity of both parties is $O(\fpoly(\kappa)|C|)$. Note that since we are working on the MBQC model where long-range interactions come with a cost, the complexity analysis needs to be careful. By analyzing the honest mapping we can confirm that the total complexity of honest behavior is $O(\fpoly(\kappa)|C|)$ even in the MBQC model.\footnote{One place that needs to be additionally careful is how to model the cost of random oracle queries. In our protocol there are constant number of queries where the input is a long string, whose length is linear to $L$. It's reasonable to consider the cost of this action to be within $O(\fpoly(\kappa)L)$.}\par
\paragraph{Optimality of $\OPT$} \begin{thm}\label{thm:opt}$\OPT$ is optimal with error tolerance $(10^{-2000},10^{-220})$ (as formalized in Definition \ref{defn:opt}).\end{thm}
For the verifiability we have:
\paragraph{Verifiability} \begin{thm}\label{thm:prerspvv} Protocol \ref{prtl:prerspv} has verifiability with error tolerance $(10^{-2000},10^{-200})$ (as formalized in Definition \ref{defn:prerspvv}).\end{thm}
We prove Theorem \ref{thm:opt} and \ref{thm:prerspvv} in Section \ref{sec:12}.
%
%
\section{Basic Notions and Analysis of Key-Pair Preparation and Standard Basis Test}\label{sec:5}
In this section we develop tools for analyzing the part of Protocol \ref{prtl:prerspv} before the verifiable state preparation step.
\begin{itemize}
\item First, in Section \ref{sec:6.1r} we give symbols for registers that appear in the protocols. Then we get a clearer set-up for later proofs.\par
\item Then in Section \ref{sec:5.1} we will give some basic notions for characterizing the state's properties. They include the \emph{efficiently preparable}, \emph{key checkable} and \emph{claw-free} properties, which are basically reviews or re-formalizations of existing notions; and we define the \emph{strong-claw-free} to characterize the uncorrelated claw-freeness of multiple key pairs. These notions will also be the basis of later proofs.
 \item Then in Section \ref{sec:6.3q} we analyze the first step of Protocol \ref{prtl:prerspv}, the parallel application of $\fNTCF$. We will show the output state of this step satisfies the properties above.\par
 \item Then in Section \ref{sec:6.2} we give the notion of \emph{basis-honest form}, which means the server holds exactly one key from each key pair that the client holds. We also define the \emph{branch} of basis-honest form formally, which will be useful in later proofs.\par
\item  In Section \ref{sec:6.3} we analyze the properties of standard basis test ($\fStdBTest$). 
  We will show the ability of passing the standard basis test with high probability implies the state is approximately isometric to a basis-honest form via a server-side isometry.
  \item Finally in Section \ref{sec:6.6} we prove some lemmas that will be used in later proofs.
  \end{itemize}
  Note that we will introduce the \emph{Setup} to organize the properties of a state family.
 \subsection{Symbols for Different Registers}\label{sec:6.1r}
 By the time of the completion of step $1$ in Protocol \ref{prtl:prerspv}, the registers include:
 \begin{itemize}
 	\item The client-side key registers: denoted by $\bK^{(\switch)}=(\bx_0^{(\switch)},\bx_1^{(\switch)}),\bK=(\bK^{(i)})_{i\in [0,L]}$, $\bK^{(i)}=(\bx_0^{(i)},\bx_1^{(i)})$.\par
 	When the security parameter is $\kappa$, the size of each $\bx_b^{(i)}$ register is $\kappa$.\par
 	For each $i\in \{\switch, [0,L]\}$, define $\Domain(\bK^{(i)})$ as $\{(x_0,x_1):x_0\in \{0,1\}^{\kappa},x_1\in \{0,1\}^{\kappa},x_0\neq x_1\}$.\par
 	As a convention, we use $\tilde\bK$ to denote $(\bK^{(i)})_{i\in [L]}$.
 	\item The transcript registers hold the server's response in the first step, which is denoted by $y^{(\switch)}y^{(0)}\cdots y^{(L)}$ in the protocol; denote the corresponding registers by $\bY=(\bY^{(\switch)},\bY^{(0)},\bY^{(1)},\cdots \bY^{(L)})$.
 	\item The client-side phase registers: although these phases are not sampled and used in the protocol until the step (a) of the verifiable state preparation step, we could assume the client has already sampled them out in advance. These information is stored in registers $\bTheta=(\bTheta^{(i)})_{i\in [0,L]}$. $\bTheta^{(i)}=(\btheta_0^{(i)},\btheta_1^{(i)})$. The domain of each $\btheta$ register is $\{0,1\cdots 7\}$.
 	\item The server's registers. In general these registers are denoted by symbol $\bS$; and we note there are many different registers appeared in different steps of the protocol. We will use superscripts to refer to these different registers. Especially, denote the registers that holds the $\fgadget(K^{(i)})$  in \eqref{eq:25prep} in the honest setting as $\bS_{bsh}^{(i)}$. Denote $\bS_{bsh}=(\bS_{bsh}^{(i)})_{i\in [0,L]}$.
 	\item The random oracle registers that store the random oracle outputs for each input. The symbols are as used in Section \ref{sec:3.3}.
 \end{itemize}
%
\subsection{Basic Notions on Joint States}\label{sec:5.1}
We give the following basic notions for characterizing the properties of a state.
\begin{defn}[Efficiently preparable states, repeated] 
We say a purified joint state $\ket{\varphi}$ is efficiently-preparable if there exists a polynomial time operator (which could include projections) $O$ such that $\ket{\varphi}=O\ket{0}$.	
\end{defn}

\begin{defn}[Key-checkable]
Suppose $\bx$ is a client-side register that holds a key. We say a purified joint state $\ket{\varphi}$ is key-checkable for $\bx$ if there exists an efficient server-side operation\footnote{In general this server-side operation is implemented with the help of transcript registers (which are the $\bY$ registers).}  that implements $\Pi^{\bS}_{=\bx}$ on $\ket{\varphi}$ where $\bS$ is an arbitrary server-side register, $\Pi^{\bS}_{=\bx}$ denotes the projection onto the space that the content of $\bS$ is equal to the content of $\bx$.\par
Suppose the client holds a tuple of keys in registers $\bK$. We say a purified joint state $\ket{\varphi}$ is key-checkable for $\bK$ if for any single key register $\bx$ within $\bK$, $\ket{\varphi}$ is key-checkable for $\bx$.

\end{defn}

Then we review the notion of claw-freeness.
\begin{defn}[Claw-free \cite{BCMVV}]
Suppose the client holds a key pair in register $\bK=(\bx_0,\bx_1)$. We say a purified joint state $\ket{\varphi}$ is claw-free for $\bK$ against adversary family $\cF$ if for any adversary $\fAdv\in \cF$, for an arbitrary server-side register $\bS$,
$$|\Pi_{=\bx_0||\bx_1}^{\bS}\fAdv\ket{\varphi}|\leq \fneg(\kappa)$$
where $\Pi^{\bS}_{=\bx_0||\bx_1}$ is a projection onto the subspace that the content of $\bS$ is equal to the content of $\bx_0||\bx_1$.\par
We omit $\cF$ when it is taken to be the set of polynomial time server-side operations.
\end{defn}
And we further generalize it to the multi key pair setting:
\begin{defn}[Strongly claw-free]
	Suppose the client holds a tuple of key pairs in register $(\bK^{(i)})_{i\in D}$ where $\bK^{(i)}=(\bx_0^{(i)},\bx_1^{(i)})$. We say a purified joint state $\ket{\varphi}$ is strongly-claw-free for any key pair in $(\bK^{(i)})_{i\in D}$ if for any $i\in D$, $\ket{\varphi}\odot ((\bK^{(i)})_{i\in D}-\bK^{(i)})$ is claw-free for $\bK^{(i)}$, where $\ket{\varphi}\odot ((\bK^{(i)})_{i\in D}-\bK^{(i)})$ means, starting from  $\ket{\varphi}$, the client sends all the keys except the $i$-th key pair to the server.
\end{defn}
\subsection{Analysis of the Key-pair Superposition Preparation Step}\label{sec:6.3q}
We have the following theorem which characterizes the properties of output states from the first step of Protocol \ref{prtl:prerspv} (a parallel $\fNTCF$ evaluation).
\begin{thm}\label{thm:gp}
For any polynomial time adversary $\fAdv$, denote the output state of the first step of Protocol \ref{prtl:prerspv} as $\ket{\varphi^1}$. Then registers described in Section \ref{sec:6.1r} are initialized, and $\ket{\varphi^1}$ is:
\begin{itemize}
\item efficiently preparable;
\item key checkable for both $\bK^{(\switch)}$ and $\bK$;
\item strongly-claw-free for any key pair in $(\bK^{(\switch)},\bK)$.
\end{itemize}
\end{thm}
Note that we will interchangeably use $(\bK^{(\switch)},\bK)$ or $(\bK^{(i)})_{i\in \{\switch,[0,L]\}}$ to denote this tuple of $2+L$ output key pairs.
\begin{proof}
The efficiently preparable property comes from the efficiency of $\fAdv$ and the protocol; the key-checkable property comes from the correctness of $\fNTCF$. We only need to prove the strong claw-freeness.\par
	If this is not true, there will be an index $i\in (\switch)\cup [0,L]$ such that an efficient operation $V$ can output both keys in $\bK^{(i)}$ given the other keys. Then we can construct an adversary that breaks the claw-free property of $\fNTCF$ as follows: \begin{enumerate}\item Instead of interacting with the client, the server simulates all the other $\fNTCF$ evaluations on its own excluding the $i$-th evaluation. In the simulated state the server has access to $((\bK^{(\switch)},\bK)-\bK^{(i)})$ since the client-side key registers for this part are simulated. Denote the simulated state as $\ket{\varphi}$.
	\item Then by the assumption $V$ applied on $\ket{\varphi}\odot ((\bK^{(\switch)},\bK)-\bK^{(i)})$ \footnote{Here we slightly abuse the notation: we use the $\odot$ notation on the simulated state to mean that the simulated client side registers are copied to the transcript registers.} outputs both keys in $\bK^{(i)}$. This contradicts the claw-free property of $\fNTCF$.
 \end{enumerate}
\end{proof}
Now we are ready to formalize a \emph{set-up} that abstracts the property of the output state of the first step of Protocol \ref{prtl:prerspv}. In the remaining proofs we could only refer to this set-up instead of applying Theorem \ref{thm:gp} again and again.
\begin{setup}\label{setup:1}
	We use Setup \ref{setup:1} to denote the set of states that satisfy:
	\begin{itemize}
	\item The parties are as described in Section \ref{sec:4.1}.
		\item The registers are as described in Section \ref{sec:6.1r}.
		\item The state is efficiently preparable;
		\item The state is key checkable for each key  in $(\bK^{(\switch)},\bK)$;
		\item The state is strongly-claw-free for each key pair in $(\bK^{(\switch)},\bK)$.
	\end{itemize}
\end{setup} 
\subsection{Basis-honest Form}\label{sec:6.2}
We will define the basis-honest form as below. Recall that in the honest setting, the client holds a tuple of keys, and the server holds superpositions of keys. Correspondingly, in the malicious setting, if the server holds keys in superpositions (which are possibly entangled with some auxiliary registers since it's malicious), we call this form the basis-honest form.\par
Let's first introduce a convenient notation, the \emph{key vector}.
\begin{nota}[Key vector and subscript vector]\label{nota:kv}
Suppose the client holds a tuple of key pairs $\tilde K=(K^{(i)})_{i\in [L]}$, $K^{(i)}=(x_0^{(i)},x_1^{(i)}),i\in [L]$. For subscript vector $\vec{b}\in \{0,1\}^L$, denote
$$\vec{x}_{\vec{b}}:=x^{(1)}_{b^{(1)}}x^{(2)}_{b^{(2)}}\cdots x^{(L)}_{b^{(L)}}$$ 
where $b^{(1)}b^{(2)}\cdots b^{(L)}$ are coordinates of $\vec{b}$, as the key vector of $\tilde K$ under subscript vector $\vec{b}$.\par
If the superscripts of keys start from $0$, the subscript vector will be $\vec{b}\in \{0,1\}^{1+L}$, and the key vector is defined correspondingly as $x^{(0)}_{b^{(0)}}x^{(1)}_{b^{(1)}}x^{(2)}_{b^{(2)}}\cdots x^{(L)}_{b^{(L)}}$.
\end{nota}

\begin{defn}[Basis-honest form]\label{defn:bsh} Below we define the basis-honest form of a tuple of key pairs $\bK$ (or $\tilde\bK$), and the basis-honest form of a single key pair. For the first two bullets we use the register setup in Setup \ref{setup:1}.
\begin{itemize}
	\item Suppose the client holds a tuple of key pairs in registers $\bK$, and correspondingly the server holds register $\bS_{bsh}$. 
 We say a purified joint state $\ket{\varphi}$ is in a basis-honest form of $\bK$ if it has the form
\begin{equation}\label{eq:basishonest}\ket{\varphi}=\sum_{K\in \Domain(\bK)}\underbrace{\ket{K}}_{\text{client-side }\bK}\otimes \sum_{\vec{b}\in \{0,1\}^{1+L}}\underbrace{\ket{\vec{x}_{\vec{b}}}}_{\text{server-side }\bS_{bsh}}\underbrace{\ket{\varphi_{K,\vec{b}}}}_{\text{remaining registers}}\end{equation}
We define $\Pi_{\basishonest(\bK)}^{\bS_{bsh}}$, or simply $\Pi_{\basishonest(\bK)}$, as the projection onto the subspace that the content of $\bS_{bsh}$ is a valid key vector of $\bK$.\par
\item Correspondingly, for key tuple $\tilde\bK=(\bK^{(i)})_{i\in [L]}$, the basis-honest form is defined to be the state where the server holds a key vector of $\tilde\bK$. $\Pi_{\basishonest(\tilde\bK)}$ is define similarly.\par
\item For this bullet we consider a general notion that is not necessarily under Setup \ref{setup:1}. Suppose the client holds a key pair $\bK=(\bx_0,\bx_1)$ while honestly the server holds a key in $\bK$ in register $\bS_{bsh}$. Define the basis-honest form of $\bK$ to be the states in the form of 
\begin{equation}\label{eq:basishonest2}\ket{\varphi}=\sum_{K\in \Domain(\bK)}\underbrace{\ket{K}}_{\text{client-side }\bK}\otimes \sum_{b\in \{0,1\}}\underbrace{\ket{x_{b}}}_{\text{server-side }\bS_{bsh}}\underbrace{\ket{\varphi_{K,b}}}_{\text{remaining registers}}\end{equation}
And we use $\Pi_{\basishonest(\bK)}^{\bS_{bsh}}$ to denote the projection onto the space that $\bS_{bsh}$ holds a valid key in $\bK$, and we could omit the register superscript when there is no ambiguity.
\end{itemize}
\end{defn}
\begin{defn}[Approximate basis-honest form]
	If $|(\bbI-\Pi_{\basishonest(\bK)})\ket{\varphi}|\leq \epsilon$, we say $\ket{\varphi}$ is in an $\epsilon$-basis-honest form of $\bK$.
\end{defn}

\begin{defn}[Branch of a basis-honest form]\label{defn:branch}
Suppose $\ket{\varphi}$ is in a basis-honest form  of a key pair $\bK$ as shown in \eqref{eq:basishonest2}. For each $b\in \{0,1\}$, we call 
$$\sum_{K\in \Domain(\bK)}\underbrace{\ket{K}}_{\text{client}}\otimes \underbrace{\ket{x_{b}}}_{\bS_{bsh}}\ket{\varphi_{K,b}}$$
the $\bx_b$-branch of this state. 

Suppose $\ket{\varphi}$ is in a basis-honest form  of a tuple of key pair $\bK$ as shown in \eqref{eq:basishonest}. For each $\vec{b}\in \{0,1\}^{1+L}$, we call
$$\sum_{K\in \Domain(\bK)}\underbrace{\ket{K}}_{\text{client}}\otimes \underbrace{\ket{\vec{x}_{\vec{b}}}}_{\bS_{bsh}}\ket{\varphi_{K,\vec{b}}}$$
the $\vec{\bx}_{\vec{b}}$-branch of this state. 
\end{defn}
Note when we replace $\bK=(\bK^{(i)})_{i\in [0,L]}$ by $\tilde\bK=(\bK^{(i)})_{i\in [L]}$ or $(\bK^{(i)})_{i\in \{\switch,[0,L]\}}$, the definition could be adapted correspondingly.\par


\subsection{$\fStdBTest$ Implies Approximate Basis-honest Form}\label{sec:6.3}
Below we analyze the implications of the standard basis test. 
We have the following lemma, which allows us to analyze the structure of a state on the server side assuming it can pass the standard basis test:
\begin{thm}\label{thm:sbt}
Suppose the client holds a tuple of key pairs in register $(\bK^{(i)})_{i\in D}$, where $D$ is a set of indices in Setup \ref{setup:1}. Suppose an efficient adversary $\fAdv$, operated on a sub-normalized purified joint state $\ket{\varphi}$, could pass $\fStdBTest((\bK^{(i)})_{i\in D})$ with probability $p$. Then there exists an efficient server-side operation $O$ such that $O\ket{\varphi}$ is in a $\sqrt{1-p}$-basis-honest form of $(\bK^{(i)})_{i\in D}$.
\end{thm}
\begin{proof}
Run the adversary's operation but do not do the final measurement. Swap the response register and $(\bS_{bsh}^{(i)})_{i\in D}$ (See Definition \ref{defn:bsh}). The fact that it passes the standard basis test with failure probability $1-p$ implies the state on subspace $\bbI-\Pi_{\basishonest((\bK^{(i)})_{i\in D})}$ has norm at most $\sqrt{1-p}$.	
\end{proof}
\subsection{Useful Lemmas}\label{sec:6.6}
\subsubsection{Collapsing property}
One useful lemma is the collapsing property, which is used in different forms in many previous works \cite{MahadevVerification,GVRSP}. Intuitively it says if the initial state is claw-free for a key pair then the superposition of two branches is indistinguishable to the mixture of two branches.
\begin{lem}\label{lem:collapse}
	Suppose a subnormalized state $\ket{\varphi}$ is in Setup \ref{setup:1}, and is in the basis-honest form of $\bK^{(0)}$. Denote the state of the $\bx_0^{(0)}$-branch as $\ket{\varphi_0}$ and $\bx_1^{(0)}$-branch as $\ket{\varphi_1}$. (That is, $\ket{\varphi_0}:=\Pi^{\bS_{bsh}^{(0)}}_{\bx^{(0)}_0}\ket{\varphi}$, $\ket{\varphi_1}:=\Pi^{\bS_{bsh}^{(0)}}_{\bx^{(0)}_1}\ket{\varphi}$.) Then for any efficient server-side operation $O$ that output a single bit in a register $\bS$:
	$$|\Pi_0^{\bS}O\ket{\varphi}|^2\approx_{\fneg(\kappa)}|\Pi_0^{\bS}O\ket{\varphi_0}|^2+|\Pi_0^{\bS}O\ket{\varphi_1}|^2$$
\end{lem}
\subsubsection{Look-up table encryptions do not affect claw-freeness}
Then we give a lemma that formalizes the following intuition. If a basis-honest form state is claw-free for $\bK^{(i)}$, if we only consider one branch, for example, $\bx^{(i)}_0$-branch, since the server already knows the value of $x^{(i)}_0$, it won't be able to predict $x^{(i)}_1$ by claw-freeness. This is still true even if the server additionally gets polynomial number of ciphertexts encrypted under $x_{1}^{(i)}$. Formally, we have the following lemma, which will be useful later.
\begin{lem}\label{lem:5.2a}
Suppose a purified joint state $\ket{\varphi}$ is in Setup \ref{setup:1} and is in a basis-honest form of $\bK^{(i)}$ with only the $\bx^{(i)}_b$-branch, for some $i\in \{\switch\}\cup [0,L],b\in \{0,1\}$. $N=\fpoly(\kappa)$. $(p_{pre}^{(t)})_{t\in [N]}$ is a tuple of strings in $\{0,1\}^\kappa$, $(p_{post}^{(i)})_{t\in [N]}$ is a tuple of strings in $\{0,1\}^\kappa\cup \{\emptyset\}$. Define $\fAuxInf$ as the following algorithm that generates a tuple of salted hash values of $\bx_{1-b}^{(i)}$:
$$\forall t\in [N]:\text{sample $R^{(t)}\leftarrow_r \{0,1\}^\kappa$, output }(R^{(t)},H(R^{(t)}||p_{pre}^{(t)}||x_{1-b}^{(i)}),H(R^{(t)}||x_{1-b}^{(i)}||p_{post}^{(t)}))$$

Then $\ket{\varphi}\odot \llbracket\fAuxInf\rrbracket$ is claw-free for $\bK$.
\end{lem}
Finally we can show, if some key is unpredictable, the blinded oracle where these entries are blinded looks the same as the original oracle:
\begin{lem}\label{lem:5.2b}
Suppose a purified joint state $\ket{\varphi}$ is in Setup \ref{setup:1} and is in a basis-honest form of $\bK^{(i)}$ with only the $\bx^{(i)}_b$-branch, for some $b\in \{0,1\}$. suppose $\bH^\prime$ is the blinded oracle where $\cdots||\bx^{(i)}_{1-b}||\cdots$ are blinded, where ``$\cdots$'' represents arbitrary strings of a fixed length.   For any efficient adversary $\fAdv$, denote $\fAdv^\prime$ as the operation that each query in $\fAdv$ is replaced by a query to $\bH^\prime$. Then
$$
\fAdv\ket{\varphi}\approx_{\fneg(\kappa)}\fAdv^\prime\ket{\varphi}$$
\end{lem}
We put the proofs to these lemmas in Appendix \ref{app:mp6}.
\subsubsection{Rigidity of basis-honest form with strong-claw-free condition}
We show a useful lemma that will be used for multiple times in later proofs. It says, if a state is in a basis-honest form with good security properties, if the adversary transforms it to another basis-honest form, each of the outcome branch should only come from the corresponding input branch:
\begin{lem}\label{lem:multisbt}
Consider a sub-normalized purified joint state $\ket{\varphi}$ in Setup \ref{setup:1}.  
 Then for any efficient server-side operation $D$ there is
\begin{equation}\label{eq:multisbt}\Pi_{\basishonest(\bK)}D\Pi_{\basishonest(\bK)}\ket{\varphi}\approx_{\fneg(\kappa)}\sum_{\vec{b}\in \{0,1\}^{1+L}}\Pi_{\vec{\bx}_{\vec{b}}}^{\bS_{bsh}}D\Pi_{\vec{\bx}_{\vec{b}}}^{\bS_{bsh}}\ket{\varphi}\end{equation}
where $\Pi_{\vec{\bx}_{\vec{b}}}^{\bS_{bsh}}$ denotes the projection onto $\vec{\bx}_{\vec{b}}$-branch (Definition \ref{defn:branch}).
\end{lem}

\begin{proof}
 $\forall i\in [0,L]$ define: 
$$\ket{\varphi^i}=\sum_{b^{(0)}b^{(1)}\cdots b^{(i)}\in \{0,1\}^{1+i}}\Pi_{b^{(0)}b^{(1)}\cdots b^{(i)}}\Pi_{\basishonest(\bK)}D\Pi_{b^{(0)}b^{(1)}\cdots b^{(i)}}\Pi_{\basishonest(\bK)}\ket{\varphi}$$
where $\Pi_{b^{(0)}b^{(1)}\cdots b^{(i)}}$ denotes the projection onto the space that the values of ${\bS_{bsh}^{(0)}\bS_{bsh}^{(1)}\cdots \bS^{(i)}_{bsh}}$ are equal to the values of $\bx_{b^{(0)}}^{(0)}\bx_{b^{(1)}}^{(1)}\cdots \bx_{b^{(i)}}^{(i)}$.\par
Additionally define
$$\ket{\varphi^{-1}}=\Pi_{\basishonest(\bK)}D\Pi_{\basishonest(\bK)}\ket{\varphi}$$
To prove \eqref{eq:multisbt} we only need to prove: \begin{equation}\label{eq:37}\forall i\in [0,L],\ket{\varphi^{i}}\approx_{\fneg(\kappa)}\ket{\varphi^{i-1}}\end{equation} 
The observation is, although $\ket{\varphi^{i-1}},\ket{\varphi^i}$ are written as linear sums, there exists an efficient operation that exactly prepares this state from $\ket{\varphi}$. Recall $\ket{\varphi}$ has the form shown in \eqref{eq:basishonest}. Let's give an operation that transforms it to $\ket{\varphi^{i}}$.\par
 Initialize server-side auxiliary registers $\bA^{(0)}\bA^{(1)}\cdots \bA^{(i)}$, and define the following operations:
\begin{itemize}\item Define $\fCOPY_{0\sim i}$ as the operation that copies (bit-wise $\fCN$) the contents of $\bS_{bsh}^{(0)}\bS_{bsh}^{(1)}\cdots \bS_{bsh}^{(i)}$ to registers $\bA^{(0)}\bA^{(1)}\cdots \bA^{(i)}$. Define $\fCOPY_{i}$ as the operator that only copies $\bS_{bsh}^{(i)}$ to $\bA^{(i)}$.
\item Define $\Pi_{\bS_{bsh}^{(0)}\bS_{bsh}^{(1)}\cdots \bS_{bsh}^{(i)}=\bA^{(0)}\bA^{(1)}\cdots \bA^{(i)}}$ as the operator that projects onto the space that the values of $\bS_{bsh}^{(0)}\bS_{bsh}^{(1)}\cdots \bS_{bsh}^{(i)}$ are equal to $\bA^{(0)}\bA^{(1)}\cdots \bA^{(i)}$. Define $\Pi_{\bS_{bsh}^{(i)}=\bA^{(i)}}$ as the operator that projects onto the space that $\bS_{bsh}^{(i)}$ is equal to $\bA^{(i)}$.
\end{itemize}
 Then:
$$\ket{\varphi^i}=\Pi_{\basishonest(\bK)}\fCOPY_{0\sim i}\circ \Pi_{\bS_{bsh}^{(0)}\bS_{bsh}^{(1)}\cdots \bS_{bsh}^{(i)}=\bA^{(0)}\bA^{(1)}\cdots \bA^{(i)}}D\circ \fCOPY_{0\sim i}\circ \Pi_{\basishonest(\bK)}\ket{\varphi}$$
Similarly
$$\ket{\varphi^{i-1}}=\Pi_{\basishonest(\bK)}\fCOPY_{0\sim i-1}\circ \Pi_{\bS_{bsh}^{(0)}\bS_{bsh}^{(1)}\cdots \bS_{bsh}^{(i-1)}=\bA^{(0)}\bA^{(1)}\cdots \bA^{(i-1)}}D\circ \fCOPY_{0\sim i-1}\circ \Pi_{\basishonest(\bK)}\ket{\varphi}$$
Define $\cP=\fCOPY_{0\sim i-1}\circ \Pi_{\bS_{bsh}^{(0)}\bS_{bsh}^{(1)}\cdots \bS_{bsh}^{(i-1)}=\bA^{(0)}\bA^{(1)}\cdots \bA^{(i-1)}}D\circ \fCOPY_{0\sim i-1}$. Then \eqref{eq:37} is reduced to proving
\begin{equation}\label{eq:38}\Pi_{\basishonest(\bK)}\cP\Pi_{\basishonest(\bK)}\ket{\varphi}\approx_{\fneg(\kappa)}\Pi_{\basishonest(\bK)}\fCOPY_{i}\circ \Pi_{\bS_{bsh}^{(i)}=\bA^{(i)}}\circ\cP\circ \fCOPY_{i}\Pi_{\basishonest(\bK)}\ket{\varphi}\end{equation}
Denote $\Pi_{\bx_b^{(i)}}^{\bS_{bsh}^{(i)}}$ as the operator that projects onto the $\bx_b^{(i)}$-branch of the basis-honest form. Then the left hand side of \eqref{eq:38} is
$$\Pi_{\basishonest(\bK)}\cP\Pi_{\basishonest(\bK)}\ket{\varphi}=\sum_{b,b^\prime\in \{0,1\}^2}\Pi_{\basishonest(\bK)}\Pi_{\bx_{b^\prime}^{(i)}}^{\bS_{bsh}^{(i)}}\cP\Pi_{\bx_b^{(i)}}^{\bS_{bsh}^{(i)}}\Pi_{\basishonest(\bK)}\ket{\varphi}$$
and the right hand side of \eqref{eq:38} is 
$$\sum_{b\in \{0,1\}}\Pi_{\basishonest(\bK)}\Pi_{\bx_{b}^{(i)}}^{\bS_{bsh}^{(i)}}\cP\Pi_{\bx_b^{(i)}}^{\bS_{bsh}^{(i)}}\Pi_{\basishonest(\bK)}\ket{\varphi}$$
then \eqref{eq:38} is further reduced to
\begin{equation}\label{eq:103}\sum_{b,b^\prime\in \{0,1\}^2,b\neq b^\prime}\Pi_{\bx^{(i)}_{b^\prime}}^{\bS_{bsh}^{(i)}}\circ\cP\circ \Pi_{\bx^{(i)}_b}^{\bS_{bsh}^{(i)}}\Pi_{\basishonest(\bK)}\ket{\varphi}\approx_{\fneg(\kappa)}0\end{equation}
which holds by the claw-free property of $\ket{\varphi}$.
\end{proof}

\section{The Switch Gadget Technique}\label{sec:6}
In this section we analyze the switch gadget technique and how it affects the security properties of Protocol \ref{prtl:prerspv}. We have discussed the switch gadget technique briefly in the introduction; below we will review it and discuss its formal analysis. The basic lemma behind the switch gadget technique is given in Section \ref{sec:7.1q}, based on the notion of \emph{blinded oracle}. In Section \ref{sec:7.2q} we give several basic lemmas on RO-padded Hadamard test. The basic lemma behind the switch gadget technique is proved in Section \ref{sec:7.3q}.\par
 As discussed in the introduction, in the switch gadget technique we design protocols that go as follows:
 \begin{enumerate}\item Encode a mapping onto a switch gadget;\item Do a Hadamard test (Protocol \ref{prtl:hadamard}) on the switch gadget. This destroys the switch gadget.\end{enumerate} The switch gadget is used as a ``switch'' that controls whether the adversary could make use of the mapping. With the switch gadget, the honest server could implement the mapping; the problem is how to characterize the adversaries' view, assuming it wants to pass the Hadamard test with high probability.\par
\paragraph{An intuitive discussion on Hadamard test} We need to formalize the intuition that
 \begin{center}
 \emph{$K$ is a key pair, and the initial state is claw-free for $K$. If the server wants to pass the RO-padded Hadamard test for $K$, it loses the keys after the test.}	
 \end{center}

As \cite{jiayu20} pointed out, passing the RO-padded Hadamard test implies the server could not do powerful things about $K$ from the post-test state. However the lemmas given in \cite{jiayu20} are not suitable for our purpose here. Here we give a lemma that captures the properties of RO-padded Hadamard test in terms of \emph{blinded oracle} (Section \ref{sec:3.3}). Our analysis of the RO-padded Hadamard test could also be of independent interest elsewhere.\par
Suppose the Hadamard test is on key pair $K^{(\switch)}$ and initial state $\ket{\varphi}$. Suppose the post-test state is $\ket{\varphi^\prime}$. To formalize the intuition above, we consider a \emph{blinded oracle} $H^\prime$ where the entries in the form of $\{0,1\}^\kappa||K^{(\switch)}||\cdots$ are blinded. Then we will show, informally, if the test passes with high probability:
\begin{center}\emph{Starting from $\ket{\varphi^\prime}$, for any  efficient adversary, querying $H$ and querying $H^\prime$ should end up with similar states. }\end{center}
The formal statement goes roughly as follows. For any efficient operator $D$ on the post-test state, define $D^\prime$ as the blinded version of $D$ where all the random oracle queries in $D$ are replaced by queries to $H^\prime$, there is 
\begin{equation}\label{eq:77}D^\prime\ket{\varphi^\prime}\approx D\ket{\varphi^\prime}\end{equation}
How is  \eqref{eq:77} related to the intuition above? If the server could still predict a key in $K$ from the post-test state, it can query the oracle with this key and make two sides of \eqref{eq:77} quite different. Thus \eqref{eq:77} describes (and strengthens) the intuition above.\par
This statement gives what we want from the switch gadget technique: the ciphertexts within our lookup tables (which encrypt the phases) are encrypted under $\{0,1\}^\kappa||K^{(\switch)}||\cdots$; and after applying this lemma, we only need to analyze a blinded oracle where this form of entries are blinded, which means the phase information encrypted under these lookup tables will become secret again.\par
The formal theorem is given below in Theorem \ref{thm:pht}.

%
\subsection{Basic Lemma Behind the Switch Gadget Technique}\label{sec:7.1q}
We first formalize the set-up of the theorem as follows. We note in this subsection we do not follow the register symbols in Section \ref{sec:6.1r}; specifically, $\bK$ will be used to denote a key pair, to make our lemmas more general.
\begin{setup}\label{setup:2}
	Suppose the parties are formalized in Section \ref{sec:4.1}. Suppose the client holds a key pair in register $\bK=(\bx_0,\bx_1)$. Define $\SETUP^2(\bK)$ as the set of purified joint states that are \begin{itemize}\item efficiently preparable;\item  key checkable for $\bK$;
	\item claw-free for $\bK$.
 \end{itemize}
In addition, a transcript register $\bd$ is initialized, whose length is equal to the length of keys in $\bK$ plus $\kappa$ (security parameter).
\end{setup}
\begin{thm}\label{thm:pht}
 Suppose the client holds a key pair in register $\bK^{(\switch)}$.  Suppose a sub-normalized purified joint state $\ket{\varphi}\in \SETUP^2(\bK^{(\switch)})$ is in the $\epsilon$-basis-honest form of $\bK^{(\switch)}$. For any polynomial time adversary $\fAdv$, denote the post-execution state of the RO-padded Hadamard test as
$$\ket{\varphi^\prime}=\fHadamardTest^\fAdv(\bK^{(\switch)};1^\kappa)\ket{\varphi}$$
Then at least one of the following two is true:
\begin{itemize}
\item (Small passing probability)
$$|\Pi_{\fpass}\ket{\varphi^\prime}|^2\leq 1-p$$
\item (RO hiding) Suppose $H^\prime$ is the blinded oracle where the entries $\{0,1\}^\kappa||K^{(\switch)}||\cdots$ are blinded. For any efficient operation $D$, suppose $D^{\prime}$ is the blinded version of $D$ where all the oracle queries are replaced by queries to $H^\prime$. Then
\begin{equation}\label{eq:pht}D^{\prime} \ket{\varphi^\prime}\approx_{8\sqrt{p+2\epsilon}+\fneg(\kappa)}D\ket{\varphi^\prime}\end{equation}	
\end{itemize}
	
\end{thm}
To prove this theorem, we need to analyze the RO-padded Hadamard test. We first prove some basic lemmas on this test. Note these lemmas will also be independently useful in later proofs.
\subsection{Analysis of RO-padded Hadamard test}\label{sec:7.2q}
We first introduce the following notations which explicitly describe the passing/failing and winning/losing conditions of these Hadamard test. Below we will use these notations to describe the output of Hadamard test.
\begin{nota}\label{nota:setup2aux}
	Suppose the client holds a key pair in register  $\bK=(\bx_0,\bx_1)$. Suppose a sub-normalized purified joint state $\ket{\varphi}\in \SETUP^2(\bK)$. Recall that in $\fHadamardTest$ the client first sample a random padding in $\{0,1\}^\kappa$, denote the register storing it as $\bpad$. Define the following operators on system $\bd$ (introduced in $\SETUP^2$):
	\begin{itemize}
		\item $\Pi_{\eqref{eq:htt}=0}^{\bd}$ denotes the projection onto the space that the values of client side registers $\bx_0,\bx_1$, the transcript registers $\bd$, $\bpad$, and the random oracle registers satisfy: the result of calculation of \eqref{eq:htt} is $0$:
		$$d\cdot (x_0||\underbrace{H(pad||x_0)}_{\text{$\kappa$ bits}})+d\cdot(x_1||H(pad||x_1))=0$$
		\item $\Pi_{\eqref{eq:htt}=1}^{\bd}$ is defined similarly.
		\item $\Pi_{= 0}^{\text{last $\kappa$ bits of }\bd}$ is the projection onto the space that the last $\kappa$ bits of $\bd$ is all zero. $\Pi_{\neq 0}^{\text{last $\kappa$ bits of }\bd}$ is defined as its complement. Note this is one of the client's checking in $\fHadamardTest$.
	\end{itemize}
\end{nota}
The following lemma studies the output property of $\fHadamardTest$ on a single branch:
\begin{lem}\label{lem:prepht}
	Suppose the client holds a key pair in register  $\bK=(\bx_0,\bx_1)$. Suppose a sub-normalized purified joint state $\ket{\varphi}\in \SETUP^2(\bK)$ and is in the basis-honest form of $\bK$ with only $\bx_0$ branch. Then for any efficient adversary $\fAdv$ there is
	$$|\Pi_{\eqref{eq:htt}=0}^{\bd}\Pi_{\neq 0}^{\text{last $\kappa$ bits of }\bd}\fHadamardTest^\fAdv(\bK;1^\kappa)\ket{\varphi}|\approx_{\fneg(\kappa)}\frac{1}{\sqrt{2}}|\Pi_{\neq 0}^{\text{last $\kappa$ bits of }\bd}\fHadamardTest^\fAdv(\bK;1^\kappa)\ket{\varphi}|$$
	$$|\Pi_{\eqref{eq:htt}=1}^{\bd}\Pi_{\neq 0}^{\text{last $\kappa$ bits of }\bd}\fHadamardTest^\fAdv(\bK;1^\kappa)\ket{\varphi}|\approx_{\fneg(\kappa)} \frac{1}{\sqrt{2}}|\Pi_{\neq 0}^{\text{last $\kappa$ bits of }\bd}\fHadamardTest^\fAdv(\bK;1^\kappa)\ket{\varphi}|$$
\end{lem}
Then we have the following corollary, which studies the relations of two branches if a state could pass the $\fHadamardTest$:
\begin{cor}\label{cor:prepht}
	Suppose the client holds a key pair in register  $\bK=(\bx_0,\bx_1)$. Suppose a sub-normalized purified joint state $\ket{\varphi}\in \SETUP^2(\bK)$ and is in an $\epsilon$-basis-honest form of $\bK$. Denote the $\bx_0$-branch as $\ket{\varphi_0}$ and $\bx_1$-branch as $\ket{\varphi_1}$. If an efficient adversary $\fAdv$ could make the client output $\fpass$ in the Hadamard test (Protocol \ref{prtl:phadamard}) from initial state $\ket{\varphi}$ with probability $\geq 1-p$, then
	\begin{align}\label{eq:98y}&\Pi_{\eqref{eq:htt}=0}^{\bd}\Pi_{\neq 0}^{\text{last $\kappa$ bits of }\bd}\fHadamardTest^\fAdv(\bK;1^\kappa)\ket{\varphi_0}\nonumber\\\approx_{\sqrt{2(p+2\epsilon)}+\fneg(\kappa)}&\Pi_{\eqref{eq:htt}=0}^{\bd}\Pi_{\neq 0}^{\text{last $\kappa$ bits of }\bd}\fHadamardTest^\fAdv(\bK;1^\kappa)\ket{\varphi_1}\end{align}
	\begin{equation*}\Pi_{=0}^{\text{last $\kappa$ bits of }\bd}\fHadamardTest^\fAdv(\bK;1^\kappa)\ket{\varphi_0}\approx_{\sqrt{p+2\epsilon}+\fneg(\kappa)}0,\end{equation*}\begin{equation}\label{eq:78ol} \Pi_{= 0}^{\text{last $\kappa$ bits of }\bd}\fHadamardTest^\fAdv(\bK;1^\kappa)\ket{\varphi_1}\approx_{\sqrt{p+2\epsilon}+\fneg(\kappa)} 0\end{equation}
		\begin{align}&\Pi_{\eqref{eq:htt}=1}^{\bd}\Pi_{\neq 0}^{\text{last $\kappa$ bits of }\bd}\fHadamardTest^\fAdv(\bK;1^\kappa)\ket{\varphi_0}\nonumber\\\approx_{\sqrt{p+\epsilon}+\fneg(\kappa)}&-\Pi_{\eqref{eq:htt}=1}^{\bd}\Pi_{\neq 0}^{\text{last $\kappa$ bits of }\bd}\fHadamardTest^\fAdv(\bK;1^\kappa)\ket{\varphi_1}\end{align}
\end{cor}
These lemmas are proved in Appendix \ref{app:b}.
\subsection{Proof of Theorem \ref{thm:pht}}\label{sec:7.3q}
Now we give the proof of Theorem \ref{thm:pht} with this lemma.
\begin{proof}[Proof of Theorem \ref{thm:pht}]
Suppose $|\Pi_{\fpass}\ket{\varphi^\prime}|^2\geq 1-p$. Denote the $\bx_0^{(\switch)}$-branch of $\ket{\varphi}$ as $\ket{\varphi_0}$ and $\bx_1^{(\switch)}$-branch of $\ket{\varphi}$ as $\ket{\varphi_1}$. Then $\ket{\varphi}\approx_{\sqrt{\epsilon}}\Pi_{\basishonest(\bK^{(\switch)})}\ket{\varphi}=\ket{\varphi_0}+\ket{\varphi_1}$. Define the corresponding output term:
$$\ket{\varphi^\prime_0}=\fHadamardTest^\fAdv(\bK^{(\switch)};1^\kappa)\ket{\varphi_0}$$
$$\ket{\varphi^\prime_1}=\fHadamardTest^\fAdv(\bK^{(\switch)};1^\kappa)\ket{\varphi_1}$$
Then \begin{equation}\label{eq:112pq}\ket{\varphi^\prime}\approx_{\sqrt{\epsilon}}\ket{\varphi^\prime_0}+\ket{\varphi_1^\prime}\end{equation} We will blind the oracle in two steps.
\begin{enumerate}
\item Define $\bH^{mid}$ as the blinded oracle where the entries $\{0,1\}^\kappa||x_0^{(\switch)}||\cdots$ are blinded. Define $D^{mid}$ as the adversary	where all the random oracle queries in $D$ are replaced by queries to $H^{mid}$. Our goal is to prove
\begin{equation}\label{eq:hgm1}D^{mid} \ket{\varphi^\prime}\approx_{4\sqrt{p+2\epsilon}+\fneg(\kappa)}D\ket{\varphi^\prime}\end{equation}
The reason is as follows.\par
We will first analyze the effect of replacing $D$ by $D^{mid}$ on $\ket{\varphi^\prime_1}$. Intuitively $\ket{\varphi_1}$ is the $\bx_1^{(\switch)}$-branch, the server could not predict $\bx_0^{(\switch)}$ by claw-freeness thus the blinding operation on $\bx_0^{(\switch)}$-related entries will not be detected. Formally speaking, the following is implied by Lemma \ref{lem:5.2b}:
\begin{equation}\label{eq:switch1}D^{mid} \ket{\varphi^\prime_1}\approx_{\fneg(\kappa)}D\ket{\varphi^\prime_1}\end{equation}
Then since the $\bd$ register, after generated in the Hadamard test, is read-only, both sides of \eqref{eq:switch1} are still close to each other if a projection on a subset of possible values of $\bd$ is applied. Concretely:
\begin{equation}\label{eq:switch2}D^{mid}\Pi_{\eqref{eq:htt}=0}^{\bd}\Pi_{\neq 0}^{\text{last $\kappa$ bits of }\bd} \ket{\varphi^\prime_1}\approx_{\fneg(\kappa)}D\Pi_{\eqref{eq:htt}=0}^{\bd}\Pi_{\neq 0}^{\text{last $\kappa$ bits of }\bd}\ket{\varphi^\prime_1},\end{equation}\begin{equation}\label{eq:82ol}D^{mid}\Pi_{\eqref{eq:htt}=1}^{\bd}\Pi_{\neq 0}^{\text{last $\kappa$ bits of }\bd} \ket{\varphi^\prime_1}\approx_{\fneg(\kappa)}D\Pi_{\eqref{eq:htt}=1}^{\bd}\Pi_{\neq 0}^{\text{last $\kappa$ bits of }\bd}\ket{\varphi^\prime_1}\end{equation}
Now we use properties of the Hadamard test to argue about the $\ket{\varphi^\prime_0}$ branch. By Corollary \ref{cor:prepht} we have
\begin{equation}\label{eq:switch3}\Pi_{\eqref{eq:htt}=0}^{\bd}\Pi_{\neq 0}^{\text{last $\kappa$ bits of }\bd}\ket{\varphi_0^\prime}\approx_{\sqrt{2(p+2\epsilon)}+\fneg(\kappa)}\Pi_{\eqref{eq:htt}=0}^{\bd}\Pi_{\neq 0}^{\text{last $\kappa$ bits of }\bd}\ket{\varphi_1^\prime}\end{equation}
\begin{equation}\label{eq:switch4}\Pi_{\eqref{eq:htt}=1}^{\bd}\Pi_{\neq 0}^{\text{last $\kappa$ bits of }\bd}\ket{\varphi_0^\prime}\approx_{\sqrt{p+\epsilon}+\fneg(\kappa)}-\Pi_{\eqref{eq:htt}=1}^{\bd}\Pi_{\neq 0}^{\text{last $\kappa$ bits of }\bd}\ket{\varphi_1^\prime}\end{equation}
which together with \eqref{eq:switch2} \eqref{eq:82ol} implies
\begin{equation}D^{mid}\Pi_{\eqref{eq:htt}=0}^{\bd}\Pi_{\neq 0}^{\text{last $\kappa$ bits of }\bd} \ket{\varphi^\prime_0}\approx_{\sqrt{2(p+2\epsilon)}+\fneg(\kappa)}D\Pi_{\eqref{eq:htt}=0}^{\bd}\Pi_{\neq 0}^{\text{last $\kappa$ bits of }\bd}\ket{\varphi^\prime_0},\end{equation}\begin{equation}D^{mid}\Pi_{\eqref{eq:htt}=1}^{\bd}\Pi_{\neq 0}^{\text{last $\kappa$ bits of }\bd} \ket{\varphi^\prime_0}\approx_{\sqrt{p+\epsilon}+\fneg(\kappa)}D\Pi_{\eqref{eq:htt}=1}^{\bd}\Pi_{\neq 0}^{\text{last $\kappa$ bits of }\bd}\ket{\varphi^\prime_0}\end{equation}
which together with \eqref{eq:78ol} implies
\begin{equation}\label{eq:88ol}D^{mid} \ket{\varphi^\prime_0}\approx_{2.5\sqrt{p+2\epsilon}+\sqrt{p+\epsilon}+\fneg(\kappa)}D\ket{\varphi^\prime_0}\end{equation}
which together with \eqref{eq:switch1}\eqref{eq:112pq} implies \eqref{eq:hgm1}.
\item Now we hide the $\bx_1^{(\switch)}$ part using similar techniques. Consider $\bH^{\prime}$ as the blinded oracle of $\bH^{mid}$ where the entries $\{0,1\}^\kappa||x_1^{(\switch)}||\cdots$ are blinded. Then $\bH^\prime$ could also be seen as a blinded version of $\bH$ where $\{0,1\}^\kappa||K^{(\switch)}||\cdots$ are blinded. Correspondingly, $D^{\prime}$ is defined to be: in $D^{mid}$ all the queries to $\bH^{mid}$ are replaced by queries to $\bH^\prime$. Then $D^\prime$ is also the blinded version of $D$ where all the queries to $\bH$ are replaced by queries to $\bH^\prime$. Our goal is to prove
\begin{equation}\label{eq:hgm2}D^{\prime}\ket{\varphi^\prime}\approx_{4\sqrt{p+2\epsilon}+\fneg(\kappa)} D^{mid}\ket{\varphi^\prime}\end{equation}
Here we will start with the $\bx_0^{(\switch)}$ branch of $\ket{\varphi^\prime}$. By Lemma \ref{lem:5.2b} there is
\begin{equation}\label{eq:90ol}D^{\prime} \ket{\varphi^\prime_0}\approx_{\fneg(\kappa)}D^{mid}\ket{\varphi^\prime_0}\end{equation}
Then arguments \eqref{eq:switch2}-\eqref{eq:88ol} hold after all the appearances of $D$ in them are replaced by $D^\prime$, $\ket{\varphi_0^\prime}$ are replaced by $\ket{\varphi_1^\prime}$ and $\ket{\varphi_1^\prime}$ are replaced by $\ket{\varphi_0^\prime}$. Thus we get
\begin{equation}D^{mid} \ket{\varphi^\prime_1}\approx_{2.5\sqrt{p+2\epsilon}+\sqrt{p+\epsilon}+\fneg(\kappa)}D^\prime\ket{\varphi^\prime_1}\end{equation}
Combining it with \eqref{eq:90ol} completes the proof of \eqref{eq:hgm2}.
\end{enumerate}
Now combining \eqref{eq:hgm1}\eqref{eq:hgm2} completes the proof.
\end{proof}
\section{Analysis of $\fAddPhaseWithswitch$, the Switch Gadget Technique Part}\label{sec:7.3}
Now we return to the $\fpreRSPV$ protocol and see how the switch gadget technique is used to argue about the security.\par
\subsection{Switch Gadget Technique Implies the Output Closeness of Original Adversary and Blinded Adversary in Later Steps}
In the following theorem, we use the switch gadget technique theorem to study the behavior of the output state of $\fAddPhaseWithswitch$ in later steps. We will see, in later steps, we can replace the adversary by its blinded version up to an approximation.

\begin{thm}\label{thm:7.3}
 Suppose a sub-normalized purified joint state $\ket{\varphi}$ in Setup \ref{setup:1} and is in an $\epsilon$-basis-honest form for $\bK$. For any polynomial-time adversary $\fAdv$, suppose the state after the $\fAddPhaseWithswitch$ step is $\ket{\varphi^{2.a}}$:
$$\ket{\varphi^{2.a}}=\fAddPhaseWithswitch^{\fAdv_{2.a}}((\bK^{(i)})_{i\in (\switch)\cup [0,L]},\bTheta;1^\kappa)\ket{\varphi}$$
where $\fAdv_{2.a}$ is the part of adversary in $\fAdv$ in the $\fAddPhaseWithswitch$ step. Suppose 
\begin{equation}\label{eq:switchprtlpassing}|\Pi_{\fpass}\ket{\varphi^{2.a}}|^2\geq 1-p.\end{equation}
Use $\bH^\prime$ to denote the blinded version of $\bH$ where the entries  $\{0,1\}^\kappa||K^{(\switch)}||\cdots$ are blinded. Denote $\fAdv_{\geq 2.b}^{blind}$ as the adversary where all the oracle queries in $\fAdv_{\geq 2.b}$ are replaced by queries to $\bH^\prime$. Then we have
\begin{equation}\label{eq:switchprtl}\fpreRSPV_{\geq 2.b}^{\fAdv_{\geq 2.b}}\ket{\varphi^{2.a}}\approx_{8\sqrt{p+2\epsilon}+\fneg(\kappa)}\fpreRSPV_{\geq 2.b}^{\fAdv_{\geq 2.b}^{blind}}\ket{\varphi^{2.a}}\end{equation}
\end{thm}

The difference of this theorem from Theorem \ref{thm:pht} is the client might send additional messages to the server. Notice that the client-side messages in $\fpreRSPV_{\geq 2.b}$ do not take $\bK^{(\switch)}$ as inputs,  we are able to reduce this theorem to Theorem \ref{thm:pht} by constructing an adversary that simulates these messages.
\begin{proof}[Proof of Theorem \ref{thm:7.3}]
Note that in each step of $\fpreRSPV_{\geq 2.b}$, the client's messages are output of some algorithms which take $(\bK^{(i)})_{i\in [0,L]}$ and $\bTheta$ as inputs, and $\bK^{(\switch)}$ is not used any more. Consider an adversary $(\fAdv_{2.a},\fAdv_{\geq 2.b})$ that violates \eqref{eq:switchprtl}. We can construct an adversary for breaking Theorem \ref{thm:pht} as follows:
\begin{enumerate}
\item[0.] The initial state is \begin{equation}\label{eq:94ol}\ket{\varphi}\odot (\bK^{(i)})_{i\in [0,L]}\odot \bTheta\odot \llbracket\fAddPhaseWithswitch_{1}((\bK^{(i)})_{i\in (\switch)\cup [0,L]},\bTheta;1^\kappa)\rrbracket\end{equation}
where $\fAddPhaseWithswitch_{1}$ is the first step in $\fAddPhaseWithswitch$ (where the client sends lookup tables that encode the phases).\par
By Lemma \ref{lem:5.2a} we know \eqref{eq:94ol} is claw-free for $\bK^{(\switch)}$.
\item The adversary executes $\fHadamardTest(\bK^{(\switch)};1^\kappa)$ with the client. It runs the code of $\fAdv_{2.a}$ in this step. 
\item The adversary simulates all the client side messages in $\fpreRSPV_{\geq 2.b}$ using $(\bK^{(i)})_{i\in [0,L]}$ and $\bTheta$.
 Run $\fAdv_{\geq 2.b}$ with the simulated messages.
\end{enumerate}
\eqref{eq:94ol} satisfies the conditions required in Theorem \ref{thm:pht}. Since \eqref{eq:switchprtlpassing} holds we know the first case in Theorem \ref{thm:pht} is not true. Then the violation of \eqref{eq:switchprtl} implies a violation of the second case in Theorem \ref{thm:pht}, where $D$, $D^\prime$ in Theorem \ref{thm:pht} translate to:
\begin{itemize}
\item $D$ corresponds to $\fAdv_{\geq 2.b}$ run on simulated messages.
\item $D^\prime$ corresponds to $\fAdv_{\geq 2.b}^{blind}$ run on simulated messages.	
\end{itemize}
This completes the proof.
\end{proof}
\paragraph{Implication for later proofs} The implication of Theorem \ref{thm:7.3} is, in the protocol analysis later, we can first use this theorem to replace the adversary by an adversary that only queries the blinded oracle.  Especially, the phase information used in $\fAddPhaseWithswitch$ is encrypted under the switch gadget keys, blinding this part of the oracle implies the secrecy of phase information in later steps.

\subsection{Set-up for the Output State of $\fAddPhaseWithswitch$}
In this section we formalize a set-up that captures the basic properties of the output states of $\fAddPhaseWithswitch$. In the later proofs when we need to further analyze the output state of $\fAddPhaseWithswitch$ we could simply refer to this set-up.
\begin{setup}\label{setup:3}
	Setup \ref{setup:3} is defined as the set of states that could be expressed as
	\begin{equation}\fAddPhaseWithswitch^\fAdv((\bK^{(\switch)},\bK),\bTheta;1^\kappa)\ket{\varphi^1},\end{equation}
	where $\fAdv$ is efficient, $\ket{\varphi^1}$ is in Setup \ref{setup:1}.
	\end{setup}
	Accompanied with Setup \ref{setup:3}, we introduce the following symbols for registers and describe the property of states in Setup \ref{setup:3}:
	\begin{nota}	We introduce the following notations that describe the transcript registers initialized in the first step of $\fAddPhaseWithswitch$:\par
	The client-side messages in the first step of $\fAddPhaseWithswitch$ contain ciphertexts that encrypt $\theta^{(i)}_b$ under $x^{(\switch)}_{b^{(\switch)}}||x_b^{(i)}$ for each $i\in [0,L],b\in \{0,1\},b^{(\switch)}\in \{0,1\}$. In the protocol it is denoted as
	\begin{equation}\label{eq:123}x^{(\switch)}_{b^{(\switch)}}||x_b^{(i)}\rightarrow \theta_b^{(i)}\end{equation}
	Recall in Definition \ref{defn:lt}, \ref{defn:enc}, the ciphertext part of \eqref{eq:123} is defined as
	\begin{equation}\label{eq:124}R^{(i)}_{b^{(\switch)},b},H(R^{(i)}_{b^{(\switch)},b}||x^{(\switch)}_{b^{(\switch)}}||x_b^{(i)})+ \theta_b^{(i)}\end{equation}
	where $R^{(i)}_{b^{(\switch)},b}$ is uniformly sampled from $\{0,1\}^\kappa$.\par
	 Denote the transcript registers that store \eqref{eq:124} during the protocol by $\bR^{(i)}_{b^{(\switch)},b},\bct^{(i)}_{b^{(\switch)},b}$ correspondingly. Then  related registers for \eqref{eq:124} are:
	\begin{equation}\label{eq:125ta}\underbrace{\bx^{(\switch)}_{b^{(\switch)}},\bx^{(i)}_{b},\btheta_b^{(i)}}_{\text{client}},\underbrace{\bH(R^{(i)}_{b^{(\switch)},b}||x^{(\switch)}_{b^{(\switch)}}||x^{(i)}_{b})}_{\text{random oracle}},\underbrace{\bR^{(i)}_{b^{(\switch)},b},\bct_{b^{(\switch)},b}^{(i)}}_{\text{transcript}}\end{equation}
	and there is (recall Notation \ref{nota:register}):
	\begin{equation}\label{eq:117x}\bct^{(i)}_{b^{(\switch)},b}=\bH(R^{(i)}_{b^{(\switch)},b}||x^{(\switch)}_{b^{(\switch)}}||x^{(i)}_{b})+\btheta_b^{(i)}\end{equation}
	 
\end{nota}
	In summary, for any state in Setup \ref{setup:3}, for any $i\in [0,L]$, $b\in \{0,1\}$, $b^{(\switch)}\in \{0,1\}$, there are registers shown in \eqref{eq:125ta} and their values satisfy \eqref{eq:117x}.
	
We also introduce the following blinded oracle that accompanies Setup \ref{setup:3}:
\begin{nota}
	Under Setup \ref{setup:3}, define $\bH^\prime$ as the blinded oracle where entries in the form of $\{0,1\}^\kappa||K^{(\switch)}||\cdots$ are blinded. Define $\cF_{blind}$ as the set of server-side operators that could only query this blinded oracle.
\end{nota}
\subsection{Preparation for the Later Proofs: De-correlate the $\bH$ Registers by $\fReviseRO$ Operator}

We introduce an operation as a preparation for later proofs. Recall that in Example \ref{exmp:r1} we discuss the randomization operators that operate on the client side phase registers and keep the state approximately invariant; but the state that we use in the example has an important difference from the outcome of the real protocol (characterized by Setup \ref{setup:3}). In a state in Setup \ref{setup:3}, the messages of $\fAddPhaseWithswitch$ are stored in the transcript registers, and the server also has access to it. Taking this into consideration, the purified joint state of both parties is generally described as
\begin{equation}\label{eq:125x} \forall i\in [0,L],b\in \{0,1\},b^{(\switch)}\in \{0,1\},\quad \underbrace{\ket{H(R^{(i)}_{b^{(\switch)},b}||x^{(\switch)}_{b^{(\switch)}}||x_b^{(i)})}}_{\bH(R^{(i)}_{b^{(\switch)},b}||x^{(\switch)}_{b^{(\switch)}}||x_b^{(i)})}\underbrace{\ket{\theta_b^{(i)}}}_{\btheta_b^{(i)}}\underbrace{\ket{H(R^{(i)}_{b^{(\switch)},b}||x^{(\switch)}_{b^{(\switch)}}||x_b^{(i)})+\theta_b^{(i)}}}_{\bct^{(0)}_{b^{(\switch)},b}}\end{equation}
together with other registers. Taking these ciphertexts and the related registers into consideration, the swapping operator that swaps $\btheta$ with freshly new randomness (as discussed in the example) does not keep the state invariant.\footnote{Note after the swapping the random oracle content, $\btheta$ register, and the $\bct$ register does not necessarily satisfy Equation \eqref{eq:117x} given in Setup \ref{setup:3}.} We introduce the following operation to make this swap-based randomization work again. This operator will erase the content of the $\bH$ registers, and thus de-correlate this register with the other parts; this operator will significantly change the state, but we can show, the original state and the new state are indistinguishable under the family of distinguishers that we care about.
\begin{defn}[$\fReviseRO$] Under Setup \ref{setup:3}, for each $i\in [0,L],b^{(\switch)}\in \{0,1\},b\in \{0,1\}$, the $\fReviseRO^{(i)}_{b^{(\switch)},b}$ operator is defined as follows. 
	\begin{align}
&\underbrace{\ket{H(R^{(i)}_{b^{(\switch)},b}||x^{(\switch)}_{b^{(\switch)}}||x_b^{(i)})}}_{\bH(R^{(i)}_{b^{(\switch)},b}||x^{(\switch)}_{b^{(\switch)}}||x_b^{(i)})}\underbrace{\ket{\theta_b^{(i)}}}_{\btheta_b^{(i)}}\underbrace{\ket{H(R^{(i)}_{b^{(\switch)},b}||x^{(\switch)}_{b^{(\switch)}}||x_b^{(i)})+\theta_b^{(i)}}}_{\bct^{(i)}_{b^{(\switch)},b}}\\
\xrightarrow{\fReviseRO^{(i)}_{b^{(\switch)},b}}&\underbrace{\ket{0}}_{\bH(R^{(i)}_{b^{(\switch)},b}||x^{(\switch)}_{b^{(\switch)}}||x_b^{(i)})}\underbrace{\ket{\theta_b^{(i)}}}_{\btheta_b^{(i)}}\underbrace{\ket{H(R^{(i)}_{b^{(\switch)},b}||x^{(\switch)}_{b^{(\switch)}}||x_b^{(i)})+\theta_b^{(i)}}}_{\bct^{(i)}_{b^{(\switch)},b}}\label{eq:136se}
\end{align}
Then define $\fReviseRO$ to be the following operation:
\begin{enumerate}
\item For each $i\in [0,L], b^{(\switch)}\in \{0,1\}, b\in \{0,1\}$:
\begin{enumerate}
\item Apply $\fReviseRO^{(i)}_{b^{(\switch)},b}$.
\item Apply Hadamard gates on each bit of RO register  $\bH(R^{(i)}_{b^{(\switch)},b}||x^{(\switch)}_{b^{(\switch)}}||x_b^{(i)})$.
\end{enumerate}
\end{enumerate}
The second step above is to re-create uniformly random coins for the random oracle, and thus preserve the validity of the state in the random oracle model (Definition \ref{defn:validityro}). 
\end{defn}
\begin{fact}
	Suppose $\ket{\varphi}$ is in Setup \ref{setup:3}. Then $\fReviseRO\ket{\varphi}$ is a valid state in QROM.
\end{fact}
\begin{proof}
We only need to prove for each 	$i\in [0,L], b^{(\switch)}\in \{0,1\}, b\in \{0,1\}$, the corresponding operation in the construction preserves the validity of the state. Suppose for the original state $\ket{\varphi}$ when  $\bH(R^{(i)}_{b^{(\switch)},b}||x^{(\switch)}_{b^{(\switch)}}||x_b^{(i)})=h$, the state in registers excluding $\bH(R^{(i)}_{b^{(\switch)},b}||x^{(\switch)}_{b^{(\switch)}}||x_b^{(i)})$ is $\ket{\varphi_h}$. Then $\ket{\varphi_h}$ for different $h$ is orthogonal to each other by \eqref{eq:117x}. By the construction of $\fReviseRO$, in $\fReviseRO\ket{\varphi}$, when $\bH(R^{(i)}_{b^{(\switch)},b}||x^{(\switch)}_{b^{(\switch)}}||x_b^{(i)})=h$, the state in registers excluding $\bH(R^{(i)}_{b^{(\switch)},b}||x^{(\switch)}_{b^{(\switch)}}||x_b^{(i)})$ is a superposition of each $\ket{\varphi_h}$ with different phases. Then since $\ket{\varphi}$ is valid the new state is also valid.
\end{proof}
Below we show the application of $\fReviseRO$ keeps the state indistinguishable under a class of operations that is sufficiently big to cover the distinguishers that we care about in the main proof.
\begin{defn}
$\cF_{cq\land blind}$ is defined to be the set of operators that take the client-side $\bTheta$ registers and the transcript registers read-only and only query the blinded oracle.
\end{defn}
The requirement ``take $\bTheta$ read-only'' corresponds to the client-side read-only requirement on the distinguisher in Definition \ref{defn:rspvv}. What's more, $\cF_{cq\land blind}\subseteq \cF_{blind}$.
\begin{fact}\label{fact:8}
	Suppose $\ket{\varphi}$ is in Setup \ref{setup:3}. Then $\fReviseRO\ket{\varphi}\approx^{ind:\cF_{cq\land blind}}\ket{\varphi}$.
\end{fact}
\begin{proof}
	By Fact \ref{fact:rra} both $\ket{\varphi}$ and $\fReviseRO\ket{\varphi}$ are indistinguishable under $\cF_{cq\land blind}$ to a state where $\bTheta$, $\bct$ registers are all cloned to the environment. Then $\fReviseRO$ becomes an operation that only operates on the registers that $\cF_{cq\land blind}$ never uses thus keeps the state indistinguishable.
\end{proof}

\subsection{Outcome of $\fReviseRO$ is Almost Efficiently-preparable}
In the subsections above we introduce the $\fReviseRO$ operators, and show the output state from the protocol remains indistinguishable under $\cF_{cq\land blind}$. But there is still one more thing to worry about: the output of $\fReviseRO$ is not obviously efficiently-preparable, since it revises the random oracle registers in a way that is not allowed in the definition of efficient-preparation. Below we will show the output state is still efficiently preparable up to an exponentially small error, and show many lemmas that we need still hold for this type of states.
	\subsubsection{Approximate efficient preparation of the output state}
	Recall that a state $\ket{\varphi}$ in Setup \ref{setup:3} could be written as
	\begin{align}
		\ket{\varphi}&=\fAddPhaseWithswitch^{\fAdv}\ket{\varphi^1}\\
		&=\fPadHadamard^{\fAdv}(\ket{\varphi^1}\odot \llbracket \fAddPhaseWithswitch^1\rrbracket)
	\end{align}
	where $\ket{\varphi^1}$ is in Setup \ref{setup:1}, $\llbracket \fAddPhaseWithswitch^1\rrbracket$ is the client-side messages in the first step of $\fAddPhaseWithswitch$ (that is, the lookup tables).\par
	Suppose the set of random padding in $\llbracket \fAddPhaseWithswitch^1\rrbracket$ is $R$. By Lemma \ref{lem:padro} we have, there exists $\ket{\tilde\varphi^1}$ that is independent to $\bH(R||\cdots)$, and $\ket{\tilde\varphi^1}\approx_{\fneg(\kappa)}\ket{\varphi^1}$.\par
	Now we claim the state
	\begin{equation}\label{eq:prq}\fReviseRO\circ\fPadHadamard^\fAdv(\ket{\tilde\varphi^1}\odot \llbracket \fAddPhaseWithswitch^1\rrbracket)\end{equation}
	is efficiently-preparable. We give the following construction. 
	\begin{enumerate}
		\item Starting from $\ket{\tilde\varphi^1}$, fill uniform superpositions in all the registers in $\bct$.
		\item Run $\fPadHadamard^\fAdv$, and for each  query to $H$, replace it by the following operation: if the input has the form of $\bH(R^{(i)}_{b^{(\switch)},b}||x^{(\switch)}_{b^{(\switch)}}||x_b^{(i)})$, use $\bct_{b^{(\switch)},b}^{(i)}-\btheta^{(i)}_b$ as the query outcome; otherwise use the corresponding output of $H$.
	\end{enumerate}
	The reason is as follows. Starting from \eqref{eq:prq}, first since in $\fPadHadamard^\fAdv$ registers $\bct,\btheta$ are read-only, and the $\bH$ registers, $\bct$ registers and $\btheta$ registers satisfy \eqref{eq:117x}, we can replace each query by the construction in the second step above, and the state does not change. Then since $\ket{\tilde\varphi^1}$ does not depend on $\bH(R||\cdots)$, $\fReviseRO$ together with the preparation of $\fAddPhaseWithswitch$ commutes with the preparation of $\ket{\tilde\varphi^1}$ and $\fPadHadamard^\fAdv$. Then a direct calculation of $\fReviseRO$ on the superpositions of all the basis components that satisfy \eqref{eq:117x} shows the output of this operator is the uniform superposition in all the $\bct$ registers and $\btheta$ registers.
	\subsubsection{A list of useful lemmas}
	The discussions above imply that, many lemmas that we proved before under the efficiently-preparable property, still hold for states in the form of
	\begin{equation}\label{eq:138r}
		\fReviseRO\ket{\phi},\ket{\phi}\in \text{ Setup \ref{setup:3}}
	\end{equation}
	Note that when we construct the efficient-preparable operator that approximately prepare \eqref{eq:138r} the operator also operates on the client-side register $\bTheta$ to simulate operators that are originally solely server-side. But we can only focus on lemmas that remain true even if all these $\bTheta$ registers are considered server-side registers. We list the following lemmas that will be used in later proofs.
	\begin{lem}[Analog of Lemma \ref{lem:multisbt}]\label{lem:multisbtre}
Consider a sub-normalized purified joint state $\ket{\varphi}$ in \eqref{eq:138r}. Suppose $\ket{\varphi^\prime}=\fPrtl^\fAdv\ket{\varphi}$ where the protocol $\fPrtl$ and adversary $\fAdv$ are both efficient and do not query the blinded part of $H$. 

Use $\Pi_{\vec{\bx}_{\vec{b}}}^{\bS_{bsh}}$
to denote the projection onto $\vec{\bx}_{\vec{b}}$-branch (Definition \ref{defn:branch}). Then for any efficient server-side operation $D\in \cF_{blind}$ there is
\begin{equation}\label{eq:multisbt}\Pi_{\basishonest(\bK)}D\Pi_{\basishonest(\bK)}\ket{\varphi}\approx_{\fneg(\kappa)}\sum_{\vec{b}\in \{0,1\}^L}\Pi_{\vec{\bx}_{\vec{b}}}^{\bS_{bsh}}D\Pi_{\vec{\bx}_{\vec{b}}}^{\bS_{bsh}}\ket{\varphi}\end{equation}
\end{lem}
The two lemmas below are analog of Lemma \ref{lem:prepht} and Corollary \ref{cor:prepht}, and we use $\bK^\prime$ to replace $\bK$ in these lemmas since the symbol $\bK$ is occupied by the key tuple appeared in the protocol.
\begin{lem}[Analog of Lemma \ref{lem:prepht}]\label{lem:prephtre}
	Consider a sub-normalized purified joint state $\ket{\varphi}$ in \eqref{eq:138r}. Suppose $\ket{\varphi^\prime}=\fPrtl^{\fAdv_0}\ket{\varphi}$ where the protocol $\fPrtl$ and adversary $\fAdv_0$ are both efficient and do not query the blinded part of $H$. Then suppose the client holds a key pair in register  $\bK^\prime=(\bx_0^\prime,\bx_1^\prime)$, and $\ket{\varphi^\prime}$ satisfies all the conditions in  $\SETUP^2(\bK^\prime)$ except the first bullet, and is in the basis-honest form of $\bK^\prime$ with only $\bx_0^\prime$ branch. Then for any efficient adversary $\fAdv\in \cF_{blind}$ there is
	$$|\Pi_{\eqref{eq:htt}=0}^{\bd}\Pi_{\neq 0}^{\text{last $\kappa$ bits of }\bd}\fHadamardTest^\fAdv(\bK^\prime;1^\kappa)\ket{\varphi^\prime}|\approx_{\fneg(\kappa)}\frac{1}{\sqrt{2}}|\Pi_{\neq 0}^{\text{last $\kappa$ bits of }\bd}\fHadamardTest^\fAdv(\bK^\prime;1^\kappa)\ket{\varphi^\prime}|$$
	$$|\Pi_{\eqref{eq:htt}=1}^{\bd}\Pi_{\neq 0}^{\text{last $\kappa$ bits of }\bd}\fHadamardTest^\fAdv(\bK^\prime;1^\kappa)\ket{\varphi^\prime}|\approx_{\fneg(\kappa)} \frac{1}{\sqrt{2}}|\Pi_{\neq 0}^{\text{last $\kappa$ bits of }\bd}\fHadamardTest^\fAdv(\bK^\prime;1^\kappa)\ket{\varphi^\prime}|$$
\end{lem}
\begin{cor}[Analog of Corollary \ref{cor:prepht}]\label{cor:prephtre}
	Consider a sub-normalized purified joint state $\ket{\varphi}$ in \eqref{eq:138r}. Suppose $\ket{\varphi^\prime}=\fPrtl^{\fAdv_0}\ket{\varphi}$ where the protocol $\fPrtl$ and adversary $\fAdv_0$ are both efficient and do not query the blinded part of $H$. Then suppose the client holds a key pair in register  $\bK^\prime=(\bx_0^\prime,\bx_1^\prime)$, and $\ket{\varphi^\prime}$ satisfies all the conditions in  $\SETUP^2(\bK^\prime)$ except the first bullet, and is in an $\epsilon$-basis-honest form of $\bK^\prime$. Denote the $\bx_0^\prime$-branch as $\ket{\varphi_0^\prime}$ and $\bx_1^\prime$-branch as $\ket{\varphi_1^\prime}$. If an efficient adversary $\fAdv\in \cF_{blind}$ could make the client output $\fpass$ in the Hadamard test (Protocol \ref{prtl:phadamard}) from initial state $\ket{\varphi^\prime}$ with probability $\geq 1-p$, then
		\begin{align}\label{eq:98y2}&\Pi_{\eqref{eq:htt}=0}^{\bd}\Pi_{\neq 0}^{\text{last $\kappa$ bits of }\bd}\fHadamardTest^\fAdv(\bK^\prime;1^\kappa)\ket{\varphi_0^\prime}\nonumber\\\approx_{\sqrt{2(p+2\epsilon)}+\fneg(\kappa)}&\Pi_{\eqref{eq:htt}=0}^{\bd}\Pi_{\neq 0}^{\text{last $\kappa$ bits of }\bd}\fHadamardTest^\fAdv(\bK^\prime;1^\kappa)\ket{\varphi_1^\prime}\end{align}
	\begin{equation*}\Pi_{=0}^{\text{last $\kappa$ bits of }\bd}\fHadamardTest^\fAdv(\bK^\prime;1^\kappa)\ket{\varphi_0^\prime}\approx_{\sqrt{p+2\epsilon}+\fneg(\kappa)}0,\end{equation*}\begin{equation}\label{eq:78ol2} \Pi_{= 0}^{\text{last $\kappa$ bits of }\bd}\fHadamardTest^\fAdv(\bK^\prime;1^\kappa)\ket{\varphi_1^\prime}\approx_{\sqrt{p+2\epsilon}+\fneg(\kappa)} 0\end{equation}
		\begin{align}&\Pi_{\eqref{eq:htt}=1}^{\bd}\Pi_{\neq 0}^{\text{last $\kappa$ bits of }\bd}\fHadamardTest^\fAdv(\bK^\prime;1^\kappa)\ket{\varphi_0^\prime}\nonumber\\\approx_{\sqrt{p+\epsilon}+\fneg(\kappa)}&-\Pi_{\eqref{eq:htt}=1}^{\bd}\Pi_{\neq 0}^{\text{last $\kappa$ bits of }\bd}\fHadamardTest^\fAdv(\bK^\prime;1^\kappa)\ket{\varphi_1^\prime}\end{align}
\end{cor}
We put their proofs in Appendix \ref{sec:addms}.
%
%

\section{Analysis of $\fAddPhaseWithswitch$, the Lookup-table Part}\label{sec:8}
 As described in the introduction, we will prove the joint state of the client and the server remain indistinguishable under a series of randomization operators.\par
In this section, we will focus on what the $\fAddPhaseWithswitch$ protocol itself could tell us. The analysis will be based on the structure of phase tables.\par
Recall that each row of the phase tables used in this protocol has the following structure:
\begin{equation}\label{eq:50}x^{(\switch)}_{b^{(\switch)}}||x^{(i)}_b\rightarrow \theta^{(i)}_b;\quad i\in [0,L],b^{(\switch)},b\in \{0,1\}\end{equation}
Intuitively, the theorems that we prove in this section is based on the following intuition: if the server does not hold $x_b^{(i)}$, it could not decrypt  \eqref{eq:50}, thus $\theta^{(i)}_b$ remains secure.	(Recall by the claw-freeness the server does not know $\bx_b$ on the $\bx^{(i)}_{1-b}$-branch, and $\btheta_{b}^{(i)}$ is the client-side register that stores the corresponding phase information.)
Furthermore, for each $\vec{b}=b^{(0)}b^{(1)}\cdots b^{(L)}\in \{0,1\}^{1+L}$, on the $\vec{\bx}_{\vec{b}}$-branch, the server holds $x^{(0)}_{b^{(0)}}x^{(1)}_{b^{(1)}}\cdots x^{(L)}_{b^{(L)}}$ but does not know $x^{(0)}_{1-b^{(0)}}x^{(1)}_{1-b^{(1)}}\cdots x^{(L)}_{1-b^{(L)}}$. Thus intuitively:
\begin{center}
\emph{On the $\vec{\bx}_{\vec{b}}$-branch, the server-side state should not depend on the values of $\btheta^{(0)}_{1-b^{(0)}}\btheta^{(1)}_{1-b^{(1)}}\cdots \btheta^{(L)}_{1-b^{(L)}}$.}	
\end{center}
In this section we build the bridge between this intuition and the protocol  as follows.
\begin{itemize}
	\item In Section \ref{sec:8.1} we define the \emph{basis-phase correspondence form}, which characterizes the states that perfectly satisfy the intuition above.
	\item In Section \ref{sec:8.2r} we construct a randomization operator $\cR_1$, which transforms a general basis-honest form to a basis-phase-correspondence form.
	\item In Section \ref{sec:8.3} we show the output states of $\fAddPhaseWithswitch$ are approximately invariant under $\cR_1$.
	\item Finally in Section \ref{sec:8.4} we formalize a new setup for the outcome of $\cR_1$, which will be used in later analysis.
\end{itemize}
\subsection{Basis-phase Correspondence Form}\label{sec:8.1}
As discussed in the technical overview, we define the \emph{basis-phase correspondence form} as follows.
\begin{defn}
Assume the parties are as in Section \ref{sec:4.1} and the registers are as in Section \ref{sec:6.1r}. We say a state $\ket{\varphi}$ is in a basis-phase correspondence form if there exists states $\ket{\varphi_{K,\vec{b},\vec{\theta}}}$ for each $K\in \Domain(\bK)$, $\vec{b}\in \{0,1\}^{1+L}$, $\vec{\theta}\in \{0,1\cdots 7\}^{1+L}$ such that 
$$\ket{\varphi}=\underbrace{\sum_{K\in \Domain(\bK)}\ket{K}\otimes \sum_{\Theta\in \Domain(\bTheta)}\ket{\Theta}}_{\text{client}}
\otimes \underbrace{\sum_{\vec{b}\in \{0,1\}^{1+L}}\ket{\vec{x}_{\vec{b}}}}_{\text{server-side register $\bS_{bsh}$}}\otimes \ket{\varphi_{K,\vec{b},\vec{\Theta}_{\vec{b}}}}$$
Recall $\vec{\Theta}_{\vec{b}}$ denotes $\theta^{(0)}_{b^{(0)}}\theta^{(1)}_{b^{(1)}}\cdots \theta^{(L)}_{b^{(L)}}$.
\end{defn}

\subsection{Randomization operator $\cR_1$}\label{sec:8.2r}
In this subsection we define the randomization operator $\cR_1$.
\subsubsection{Intuitive discussion}
\begin{exmp}First we can recall that an intuitive introduction of the $\cR_1$ operator is given in Example \ref{exmp:r1}. Considering the $x_0$-branch in the example, if we omit the unused registers, Equation \eqref{eq:45v} in the example becomes
	\begin{equation}\label{eq:119g}\ket{\Delta_1}\ket{\theta_1}\ket{x_0}\rightarrow \ket{\theta_1}\ket{\Delta_1}\ket{x_0}\end{equation}
Under Setup \ref{setup:3}, considering the $\bx^{(0)}_0$-branch as an example, \eqref{eq:119g} becomes
	\begin{equation}\label{eq:120g}\ket{\Delta_1^{(0)}}\underbrace{\ket{\theta_1^{(0)}}}_{\btheta^{(0)}_0}\underbrace{\ket{x_0^{(0)}}}_{\text{server-side register $\bS_{bsh}^{(0)}$}}\rightarrow \ket{\theta_1^{(0)}}\ket{\Delta_1^{(0)}}\ket{x_0^{(0)}}\end{equation}\end{exmp}
 We will see $\cR_1$ is defined to be this type of operations applied on each possible superscript in $[0,L]$ and each possible subscript in $\{0,1\}$.
 
\subsubsection{Formal definition}
To formalize this operator, let's define an operator that operates on a $\bx^{(i)}_{1-b}$ branch of a basis-honest state, and randomizes the phase information register $\btheta_{b}^{(i)}$ by swapping it with a completely new random value $\Delta^{(i)}_{b}$:
\begin{defn}\label{defn:swap}
Recall in Setup \ref{setup:3} the client holds a tuple of key pairs in register $\bK=(\bK^{(i)})_{i\in [0,L]}$, $\bK^{(i)}=(\bx^{(i)}_b)_{b\in \{0,1\}}$. And the client additionally holds a tuple of phase pairs in register $\bTheta=(\bTheta^{(i)})_{i\in [0,L]}$, $\bTheta^{(i)}=(\btheta^{(i)}_{b})_{b\in \{0,1\}}$, $\btheta^{(i)}_b\in \{0,1\cdots 7\}$. 
  For a purified joint state $\ket{\varphi}$, expand the basis-honest part:
\begin{equation}\label{eq:r1}\Pi_{\basishonest(\bK)}\ket{\varphi}=\underbrace{\sum_{K\in \Domain(\bK)}\ket{K}\otimes \sum_{\Theta\in \Domain(\bTheta)}\ket{\Theta}}_{\text{client}}
\otimes \underbrace{\sum_{\vec{b}\in \{0,1\}^{1+L}}\ket{\vec{x}_{\vec{b}}}}_{\text{server-side register $\bS_{bsh}$}}\otimes \ket{\varphi_{K,\Theta,\vec{b}}}\end{equation}

For any $i\in [0,L], b\in \{0,1\}$, define $\fSWAP_{1-b,b}^{(i)}$ as follows. First initialize randomness register $\bDelta^{(i)}_{b}$ to hold uniformly distributed value $\in \{0,1\cdots 7\}$. Then $\fSWAP_{1-b,b}^{(i)}$ is the following control-swap operation that acts nontrivially on the $\bx^{(i)}_{1-b}$ branch of \eqref{eq:r1} and swaps the value of $\btheta^{(i)}_{b}$ with the value of $\bDelta_b^{(i)}$:
\begin{align}&\underbrace{\ket{\Delta^{(i)}_b}}_{\bDelta^{(i)}_{b}}\underbrace{\ket{K}\ket{\cdots,\underbrace{\theta^{(i)}_{b}}_{\btheta^{(i)}_b},\cdots  }}_{\text{client}}\otimes \underbrace{\ket{\cdots x^{(i)}_{1- b}\cdots }}_{\text{server-side register $\bS_{bsh}$}}\quad\quad(K=(K^{(i)})_{i\in [0,L]},K^{(i)}=(x^{(i)}_b)_{b\in \{0,1\}})\\
\xrightarrow{\fSWAP_{1-b,b}^{(i)}} &	\ket{\theta_b^{(i)}}\ket{K}\ket{\cdots,\Delta^{(i)}_b,\cdots  }\otimes \ket{\cdots x^{(i)}_{1- b}\cdots }
\end{align}
The operator acts as identity on the other branch and outside $\Pi_{\basishonest(\bK)}$.\par
\end{defn}

Then the randomization operator $\cR_1$ is to apply $\fSWAP^{(i)}_{1-b,b}$ for all the possible $i,b$: 
\begin{defn} Consider the same register set-up as Definition \ref{defn:swap}. Additionally introduce registers $\bDelta^{(i)}_b$, $i\in [0,L],b\in \{0,1\}$, which are initialized to hold the state
$$\ket{\$_1}=\frac{1}{\sqrt{8^{2(L+1)}}}\sum_{\forall i\in [0,L],\forall b\in \{0,1\}: \Delta^{(i)}_{b}\in \{0,1\cdots 7\}}\ket{(\underbrace{\Delta^{(i)}_{b}}_{\bDelta_b^{(i)}})_{i\in [0,L], b\in \{0,1\}}}$$
Define randomization operator $\cR_1$ as
$$\circ_{\forall i\in [0,L],\forall b\in \{0,1\}}\fSWAP^{(i)}_{1-b,b}$$

\end{defn}
Note that (1) for each $i,b$ the $\fSWAP$ operator uses freshly new randomness; (2) $\circ$ denote the operator composition; since these $\fSWAP$ operators commute with each other the order of applying these $\fSWAP$ operators does not matter; (3) when we say ``randomizing a state $\ket{\varphi}$ with $\cR_1$'', we mean applying $\cR_1$ on $\ket{\$_1}\otimes \ket{\varphi}$.\par
 We will show:
 \begin{itemize}\item Applying $\cR_1$ to the output of $\fAddPhaseWithswitch$ keeps the state indistinguishable;
 	\item Applying $\cR_1$ takes the state to a  class of states that satisfy specific properties. 
 \end{itemize}
\subsection{Phase Table Structure Implies Approximate Invariance Under Randomization of $\cR_1$}\label{sec:8.3}
The following theorem says the application of $\cR_1$ keeps the state approximately invariant.
\begin{thm}\label{thm:indofr1}
	 Suppose a sub-normalized purified joint state $\ket{\varphi^{2.a}}$ is in Setup \ref{setup:3}. 
 Then
\begin{equation}\label{eq:124v}\cR_1(\ket{\$_1}\otimes \fReviseRO\ket{\varphi^{2.a}})\approx_{\fneg(\kappa)}\fReviseRO\ket{\varphi^{2.a}}\end{equation}
\end{thm}
\subsubsection{Linear algebra fact that connects state form to approximate-invariance of operator}
To prove this theorem, we give the following linear algebra fact that says, if the state is close to a specific form, it is approximate-invariant under a randomization operator.
\begin{fact}\label{fact:otpro}
Suppose 	$\ket{\varphi}$ satisfies:\begin{equation}\label{eq:otpro1}\ket{\varphi}\approx_{\epsilon}\frac{1}{\sqrt{8}}\sum_{\theta\in \{0,1\cdots 7\}}\ket{\theta}\otimes\ket{\psi}\end{equation}
Then
$$\frac{1}{\sqrt{8}}\sum_{\Delta\in \{0,1\cdots 7\}}\ket{\Delta}\otimes \ket{\varphi}$$
is $2\epsilon$-invariant under the following operator:
\begin{equation}
	\ket{\Delta}\ket{\theta}\rightarrow \ket{\theta}\ket{\Delta}
\end{equation}
\end{fact}
\subsubsection{Proof of Theorem \ref{thm:indofr1}}
\begin{proof}
	Unrolling the definition of $\cR_1$, we only need to prove, for all $i\in [0,L]$,
	$b\in \{0,1\}$:
	\begin{equation}\label{eq:115pl}\fSWAP^{(i)}_{1-b,b}(\frac{1}{\sqrt{8}}\sum_{\Delta_b^{(i)}\in \{0,1\cdots 7\}}\underbrace{\ket{\Delta_b^{(i)}}}_{\bDelta^{(i)}_b}\otimes \fReviseRO\ket{\varphi^{2.a}})\approx_{\fneg(\kappa)}\fReviseRO\ket{\varphi^{2.a}}\end{equation}
	Denote the $\bx^{(i)}_{1-b}$ branch of $\ket{\varphi^{2.a}}$ as $\ket{\varphi_{(i),1-b}^{2.a}}$. Since $\fSWAP^{(i)}_{1-b,b}$ only operates nontrivially on $\ket{\varphi_{(i),1-b}^{2.a}}$, \eqref{eq:115pl} is further reduced to proving 
	\begin{equation}\label{eq:76r}\fSWAP^{(i)}_{1-b,b}(\frac{1}{\sqrt{8}}\sum_{\Delta_b^{(i)}\in \{0,1\cdots 7\}}\ket{\Delta_0^{(i)}}\otimes\fReviseRO\ket{\varphi_{(i),1-b}^{2.a}})\approx_{\fneg(\kappa)}\fReviseRO\ket{\varphi_{(i),1-b}^{2.a}}\end{equation}
 Recall the definition of $\ket{\varphi^{2.a}}$ in Setup \ref{setup:3}:
$$\ket{\varphi^{2.a}}=\fAddPhaseWithswitch^{\fAdv}(\bK^{(\switch)},\bK,\bTheta;1^\kappa)\ket{\varphi^1}$$
where $\ket{\varphi^1}$ is in Setup \ref{setup:1}. Applying Lemma \ref{lem:multisbt} we get

\begin{equation}\label{eq:133d}\ket{\varphi^{2.a}_{(i),1-b}}\approx_{\fneg(\kappa)}\fAddPhaseWithswitch^{\fAdv}(\bK^{(\switch)},\bK,\bTheta;1^\kappa)\ket{\varphi^1_{(i),1-b}}\end{equation}
where we use $\ket{\varphi^1_{(i),1-b}}$ to denote the $\bx_{1-b}^{(i)}$-branch of $\ket{\varphi^1}$.\par
To prove it, we are going to use Fact \ref{fact:otpro}, a linear algebra fact that is proved previously for the preparation of this proof. We need to show $\ket{\varphi_{(i),1-b}^{2.a}}$ is close to a state that satisfies the condition of Fact \ref{fact:otpro} (that is, does not depend on $\btheta^{(i)}_b$). This is by replacing state and operations in \eqref{eq:133d} by states and operators that does not depend on the values of $\btheta^{(i)}_b$ step by step: 
	\begin{enumerate}
\item Suppose the random paddings used within the $\fAddPhaseWithswitch$ are sampled and stored in client-side register $\bR$. Since $\ket{\varphi_{(i),1-b}^1}$ is efficiently-preparable, by Lemma \ref{lem:padro}:
\begin{equation}\label{eq:130s}\exists \ket{\tilde\varphi^1_{(i),1-b}}\text{ independent of $\bH(R||\cdots)$ }: \ket{\tilde\varphi^1_{(i),1-b}}\approx_{\fneg(\kappa)}\ket{\varphi^1_{(i),1-b}}\end{equation}
\item Define $\fAdv^\prime$ as the adversary that compared to $\fAdv$, all the queries to $\bH$ are replaced by $\bH^{mid}$ where $\bH^{mid}$ is a blinded oracle where entries in the form of $\{0,1\}^{2\kappa} || x^{(i)}_{b}$ are blinded. (Recall the key length of the switch gadget is also $\kappa$.)  Intuitively, recall we are studying the $\bx^{(i)}_{1-b}$-branch thus $\bx^{(i)}_{b}$ is not predictable by the server. Formally, by Lemma \ref{lem:5.2b} we have
\begin{align}&\fAddPhaseWithswitch^\fAdv(\bK^{(\switch)},\bK,\bTheta;1^\kappa)\ket{\varphi_{(i),1-b}^1}\\\approx_{\fneg(\kappa)}&\fAddPhaseWithswitch^{\fAdv^\prime}(\bK^{(\switch)},\bK,\bTheta;1^\kappa)\ket{\varphi_{(i),1-b}^1}\label{eq:131s}\end{align}
\item Combining \eqref{eq:130s}\eqref{eq:131s} above we have
$$\ket{\tilde\varphi^{2.a}_{(i),1-b}}:=\fAddPhaseWithswitch^{\fAdv^\prime}(\bK^{(\switch)},\bK,\bTheta;1^\kappa)\ket{\tilde\varphi^1_{(i),1-b}}\approx_{\fneg(\kappa)}\ket{\varphi^{2.a}_{(i),1-b}}$$
and by its definition it is independent to  random oracle output register $\bH(R||\{0,1\}^{\kappa}||x_b^{(i)})$. 
Recall the encryptions of $\theta_b^{(i)}$ in the client-side messages of $ \fAddPhaseWithswitch$ have the form
$$ct^{(i)}_{b^{(\switch)},b}=h+\theta_b^{(i)},h=H(R^{(i)}_{b^{(\switch)},b}||x^{(\switch)}_{b^{(\switch)}}||x_b^{(i)})$$
Thus after the application of $\fReviseRO^{(i)}_{b^{(\switch)},b}$ on $\ket{\tilde\varphi^{2.a}_{(i),1-b}}$ the overall state could be written as
$$\sum_{\theta^{(i)}_{b^{(\switch)},b}\in \{0,1\cdots 7\}}\sum_{h\in \{0,1\cdots 7\}}\underbrace{\ket{\theta^{(i)}_{b^{(\switch)},b}}}_{\btheta^{(i)}_{b^{(\switch)},b}}\underbrace{\ket{ct^{(i)}_{b^{(\switch)},b}}}_{\bct^{(i)}_{b^{(\switch)},b}}\underbrace{\ket{\psi_{ct^{(i)}_{b^{(\switch)},b}}}}_{\text{other parts}}$$
That is, the rightmost term does not depend on the value of register $\btheta^{(i)}_{b^{(\switch)},b}$ (but could depend on the transcript register $\bct^{(i)}_{b^{(\switch)},b}$). Now the condition for applying Fact \ref{fact:otpro} is satisfied. This completes the proof of \eqref{eq:76r}. 
\end{enumerate}
This completes the proof of \eqref{eq:115pl} and completes the whole proof.
\end{proof}
\subsection{New Set-up}\label{sec:8.4}
Now we are going to formalize a new setup that captures the properties of states after the randomization of $\cR_1$.
\begin{setup}\label{setup:4}
	
	Define Setup \ref{setup:4} as the set of states that could be written as 
	$$\cR_1(\ket{\$_1}\otimes \fReviseRO\ket{\varphi}),\ket{\varphi}\text{ is in Setup \ref{setup:3}}.$$
	\end{setup}
Since there have been lots of nesting in the definition, let's give a recap of the properties of states in Setup \ref{setup:4}.\par
In Setup \ref{setup:4} the client holds a tuple of key pairs in register $\bK=(\bK^{(i)})_{i\in [0,L]}$, $\bK^{(i)}=(\bx^{(i)}_b)_{b\in \{0,1\}}$. And the client additionally holds a tuple of phase pairs in register $\bTheta=(\bTheta^{(i)})_{i\in [0,L]}$, $\bTheta^{(i)}=(\btheta^{(i)}_{b})_{b\in \{0,1\}}$, $\btheta^{(i)}_b\in \{0,1\cdots 7\}$. $\bH^\prime$ is defined to be the blinded oracle where entries in the form of $\{0,1\}^\kappa||K^{(\switch)}||\cdots$ are blinded, which covers the target registers of $\fReviseRO$. $\cF_{blind}$ is defined to be the set of server-side operators that could only query this blinded oracle.\par
 A state $\ket{\varphi}$ in Setup \ref{setup:4} satisfies:
	\begin{itemize}
		\item It is key checkable for each key pair in $\bK$ by an operator in $\cF_{blind}$;
		\item It is strongly-claw-free for each key pair in $\bK$ against efficient operators in $\cF_{blind}$;
		\item For each possible value $\Theta$ of registers $\bTheta$, the corresponding component of $\ket{\varphi}$ has norm $\frac{1}{\sqrt{|\Domain(\bTheta)|}}$.
		\item For any $i\in [0,L],b\in \{0,1\}$, the $\bx^{(i)}_{1-b}$-branch of $\ket{\varphi}$ does not depend on the value of $\btheta^{(i)}_{b}$.
	\end{itemize}
\section{Analysis of Collective Phase Test ($\fCoPhTest$)}
In this section, we will analyze the implication of collective phase test.\par
Before going to the protocol analysis, in Section \ref{sec:9.1} we will define a series of notions including the \emph{pre-phase-honest form}, and the \emph{phase-honest form}. 
Then in Section \ref{sec:8.2}, we will define a new randomization operator $\cR_2$.
\begin{itemize}
\item Following the definition, we will show this operator will transform a state in Setup \ref{setup:4} (a basis-honest form with a specific property) to a pre-phase-honest form.\par
\item In Section \ref{sec:9.3} we will show if an efficient adversary could pass the collective phase test from a state in Setup \ref{setup:4}, the overall state could be further randomized under $\cR_2$.
\end{itemize}

\subsection{Pre-phase-honest Form and Phase-honest Form}\label{sec:9.1}
Let's assume the Setup \ref{setup:4} and assume the state is in a basis-honest form for $\bK$. We say this state is a pre-phase-honest form or a phase-honest form if it is a basis-honest form and has additional structure related to the phase information in $\bTheta$. Recall by Definition \ref{defn:branch} and Notation \ref{nota:compo} the basis-honest form could be written as a linear sum of different branches, and each branch could be written as the sum of different components based on different values of $\bTheta$ register.\par 
For preparation, let's first define the \emph{honest joint phase} of a branch:
\begin{defn}[Honest joint phase]\label{defn:hjp}
Suppose the client holds a tuple of key pairs $\bK=(\bK^{(i)})_{i\in [0,L]}$ and a tuple of phase pairs $\bTheta=(\bTheta^{(i)})_{ i\in [0,L]},\bTheta^{(i)}=(\btheta_0^{(i)},\btheta^{(i)}_1)$. For subscript vector $\vec{b}=b^{(0)}b^{(1)}b^{(2)}\cdots b^{(L)}\in \{0,1\}^{1+L}$, we call $\tSUM(\vec{\Theta}_{\vec{b}})=\theta^{(0)}_{b^{(0)}}+\theta^{(1)}_{b^{(1)}}+\theta^{(2)}_{b^{(2)}}+\cdots +\theta^{(L)}_{b^{(L)}}$  the \emph{honest joint phase} for $\vec{\bx}_{\vec{b}}$-branch when the value of $\bTheta$ is $\Theta$.
\end{defn}

Then informally:\begin{itemize} \item A pre-phase honest form satisfies: if $\tSUM(\vec{\Theta}_{\vec{b}_1})=\tSUM(\vec{\Theta}_{\vec{b}_2})$, then the $\vec{x}_{\vec{b}_1}$ branch is the same as the $\vec{x}_{\vec{b}_2}$ branch, excluding the registers that are necessarily different (which is the server-side key vector register $\bS_{bsh}$ and the client side phase register $\bTheta$). But we do not restrict the phases of states with different $\tSUM(\vec{\Theta}_{\vec{b}})$.
\item A state in the phase-honest form means the branch $\vec{\bx}_{\vec{b}}$ has phase $\tSUM(\vec{\Theta}_{\vec{b}})$ besides the requirement in pre-phase-honest form.
\end{itemize}
\begin{defn}[Pre-phase-honest form]\label{defn:pphf}
		We say a purified joint state $\ket{\varphi}$ in Setup \ref{setup:4} is in the pre-phase honest form if there exists a family of states $\ket{\varphi_{K,\vec{b},sum}}$ for each $K\in \Domain(\bK),\vec{b}\in \{0,1\}^{1+L},sum\in \{0,1\cdots 7\}$ such that $\ket{\varphi}$ could be written as
\begin{equation}\label{eq:pphf}\underbrace{\sum_{K\in \Domain(\bK)}\ket{K}\otimes \sum_{\Theta\in \Domain(\bTheta)}\ket{\Theta}}_{\text{client}}
\otimes \underbrace{\sum_{\vec{b}\in \{0,1\}^{1+L}}\ket{\vec{x}_{\vec{b}}}}_{\text{server-side $\bS_{bsh}$}}\otimes \ket{\varphi_{K,\vec{b},\tSUM(\vec{\Theta}_{\vec{b}})}}\end{equation}
\end{defn}

\begin{defn}[Phase-honest form]\label{defn:phf}	
	We say a purified joint state $\ket{\varphi}$ in Setup \ref{setup:4} is in the pre-phase honest form if there exists a family of states $\ket{\varphi_{K,\vec{b},+}},\ket{\varphi_{K,\vec{b},-}}$ for each $K\in \Domain(\bK),\vec{b}\in \{0,1\}^{1+L}$ such that $\ket{\varphi}$ could be written as
\begin{equation}\label{eq:phf}\underbrace{\sum_{K\in \Domain(\bK)}\ket{K}\otimes \sum_{\Theta\in \Domain(\bTheta)}\ket{\Theta}}_{\text{client}}\otimes \underbrace{\sum_{\vec{b}\in \{0,1\}^{1+L}}\ket{\vec{x}_{\vec{b}}}}_{\text{server-side $\bS_{bsh}$}}\otimes (e^{\tSUM(\vec{\Theta}_{\vec{b}})\mi \pi/4}\ket{\varphi_{K,\vec{b},+}}+e^{-\tSUM(\vec{\Theta}_{\vec{b}})\mi \pi/4}\ket{\varphi_{K,\vec{b},-}})\end{equation}

\end{defn}
As expected, in the formal definition we need to take the complex-conjugate term into consideration.
\subsection{Randomization Operator $\cR_2$}\label{sec:8.2}
Let's start to define the randomization operator that randomizes a basis-honest form in Setup \ref{setup:4} to a pre-phase-honest form. To do that, we will define an operator $\fAdd_{\vec{b}}$ operated on the $\vec{\bx}_{\vec{b}}$-branch, for $\vec{b}\in \{0,1\}^{1+L}$.
\begin{defn} 
Consider a purified joint state $\ket{\varphi}$ in Setup \ref{setup:4} and a basis-honest form of $\bK$:
$$\ket{\varphi}=\underbrace{\sum_{K\in \Domain(\bK)}\ket{K}\otimes \sum_{\Theta\in \Domain(\bTheta)}\ket{\Theta}}_{\text{client}}\otimes \underbrace{\sum_{\vec{b}\in \{0,1\}^{1+L}}\ket{\vec{x}_{\vec{b}}}}_{\text{server-side $\bS_{bsh}$}}\otimes \ket{\varphi_{K,\Theta,\vec{b}}}$$
 Define the operator $\fAdd_{\vec{b}}$ controlled on a specific branch indexed by $\vec{b}=b^{(0)}b^{(1)}b^{(2)}\cdots b^{(L)}$, with randomness $\Delta^{(1)}_{b^{(1)}}\cdots \Delta^{(L)}_{b^{(L)}}$:
\begin{align}&\underbrace{\ket{\Delta^{(1)}_{b^{(1)}}\cdots \Delta^{(L)}_{b^{(L)}}}}_{\bDelta^{(1)}_{b^{(1)}}\cdots \bDelta^{(L)}_{b^{(L)}}}\underbrace{\ket{\theta^{(0)}_{b^{(0)}}\theta^{(1)}_{b^{(1)}}\cdots \theta^{(L)}_{b^{(L)}}}}_{\text{client-side registers $\btheta^{(0)}_{b^{(0)}}\btheta^{(1)}_{b^{(1)}}\cdots \btheta^{(L)}_{b^{(L)}}$}}\underbrace{\ket{\vec{x}_{\vec{b}}}}_{\text{server-side $\bS_{bsh}$}}\\\xrightarrow{\fAdd_{\vec{b}}}& \ket{\theta^{(1)}_{b^{(1)}}\cdots \theta^{(L)}_{b^{(L)}}}\underbrace{\ket{(\sum_{i\in [0,L]}\theta^{(i)}_{b^{(i)}}-\sum_{i\in [L]}\Delta^{(i)}_{b^{(i)}})\Delta^{(1)}_{b^{(1)}}\cdots \Delta^{(L)}_{b^{(L)}}}}_{\btheta^{(0)}_{b^{(0)}}\btheta^{(1)}_{b^{(1)}}\cdots \btheta^{(L)}_{b^{(L)}}}\ket{\vec{x}_{\vec{b}}}\end{align}
and acts as identity on the other branches and outside $\Pi_{\basishonest(\bK)}$.\end{defn}
This means:
\begin{itemize}\item This operator only operates on the ${\vec{\bx}_{\vec{b}}}$-branch of the state, and randomizes the $\theta^{(0)}_{b^{(0)}}\theta^{(1)}_{b^{(1)}}\cdots \theta^{(L)}_{b^{(L)}}$ information stored on the client side with randomness $\Delta$. Recall that in the $\fpreRSPV$ protocol $\theta_{b^{(i)}}^{(i)}$ could be decrypted with keys in $\vec{\bx}_{\vec{b}}$. Thus this operator aims at randomizing the phases that \emph{could} be decrypted by the server's keys, which is different from the operator $\cR_1$.
\item The randomization is done in a way that for each branch, the honest joint phase (Definition \ref{defn:hjp}) of this branch remains the same.

\end{itemize}
Then define the overall randomization operator for the collective phase test as the composition of $\fAdd_{\vec{b}}$ for each $\vec{b}$, with suitable choice of randomness:
\begin{defn}[$\cR_2$]
For each $i\in [L]$, $b\in \{0,1\}$, initialize register $\bDelta^{(i)}_{b}$ to store a uniform superposition of $\{0,1\cdots 7\}$. Overall these registers are initialized to hold the state
\begin{equation}\label{eq:139k}\ket{\$_2}=\frac{1}{\sqrt{8^{2L}}}\sum_{\forall b\in \{0,1\},i\in [L]:\Delta^{(i)}_{b}\in \{0,1\cdots 7\}}\ket{(\underbrace{\Delta_{b}^{(i)}}_{\bDelta_{b}^{(i)}})_{i\in [L],b\in \{0,1\}}}\end{equation}
Define $\cR_2$ as
$$\cR_2=\circ_{\vec{b}\in \{0,1\}^{1+L}}\fAdd_{\vec{b}}$$
which operates on a state in Setup \ref{setup:4} together with $\ket{\$_2}$.
\end{defn}
We have the following theorems about $\cR_2$.
\begin{thm}
On state $\ket{\varphi}$ in Setup \ref{setup:4}, $\cR_2$ could be efficiently implemented with access to the transcript registers, $\bS_{bsh}$ and registers of $\ket{\$_2}$.
\end{thm}
\begin{proof}
$\cR_2$ can be implemented through the following operations.
\begin{enumerate}
	\item Use the key-checkable operators to calculate the subscript vector for keys in $\bS_{bsh}$.
	\item For the $\vec{\bx}_{\vec{b}}$-branch, controlled by the subscript vector $\vec{b}$, apply operator $\fAdd_{\vec{b}}$.
	\item Redo the first step to erase the temporary register that stores the subscripts.
\end{enumerate}
\end{proof}
\begin{thm}\label{thm:9.pre}
If $\ket{\varphi}$ is in Setup \ref{setup:4} and is in a basis-honest form, $\cR_2(\ket{\$_2}\otimes\ket{\varphi})$ is in a pre-phase-honest form.
\end{thm}
\begin{proof}
By the definition of pre-phase-honest form we can study the structure of each branch separately. For $\vec{b}\in \{0,1\}^{1+L}$, denote the $\vec{\bx}_{\vec{b}}$-branch of $\ket{\varphi}$ as
$$\ket{\varphi_{\vec{b}}}=\sum_{\Theta\in \Domain(\bTheta)}\ket{\Theta}\otimes \ket{\varphi_{\Theta,\vec{b}}}$$
where we make the client-side key registers implicit and make the phase registers explicit. By the condition that $\ket{\varphi}$ is in Setup \ref{setup:4}, we know $\ket{\varphi_{\vec{b}}}$ does not depend on the values of registers $\btheta_{1-b^{(0)}}^{(0)}\btheta_{1-b^{(1)}}^{(1)}\cdots \btheta_{1-b^{(L)}}^{(L)}$. Thus we could write $\ket{\varphi_{\vec{b}}}$ as:
$$\ket{\varphi_{\vec{b}}}=\sum_{\Theta\in \Domain(\bTheta)}\ket{\Theta}\otimes \ket{\varphi_{\vec{b},\theta_{b^{(0)}}^{(0)}\theta_{b^{(1)}}^{(1)}\cdots \theta_{b^{(L)}}^{(L)}}}$$
By direct calculation, after the application of $\fAdd_{\vec{b}}$ it becomes (note that in the calculation below we omit some unused registers)
\begin{align}&\fAdd_{\vec{b}}(\ket{\$_2}\otimes\ket{\varphi_{\vec{b}}})\\
=&\sum_{\forall i\in [L]:\Delta^{(i)}_{b^{(i)}}\in \{0,1\cdots 7\}}\sum_{\forall i\in [0,L]:\theta^{(i)}_{b^{(i)}}\in \{0,1\cdots 7\}}\underbrace{\ket{\theta_{b^{(0)}}^{(0)}\theta^{(1)}_{b^{(1)}}\cdots \theta^{(L)}_{b^{(L)}}}}_{\bDelta^{(0)}_{b^{(0)}}\bDelta^{(1)}_{b^{(1)}}\cdots \bDelta^{(L)}_{b^{(L)}}}\underbrace{\ket{(\sum_{i\in [0,L]}\theta^{(i)}_{b^{(i)}}-\sum_{i\in [L]}\Delta^{(i)}_{b^{(i)}})\Delta^{(1)}_{b^{(1)}}\cdots \Delta^{(L)}_{b^{(L)}}}}_{\btheta^{(0)}_{b^{(0)}}\btheta^{(1)}_{b^{(1)}}\cdots \btheta^{(L)}_{b^{(L)}}}\otimes \ket{\varphi_{\vec{b},\theta_{b^{(0)}}^{(0)}\theta_{b^{(1)}}^{(1)}\cdots \theta_{b^{(L)}}^{(L)}}}\\
	=&\sum_{\forall i\in [0,L]:\tilde\Delta^{(i)}_{b}\in \{0,1\cdots 7\},\theta^{(i)}_{b^{(i)}}\in \{0,1\cdots 7\},\sum_{i\in [0,L]}\tilde\Delta^{(i)}_{b^{(i)}}=\sum_{i\in 0,L]}\theta^{(i)}_{b^{(i)}}}\ket{\theta_{b^{(0)}}^{(0)}\theta^{(1)}_{b^{(1)}}\cdots \theta^{(L)}_{b^{(L)}}}\ket{\tilde\Delta^{(0)}_{b^{(0)}}\tilde\Delta^{(1)}_{b^{(1)}}\cdots \tilde\Delta^{(L)}_{b^{(L)}}}\otimes \ket{\varphi_{\vec{b},\theta_{b^{(0)}}^{(0)}\theta_{b^{(1)}}^{(1)}\cdots \theta_{b^{(L)}}^{(L)}}}
\end{align}

which has the form required in the pre-phase-honest form if we define $$\ket{\varphi_{\vec{b},sum}}=\sum_{\forall i\in [0,L]:\theta^{(i)}_{b^{(i)}}\in \{0,1\cdots 7\},\sum_{i\in 0,L]}\theta^{(i)}_{b^{(i)}}=sum}\ket{\theta_{b^{(0)}}^{(0)}\theta^{(1)}_{b^{(1)}}\cdots \theta^{(L)}_{b^{(L)}}}\ket{\varphi_{\vec{b},\theta_{b^{(0)}}^{(0)}\theta_{b^{(1)}}^{(1)}\cdots \theta_{b^{(L)}}^{(L)}}}$$
\end{proof}

\subsection{$\fCoPhTest$ Implies Approximate Invariance Under Randomization of $\cR_2$}\label{sec:9.3}
Now we give the following theorem, which says the ability of passing $\fCoPhTest$ implies approximate invariance of the initial state under $\cR_2$.
\begin{thm}\label{thm:8.1}
 Suppose a sub-normalized purified joint state $\ket{\varphi}$ is in Setup \ref{setup:4} and is in $\epsilon_0$-basis-honest form. 
Suppose $\fAdv$ is an efficient adversary that could make the client output $\fpass$ in the collective phase test with probability $\geq 1-\epsilon_1$ from initial state $\ket{\varphi}$. Then there is
\begin{equation}\label{eq:cophtest59}\cR_2(\ket{\$_2}\otimes \ket{\varphi})\approx_{(12\epsilon_1^{1/4}+\epsilon_0+\fneg(\kappa))}\ket{\$_2}\otimes\ket{\varphi}\end{equation}
\end{thm}
\subsubsection{A linear algebra lemma that connects state structure with randomization}
Before giving the formal proof, we give a linear algebra lemma that connects the structure of states implies approximate invariance of an operation.
\begin{fact}\label{fact:cophtool}
Suppose $\cC=\{0,1\cdots 7\}^{1+N}$.	Suppose $\ket{\varphi}$, $\ket{\phi}$ satisfy, there exist states $\ket{\varphi_{c_0c_1c_2\cdots c_N}}$, $\ket{\phi_{sum}}$ for each $c_0c_1c_2\cdots c_N\in \cC$, $sum\in \{0,1\cdots 7\}$ such that
	$$\ket{\varphi}=\sum_{c_0c_1c_2\cdots c_N\in \cC}\ket{c_0c_1c_2\cdots c_N}\otimes \ket{\varphi_{c_0c_1c_2\cdots c_N}}$$
	$$\ket{\phi}=\sum_{c_0c_1c_2\cdots c_N\in \cC}\ket{c_0c_1c_2\cdots c_N}\otimes \ket{\phi_{\tSUM(c_0c_1c_2\cdots c_N)}}$$
	$$\ket{\varphi}\approx_\epsilon\ket{\phi}$$ Then 
	$$\sum_{\Delta_1\Delta_2\cdots \Delta_N\in \{0,1\cdots 7\}^N}\frac{1}{\sqrt{8^N}}\ket{\Delta_1\Delta_2\cdots \Delta_N}\otimes\ket{\varphi}$$
	is $2\epsilon$-invariant under the following operator:
	\begin{equation}
		\ket{\Delta_1\Delta_2\cdots \Delta_N}\ket{c_0c_1c_2\cdots c_N}\rightarrow \ket{c_1c_2\cdots c_N}\ket{(\sum_{i\in [0,N]}c_i-\sum_{i\in [N]}\Delta_i)\Delta_1\Delta_2\cdots \Delta_N}
	\end{equation}
\end{fact}
Also as a preparation, we generalize Notation \ref{nota:3.16} a little bit:
\begin{nota}\label{nota:3.16r}
We say a purified joint state $\ket{\varphi}$ does not depend on the value of register $\bbC$ for the same $f(\bbC)$ if it can be written as
$$\ket{\varphi}=\sum_{c\in \cC}\underbrace{\ket{c}}_{\bbC}\otimes \ket{\psi_{f(c)}}$$	
\end{nota}
\subsubsection{Proof of Theorem \ref{thm:8.1}}
\begin{proof}
Let's use $\ket{\varphi^\prime}$ to denote the output state of running $\fCoPhTest$ on $\ket{\varphi}$ against $\fAdv$:
$$\ket{\varphi^\prime}=\fCoPhTest^\fAdv(\bK,\bTheta;1^\kappa)\ket{\varphi}$$
\begin{equation}\label{eq:cophtest0}|\Pi_{\fpass}\ket{\varphi^\prime}|^2\geq 1-\epsilon_1\end{equation}
	The first step in the $\fCoPhTest$ protocol is a  call to the $\fCombine$ protocol. This protocol combines $(1+L)$ gadgets into one single gadget. Denote the output state after this step as $\ket{\varphi^{mid}}$:
	\begin{equation}\label{eq:159}\ket{\varphi^{mid}}:=\fCalc(\bK^{(combined)},\bTheta^{(combined)})\circ\fResponse\circ\fAdv_1(\ket{\varphi}\odot \llbracket\fCombine\rrbracket)\end{equation}
	where \begin{itemize}\item $\llbracket\fCombine\rrbracket$ is the client-side messages in this step; recall that in this step the client samples $(r_0^{(i)},r_1^{(i)})$ for each $i\in [L]$ and prepares many look-up tables that encodes these $r$-values.	
	\item  $\fAdv_1$ is the adversary's operation in this step; \item $\fResponse$ is the operation that the server sends back a response (which is the output of $\fAdv_1$); \item $\fCalc(\bK^{(combined)},\bTheta^{(combined)})$ is the client-side operation that calculates $\bK^{(combined)},\bTheta^{(combined)}$ based on the server's response.
 	
 \end{itemize}

	Let's first define some notations. In the passing space, suppose the server's measurement outcome (in other words, the output of $\fAdv_1$ above) is
	\begin{equation}\label{eq:159t}r_{b^{(1)}}^{(1)}r_{b^{(2)}}^{(2)}\cdots r_{b^{(L)}}^{(L)},b^{(1)}\cdots b^{(L)}\in \{0,1\}^L\end{equation}
	Recall the key vector notation in Notation \ref{nota:kv}, when the key tuple is $K=(x_0^{(i)},x_1^{(i)})_{i\in [0,L]}\in \Domain(\bK)$ and the phase tuple is $\Theta=(\theta_0^{(i)},\theta_1^{(i)})_{i\in [0,L]}\in \Domain(\bTheta)$, the final $K^{(combined)}$ could be denoted as  
$$K^{(combined)}=(\vec{x}_{\vec{b}_0},\vec{x}_{\vec{b}_1})$$
where \begin{equation}\label{eq:coph60}\vec{b}_0=0b^{(1)}b^{(2)}\cdots b^{(L)},\vec{b}_1=1(1-b^{(1)})(1-b^{(2)})\cdots (1-b^{(L)})\end{equation}
Below we will use
$$\vec{b}_0+\vec{b}_1=\vec{1}\text{ (or $\vec{b}_1=\vec{1}-\vec{b}_0$)}$$
to denote $\vec{b}_0,\vec{b}_1$ that satisfy \eqref{eq:coph60}, and use $0||\{0,1\}^L$, $1||\{0,1\}^L$ to denote the domain of $\vec{b}_0,\vec{b}_1$ above.\par
Correspondingly,  the combined phases are
$$\Theta^{(combined)}=(\theta^{(combined)}_0,\theta^{(combined)}_1)=(\sum_{i\in [0,L]}\theta^{(i)}_{b^{(i)}},\sum_{i\in [0,L]}\theta^{(i)}_{1-b^{(i)}})$$
Recall the notation of honest joint phase in Definition \ref{defn:hjp}, this could be written as
$$\Theta^{(combined)}=(\tSUM(\vec{\Theta}_{\vec{b}_0}),\tSUM(\vec{\Theta}_{\vec{b}_1}))$$
	Now we analyze the protocol. Starting from $\ket{\varphi^{mid}}$, with $1/2$ probability both parties will do a standard basis test on $\bK^{(combined)}$. By \eqref{eq:cophtest0} the passing probability of this step should be $\geq 1-2\epsilon_1$, by Theorem \ref{thm:sbt} we can expand the state $\ket{\varphi^{mid}}$ based on the combined keys:
	\begin{equation}\label{eq:cophtest1}\exists \text{ efficient server-side } O:\ket{\tilde\varphi^{mid}}:=O\ket{\varphi^{mid}},\ket{\tilde\varphi^{mid}}\text{ is $1.5\sqrt{\epsilon_1}$-basis-honest for $\bK^{(combined)}$}\end{equation}
	We can assume the server-side register that holds one combined key is still $\bS_{bsh}$. Then we could expand the basis-honest part of $\ket{\tilde\varphi^{mid}}$ as follows:
	$$\Pi_{\basishonest(\bK^{(combined)})}^{\bS_{bsh}}\ket{\tilde\varphi^{mid}}=\sum_{K\in \Domain(\bK)}\ket{K}\otimes \sum_{\Theta\in \Domain(\bTheta)}\ket{\Theta}$$
	$$\otimes\sum_{\vec{b}_0\in 0||\{0,1\}^{L},\vec{b}_1\in 1||\{0,1\}^L:\vec{b}_0+\vec{b}_1=\vec{1}}\underbrace{\ket{\vec{x}_{\vec{b}_0},\vec{x}_{\vec{b}_1}}}_{\bK^{(combined)}}\underbrace{\ket{\tSUM(\vec{\Theta}_{\vec{b}_0}),\tSUM(\vec{\Theta}_{\vec{b}_1})}}_{\bTheta^{(combined)}}\otimes \underbrace{\ket{r_{b^{(1)}}^{(1)}r_{b^{(2)}}^{(2)}\cdots r_{b^{(L)}}^{(L)}}}_{\text{transcript}}$$
	$$\otimes\underbrace{\sum_{\vec{b}\in (\vec{b}_0,\vec{b}_1)}\ket{\vec{x}_{\vec{b}}}}_{\bS_{bsh}}\otimes \ket{\varphi_{K,\Theta,\vec{b}}}$$
	To show \eqref{eq:cophtest59}, we will first show:
	\begin{equation}\label{eq:64}\cR_2(\ket{\$_2}\otimes\ket{\tilde\varphi^{mid}})\approx_{11\epsilon_1^{1/4}+\fneg(\kappa)}\ket{\$_2}\otimes\ket{\tilde\varphi^{mid}}\end{equation}
	Note that in the above expansion of $\Pi_{\basishonest(\bK^{(combined)})}^{\bS_{bsh}}\ket{\tilde\varphi^{mid}}$ we are considering the basis-honest form for the combined key, and there are only two branches corresponding to two keys in $\bK^{(combined)}$. Denote the two branches as $\ket{\tilde\varphi_0^{mid}},\ket{\tilde\varphi_1^{mid}}$:
$$\ket{\tilde\varphi^{mid}_0}=\sum_{K\in \Domain(\bK)}\ket{K}\otimes \sum_{\Theta\in \Domain(\bTheta)}\ket{\Theta}\otimes$$ $$ \sum_{\vec{b}_0\in 0||\{0,1\}^L,\vec{b}_1\in 1||\{0,1\}^{L}:\vec{b}_0+\vec{b}_1=\vec{1}}\ket{\vec{x}_{\vec{b}_0},\vec{x}_{\vec{b}_1}}\ket{\tSUM(\vec{\Theta}_{\vec{b}_0}),\tSUM(\vec{\Theta}_{\vec{b}_1})}\otimes\ket{r_{b^{(1)}}^{(1)}r_{b^{(2)}}^{(2)}\cdots r_{b^{(L)}}^{(L)}}\otimes\ket{\vec{x}_{\vec{b_0}}}\otimes \ket{\varphi_{K,\Theta,\vec{b_0}}}$$
$$\ket{\tilde\varphi^{mid}_1}=\sum_{K\in \Domain(\bK)}\ket{K}\otimes \sum_{\Theta\in \Domain(\bTheta)}\ket{\Theta}\otimes $$ $$\sum_{\vec{b}_0\in 0||\{0,1\}^L,\vec{b}_1\in 1||\{0,1\}^{L}:\vec{b}_0+\vec{b}_1=\vec{1}}\ket{\vec{x}_{\vec{b}_0},\vec{x}_{\vec{b}_1}}\ket{\tSUM(\vec{\Theta}_{\vec{b}_0}),\tSUM(\vec{\Theta}_{\vec{b}_1})}\otimes\ket{r_{b^{(1)}}^{(1)}r_{b^{(2)}}^{(2)}\cdots r_{b^{(L)}}^{(L)}}\otimes\ket{\vec{x}_{\vec{b_1}}}\otimes \ket{\varphi_{K,\Theta,\vec{b_1}}}$$
	Thus $\Pi_{\basishonest(\bK^{(combined)})}^{\bS_{bsh}}\ket{\tilde\varphi^{mid}}=\ket{\tilde\varphi^{mid}_0}+\ket{\tilde\varphi^{mid}_1}$.
	
Then \eqref{eq:64} is further reduced to proving:
	\begin{equation}\label{eq:cophtestmid}\cR_2(\ket{\$_2}\otimes\ket{\tilde\varphi^{mid}_0})\approx_{5\epsilon_1^{1/4}+\fneg(\kappa)} \ket{\$_2}\otimes\ket{\tilde\varphi^{mid}_0}\end{equation}
	\begin{equation}\label{eq:cophtestmid2}\cR_2(\ket{\$_2}\otimes\ket{\tilde\varphi^{mid}_1})\approx_{5\epsilon_1^{1/4}+\fneg(\kappa)}\ket{\$_2}\otimes\ket{\tilde\varphi^{mid}_1}\end{equation}
	Without loss of generality we prove \eqref{eq:cophtestmid}. 
	We first make use of the fact that the state and the adversary can also pass the Hadamard test (the other choice of the client in $\fCoPhTest$). By \eqref{eq:cophtest0} the passing probability of Hadamard test is $\geq 1-2\epsilon_1$. 
%
By Lemma \ref{lem:5.2a} we know	$\ket{\tilde\varphi^{mid}}$ is claw-free for $\bK^{(combined)}$. Then together with \eqref{eq:cophtest1} by Corollary \ref{cor:prepht} there is:
	\begin{align}\label{eq:68}&\Pi_{\fpass}\fHadamardTest^{\fAdv_{HT}\circ O^{-1}}(\bK^{(combined)},\bTheta^{(combined)};1^\kappa)\ket{\tilde\varphi_0^{mid}}\\\approx_{2.5\epsilon_1^{1/4}+\fneg(\kappa)}&\Pi_{\fpass}\fHadamardTest^{\fAdv_{HT}\circ O^{-1}}(\bK^{(combined)},\bTheta^{(combined)};1^\kappa)\ket{\tilde\varphi_1^{mid}}.\label{eq:166}\end{align}
		\begin{align}\label{eq:68b}&\Pi_{\ffail}\fHadamardTest^{\fAdv_{HT}\circ O^{-1}}(\bK^{(combined)},\bTheta^{(combined)};1^\kappa)\ket{\tilde\varphi_0^{mid}}\\\approx_{2\epsilon_1^{1/4}+\fneg(\kappa)}&-\Pi_{\ffail}\fHadamardTest^{\fAdv_{HT}\circ O^{-1}}(\bK^{(combined)},\bTheta^{(combined)};1^\kappa)\ket{\tilde\varphi_1^{mid}}.\end{align}
	Applying Fact \ref{fact:cophtool} (with details in the box below) we know:
		\begin{align}&\cR_2(\ket{\$_2}\otimes\Pi_{\fpass}\fHadamardTest^{\fAdv_{HT}\circ O^{-1}}(\bK^{(combined)},\bTheta^{(combined)};1^\kappa)\ket{\tilde\varphi_0^{mid}})\\\approx_{2.5\epsilon_1^{1/4}+\fneg(\kappa)}&\ket{\$_2}\otimes\Pi_{\fpass}\fHadamardTest^{\fAdv_{HT}\circ O^{-1}}(\bK^{(combined)},\bTheta^{(combined)};1^\kappa)\ket{\tilde\varphi_1^{mid}}.\label{eq:151rr}\end{align}
		\begin{align}&\cR_2(\ket{\$_2}\otimes\Pi_{\ffail}\fHadamardTest^{\fAdv_{HT}\circ O^{-1}}(\bK^{(combined)},\bTheta^{(combined)};1^\kappa)\ket{\tilde\varphi_0^{mid}})\\\approx_{2\epsilon_1^{1/4}+\fneg(\kappa)}&-\ket{\$_2}\otimes\Pi_{\ffail}\fHadamardTest^{\fAdv_{HT}\circ O^{-1}}(\bK^{(combined)},\bTheta^{(combined)};1^\kappa)\ket{\tilde\varphi_1^{mid}}.\label{eq:171}\end{align}
		\begin{mdframed}
		Proof of \eqref{eq:151rr}:\par
		Recall \eqref{eq:159}. By Lemma \ref{lem:multisbt} we have
		\begin{equation}\label{eq:146op}
		\ket{\tilde\varphi^{mid}_0}\approx_{\fneg(\kappa)}\sum_{\vec{b}_0\in 0||\{0,1\}^L}\Pi_{\vec{\bx}_{\vec{b}_0}}^{\bS_{bsh}}\Pi^{\bS_{bsh}}_{\bx_0^{(combined)}}\circ O\circ\fCalc\circ \fResponse\circ\fAdv_1\circ\Pi_{\vec{\bx}_{\vec{b}_0}}^{\bS_{bsh}}(\ket{\varphi}\odot \llbracket\fCombine\rrbracket)
	\end{equation}
	\begin{equation}\label{eq:146p}
		\ket{\tilde\varphi^{mid}_1}\approx_{\fneg(\kappa)}\sum_{\vec{b}_1\in 1||\{0,1\}^L}\Pi_{\vec{\bx}_{\vec{b}_1}}^{\bS_{bsh}}\Pi^{\bS_{bsh}}_{\bx_1^{(combined)}}\circ O\circ\fCalc\circ \fResponse\circ\fAdv_1\circ\Pi_{\vec{\bx}_{\vec{b}_1}}^{\bS_{bsh}}(\ket{\varphi}\odot \llbracket\fCombine\rrbracket)
	\end{equation}
	Then substituting \eqref{eq:146op}\eqref{eq:146p} into \eqref{eq:68}\eqref{eq:166} we get
		\begin{align}&\Pi_{\fpass}\fHadamardTest^{\fAdv_{HT}\circ O^{-1}}(\bK^{(combined)},\bTheta^{(combined)};1^\kappa)\label{eq:180}\\
		&\sum_{\vec{b}_0\in 0||\{0,1\}^L}\Pi_{\vec{\bx}_{\vec{b}_0}}^{\bS_{bsh}}\Pi^{\bS_{bsh}}_{\bx_0^{(combined)}}\circ O\circ\fCalc\circ \fResponse\circ\fAdv_1\circ\Pi_{\vec{\bx}_{\vec{b}_0}}^{\bS_{bsh}}(\ket{\varphi}\odot \llbracket\fCombine\rrbracket)\label{eq:181}\\
		\approx_{2.5\epsilon_1^{1/4}+\fneg(\kappa)}&\Pi_{\fpass}\fHadamardTest^{\fAdv_{HT}\circ O^{-1}}(\bK^{(combined)},\bTheta^{(combined)};1^\kappa)\\
		&\sum_{\vec{b}_1\in 1||\{0,1\}^L}\Pi_{\vec{\bx}_{\vec{b}_1}}^{\bS_{bsh}}\Pi^{\bS_{bsh}}_{\bx_1^{(combined)}}\circ O\circ\fCalc\circ \fResponse\circ\fAdv_1\circ\Pi_{\vec{\bx}_{\vec{b}_1}}^{\bS_{bsh}}(\ket{\varphi}\odot \llbracket\fCombine\rrbracket)\label{eq:183}\end{align}
	Recall that after the $\fCombine$ process there will be a series of transcript registers that stores the server's response \eqref{eq:159t}. Its subscripts $b^{(1)}b^{(2)}\cdots b^{(L)}$ determine the value of $\bK^{(combine)}$.   Then we can decompose \eqref{eq:180}-\eqref{eq:183} into a series of approximation relation for each  $b^{(1)}b^{(2)}\cdots b^{(L)}\in\{0,1\}^L$ by Fact \ref{fact:decomp}. Then notice $\Pi_{\vec{\bx}_{\vec{b}_0}}^{\bS_{bsh}}\Pi^{\bS_{bsh}}_{\bx_0^{(combined)}}$ in \eqref{eq:181} is non-zero if and only if the last $L$ bits of $\vec{b}_0$ is equal to the subscripts of $r$-registers and  $\Pi_{\vec{\bx}_{\vec{b}_1}}^{\bS_{bsh}}\Pi^{\bS_{bsh}}_{\bx_1^{(combined)}}$ in \eqref{eq:183} is non-zero if and only if the last $L$ bits of $\vec{b}_1$ is equal to $\vec{1}$ minus the subscripts of $r$-registers. Thus the decomposed approximate equation is as follows:\par
	There exists a set of non-negative real values $t_{\vec{b}}$ for each $\vec{b}\in 0||\{0,1\}^L$ such that, $\sum_{\vec{b}\in 0||\{0,1\}^L}t_{\vec{b}}\leq 2.5\epsilon_1^{1/4}$, and for all $\vec{b}\in 0||\{0,1\}^L$,
	\begin{align}\label{eq:68re}&\Pi_{\fpass}\fHadamardTest^{\fAdv_{HT}\circ O^{-1}}(\bK^{(combined)},\bTheta^{(combined)};1^\kappa)\nonumber\\
	&\qquad\Pi_{\vec{\bx}_{\vec{b}}}^{\bS_{bsh}}O\circ\fCalc\circ \fResponse\circ\fAdv_1\Pi_{\vec{\bx}_{\vec{b}}}^{\bS_{bsh}}(\ket{\varphi}\odot \llbracket\fCombine\rrbracket)
\\\approx_{t_{\vec{b}}+\fneg(\kappa)}&\Pi_{\fpass}\fHadamardTest^{\fAdv_{HT}\circ O^{-1}}(\bK^{(combined)},\bTheta^{(combined)};1^\kappa)\nonumber\\
&\qquad\Pi_{\vec{\bx}_{\vec{1}-\vec{b}}}^{\bS_{bsh}}O\circ\fCalc\circ \fResponse\circ\fAdv_1\Pi_{\vec{\bx}_{\vec{1}-\vec{b}}}^{\bS_{bsh}}(\ket{\varphi}\odot \llbracket\fCombine\rrbracket)\label{eq:158dk}
.\end{align}
Expanding the operation of $\fHadamardTest(\bK^{(combined)},\bTheta^{(combined)})$ we get
	\begin{align}&\Pi_{\fpass}\fHadamardTest^{\fAdv_{HT}\circ O^{-1}}(\bK^{(combined)};1^\kappa)\nonumber\\
	&\qquad\Pi_{\vec{\bx}_{\vec{b}}}^{\bS_{bsh}}O\circ\fCalc\circ\fResponse\circ\fAdv_1\Pi_{\vec{\bx}_{\vec{b}}}^{\bS_{bsh}}(\ket{\varphi}\odot \llbracket\fCombine\rrbracket\odot (\btheta^{(combined)}_1-\btheta^{(combined)}_0))\label{eq:168dk}
\\\approx_{t_{\vec{b}}+\fneg(\kappa)}&\Pi_{\fpass}\fHadamardTest^{\fAdv_{HT}\circ O^{-1}}(\bK^{(combined)};1^\kappa)\nonumber\\
&\qquad\Pi_{\vec{\bx}_{\vec{1}-\vec{b}}}^{\bS_{bsh}}O\circ\fCalc\circ\fResponse\circ\fAdv_1\Pi_{\vec{\bx}_{\vec{1}-\vec{b}}}^{\bS_{bsh}}(\ket{\varphi}\odot \llbracket\fCombine\rrbracket\odot (\btheta^{(combined)}_1-\btheta^{(combined)}_0))\label{eq:169dk}
\end{align}
	and by the condition that $\ket{\varphi}$ in Setup \ref{setup:4} and the defining equation of the combined phases we know \begin{equation}\label{eq:147p}
		\text{\eqref{eq:169dk} does not depend on the values of registers $\btheta^{(0)}_{0},\btheta^{(1)}_{b^{(1)}},\cdots \btheta^{(L)}_{b^{(L)}}$ for the same $(\btheta^{(0)}_{0}+\sum_{i\in [L]}\btheta_{b^{(i)}}^{(i)})$.}
	\end{equation}
	Thus applying Fact \ref{fact:cophtool} we get, for each $\vec{b}\in 0||\{0,1\}^L$,
		\begin{align}\label{eq:68ree}&\cR_2(\ket{\$_2}\otimes\eqref{eq:168dk})
\\\approx_{t_{\vec{b}}+\fneg(\kappa)}&\ket{\$_2}\otimes\eqref{eq:169dk}
\end{align}
summing up for all the possible $\vec{b}$, substituting \eqref{eq:146op}\eqref{eq:146p} completes the proof of \eqref{eq:151rr}.
		\end{mdframed}

		Summing \eqref{eq:151rr}\eqref{eq:171} implies \eqref{eq:cophtestmid}. Similarly we prove \eqref{eq:cophtestmid2}. Thus
	we complete the proof of \eqref{eq:64}.\par
	Then we port \eqref{eq:64} to $\ket{\varphi}$. Recall $\ket{\tilde\varphi^{mid}}$ is defined by
	$$\ket{\tilde\varphi^{mid}}:=O\circ\fCalc\circ\fResponse\circ \fAdv_{1}(\ket{\varphi}\odot \llbracket\fCombine\rrbracket)$$
	Note that $\fCalc$, the client-side calculation of combined keys and combined phases, commute with $O$ and $\cR_2$. (Note $\cR_2$ preserves the combined phases on each branch.) Thus we can omit it and get
	\begin{equation}\label{eq:64eqw}\cR_2(\ket{\$_2}\otimes O\circ\fResponse\circ \fAdv_{1}(\ket{\varphi}\odot \llbracket\fCombine\rrbracket))\approx_{11\epsilon_1^{1/4}+\fneg(\kappa)}\ket{\$_2}\otimes O\circ\fResponse\circ \fAdv_{1}(\ket{\varphi}\odot \llbracket\fCombine\rrbracket)\end{equation}
	 Starting from \eqref{eq:64eqw}, apply the inverse of $O\circ\fResponse\circ\fAdv_{1}$, we have
	\begin{equation}\label{eq:152j}\fAdv^{-1}_{1}\fResponse^{-1}\circ O^{-1}(\cR_2(\ket{\$_2}\otimes O\circ\fResponse\circ \fAdv_{1}(\ket{\varphi}\odot \llbracket\fCombine\rrbracket)))\approx_{11\epsilon_1^{1/4}+\fneg(\kappa)} \ket{\$_2}\otimes\ket{\varphi}\odot\llbracket\fCombine\rrbracket\end{equation}
	\begin{align}\label{eq:183pqs}&\Rightarrow\quad \Pi_{\basishonest(\bK)}^{\bS_{bsh}}\fAdv^{-1}_{1}\fResponse^{-1}\circ O^{-1}(\cR_2(\ket{\$_2}\otimes O\circ\fResponse\circ \fAdv_{1}(\Pi_{\basishonest(\bK)}^{\bS_{bsh}}\ket{\varphi}\odot \llbracket\fCombine\rrbracket)))\\
	&\approx_{12\epsilon_1^{1/4}+\epsilon_0+\fneg(\kappa)} \ket{\$_2}\otimes\ket{\varphi}\odot\llbracket\fCombine\rrbracket\end{align}
	Applying Theorem \ref{lem:multisbt} we get
	\begin{align}&\Pi_{\basishonest(\bK)}^{\bS_{bsh}}\fAdv^{-1}_{1}\fResponse^{-1}\circ O^{-1}(\cR_2(\ket{\$_2}\otimes O\circ\fResponse\circ \fAdv_{1}(\Pi_{\basishonest(\bK)}^{\bS_{bsh}}\ket{\varphi}\odot \llbracket\fCombine\rrbracket)))\label{eq:175ew}\\
	\approx_{\fneg(\kappa)}&\sum_{\vec{b}\in \{0,1\}^{1+L}}\Pi_{\vec{\bx}_{\vec{b}}}^{\bS_{bsh}}\fAdv^{-1}_{1}\fResponse^{-1} \circ O^{-1}(\cR_2(\ket{\$_2}\otimes O\circ\fResponse\circ \fAdv_{1}(\Pi_{\vec{\bx}_{\vec{b}}}^{\bS_{bsh}}\ket{\varphi}\odot \llbracket\fCombine\rrbracket)))\label{eq:96ew}\\
	=&\cR_2\sum_{\vec{b}\in \{0,1\}^{1+L}}\Pi_{\vec{\bx}_{\vec{b}}}^{\bS_{bsh}}(\fAdv^{-1}_{1}\fResponse^{-1} \circ O^{-1}(\ket{\$_2}\otimes O\circ\fResponse\circ \fAdv_{1}(\Pi_{\vec{\bx}_{\vec{b}}}^{\bS_{bsh}}\ket{\varphi}\odot \llbracket\fCombine\rrbracket)))\label{eq:177ew}\\
	\approx_{\fneg(\kappa)}&\cR_2\Pi_{\basishonest(\bK)}^{\bS_{bsh}}(\fAdv^{-1}_{1}\fResponse^{-1} \circ  O^{-1}(\ket{\$_2}\otimes O\circ\fResponse\circ \fAdv_{1}(\Pi_{\basishonest(\bK)}^{\bS_{bsh}}\ket{\varphi}\odot \llbracket\fCombine\rrbracket)))\label{eq:178ew}
	\end{align}
	where \eqref{eq:96ew}-\eqref{eq:177ew} is because $\cR_2$ applying only on the $\vec{\bx}_{\vec{b}}$-branch state could be seen as a local operator $\fAdd_{\vec{b}}$ applying on the $\bTheta$ registers, and thus commutes with operations $\fAdv,\fResponse, O$.\par 
	
	\eqref{eq:175ew}-\eqref{eq:178ew} together with \eqref{eq:183pqs} implies
	$$\cR_2(\ket{\$_2}\otimes\ket{\varphi}\odot \llbracket\fCombine\rrbracket)\approx_{12\epsilon_1^{1/4}+\epsilon_0+\fneg(\kappa)} \ket{\$_2}\otimes\ket{\varphi}\odot \llbracket\fCombine\rrbracket$$
	which completes the proof.
\end{proof}

\section{Analysis of the Individual Phase Test ($\fInPhTest$)}\label{sec:10}
In this section we analyze the implication of the individual phase test ($\fInPhTest$). 
\begin{enumerate} \item In Section \ref{sec:10.1} we will give two linear algebra lemmas which are the basis for later proofs.
\item In Section \ref{sec:10.2} we show the optimal winning probability of $\fInPhTest$ is $\OPT$. Recall that the $\OPT$ is defined to be $\frac{1}{3}\cos^2(\pi/8)$ in Section \ref{sec:4}.
\item In Section \ref{sec:10.3} we construct a randomization operator $\cP$. We will show $\cP$ transforms a pre-phase-honest form to a phase-honest form.
\item In Section \ref{sec:10.4} we show $\fInPhTest$ implies approximate invariance of a state (under Setup \ref{setup:4}) under $\cP^\dagger\cP$.
\end{enumerate}
We give a brief review and overview of our analysis of the $\fInPhTest$. Recall that in the $\fInPhTest$ the client-side inputs are values of registers $\bK^{(0)},\bTheta^{(0)}$. Note that different values of $\bTheta^{(0)}$  corresponds to different server-side states. In more detail, when the value of $\bTheta^{(0)}$ is $(\theta^{(0)}_0,\theta^{(0)}_1)$, the malicious server's state, assuming it's in Setup \ref{setup:4} and basis-honest form of $\bK^{(0)}$, could be expressed as (where we omit the client-side registers):
$$\ket{x_0^{(0)}}\ket{\varphi_{0,\theta_0^{(0)}}}+\ket{x_1^{(0)}}\ket{\varphi_{1,\theta_1^{(0)}}}$$
Then $\fInPhTest$ aims at testing the following property on $\ket{\varphi_{0,\theta_0^{(0)}}}$, $\ket{\varphi_{1,\theta_1^{(0)}}}$: there exist states $\ket{\varphi_{0,+}},\ket{\varphi_{0,-}},\ket{\varphi_{1,+}},\ket{\varphi_{1,-}}$ such that, up to a server-side isometry,\begin{equation}\label{eq:158d}\text{for each $\theta_0,\theta_1$, $b\in \{0,1\}$, }\ket{\varphi_{b,\theta_b}}\approx e^{\theta_b^{(0)}\mi\pi/4}\ket{\varphi_{b,+}}+e^{-\theta_b^{(0)}\mi\pi/4}\ket{\varphi_{b,-}}.\end{equation}
The first term corresponds to the honest state case and the second term corresponds to the complex-conjugated honest state case. Then starting from \eqref{eq:158d}, we can show the state is approximate invariant under $\cP$ and thus close to a phase-honest form.\par
The analysis towards proving \eqref{eq:158d} roughly goes as follows. In $\fInPhTest$ (Protocol \ref{prtl:rephtest}) both parties execute the  Hadamard test with an extra phase bias. One of three choices of extra phase bias $\delta\in \{0,1,4\}$ is chosen randomly, which corresponds to two cases:
\begin{itemize}
	\item $\delta\in \{0,4\}$: the server is required to make the client output $\fpass$ in these two tests.\par
	Considering the $\delta=0$ case first. Suppose the adversary's operation maps $\ket{\varphi_{0,\theta}}$ to $\ket{\varphi^\prime_{0,\theta}}$ and maps $\ket{\varphi_{1,\theta}}$ to $\ket{\varphi^\prime_{1,\theta}}$. (Here $\theta$ in $\ket{\varphi_{0,\theta}}$ stands for $\theta_0^{(0)}$ and $\theta$ in $\ket{\varphi_{1,\theta}}$ stands for $\theta_1^{(0)}$.) Then use $\Pi_{0}$ to denote the projection onto the passing conditions (that is, $\Pi_{\eqref{eq:htt}=0}^{\bd}\Pi_{\neq 0}^{\text{last $\kappa$ bits of }\bd}$ in Notation \ref{nota:setup2aux}). Then the server could pass the $\delta=0$ case implies
	\begin{equation}\label{eq:190g}
		\text{for each $\theta$, }\Pi_0\ket{\varphi^\prime_{0,\theta}}\approx \Pi_0\ket{\varphi^\prime_{1,\theta}}
	\end{equation}
	Then in the $\delta=4$ case the passing and failing conditions are (almost) opposite to each other, we can show
		\begin{equation}\label{eq:191g}
		\text{for each $\theta$, }\Pi_0\ket{\varphi^\prime_{0,\theta}}\approx -\Pi_0\ket{\varphi^\prime_{1,\theta+4}}
	\end{equation}
	(Recall that in the extra-phase-biased Hadamard test the adversary does not know $\delta$ directly, it only knows $\theta_1^{(0)}-\theta_0^{(0)}-\delta$; thus if we fix the transcript and make $\delta$ varies from $0$ to $4$ the client-side phase changes by $4$ too.) 
	\item $\delta=1$: the server is expected to make the client output $\fwin$ to the score register with sufficiently high probability. As before we fix the transcript and the adversary's operation is applied on a suitable superposition of $\ket{\varphi_{0,\theta}}$ and $\ket{\varphi_{1,\theta+1}}$. Then could calculate the winning probability and get
	 \begin{equation}\label{eq:191}\text{the winning probability is the sum of }|\Pi_0\ket{\varphi^\prime_{0,\theta}}+\Pi_0\ket{\varphi^\prime_{1,\theta+1}}|^2\text{ for each $\theta$}\end{equation}
	 \end{itemize}
	 Then:
	 \begin{itemize}
	\item The optimality of $\OPT$ comes from bounding \eqref{eq:191} under the conditions \eqref{eq:190g}\eqref{eq:191g} (and necessary properties on their norms). This is done in Section \ref{sec:10.2}.
	\item When \eqref{eq:191} is close to $\OPT$, with this property together with \eqref{eq:190g}\eqref{eq:191g} we can prove these states satisfy some relations, which allow us to derive \eqref{eq:158d}. This is done in Section \ref{sec:10.4}.
	 \end{itemize}

\subsection{Linear Algebra Lemmas for Self-testing of State Sequences}\label{sec:10.1}

In this subsection we prove two lemmas by linear algebra. These lemmas will be used in the analysis of $\fInPhTest$.\par
First, we have the following lemma which leads to the optimality of $\OPT$ for winning probability:
\begin{lem}\label{lem:optla}
Suppose $\ket{\phi_{0,0}},\ket{\phi_{0,1}},\ket{\phi_{0,2}},\cdots, \ket{\phi_{0,7}}$, $\ket{\phi_{1,0}},\ket{\phi_{1,1}},\ket{\phi_{1,2}},\cdots, \ket{\phi_{1,7}}$ satisfy the following for $\epsilon<0.001$:
\begin{itemize}
\item There exist two non-negative real number $A_0,A_1$ such that $\forall i\in \{0,1\cdots 7\}$, $\frac{1}{8}A_0-\epsilon\leq |\ket{\phi_{0,i}}|^2\leq \frac{1}{8}A_0$, $\frac{1}{8}A_1-\epsilon\leq |\ket{\phi_{1,i}}|^2\leq \frac{1}{8}A_1$, $A_0+A_1\leq \frac{1}{2}$.
\item $\sum_{i\in \{0,1\cdots 7\}}|\ket{\phi_{0,i}}-\ket{\phi_{1,i}}|^2\leq \epsilon$
\item $\sum_{i\in \{0,1\cdots 7\}}|\ket{\phi_{0,i}}+\ket{\phi_{1,i+4}}|^2\leq \epsilon$	
\end{itemize}
Then \begin{equation}\label{eq:103pe}\sum_{i\in \{0,1\cdots 7\}}|\ket{\phi_{0,i}}+\ket{\phi_{1,i+1}}|^2\leq \cos^2(\pi/8)+11\sqrt{\epsilon}\end{equation}
\end{lem}
We also have the inverse version of this lemma, which characterize the self-testing property of $\fInPhTest$. In more detail, if the left hand side of \eqref{eq:103pe} is close-to-optimal, then these states should have a specific form, which is given in the following lemma.
\begin{lem}\label{lem:10.2rr}
Suppose $\ket{\phi_{0,0}},\ket{\phi_{0,1}},\ket{\phi_{0,2}},\cdots, \ket{\phi_{0,7}}$, $\ket{\phi_{1,0}},\ket{\phi_{1,1}},\ket{\phi_{1,2}},\cdots, \ket{\phi_{1,7}}$ satisfy  the following for $\epsilon<10^{-5}$:
\begin{itemize}
\item There exist two non-negative real number $A_0,A_1$ such that $\forall i\in \{0,1\cdots 7\}$, $\frac{1}{8}A_0-\epsilon\leq|\ket{\phi_{0,i}}|^2\leq \frac{1}{8}A_0$, $\frac{1}{8}A_1-\epsilon\leq |\ket{\phi_{1,i}}|^2\leq \frac{1}{8}A_1$, $A_0+A_1\leq \frac{1}{2}$.
\item $\sum_{i\in \{0,1\cdots 7\}}|\ket{\phi_{0,i}}-\ket{\phi_{1,i}}|^2\leq \epsilon$
\item $\sum_{i\in \{0,1\cdots 7\}}|\ket{\phi_{0,i}}+\ket{\phi_{1,i+4}}|^2\leq \epsilon$
\item $\sum_{i\in \{0,1\cdots 7\}}|\ket{\phi_{0,i-1}}+\ket{\phi_{1,i}}|^2\geq \cos^2(\pi/8)-\epsilon$
		
\end{itemize}
	Then define
	\begin{equation}\label{eq:163s}\ket{\phi_{0,+}}=\frac{1}{8}\sum_{i\in \{0,1\cdots 7\}}e^{-i\mi\pi/4}\ket{\phi_{0,i}}\end{equation}
	\begin{equation}\label{eq:164s}\ket{\phi_{0,-}}=\frac{1}{8}\sum_{i\in \{0,1\cdots 7\}}e^{i\mi\pi/4}\ket{\phi_{0,i}}\end{equation}
		\begin{equation}\label{eq:165s}\ket{\phi_{1,+}}=\frac{1}{8}\sum_{i\in \{0,1\cdots 7\}}e^{-i\mi\pi/4}\ket{\phi_{1,i}}\end{equation}
	\begin{equation}\label{eq:166s}\ket{\phi_{1,-}}=\frac{1}{8}\sum_{i\in \{0,1\cdots 7\}}e^{i\mi\pi/4}\ket{\phi_{1,i}}\end{equation}
	there is
	$$\sum_{b\in \{0,1\}}\sum_{i\in \{0,1\cdots 7\}} |\ket{\phi_{b,i}}-(e^{i\cdot \mi \pi/4}\ket{\phi_{b,+}}+e^{-i\cdot \mi \pi/4}\ket{\phi_{b,-}})|^2\leq 640\epsilon^{1/4}$$
\end{lem}
The proofs of Lemma \ref{lem:optla} and \ref{lem:10.2rr} are given in Appendix \ref{sec:missing10}.\par
\subsection{Optimality of $\OPT$ in $\fInPhTest$}\label{sec:10.2}
In this section we prove the optimality of winning probability in the individual phase test.\par
\begin{thm}[Optimality of $\OPT$ in $\fInPhTest$]\label{thm:9.2}
Assume a sub-normalized purified joint state $\ket{\varphi}$ is in Setup \ref{setup:4} and is in $\epsilon_0$-basis-honest form for $\bK^{(0)}$. Then for any $\epsilon_1<10^{-5}-2\epsilon_0$, any efficient adversary $\fAdv$, at least one of the following two is true:
\begin{itemize}
\item (Small passing probability) $$|\Pi_{\fpass}\fInPhTest^\fAdv(\bK,\bTheta;1^\kappa)\ket{\varphi}|^2\leq 1-\epsilon_1$$
\item (Small winning probability) 
 $$|\Pi_{\fwin}\fInPhTest^\fAdv(\bK, \bTheta;1^\kappa)\ket{\varphi}|^2\leq OPT+10(\epsilon_1+\epsilon_0)^{1/4}+\fneg(\kappa)$$
\end{itemize}

\end{thm}
%
\begin{proof}
Expand $\ket{\varphi}$ on the basis-honest form of $\bK^{(0)}$:
$$\Pi_{\basishonest(\bK^{(0)})}\ket{\varphi}=\underbrace{\sum_{K\in \Domain(\bK)}\ket{K}\otimes \sum_{\Theta\in \Domain(\bTheta)}\ket{\Theta}}_{\text{client}}\otimes \sum_{b\in \{0,1\}}\underbrace{\ket{x_b^{(0)}}}_{\bS_{bsh}^{(0)}}\otimes \ket{\varphi_{K,\Theta,b}}$$
Define $\ket{\varphi_0}$, $\ket{\varphi_1}$ as the $\bx^{(0)}_0$, $\bx^{(0)}_1$ branches:
$$\ket{\varphi_0}=\sum_{K\in \Domain(\bK)}\ket{K}\otimes \sum_{\Theta\in \Domain(\bTheta)}\ket{\Theta}\otimes  \ket{x_{0}^{(0)}}\otimes \ket{\varphi_{K,\Theta,{0}}}$$
$$\ket{\varphi_1}=\sum_{K\in \Domain(\bK)}\ket{K}\otimes \sum_{\Theta\in \Domain(\bTheta)}\ket{\Theta}\otimes \ket{x_{1}^{(0)}}\otimes \ket{\varphi_{K,\Theta,{1}}}.$$
Then $\Pi_{\basishonest(\bK^{(0)})}\ket{\varphi}=\ket{\varphi_0}+\ket{\varphi_1}$.\par
Define $\ket{\varphi_{0,\theta_0,\theta_1}}$ as the component of $\ket{\varphi_0}$ where the $\btheta^{(0)}_{0}$ register is in value $\theta_0$ and $\btheta^{(0)}_1$ register is in value $\theta_1$:
\begin{equation}\label{eq:95rea}\ket{\varphi_{0,\theta_0,\theta_1}}=\sum_{K\in \Domain(\bK)}\ket{K}\otimes \sum_{\Theta:\btheta^{(0)}_{0}=\theta_0,\btheta^{(0)}_{1}=\theta_1}\ket{\Theta}\otimes \ket{x_{0}^{(0)}}\otimes \ket{\varphi_{K,\Theta,{0}}}\end{equation}
Similarly define the $\ket{\varphi_{1,\theta_0,\theta_1}}$:
\begin{equation}\label{eq:96rea}\ket{\varphi_{1,\theta_0,\theta_1}}=\sum_{K\in \Domain(\bK)}\ket{K}\otimes \sum_{\Theta:\btheta^{(0)}_{0}=\theta_0,\btheta^{(0)}_{1}=\theta_1}\ket{\Theta}\otimes \ket{x_{1}^{(0)}}\otimes \ket{\varphi_{K,\Theta,{1}}}\end{equation}
Then we have
$$\ket{\varphi_0}=\sum_{\theta_0,\theta_1\in \{0,1\cdots 7\}^2}\ket{\varphi_{0,\theta_0,\theta_1}},\ket{\varphi_1}=\sum_{\theta_0,\theta_1\in \{0,1\cdots 7\}^2}\ket{\varphi_{1,\theta_0,\theta_1}}$$
 Suppose \begin{equation}\label{eq:reph72}|\Pi_{\fpass}\fInPhTest^\fAdv(\bK,\bTheta;1^\kappa)\ket{\varphi}|^2> 1-\epsilon_1\end{equation}
 Let's calculate the probability of winning, which happens in the $\delta=1$ case in the extra-phase-bias Hadamard test (Protocol \ref{prtl:hadamard}).
 \begin{align}
 	&|\Pi_{\fwin}\fInPhTest^{\fAdv}(\bK,\bTheta;1^\kappa)\ket{\varphi}|^2\label{eq:158k}\\
\approx_{2\epsilon_0}&|\Pi_{\fwin}\fInPhTest^{\fAdv}(\bK,\bTheta;1^\kappa)\Pi_{\basishonest(\bK^{(0)})}\ket{\varphi}|^2\label{eq:159k}\\
=&\sum_{\theta_0\in \{0,1\cdots 7\}}\sum_{\theta_1\in \{0,1\cdots 7\}}|\Pi_{\fwin}\fInPhTest^{\fAdv}(\bK,\bTheta;1^\kappa)(\ket{\varphi_{0,\theta_0,\theta_1}}+\ket{\varphi_{1,\theta_0,\theta_1}})|^2\\
=&\sum_{\theta_0\in \{0,1\cdots 7\}}\sum_{\theta_1\in \{0,1\cdots 7\}}\frac{1}{3}|\Pi_{\fwin}\fHadamardTest^{\fAdv}(\bK^{(0)},\bTheta^{(0)},1;1^\kappa)(\ket{\varphi_{0,\theta_0,\theta_1}}+\ket{\varphi_{1,\theta_0,\theta_1}})|^2\\
=&\sum_{\theta_0\in \{0,1\cdots 7\}}\sum_{\theta_1\in \{0,1\cdots 7\}}\frac{1}{3}|\Pi_{\eqref{eq:htt}=0}^{\bd}\Pi_{\neq 0}^{\text{last $\kappa$ bits of }\bd}\fHadamardTest^{\fAdv}(\bK^{(0)};1^\kappa)\nonumber\\
&\qquad\qquad\qquad(\ket{\varphi_{0,\theta_0,\theta_1}}+\ket{\varphi_{1,\theta_0,\theta_1}})\odot (\theta_1-\theta_0-1)|^2\label{eq:99rr}
 \end{align}
 Recall $\Pi_{\eqref{eq:htt}=0}^{\bd},\Pi_{\eqref{eq:htt}=1}^{\bd},\Pi_{\neq 0}^{\text{last $\kappa$ bits of }\bd}$ are defined when we define the set-up of Hadamard test (Setup \ref{setup:2}, Notation \ref{nota:setup2aux}). In \eqref{eq:99rr} we write explicitly the operations in which the client sends the relative phase with extra phase bias to the server, runs the protocol and projects onto the winning space.\par
 Then we consider the $\delta=0$ and $\delta=4$ case in the $\fHadamardTest(\bK^{(0)},\bTheta^{(0)},\delta)$ subprotocol. By \eqref{eq:reph72} the passing probability for these two cases should both be $\geq 1-3\epsilon_1$. We can do a similar calculation as \eqref{eq:158k}-\eqref{eq:99rr} and get:
 \begin{itemize}
 \item When $\delta=0$,  
  $|\Pi_{\fpass}\fHadamardTest^\fAdv(\bK^{(0)},\bTheta^{(0)},0;1^\kappa)\circ\ket{\varphi}|^2\geq 1-3\epsilon_1$. Recall the passing space is $\Pi_{\eqref{eq:htt}=0}^{\bd}\Pi_{\neq 0}^{\text{last $\kappa$ bits of }\bd}$. By Corollary \ref{cor:prepht}:
 \begin{align}\label{eq:170rw}&\sum_{\theta_0\in \{0,1\cdots 7\}}\sum_{\theta_1\in \{0,1\cdots 7\}}\frac{1}{3}|\Pi_{\eqref{eq:htt}=0}^{\bd}\Pi_{\neq 0}^{\text{last $\kappa$ bits of }\bd}\fHadamardTest^\fAdv(\bK^{(0)};1^\kappa)(\ket{\varphi_{0,\theta_0,\theta_1}}-\ket{\varphi_{1,\theta_0,\theta_1}})\odot (\theta_1-\theta_0)|^2\\\leq& 6\sqrt{\epsilon_1+\epsilon_0}+\fneg(\kappa)\label{eq:171ma}\end{align}
  \item When $\delta=4$, the passing space is $\Pi_{\eqref{eq:htt}=1}^{\bd}\Pi_{\neq 0}^{\text{last $\kappa$ bits of }\bd}$ and $\Pi_{\eqref{eq:htt}=0}^{\bd}\Pi_{\neq 0}^{\text{last $\kappa$ bits of }\bd}$ is in the failing space. Thus
 \begin{align}&\sum_{\theta_0\in \{0,1\cdots 7\}}\sum_{\theta_1\in \{0,1\cdots 7\}}\frac{1}{3}|\Pi_{\eqref{eq:htt}=0}^{\bd}\Pi_{\neq 0}^{\text{last $\kappa$ bits of }\bd}\fHadamardTest^\fAdv(\bK^{(0)};1^\kappa)(\ket{\varphi_{0,\theta_0,\theta_1}}+\ket{\varphi_{1,\theta_0,\theta_1}})\odot (\theta_1-\theta_0-4)|^2\\\leq& 3\sqrt{\epsilon_1+\epsilon_0}\label{eq:173ma}\end{align}
 \end{itemize}
  Define
  \begin{equation}\label{eq:174rw}\ket{\varphi^\prime_{b,\theta_0,\theta_1,\alpha}}=\Pi_{\eqref{eq:htt}=0}^{\bd}\Pi_{\neq 0}^{\text{last $\kappa$ bits of }\bd}\fHadamardTest^\fAdv(\bK^{(0)};1^\kappa)(\ket{\varphi_{b,\theta_0,\theta_1}}\odot \alpha) \end{equation}
  Recall its subscripts correspond to the values of $\bTheta^{(0)}$ register and the transcript $\alpha$ ($=\theta_1-\theta_0-\delta$) of a specific branch.\par
  Then define $\ket{\psi^\prime_{b,\theta_0,\theta_1,\alpha}}$ as the part of state in $\ket{\varphi^\prime_{b,\theta_0,\theta_1,\alpha}}$ excluding the registers $\btheta^{(0)}_0,\btheta^{(0)}_1$. That is,
  \begin{equation}\label{eq:174rx}\ket{\varphi^\prime_{b,\theta_0,\theta_1,\alpha}}=\underbrace{\ket{\theta_0}}_{\btheta_0^{(0)}}\underbrace{\ket{\theta_1}}_{\btheta_1^{(0)}}\otimes\ket{\psi^\prime_{b,\theta_0,\theta_1,\alpha}}\end{equation}
  By the basis-phase correpondance property described in Setup \ref{setup:4} we know $\ket{\psi^\prime_{0,\theta_0,\theta_1,\alpha}}$ is the same for different $\theta_1$ and $\ket{\psi^\prime_{1,\theta_0,\theta_1,\alpha}}$ is the same for different $\theta_0$. Thus we could introduce the following notation for these two states, where these unnecessary parameters are omitted:
  \begin{equation}\label{eq:174tc}\ket{\psi^\prime_{0,\theta_0,\cdot,\alpha}},\ket{\psi^\prime_{1,\cdot,\theta_1,\alpha}}\end{equation}
Then we can continue to calculate from \eqref{eq:99rr}:
\begin{align}
&\eqref{eq:99rr}\\
=&\sum_{\theta_0\in \{0,1\cdots 7\}}\sum_{\theta_1\in \{0,1\cdots 7\}}\frac{1}{3}|\ket{\psi^\prime_{0,\theta_0,\cdot,\theta_1-\theta_0-1}}+\ket{\psi^\prime_{1,\cdot,\theta_1,\theta_1-\theta_0-1}}|^2\\
=&\sum_{\theta_0\in \{0,1\cdots 7\}}\sum_{\alpha\in \{0,1\cdots 7\}}\frac{1}{3}|\ket{\psi^\prime_{0,\theta_0,\cdot,\alpha}}+\ket{\psi^\prime_{1,\cdot,\theta_0+1+\alpha,\alpha}}|^2\label{eq:105rr}
\end{align}
We will make use of Lemma \ref{lem:optla} to bound the expression above. For the conditions to apply this lemma, we know:
\begin{equation}\label{eq:106r}
\eqref{eq:171ma}\Rightarrow\sum_{\theta_0\in \{0,1\cdots 7\}}\sum_{\theta_1\in \{0,1\cdots 7\}}|\ket{\psi_{0,\theta_0,\cdot,\theta_1-\theta_0}^\prime}-\ket{\psi_{1,\cdot,\theta_1,\theta_1-\theta_0}^\prime}|^2\leq 6\sqrt{\epsilon_1+\epsilon_0}+\fneg(\kappa)	
\end{equation}
\begin{equation}\label{eq:107r}
\eqref{eq:173ma}\Rightarrow \sum_{\theta_0\in \{0,1\cdots 7\}}\sum_{\theta_1\in \{0,1\cdots 7\}}|\ket{\psi_{0,\theta_0,\cdot,\theta_1-\theta_0}^\prime}+\ket{\psi_{1,\cdot,\theta_1,\theta_1-\theta_0-4}^\prime}|^2\leq 3\sqrt{\epsilon_1+\epsilon_0}	
\end{equation}
Finally by Lemma \ref{lem:prepht} we can argue about the norms of these states:
\begin{equation}\label{eq:179}\forall b\in \{0,1\},\theta_0,\theta_1\in \{0,1\cdots 7\}^2, \alpha\in \{0,1\cdots 7\},|\ket{\psi^\prime_{b,\theta_0,\theta_1,\alpha}}|^2\leq \frac{1}{2}|\ket{\varphi_{b,\theta_0,\theta_1}}|^2+\fneg(\kappa)\leq \frac{1}{128}|\ket{\varphi_b}|^2+\fneg(\kappa)\end{equation}
which together with \eqref{eq:106r}\eqref{eq:107r} allows us to apply Lemma \ref{lem:optla} as follows. Define
\begin{equation}\label{eq:171j}\ket{\psi_{0,\theta_0}^\prime}=\sum_{\alpha\in \{0,1\cdots 7\}} \ket{\psi^\prime_{0,\theta_0,\cdot,\alpha}},\end{equation}
\begin{equation}\label{eq:172j}\ket{\psi_{1,\theta_0}^\prime}=\sum_{\alpha\in \{0,1\cdots 7\}}\ket{\psi^\prime_{1,\cdot,\theta_0+\alpha,\alpha}}\end{equation} 
Then \eqref{eq:106r}\eqref{eq:107r}\eqref{eq:179} translate to
\begin{equation}\label{eq:230}\sum_{\theta_0\in \{0,1\cdots 7\}}|\ket{\psi_{0,\theta_0}^\prime}-\ket{\psi_{1,\theta_0}^\prime}|^2\leq 6\sqrt{\epsilon_1+\epsilon_0}+\fneg(\kappa)\end{equation}
\begin{equation}\label{eq:231}\sum_{\theta_0\in \{0,1\cdots 7\}}|\ket{\psi_{0,\theta_0}^\prime}+\ket{\psi_{1,\theta_0+4}^\prime}|^2\leq 6\sqrt{\epsilon_1+\epsilon_0}+\fneg(\kappa)\end{equation}
\begin{equation}\label{eq:232}\forall b\in \{0,1\},\theta_0\in \{0,1\cdots 7\},|\ket{\psi_{b,\theta_0}^\prime}|^2\leq \frac{1}{16}|\ket{\varphi_{b}}|^2+\fneg(\kappa)\end{equation}
 then  we get 
\begin{align}
&\text{\eqref{eq:105rr} }\\
=&\frac{1}{3}\sum_{\theta_0\in \{0,1\cdots 7\}}|\ket{\psi^\prime_{0,\theta_0}}+\ket{\psi^\prime_{1,\theta_0+1}}|^2\\
\text{(apply Lemma \ref{lem:optla}) }\leq& \OPT +10(\epsilon_1+\epsilon_0)^{1/4}+\fneg(\kappa)	
\end{align}
Substituting \eqref{eq:158k}\eqref{eq:159k} completes the proof.
\end{proof}

\subsection{Randomization Operator $\cP$ for $\fInPhTest$}\label{sec:10.3}
We will define the randomization operator $\cP$ for the $\fInPhTest$. Different from $\cR_1$, $\cR_2$, this operator will be a projection operator and will not use additional randomness. Let's first give an intuitive discussion, and formalize it in Definition \ref{defn:insub}.
\subsubsection{Intuitive discussion}\label{sec:10.3.1}
Let's first consider the honest setting. The joint state of client side phase registers and the server-side state could be jointly written as
\begin{equation}\label{eq:192rh}\sum_{\theta_0^{(0)},\theta_1^{(0)}\in \{0,1\cdots 7\}^2}\frac{1}{8}\underbrace{\underbrace{\ket{\theta_0^{(0)}}}_{\btheta_0^{(0)}}\underbrace{\ket{\theta_1^{(0)}}}_{\btheta^{(0)}_1}}_{\text{client}}\otimes\underbrace{\frac{1}{\sqrt{2}}(e^{\theta_0^{(0)}\mi\pi/4}\ket{x_0^{(0)}}+e^{\theta_1^{(0)}\mi\pi/4}\ket{x_1^{(0)}})}_{\text{server}}\end{equation}
which is equal to the sum of two branches:
\begin{equation}\label{eq:105}(\sum_{\theta_0^{(0)},\theta_1^{(0)}\in \{0,1\cdots 7\}^2}\frac{1}{8}\underbrace{\frac{1}{\sqrt{2}}e^{\theta_0^{(0)}\mi\pi/4}\ket{\theta_0^{(0)}}\ket{\theta_1^{(0)}}}_{\text{client}}\otimes\underbrace{\ket{x_0^{(0)}}}_{\text{server}})+(\sum_{\theta_0^{(0)},\theta_1^{(0)}\in \{0,1\cdots 7\}^2}\frac{1}{8}\underbrace{\frac{1}{\sqrt{2}}\ket{\theta_0^{(0)}}e^{\theta_1^{(0)}\mi\pi/4}\ket{\theta_1^{(0)}}}_{\text{client}}\otimes\underbrace{\ket{x_1^{(0)}}}_{\text{server}})\end{equation}
We will define two sub-operators $\cP_{0,+},\cP_{1,+}$ that operate on the two branches correspondingly. 
$\cP_{0,+}$ operates nontrivially only on the $\bx_0^{(0)}$-branch (the first term of \eqref{eq:105}) while $\cP_{1,+}$ operates nontrivially only on the $\bx_1^{(0)}$-branch (the second term of \eqref{eq:105}). Without loss of generality, let's show the design of $\cP_{0,+}$. As before, we will also see the honest input state is indeed invariant under this operator.\par
 First note the first term of \eqref{eq:105} could be seen as the following state on $\btheta_0^{(0)}$ register tensoring other registers:
\begin{equation}\label{eq:106}\sum_{\theta_0^{(0)}\in \{0,1\cdots 7\}}\frac{1}{\sqrt{8}}e^{\theta_0^{(0)}\mi\pi/4}\ket{\theta_0^{(0)}}\end{equation}
Introduce an indicator register $\bindic_+$ which hold $\ffalse$ value by default. Recall that $\cP_{0,+}$, $\cP_{1,+}$ are projections; we will use this indicator register to record whether these projections are successful (which means, if these projections are applied on honest inputs, the projections will always be successful, and this register will be flipped to $\ftrue$ deterministically). Then $\cP_{0,+}$ applied on \eqref{eq:106} goes as follows:
\begin{align}
&\text{Equation \eqref{eq:106}}\otimes \underbrace{\ket{\ffalse}}_{\bindic_{+}}\label{eq:108}\\
&\text{(Control phase gate that adds phase $e^{-\theta_0^{(0)}\mi\pi/4}$ on $\btheta_0^{(0)}$ register when it has value $\theta_0^{(0)}$)}\label{eq:111}\\
\rightarrow&\sum_{\theta_0^{(0)}\in \{0,1\cdots 7\}}\frac{1}{\sqrt{8}}\ket{\theta_0^{(0)}}\quad \text{(Note that it is $=\ket{+}\ket{+}\ket{+}$)}\\
&\text{($\fH^{\otimes 3}$, followed by a projection measurement on $\{\ket{0}\bra{0},\bbI-\ket{0}\bra{0}\}$; use $\ftrue$ to indicate $\ket{0}\bra{0}$ and $\ffalse$ otherwise) }\\
\rightarrow &\ket{0}\ket{0}\ket{0}\otimes \underbrace{\ket{\ftrue}}_{\bindic_+}\label{eq:114re}\\
&\text{(Reverse the Hadamard and phase operations)}\label{eq:114}\\
\rightarrow & \text{Equation \eqref{eq:106}}\otimes \ket{\ftrue}\label{eq:182j}
\end{align}
Besides $\cP_{0,+},\cP_{1,+}$, we also need to define $\cP_{0,-},\cP_{1,-}$, as follows. 
 As discussed in Section \ref{sec:2.1}, in the malicious setting there is no way so far to rule out the \emph{complex conjugate attack}. Correspondingly, we define the sub-operators $\cP_{0,-},\cP_{1,-}$ that fix the complex conjugate of the honest state, which is 
\begin{equation}\label{eq:115}\sum_{\theta_0^{(0)},\theta_1^{(0)}\in \{0,1\cdots 7\}^2}\frac{1}{8}\underbrace{\ket{\theta_0^{(0)}}\ket{\theta_1^{(0)}}}_{\text{client}}\otimes\underbrace{\frac{1}{\sqrt{2}}(e^{-\theta_0^{(0)}\mi\pi/4}\ket{x_0^{(0)}}+e^{-\theta_1^{(0)}\mi\pi/4}\ket{x_1^{(0)}})}_{\text{server}}\end{equation}
The construction is similar to \eqref{eq:108}-\eqref{eq:182j} with the following differences: \begin{itemize}\item Compared to \eqref{eq:111} the phase is changed to $e^{\theta_0^{(0)}\mi\pi/4}$.
 \item The measurement results are stored in a different indicator register denoted by $\bindic_-$.	
 \end{itemize}
Importantly, our construction of $\cP_{b,+}$ and $\cP_{b,-}$, $b\in \{0,1\}$ have the following properties, which can be verified by a direct calculation: \begin{itemize}\item When $\cP_{0,+}$, $\cP_{1,+}$ are applied on \eqref{eq:115}, the outcome value in the indicator register $\bindic_+$ is deterministically $\ffalse$. When they are applied on \eqref{eq:192rh} $\bindic_+$ is deterministically $\ftrue$.
\item When $\cP_{0,-}$, $\cP_{1,-}$ is applied on \eqref{eq:192rh}, the outcome value in the indicator register $\bindic_-$ is deterministically $\ffalse$. When they are applied on \eqref{eq:115} $\bindic_-$ is deterministically $\ftrue$.
 \end{itemize}
 Finally define $\cP$ to be the sequential application of each of these four suboperators with suitable projections. (The constructions guarantee that the order of these suboperators does not really matter as long as the initial state is in some specific form). And the outcome of the indicator registers will indicate whether the server-side state has honest phases or conjugated phases.\par
Below we give the formal definitions.
\subsubsection{Formalization}
\begin{defn}\label{defn:insub}
Consider the register setup in Section \ref{sec:6.1r}. Explicitly, we can assume a purified joint state $\ket{\varphi}$ where the operators will act nontrivially on is in the basis-honest form of key pair $\bK^{(0)}$ (since our operators will act as identity on spaces outside this form):
$$\ket{\varphi}=\underbrace{\sum_{K^{(0)}\in \Domain(\bK^{(0)})}\ket{K^{(0)}}\otimes \sum_{\Theta^{(0)}\in \Domain(\bTheta^{(0)})}\ket{\Theta^{(0)}}}_{\text{client}}\otimes \underbrace{\sum_{b\in \{0,1\}}\ket{x_b^{(0)}}}_{\bS^{(0)}_{bsh}}\otimes \ket{\varphi_{K,\Theta,b}}$$
First define some intermediate operators that will be used in our construction.
\begin{itemize}\item Define the control-phase operator that controlled on the $\bx^{(0)}_b$-branch, adds a phase determined by $\btheta_b^{(0)}$:
$$\fCPhase_{b}(+): \underbrace{\ket{(x_0^{(0)},x_1^{(0)})}}_{\bK^{(0)}}\underbrace{\ket{\theta_b^{(0)}}}_{\btheta^{(0)}_b}\underbrace{\ket{x_b^{(0)}}}_{\bS^{(0)}_{bsh}}\rightarrow \ket{(x_0^{(0)},x_1^{(0)})}e^{\theta_b^{(0)} \mi\pi/4}\ket{\theta_b^{(0)}}\ket{x_b^{(0)}}$$
$$\fCPhase_b(-):\ket{(x_0^{(0)},x_1^{(0)})}\ket{\theta_b^{(0)}}\ket{x_b^{(0)}}\rightarrow \ket{(x_0^{(0)},x_1^{(0)})}e^{-\theta_b^{(0)} \mi\pi/4}\ket{\theta_b^{(0)}}\ket{x_b^{(0)}}$$
\item Define $\fM_{\btheta^{(0)}_b, \bindic}$ as the control-flip operator that flip the state of $\bindic$ when $\btheta^{(0)}_b$ register is in all-zero state.
\item We use $\fH^{\otimes 3}_{\btheta_b^{(0)}}$ to denote the bit-wise Hadamard on client-side register $\btheta_b^{(0)}$.
\end{itemize}
Initialize single-bit registers $\bindic_{+},\bindic_{-}$ to hold $\ffalse$ by default. 
Define suboperators as follows:
$$\cP_{0,+}=\fCPhase_0(+)\fH^{\otimes 3}_{\btheta^{(0)}_0}\fM_{\btheta^{(0)}_0, \bindic_{+}}\fH^{\otimes 3}_{\btheta^{(0)}_0}\fCPhase_0(-)$$
$$\cP_{0,-}=\fCPhase_0(-)\fH^{\otimes 3}_{\btheta^{(0)}_0}\fM_{\btheta^{(0)}_0, \bindic_{-}}\fH^{\otimes 3}_{\btheta^{(0)}_0}\fCPhase_0(+)$$
$$\cP_{1,+}=\fCPhase_1(+)\fH^{\otimes 3}_{\btheta^{(0)}_1}\fM_{\btheta^{(0)}_1,\bindic_{+}}\fH^{\otimes 3}_{\btheta^{(0)}_1}\fCPhase_1(-)$$
$$\cP_{1,-}=\fCPhase_1(-)\fH^{\otimes 3}_{\btheta^{(0)}_1}\fM_{\btheta^{(0)}_1,\bindic_{-}}\fH^{\otimes 3}_{\btheta^{(0)}_1}\fCPhase_1(+)$$
Finally the overall randomization operator is:
\begin{equation}\label{eq:204ta}\cP=(\Pi_{\bx_0^{(0)}}^{\bS_{bsh}^{(0)}}(\Pi^{\bindic_{+}}_{\ftrue}\Pi^{\bindic_{-}}_{\ffalse}+ \Pi^{\bindic_{+}}_{\ffalse}\Pi^{\bindic_{-}}_{\ftrue})+\Pi_{\bx_1^{(0)}}^{\bS_{bsh}^{(0)}}(\Pi^{\bindic_{+}}_{\ftrue}\Pi^{\bindic_{-}}_{\ffalse}+ \Pi^{\bindic_{+}}_{\ffalse}\Pi^{\bindic_{-}}_{\ftrue}))\cP_{1,-}\cP_{1,+}\cP_{0,-}\cP_{0,+}\end{equation}
Recall $\Pi_{\bx_b^{(0)}}^{\bS_{bsh}^{(0)}}$ is the projection onto the $\bx_b^{(0)}$-branch of the basis-honest part.   
\end{defn}
We add some comments for understanding this definition. First note $\cP_{1,-}\cP_{1,+}\cP_{0,-}\cP_{0,+}$ are all defined to be unitary operators here and projections happen in \eqref{eq:204ta}. Note multiplication of commuting projectors is equivalent to logical and and summation of orthogonal projectors is equivalent to logical or. If we focus on each single term of \eqref{eq:204ta}, and ignore the possible interference between different $\cP_{b,\pm}$, this construction is the same as what we intuitively discussed in Section \ref{sec:10.3.1}. In the next subsubsection we will formally analyze how the honest state, or its complex conjugate, evolves under this operator.\par
\subsubsection{$\cP$ behaves well on states with honest phases or its complex conjugates}
As discussed in the beginning of Section \ref{sec:10}, we aim at testing \eqref{eq:158d}. Here we first prove the form of states shown in the right hand side of \eqref{eq:158d} do behave well under $\cP$. In more detail, $\cP$ will revise the indicator registers $\bindic_+$, $\bindic_-$; if the phase is $e^{\theta_b^{(0)}\mi\pi/4}$ as shown in the first term in \eqref{eq:158d} the $\bindic_+$ register will be flipped to $\ftrue$, and if the phase is  $e^{-\theta_b^{(0)}\mi\pi/4}$ as shown in the second term in \eqref{eq:158d} the $\bindic_-$ register will be flipped to $\ftrue$. In addition, when the initial state is in the form of \eqref{eq:158d}, $\cP$ behaves as a unitary; thus $\cP^\dagger\cP$ maps the state to the original state, which is similar to the case where $\cR_1,\cR_2$ are applied on corresponding honest states.\par
Below we formalize this discussion as a lemma.
\begin{lem}\label{lem:10.5}
Suppose the register setup is the same as Section \ref{sec:6.1r}, especially, the client holds key pair $\bK^{(0)}=(\bx_0^{(0)},\bx_1^{(0)})$ and the corresponding phase pair $\bTheta^{(0)}=(\btheta^{(0)}_0,\btheta^{(0)}_1)$. Suppose the purified joint state $\ket{\varphi}$ satisfies: for each $b\in \{0,1\},K^{(0)}\in \Domain(\bK^{(0)})$ there exist states $\ket{\varphi_{K^{(0)},b,+}}$, $\ket{\varphi_{K^{(0)},b,-}}$ such that:
$$\ket{\varphi}=\underbrace{\sum_{K^{(0)}\in \Domain(\bK^{(0)})}\ket{K^{(0)}}\otimes \sum_{\Theta^{(0)}\in \Domain(\bTheta^{(0)})}\ket{\Theta^{(0)}}}_{\text{client}}\otimes \sum_{b\in \{0,1\}}\underbrace{\ket{x_b^{(0)}}}_{\bS_{bsh}^{(0)}}\otimes (e^{\theta_b^{(0)}\mi\pi/4}\ket{\varphi_{K^{(0)},b,+}}+e^{-\theta_b^{(0)}\mi\pi/4}\ket{\varphi_{K^{(0)},b,-}})$$
Then $\cP\ket{\varphi}$ is the linear sum of the following states:
\begin{equation}\label{eq:118}\sum_{K^{(0)}\in \Domain(\bK^{(0)})}\ket{K^{(0)}}\otimes \underbrace{\ket{\ftrue}}_{\bindic_{+}}\ket{\ffalse}\otimes \sum_{\theta_0^{(0)},\theta_1^{(0)}\in \{0,1\cdots 7\}^2}\underbrace{\ket{\theta_0^{(0)}}\ket{\theta_1^{(0)}}}_{\bTheta^{(0)}}\otimes \underbrace{\ket{x_0^{(0)}}}_{\bS_{bsh}^{(0)}}\otimes e^{\theta^{(0)}_0\mi\pi/4}\ket{\varphi_{K^{(0)},0,+}}\end{equation}
\begin{equation}\label{eq:119}\sum_{K^{(0)}\in \Domain(\bK^{(0)})}\ket{K^{(0)}}\otimes \ket{\ffalse}\underbrace{\ket{\ftrue}}_{\bindic_{-}}\otimes \sum_{\theta_0^{(0)},\theta_1^{(0)}\in \{0,1\cdots 7\}^2}\ket{\theta_0^{(0)}}\ket{\theta_1^{(0)}}\otimes \ket{x_0^{(0)}}\otimes e^{-\theta^{(0)}_0\mi\pi/4}\ket{\varphi_{K^{(0)},0,-}}\end{equation}
$$\sum_{K^{(0)}\in \Domain(\bK^{(0)})}\ket{K^{(0)}}\otimes \underbrace{\ket{\ftrue}}_{\bindic_{+}}\ket{\ffalse}\otimes \sum_{\theta_0^{(0)},\theta_1^{(0)}\in \{0,1\cdots 7\}^2}\ket{\theta_0^{(0)}}\ket{\theta_1^{(0)}}\otimes \ket{x_1^{(0)}}\otimes e^{\theta^{(0)}_1\mi\pi/4}\ket{\varphi_{K^{(0)},1,+}}$$
$$\sum_{K^{(0)}\in \Domain(\bK^{(0)})}\ket{K^{(0)}}\otimes \ket{\ffalse}\underbrace{\ket{\ftrue}}_{\bindic_{-}}\otimes \sum_{\theta_0^{(0)},\theta_1^{(0)}\in \{0,1\cdots 7\}^2}\ket{\theta_0^{(0)}}\ket{\theta_1^{(0)}}\otimes \ket{x_1^{(0)}}\otimes e^{-\theta^{(0)}_1\mi\pi/4}\ket{\varphi_{K^{(0)},1,-}}$$
And
\begin{equation}\label{eq:247}\cP^\dagger\cP\ket{\varphi}= \ket{\varphi}\end{equation}	
\end{lem}
\begin{proof}
The proof is by a direct calculation as discussed before this theorem. Without loss of generality we calculate on the $\bx_0^{(0)}$ branch. Then we only need to prove
\begin{align}&(\Pi^{\bindic_{+}}_{\ftrue}\Pi^{\bindic_{-}}_{\ffalse}+ \Pi^{\bindic_{+}}_{\ffalse}\Pi^{\bindic_{-}}_{\ftrue})\cP_{0,-}\cP_{0,+}\sum_{K^{(0)}\in \Domain(\bK^{(0)})}\ket{K^{(0)}}\otimes \sum_{\Theta^{(0)}\in \Domain(\bTheta^{(0)})}\ket{\Theta^{(0)}}\otimes\nonumber\\
&\qquad\qquad \ket{x_0^{(0)}}\otimes (e^{\theta_0^{(0)}\mi\pi/4}\ket{\varphi_{K^{(0)},0,+}}+e^{-\theta_0^{(0)}\mi\pi/4}\ket{\varphi_{K^{(0)},0,-}})\\
=&\eqref{eq:118}+\eqref{eq:119}	
\end{align}
which is further reduced to
\begin{align}&\cP_{0,-}\cP_{0,+}\ket{\ffalse}\ket{\ffalse}e^{\theta_0^{(0)}\mi\pi/4}\sum_{\theta_0^{(0)},\theta_1^{(0)}\in \{0,1\cdots 7\}^2}\ket{\theta_0^{(0)}}\ket{\theta_1^{(0)}}\label{eq:249}\\
=&\underbrace{\ket{\ftrue}}_{\bindic_{+}}\ket{\ffalse}e^{\theta_0^{(0)}\mi\pi/4}\sum_{\theta_0^{(0)},\theta_1^{(0)}\in \{0,1\cdots 7\}^2}\ket{\theta_0^{(0)}}\ket{\theta_1^{(0)}}\label{eq:250}
\end{align}

\begin{align}&\cP_{0,-}\cP_{0,+}\ket{\ffalse}\ket{\ffalse}e^{-\theta_0^{(0)}\mi\pi/4}\sum_{\theta_0^{(0)},\theta_1^{(0)}\in \{0,1\cdots 7\}^2}\ket{\theta_0^{(0)}}\ket{\theta_1^{(0)}}\\
=&\ket{\ffalse}\underbrace{\ket{\ftrue}}_{\bindic_{-}}e^{-\theta_0^{(0)}\mi\pi/4}\sum_{\theta_0^{(0)},\theta_1^{(0)}\in \{0,1\cdots 7\}^2}\ket{\theta_0^{(0)}}\ket{\theta_1^{(0)}}	
\end{align}

which are true by direct calculations as \eqref{eq:108}-\eqref{eq:182j}. (As an example, in \eqref{eq:249} the application of $\cP_{0,+}$ maps the state to \eqref{eq:250}, and then $\cP_{0,-}$ will keep the state invariant since the first Hadamard transform in $\cP_{0,-}$ maps the $\btheta_0^{(0)}$ register to $\ket{110}$ on which $\fM$ operator acts as identity.)\par
Then the projections in the definition of $\cP$ acts as identity thus \eqref{eq:247} follows.
\end{proof}
\subsubsection{$\cP$ projects a pre-phase-honest form to a phase-honest form}
The restriction so far is, we are only focusing on a single pair of phases $\Theta^{(0)}$, which are the phases got tested in $\fInPhTest$. $\cP$ only operates on a single phase pair, but we want to argue about the overall property of the whole state---that is, the overall state should be in the phase-honest form (Definition \ref{defn:phf}). The next lemma says, when the input state is a pre-phase-honest form, the output of $\cP$  will be a phase-honest form:
\begin{lem}\label{lem:10.6}
Suppose the register setup is the same as Section \ref{sec:6.1r}, especially, the client holds a tuple of key pairs $\bK=(\bK^{(i)})_{i\in [0,L]}$, $\bK^{(i)}=(\bx^{(i)}_0,\bx^{(i)}_1)$, and a tuple of phase pairs $\bTheta=(\bTheta^{(i)})_{i\in [0,L]}$, $\bTheta^{(i)}=(\btheta^{(i)}_0,\btheta^{(i)}_1)$. Suppose a purified joint state $\ket{\varphi}$ is in the pre-phase-honest form. Then there exist states $\ket{\varphi_{K,\vec{b},+}}$, $\ket{\varphi_{K,\vec{b},-}}$, (for each $K\in \Domain(\bK)$, $\vec{b}\in \{0,1\}^{1+L}$) such that:
\begin{align}
&\cP\ket{\varphi}\\
=&\underbrace{\ket{\ftrue}}_{\bindic_{+}}\ket{\ffalse}\sum_{K\in \Domain(\bK)}\ket{K}\otimes \sum_{\Theta\in \Domain(\bTheta)}\ket{\Theta}\otimes \sum_{\vec{b}\in \{0,1\}^{1+L}}\underbrace{\ket{\vec{x}_{\vec{b}}}}_{\bS_{bsh}}\otimes e^{\tSUM(\vec{\Theta}_{\vec{b}})\pi\mi/4}\ket{\varphi_{K,\vec{b},+}}\label{eq:254}\\	
&+\ket{\ffalse}\underbrace{\ket{\ftrue}}_{\bindic_{-}}\sum_{K\in \Domain(\bK)}\ket{K}\otimes \sum_{\Theta\in \Domain(\bTheta)}\ket{\Theta}\otimes \sum_{\vec{b}\in \{0,1\}^{1+L}}\ket{\vec{x}_{\vec{b}}}\otimes e^{-\tSUM(\vec{\Theta}_{\vec{b}})\pi\mi/4}\ket{\varphi_{K,\vec{b},-}}\label{eq:255}
\end{align}
\end{lem}
\begin{proof}
Since $\ket{\varphi}$ is in a pre-phase-honest form, we can assume there exist states $\ket{\varphi_{K,\vec{b},sum}}$ (for each $K\in \Domain(\bK)$, $\vec{b}\in \{0,1\}^{1+L}$, $sum\in \{0,1\cdots 7\}$) such that $\ket{\varphi}$  has the form
\begin{equation}
	\sum_{K\in \Domain(\bK)}\ket{K}\otimes \sum_{\Theta\in \Domain(\bTheta)}\ket{\Theta}\otimes \sum_{\vec{b}\in \{0,1\}^{1+L}}\underbrace{\ket{\vec{x}_{\vec{b}}}}_{\bS_{bsh}}\otimes \ket{\varphi_{K,\vec{b},\tSUM(\vec{\Theta}_{\vec{b}})}}
\end{equation}
Without loss of generality consider $\vec{b}\in \{0,1\}^{1+L}$ whose first bit is $0$. Then from the definition of $\cP$ we have two terms to calculate. Let's first calculate $\Pi^{\bindic_{+}}_{\ftrue}\Pi^{\bindic_{-}}_{\ffalse}\cP_{0,-}\cP_{0,+}\Pi_{\bx_{\vec{b}}}^{\bS_{bsh}}\ket{\varphi}$. First we know:
	\begin{align}
		&\Pi^{\bindic_{+}}_{\ftrue}\cP_{0,+}\Pi_{\bx_{\vec{b}}}^{\bS_{bsh}}\ket{\varphi}\label{eq:258r}\\
		=&\fCPhase_0(+)\fH^{\otimes 3}_{\btheta^{(0)}_0}\Pi^{\bindic_{+}}_{\ftrue}\fM_{\btheta^{(0)}_0,\bindic_+}\fH^{\otimes 3}_{\btheta^{(0)}_0}\fCPhase_0(-)\sum_{K\in \Domain(\bK)}\ket{K}\otimes \sum_{\Theta\in \Domain(\bTheta)}\ket{\Theta}\otimes \underbrace{\ket{\vec{x}_{\vec{b}}}}_{\bS_{bsh}}\otimes \ket{\varphi_{K,\vec{b},\tSUM(\vec{\Theta}_{\vec{b}})}}\\
				=&\fCPhase_0(+)\fH^{\otimes 3}_{\btheta^{(0)}_0}\Pi^{\bindic_{+}}_{\ftrue}\fM_{\btheta^{(0)}_0,\bindic_+}\fH^{\otimes 3}_{\btheta^{(0)}_0}\sum_{K\in \Domain(\bK)}\ket{K}\otimes \sum_{\Theta\in \Domain(\bTheta)}\ket{\Theta}\otimes \underbrace{\ket{\vec{x}_{\vec{b}}}}_{\bS_{bsh}}\otimes e^{-\theta^{(0)}_0\pi\mi/4}\ket{\varphi_{K,\vec{b},\tSUM(\vec{\Theta}_{\vec{b}})}}\\
			=&\fCPhase_0(+)\underbrace{\ket{\ftrue}}_{\bindic_{+}}\underbrace{\ket{\ffalse}}_{\bindic_-}\sum_{K\in \Domain(\bK)}\ket{K}\otimes \sum_{\Theta\in \Domain(\bTheta)}\ket{\Theta}\otimes \underbrace{\ket{\vec{x}_{\vec{b}}}}_{\bS_{bsh}}\otimes \frac{1}{8}\sum_{\alpha\in \{0,1\cdots 7\}} e^{-\alpha\pi\mi/4}\ket{\varphi_{K,\vec{b},\alpha-\theta^{(0)}_0+\tSUM(\vec{\Theta}_{\vec{b}})}}\\
		=&\underbrace{\ket{\ftrue}}_{\bindic_{+}}\ket{\ffalse}\sum_{K\in \Domain(\bK)}\ket{K}\otimes \sum_{\Theta\in \Domain(\bTheta)}\ket{\Theta}\otimes \underbrace{\ket{\vec{x}_{\vec{b}}}}_{\bS_{bsh}}\otimes \frac{1}{8}\sum_{\alpha\in \{0,1\cdots 7\}} e^{-(\alpha-\theta^{(0)}_0)\pi\mi/4}\ket{\varphi_{K,\vec{b},\alpha-\theta^{(0)}_0+\tSUM(\vec{\Theta}_{\vec{b}})}}\\
		=&\underbrace{\ket{\ftrue}}_{\bindic_{+}}\ket{\ffalse}\sum_{K\in \Domain(\bK)}\ket{K}\otimes \sum_{\Theta\in \Domain(\bTheta)}\ket{\Theta}\otimes \underbrace{\ket{\vec{x}_{\vec{b}}}}_{\bS_{bsh}}\otimes e^{\tSUM(\vec{\Theta}_{\vec{b}})\mi\pi/4}\frac{1}{8}\sum_{\beta\in \{0,1\cdots 7\}} e^{-\beta\pi\mi/4}\ket{\varphi_{K,\vec{b},\beta}}\label{eq:220}
	\end{align}
	Then as shown in Lemma \ref{lem:10.5} $\cP_{0,-}$ keeps \eqref{eq:220} invariant. Thus we get
	$$\Pi^{\bindic_{+}}_{\ftrue}\Pi^{\bindic_{-}}_{\ffalse}\cP_{0,-}\cP_{0,+}\Pi_{\bx_{\vec{b}}}^{\bS_{bsh}}\ket{\varphi}=\eqref{eq:220}$$
	which has the form of state required in \eqref{eq:254} if we define
	$$\ket{\varphi_{K,\vec{b},+}}:=\underbrace{\ket{\ftrue}}_{\bindic_{+}}\ket{\ffalse}\otimes\frac{1}{8}\sum_{\beta\in \{0,1\cdots 7\}} e^{-\beta\pi\mi/4}\ket{\varphi_{K,\vec{b},\beta}}$$
		Now we calculate $\Pi^{\bindic_{+}}_{\ffalse}\Pi^{\bindic_{-}}_{\ftrue}\cP_{0,-}\cP_{0,+}\Pi_{\bx_{\vec{b}}}^{\bS_{bsh}}\ket{\varphi}$ and show it has the form required in \eqref{eq:255}. We could do a similar direct calculation; the case here is slightly more complicated (note that $\cP_{0,+}$ and $\cP_{0,-}$ are not known to be commutative) but still possible; but here we choose a short path where we re-use the calculations we did just now.\par
		Recall the construction of $\cP_{0,+}$, it has the form of $U^\dagger \fM_{\btheta^{(0)}_0,\bindic_{+}}U$. That implies when the initial state has value $\ket{\ffalse}$ in $\bindic_+$, $\Pi^{\bindic_{+}}_{\ffalse}\cP_{0,+}$ is the same as $\bbI-\fX_{\bindic_+}\Pi^{\bindic_{+}}_{\ftrue}\cP_{0,+}$, where $\fX_{\bindic_+}$ is an operator that flips the value of register $\bindic_+$. Thus
		$$\Pi^{\bindic_{+}}_{\ffalse}\cP_{0,+}\Pi_{\bx_{\vec{b}}}^{\bS_{bsh}}\ket{\varphi}=\Pi_{\bx_{\vec{b}}}^{\bS_{bsh}}\ket{\varphi}-\fX_{\bindic_+}\Pi^{\bindic_{+}}_{\ftrue}\cP_{0,+}\Pi_{\bx_{\vec{b}}}^{\bS_{bsh}}\ket{\varphi}=\Pi_{\bx_{\vec{b}}}^{\bS_{bsh}}\ket{\varphi}-\fX_{\bindic_+}\eqref{eq:220}$$
		Note the second term is invariant under $\cP_{0,-}$ thus satisfies $\Pi^{\bindic_{-}}_{\ftrue}\cP_{0,-}\eqref{eq:220}=0$. This implies
		$$\Pi^{\bindic_{+}}_{\ffalse}\Pi^{\bindic_{-}}_{\ftrue}\cP_{0,-}\cP_{0,+}\ket{\varphi}=\Pi^{\bindic_{-}}_{\ftrue}\cP_{0,-}\ket{\varphi}$$
		Then a calculation similar to \eqref{eq:258r} to \eqref{eq:220} shows this term has the form of \eqref{eq:255} if we define
			$$\ket{\varphi_{K,\vec{b},-}}:=\ket{\ffalse}\underbrace{\ket{\ftrue}}_{\bindic_{-}}\otimes\frac{1}{8}\sum_{\beta\in \{0,1\cdots 7\}} e^{\beta\pi\mi/4}\ket{\varphi_{K,\vec{b},\beta}}$$
\end{proof}

\subsection{$\fInPhTest$ Implies Approximate Invariance Under $\cP^\dagger\cP$}\label{sec:10.4}
In this section we show passing $\fInPhTest$ implies approximate invariance under $\cP^\dagger\cP$.
\begin{thm}\label{thm:9.8}
Suppose a sub-normalized purified joint state $\ket{\varphi}$ is in Setup \ref{setup:4} and is in a $\epsilon$-basis-honest form. Suppose an efficient adversary $\fAdv$ on initial state $\ket{\varphi}$ in $\fInPhTest$ can make the client output $\fpass$ as the flag with probability $\geq 1-\epsilon$ and make the client output $\fwin$ as the score with probability $\geq \OPT-\epsilon$. Then we have
$$|(\bbI-\cP^\dagger\cP)\ket{\varphi}|\leq 50\epsilon^{1/16}+\fneg(\kappa)$$
\end{thm}
\begin{proof}
First we use a similar argument to the proof of Theorem \ref{thm:9.2}. Similarly define $\ket{\psi^\prime_{0,\theta_0}},\ket{\psi^\prime_{1,\theta_0}}$ as \eqref{eq:171j}\eqref{eq:172j}. 
As given in the condition, suppose an adversary $\fAdv$ can pass the individual phase test with high probability:
$$|\Pi_{\fpass}\fInPhTest^\fAdv(\bK,\bTheta;1^\kappa)\ket{\varphi}|^2\geq 1-\epsilon$$
 By the same argument we know
$$\sum_{\theta\in \{0,1\cdots 7\}}| \ket{\psi_{0,\theta_0}^\prime}-\ket{\psi_{1,\theta_0}^\prime}|^2\leq 6\sqrt{2\epsilon}+\fneg(\kappa)$$
$$\sum_{\theta\in \{0,1\cdots 7\}}| \ket{\psi_{0,\theta_0}^\prime}+\ket{\psi_{1,\theta_0+4}^\prime}|^2\leq 6\sqrt{2\epsilon}+\fneg(\kappa)$$
$$\forall b\in \{0,1\},\theta_0\in \{0,1\cdots 7\},|\ket{\psi_{b,\theta_0}^\prime}|^2\leq \frac{1}{16}|\ket{\varphi_{b}}|^2+\fneg(\kappa)$$
where $|\ket{\varphi_{b}}|$ is the norm of $\bx_b^{(0)}$-branch of $\ket{\varphi}$.\par
If the server can win with significant probability:
$$|\Pi_{\fwin}\fInPhTest^\fAdv(\bK,\bTheta;1^\kappa)\ket{\varphi}|^2\geq \OPT-\epsilon$$
$$\Rightarrow|\Pi_{\fwin}\fInPhTest^\fAdv(\bK,\bTheta;1^\kappa)\Pi_{\basishonest(\bK^{(0)})}\ket{\varphi}|^2\geq \OPT-3\epsilon$$
As in the proof of Theorem \ref{thm:9.2}, it implies
$$\sum_{\theta\in \{0,1\cdots 7\}}|\ket{\psi^\prime_{0,\theta_0}}+\ket{\psi^\prime_{1,\theta_0+1}}|^2\geq \cos^2(\pi/8)-9\epsilon$$
Now applying Lemma \ref{lem:10.2rr} we know
\begin{equation}\label{eq:122}\sum_{b\in \{0,1\}}\sum_{ \theta_0\in \{0,1\cdots 7\}}|\ket{\psi^\prime_{b,\theta_0}}-(e^{\theta_0\mi\pi/4}\ket{\psi^\prime_{b,+}}+e^{-\theta_0\mi\pi/4}\ket{\psi^\prime_{b,-}})|^2\leq 1200\epsilon^{1/8}+\fneg(\kappa)\end{equation}
where 
\begin{equation}\label{eq:268}\ket{\psi^\prime_{b,+}}:=\frac{1}{8}\sum_{\theta\in \{0,1\cdots 7\}}e^{-\theta\mi\pi/4}\ket{\psi^\prime_{b,\theta}},\ket{\psi^\prime_{b,-}}:=\frac{1}{8}\sum_{\theta\in \{0,1\cdots 7\}}e^{\theta\mi\pi/4}\ket{\psi^\prime_{b,\theta}}\end{equation}
Recall the definition of $\ket{\psi^\prime_{0,\theta_0}},\ket{\psi^\prime_{1,\theta_0}}$ in \eqref{eq:171j}\eqref{eq:172j}. 
Unroll \eqref{eq:122} by substituting definitions \eqref{eq:171j}\eqref{eq:172j}, we get
\begin{align}\label{eq:122j}&\sum_{ \theta_0\in \{0,1\cdots 7\}}|\sum_{\alpha\in \{0,1\cdots 7\}}\ket{\psi^\prime_{0,\theta_0,\cdot,\alpha}}-(e^{\theta_0\mi\pi/4}\ket{\psi^\prime_{0,+}}+e^{-\theta_0\mi\pi/4}\ket{\psi^\prime_{0,-}})|^2\\
&+\sum_{ \theta_0\in \{0,1\cdots 7\}}|\sum_{\alpha\in \{0,1\cdots 7\}}\ket{\psi^\prime_{1,\cdot,\theta_0+\alpha,\alpha}}-(e^{\theta_0\mi\pi/4}\ket{\psi^\prime_{1,+}}+e^{-\theta_0\mi\pi/4}\ket{\psi^\prime_{1,-}})|^2\\
\leq& 1200\epsilon^{1/8}+\fneg(\kappa)\end{align}
Recall in the definition of $\ket{\psi^\prime_{0,\theta_0}},\ket{\psi^\prime_{1,\theta_0}}$, $\alpha$ is store in a separate transcript register thus $\ket{\psi^\prime_{0,\theta_0,\cdot,\alpha}}$ (and also $\ket{\psi^\prime_{1,\cdot,\theta_0+\alpha,\alpha}}$) for different $\alpha$ are orthogonal states. Thus we can expand the norm-square-of-sum in \eqref{eq:122j} to sum-of-norm-square on values of $\alpha$:
\begin{align}\label{eq:122jqaz}&\sum_{ \theta_0\in \{0,1\cdots 7\}}\sum_{\alpha\in \{0,1\cdots 7\}}|\ket{\psi^\prime_{0,\theta_0,\cdot,\alpha}}-(e^{\theta_0\mi\pi/4}\ket{\psi^\prime_{0,+,\alpha}}+e^{-\theta_0\mi\pi/4}\ket{\psi^\prime_{0,-,\alpha}})|^2\\
&+\sum_{ \theta_0\in \{0,1\cdots 7\}}\sum_{\alpha\in \{0,1\cdots 7\}}|\ket{\psi^\prime_{1,\cdot,\theta_0+\alpha,\alpha}}-(e^{\theta_0\mi\pi/4}\ket{\psi^\prime_{1,+,\alpha}}+e^{-\theta_0\mi\pi/4}\ket{\psi^\prime_{1,-,\alpha}})|^2\\
\leq& 1200\epsilon^{1/8}+\fneg(\kappa)\label{eq:274qaz}\end{align}
where $\ket{\psi^\prime_{b,\pm,\alpha}}$ are defined to be the corresponding component of \eqref{eq:268}. Explicitly, they are:
\begin{equation}\label{eq:274qa}\ket{\psi^\prime_{0,+,\alpha}}:=\frac{1}{8}\sum_{\theta\in \{0,1\cdots 7\}}e^{-\theta\mi\pi/4}\ket{\psi^\prime_{0,\theta,\cdot,\alpha}},\ket{\psi^\prime_{0,-,\alpha}}:=\frac{1}{8}\sum_{\theta\in \{0,1\cdots 7\}}e^{\theta\mi\pi/4}\ket{\psi^\prime_{0,\theta,\cdot,\alpha}}\end{equation}
\begin{equation}\label{eq:275qa}\ket{\psi^\prime_{1,+,\alpha}}:=\frac{1}{8}\sum_{\theta\in \{0,1\cdots 7\}}e^{-\theta\mi\pi/4}\ket{\psi^\prime_{1,\cdot,\theta+\alpha,\alpha}},\ket{\psi^\prime_{1,-,\alpha}}:=\frac{1}{8}\sum_{\theta\in \{0,1\cdots 7\}}e^{\theta\mi\pi/4}\ket{\psi^\prime_{1,\cdot,\theta+\alpha,\alpha}}\end{equation}
So far we are working on variables $\theta_0$ and $\alpha$, which corresponds to the value of client-side register $\btheta_0^{(0)}$ and client's message in the Hadamard test. We use a change-of-variable to introduce $\theta_1$ to replace $\alpha$. First define
$$\ket{\tilde\psi^\prime_{0,+,\alpha}},\ket{\tilde\psi^\prime_{0,-,\alpha}}\text{ the same as \eqref{eq:274qa}}$$
$$\ket{\tilde\psi^\prime_{1,+,\alpha}}:=\frac{1}{8}\sum_{\theta\in \{0,1\cdots 7\}}e^{-\theta\mi\pi/4}\ket{\psi^\prime_{1,\cdot,\theta,\alpha}},\ket{\tilde\psi^\prime_{1,-,\alpha}}:=\frac{1}{8}\sum_{\theta\in \{0,1\cdots 7\}}e^{\theta\mi\pi/4}\ket{\psi^\prime_{1,\cdot,\theta,\alpha}}$$
Then \eqref{eq:122jqaz}-\eqref{eq:274qaz} could be re-written as
\begin{align}\label{eq:122jqaz3}&\sum_{ \theta_0\in \{0,1\cdots 7\}}\sum_{\theta_1\in \{0,1\cdots 7\}}|\ket{\psi^\prime_{0,\theta_0,\cdot,\theta_1-\theta_0}}-(e^{\theta_0\mi\pi/4}\ket{\tilde\psi^\prime_{0,+,\theta_1-\theta_0}}+e^{-\theta_0\mi\pi/4}\ket{\psi^\prime_{0,-,\theta_1-\theta_0}})|^2\\
&+\sum_{ \theta_0\in \{0,1\cdots 7\}}\sum_{\theta_1\in \{0,1\cdots 7\}}|\ket{\psi^\prime_{1,\cdot,\theta_1,\theta_1-\theta_0}}-(e^{\theta_1\mi\pi/4}\ket{\tilde\psi^\prime_{1,+,\theta_1-\theta_0}}+e^{-\theta_1\mi\pi/4}\ket{\tilde\psi^\prime_{1,-,\theta_1-\theta_0}})|^2\\
\leq& 1200\epsilon^{1/8}+\fneg(\kappa)\label{eq:274qaz3}\end{align}
As discussed in the paragraph above \eqref{eq:174tc}, the ``$\cdot$'' could be replaced by any value without changing the state. This implies \eqref{eq:122jqaz3}-\eqref{eq:274qaz3} could be further re-written as
\begin{align}\label{eq:122jqaz2}&\sum_{ \theta_0\in \{0,1\cdots 7\}}\sum_{\theta_1\in \{0,1\cdots 7\}}|\ket{\psi^\prime_{0,\theta_0,\theta_1,\theta_1-\theta_0}}-(e^{\theta_0\mi\pi/4}\ket{\tilde\psi^\prime_{0,+,\theta_1-\theta_0}}+e^{-\theta_0\mi\pi/4}\ket{\tilde\psi^\prime_{0,-,\theta_1-\theta_0}})|^2\\
&+\sum_{ \theta_0\in \{0,1\cdots 7\}}\sum_{\theta_1\in \{0,1\cdots 7\}}|\ket{\psi^\prime_{1,\theta_0,\theta_1,\theta_1-\theta_0}}-(e^{\theta_1\mi\pi/4}\ket{\tilde\psi^\prime_{1,+,\theta_1-\theta_0}}+e^{-\theta_1\mi\pi/4}\ket{\tilde\psi^\prime_{1,-,\theta_1-\theta_0}})|^2\\
\leq& 1200\epsilon^{1/8}+\fneg(\kappa)\end{align}
This can be further unrolled by \eqref{eq:174rx}:
\begin{equation}\label{eq:205yb}\sum_{b\in \{0,1\}}\sum_{ \theta_0\in \{0,1\cdots 7\}}\sum_{ \theta_1\in \{0,1\cdots 7\}}|\ket{\varphi^\prime_{b,\theta_0,\theta_1,\theta_1-\theta_0}}-\underbrace{\ket{\theta_0}\ket{\theta_1}}_{\bTheta^{(0)}}\otimes(e^{\theta_b\mi\pi/4}\ket{\tilde\psi^\prime_{b,+,\theta_1-\theta_0}}+e^{-\theta_b\mi\pi/4}\ket{\tilde\psi^\prime_{b,-,\theta_1-\theta_0}})|^2\leq 1200\epsilon^{1/8}+\fneg(\kappa)\end{equation}
Recall \eqref{eq:170rw} translates to 
$$\sum_{ \theta_0\in \{0,1\cdots 7\}}\sum_{ \theta_1\in \{0,1\cdots 7\}}|\ket{\varphi^\prime_{0,\theta_0,\theta_1,\theta_1-\theta_0}}-\ket{\varphi^\prime_{1,\theta_0,\theta_1,\theta_1-\theta_0}}|^2\leq 6\sqrt{2\epsilon}+\fneg(\kappa)$$
which together with \eqref{eq:205yb} implies
\begin{equation*}\sum_{ \theta_0\in \{0,1\cdots 7\}}\sum_{ \theta_1\in \{0,1\cdots 7\}}|\ket{\varphi^\prime_{0,\theta_0,\theta_1,\theta_1-\theta_0}}+\ket{\varphi^\prime_{1,\theta_0,\theta_1,\theta_1-\theta_0}}
\end{equation*}
\begin{equation}\label{eq:192e}
-\ket{\theta_0}\ket{\theta_1}\otimes(e^{\theta_0\mi\pi/4}\ket{\tilde\psi^\prime_{0,+,\theta_1-\theta_0}}+e^{\theta_1\mi\pi/4}\ket{\tilde\psi^\prime_{1,+,\theta_1-\theta_0}}+e^{-\theta_0\mi\pi/4}\ket{\tilde\psi^\prime_{0,-,\theta_1-\theta_0}}+e^{-\theta_1\mi\pi/4}\ket{\tilde\psi^\prime_{1,-,\theta_1-\theta_0}})|^2\leq 2450\epsilon^{1/8}+\fneg(\kappa)\end{equation}
Recall
\begin{align}&\ket{\varphi^\prime_{0,\theta_0,\theta_1,\theta_1-\theta_0}}+\ket{\varphi^\prime_{1,\theta_0,\theta_1,\theta_1-\theta_0}}\\
=&\Pi_{\fpass}\fHadamardTest^{\fAdv}(\bK,\bTheta,\delta=0)(\ket{\varphi_{0,\theta_0,\theta_1}}+\ket{\varphi_{1,\theta_0,\theta_1}})\\
\approx_{\sqrt{3\epsilon}}& \fHadamardTest^{\fAdv}(\bK,\bTheta,\delta=0)(\ket{\varphi_{0,\theta_0,\theta_1}}+\ket{\varphi_{1,\theta_0,\theta_1}})\\
=&\fResponse\circ\fAdv((\ket{\varphi_{0,\theta_0,\theta_1}}+\ket{\varphi_{1,\theta_0,\theta_1}})\odot (\theta_1-\theta_0))
\end{align}
And we have
\begin{align}
&e^{\theta_0\mi\pi/4}\ket{\tilde\psi^\prime_{0,+,\theta_1-\theta_0}}+e^{\theta_1\mi\pi/4}\ket{\tilde\psi^\prime_{1,+,\theta_1-\theta_0}}\\
=&e^{\theta_0\mi\pi/4}\frac{1}{8}\sum_{\theta\in \{0,1\cdots 7\}}e^{-\theta\mi\pi/4}\ket{\psi^\prime_{0,\theta,\cdot,\theta_1-\theta_0}}+e^{\theta_1\mi\pi/4}\frac{1}{8}\sum_{\theta\in \{0,1\cdots 7\}}e^{-\theta\mi\pi/4}\ket{\psi^\prime_{1,\cdot,\theta,\theta_1-\theta_0}}\\
=&\frac{1}{8}\sum_{\theta\in \{0,1\cdots 7\}}e^{-\theta\mi\pi/4}\ket{\psi^\prime_{0,\theta+\theta_0,\theta+\theta_1,\theta_1-\theta_0}}+\frac{1}{8}\sum_{\theta\in \{0,1\cdots 7\}}e^{-\theta\mi\pi/4}\ket{\psi^\prime_{1,\theta+\theta_0,\theta+\theta_1,\theta_1-\theta_0}}\\
=&\frac{1}{8}\sum_{\theta\in \{0,1\cdots 7\}}(\Pi_{\fpass}\fHadamardTest^{\fAdv}(\bK,\bTheta,\delta=0)e^{-\theta\mi\pi/4}(\ket{\varphi_{0,\theta_0+\theta,\theta_1+\theta}}+\ket{\varphi_{1,\theta_0+\theta,\theta_1+\theta}})\\
\approx_{\sqrt{3\epsilon}}&\frac{1}{8}\sum_{\theta\in \{0,1\cdots 7\}}(\fResponse\circ\fAdv\circ  e^{-\theta\mi\pi/4}(\ket{\varphi_{0,\theta_0+\theta,\theta_1+\theta}}+\ket{\varphi_{1,\theta_0+\theta,\theta_1+\theta}})\odot (\theta_1-\theta_0))\\
=&	\fResponse\circ\fAdv((e^{\theta_0\mi\pi/4}\ket{\psi_{0,+}}+e^{\theta_1\mi\pi/4}\ket{\psi_{1,+}})\odot (\theta_1-\theta_0))
\end{align}

where 
$$\ket{\psi_{0,+}}=\frac{1}{8}\sum_{\theta_0\in \{0,1\cdots 7\}}e^{-\theta_0\mi\pi/4}\ket{\psi_{0,\theta_0,\cdot}},\ket{\psi_{1,+}}=\frac{1}{8}\sum_{\theta_1\in \{0,1\cdots 7\}}e^{-\theta_1\mi\pi/4}\ket{\psi_{1,\cdot,\theta_1}}$$
where $\ket{\psi_{b,\theta_0,\theta_1}}$ is defined by the $\bTheta^{(0)}=\theta_0,\theta_1$ component of $\ket{\varphi_b}$ excluding the $\btheta_0,\btheta_1$ registers. (That is, $\ket{\varphi_{b,\theta_0,\theta_1}}=\ket{\theta_0}\ket{\theta_1}\ket{\psi_{b,\theta_0,\theta_1}}$.), and by the basis-phase correspondence property a subscript could be omitted.\par
Similarly
\begin{align}
&e^{-\theta_0\mi\pi/4}\ket{\tilde\psi^\prime_{0,-,\theta_1-\theta_0}}+e^{-\theta_1\mi\pi/4}\ket{\tilde\psi^\prime_{1,-,\theta_1-\theta_0}}\\
\approx_{\sqrt{3\epsilon}}&	\fResponse\circ\fAdv((e^{-\theta_0\mi\pi/4}\ket{\psi_{0,-}}+e^{-\theta_1\mi\pi/4}\ket{\psi_{1,-}})\odot (\theta_1-\theta_0))	
\end{align}
where
$$\ket{\psi_{0,-}}=\frac{1}{8}\sum_{\theta_0\in \{0,1\cdots 7\}}e^{\theta_0\mi\pi/4}\ket{\psi_{0,\theta_0,\cdot}},\ket{\psi_{1,-}}=\frac{1}{8}\sum_{\theta_1\in \{0,1\cdots 7\}}e^{\theta_1\mi\pi/4}\ket{\psi_{1,\cdot,\theta_1}}$$

Substitute these approximations and reverse $\fHadamardTest^\fAdv$ from \eqref{eq:192e} we get 
\begin{align}&\sum_{\theta_0,\theta_1\in \{0,1\cdots 7\}^2}|\ket{\varphi_{0,\theta_0,\theta_1}}+\ket{\varphi_{1,\theta_0,\theta_1}}-\ket{\theta_0}\ket{\theta_1}\otimes(e^{\theta_0\mi\pi/4}\ket{\psi_{0,+}}+e^{\theta_1\mi\pi/4}\ket{\psi_{1,+}}+e^{-\theta_0\mi\pi/4}\ket{\psi_{0,-}}+e^{-\theta_1\mi\pi/4}\ket{\psi_{1,-}})|^2\nonumber\\\leq& 2500\epsilon^{1/8}+\fneg(\kappa)\label{eq:127rr}\end{align}
Thus \eqref{eq:127rr} says $\Pi_{\basishonest(\bK^{(0)})}\ket{\varphi}$ is close to a state in the form of \eqref{eq:158d}. By Lemma \ref{lem:10.5} a state in the form of \eqref{eq:158d} is invariant under $\cP^\dagger\cP$. This implies
$$|(\bbI-\cP^\dagger\cP)\ket{\varphi}|\leq 50\epsilon^{1/16}+\fneg(\kappa)$$
\end{proof}

\section{Analysis of the basis uniformity test ($\fBNTest$)}\label{sec:11}
In this section we analyze the implication of passing the basis uniformity test ($\fBNTest$) with high probability. As discussed in the introduction, it implies different standard basis components of the initial state have approximately equal norms.
\subsection{Initial Setup of $\fBNTest$}
Recall in the formal protocol $\fBNTest$ is defined in two steps: $\fBNTest(\tilde\bK,\tilde\bTheta)$ is applied on the output state of $\fAddPhaseWithswitch$, and in this protocol the client reveals $\tilde\bTheta$ to the server and allows it to remove the phases; then both parties run $\fBNTest(\tilde\bK)$. In the next subsection we will analyze the property of $\fBNTest(\tilde\bK)$. (Recall that $\tilde\bK$ denotes $(\bK^{(i)})_{i\in [L]}$ and $\tilde\bTheta$ denotes $(\bTheta^{(i)})_{i\in [L]}$; see Section \ref{sec:6.1r}.)\par
The initial state of $\fBNTest(\tilde\bK,\tilde\bTheta)$ is in Setup \ref{setup:3}, while the initial state of $\fBNTest(\tilde\bK)$ does not have a corresponding setup that describes it. Below we formalize the properties of input states of $\fBNTest(\tilde\bK)$ as a setup. 
\begin{setup}\label{setup:5}
		Setup \ref{setup:5} is defined to be the set of states that satisfy:
	\begin{itemize}
	\item The parties are as described in Section \ref{sec:4.1}.
		\item The client holds a tuple of key pair registers $\tilde\bK=\bK^{(1)},\bK^{(2)}\cdots \bK^{(L)}$, where each $\bK^{(i)}=(\bx_0^{(i)},\bx_1^{(i)})$. Each key has length $\kappa$. Correspondingly the server holds registers $\bS_{bsh}^{(1)},\bS_{bsh}^{(2)}\cdots \bS_{bsh}^{(L)}$, where each register has size $\kappa$. Denote $\tilde\bS_{bsh}$ as the tuple of registers $\bS_{bsh}^{(1)},\bS_{bsh}^{(2)}\cdots \bS_{bsh}^{(L)}$.
		\item The state is efficiently preparable;
		\item The state is key checkable for any key in $\tilde \bK$;
		\item The state is strongly-claw-free for any key pair in $\tilde \bK$.
	\end{itemize}
\end{setup}
\begin{lem}
If a sub-normalized state $\ket{\varphi}$ is in Setup \ref{setup:3}, $\ket{\varphi}\odot \tilde\bTheta$ is in Setup \ref{setup:5}.	
\end{lem}
\begin{proof}
By the condition we know $\ket{\varphi}=\fAddPhaseWithswitch^\fAdv((\bK^{(\switch)},\bK),\bTheta;1^\kappa)\ket{\varphi^1}$ where $\ket{\varphi^1}$ is in Setup \ref{setup:1}. By the definition of Setup \ref{setup:1} $\ket{\varphi^1}$ is strongly-claw-free for any key pair in $(\bK^{(\switch)},\bK)$. Then the claw-freeness of $\ket{\varphi}$ follows from Lemma \ref{lem:5.2a}.
\end{proof}

 Below we give the theorem for $\fBNTest$.
\subsection{$\fBNTest$ Implies Basis Norms Are Close to Uniform Vectors}
\begin{thm}[Implication of $\fBNTest$]\label{thm:bntest}
Suppose a sub-normalized purified joint state $\ket{\varphi}$ is in Setup \ref{setup:5} and is in a basis-honest form. 
Suppose $\fAdv$ is an efficient adversary such that 
\begin{equation}\label{eq:bnpass}|\Pi_{\ffail}\fBNTest^\fAdv(\tilde\bK;1^\kappa)\ket{\varphi}|^2\leq \epsilon\end{equation}
Expand the basis-honest part of $\ket{\varphi}$:
\begin{equation}\label{eq:195t}\Pi_{\basishonest(\tilde\bK)}\ket{\varphi}=\sum_{\tilde K\in \Domain(\tilde\bK)}\underbrace{\ket{\tilde K}}_{\tilde\bK}\otimes \sum_{\vec{b}\in \{0,1\}^{L}}\underbrace{\ket{\vec{x}_{\vec{b}}}}_{\tilde\bS_{bsh}}\otimes \ket{\varphi_{\tilde K,\vec{b}}}\end{equation}
and define the norm of $\vec{\bx}_{\vec{b}}$-branch  as $c_{\vec{b}}$:
$$c_{\vec{b}}:=|\sum_{\tilde K\in \Domain(\tilde \bK)}\ket{\tilde K}\otimes \ket{\vec{x}_{\vec{b}}}\otimes \ket{\varphi_{\tilde K,\vec{b}}}|$$
Additionally define
$$c:=|\ket{\varphi}|$$
 Then we have
$$\sum_{\vec{b}\in \{0,1\}^{L}}|c_{\vec{b}}-\frac{1}{\sqrt{2^{L}}}c|^2\leq 800\sqrt{\epsilon}+\fneg(\kappa)$$
\end{thm}
Note that when we apply this theorem the initial state $\ket{\varphi}$ might be far from normalized. Due to this reason, we use the failing probability in \eqref{eq:bnpass} instead of the passing probability, which makes the expression simpler and easier to use later.
\begin{proof}
The $\fBNTest$ will first combine the state into a two-key superposition state on an index set $I=(i_1i_2\cdots i_{|I|})\subseteq [L] $ randomly selected by the client. In this process (see Protocol \ref{prtl:combine2}) the server will send back $|I|-1$ output strings denoted by symbol $r$ in the protocol, and these output strings, in the passing space, will help the client determine $K^{(combined)}$ as follows:
\begin{equation}\label{eq:84bn}K^{(combined)}=(x^{(combined)}_0,x^{(combined)}_1)=(x_{0}^{(i_1)}x_{b^{(i_2)}}^{(i_2)}\cdots x_{b^{(i_{|I|})}}^{(i_{|I|})},x_{1}^{(i_1)}x_{1- b^{(i_2)}}^{(i_2)}\cdots x_{1- b^{(i_{|I|})}}^{(i_{|I|})}),\quad b^{(i_2)}b^{(i_3)}\cdots b^{(i_{|I|})}\in \{0,1\}^{|I|-1}\end{equation}
where the values of these subscripts (that is, $b$) are determined by the server's response (that is, $r$.)\par
Denote the output state of the first step (the key combination) of $\fBNTest$ as $\ket{\varphi^1}$.\par
 Then starting from $\ket{\varphi^1}$, the server is supposed to measure all the keys with superscripts in $[L]-I$ in the standard basis and send back the result. Denote $[L]-I=j_1j_2\cdots j_{L-|I|}$. The server's response, on the passing space, could be expressed as 
\begin{equation}\label{eq:85bn}x^{(j_1)}_{b^{(j_1)}}x^{(j_2)}_{b^{(j_2)}}\cdots x^{(j_{L-|I|})}_{b^{(j_{L-|I|})}},\quad b^{(j_1)}\cdots b^{(j_{L-|I|})}\in \{0,1\}^{L-|I|}\end{equation}
Denote the output state after this step as $\ket{\varphi^2}$.\par
Let's introduce notations that describe \eqref{eq:84bn}\eqref{eq:85bn} more concisely, with $I=(i_1i_2\cdots i_{|I|})\subseteq [L]$ and $\vec{b}_0\in \{0,1\}^L$, where the $i_1$-th bit of $\vec{b}_0$ is $0$.\par
 Define $X{(I)}\vec{b}_0$ as the output vector of doing a logical-not on each bit of $\vec{b}_0$ whose index is in $I$. And define
$$\vec{b}_0|_{I}= b^{(i_1)}b^{(i_2)}b^{(i_3)}\cdots b^{(i_{|I|})},\vec{x}_{\vec{b}_0|I}=x_{b^{(i_1)}}^{(i_1)}x_{b^{(i_2)}}^{(i_2)}\cdots x_{b^{(i_{|I|})}}^{(i_{|I|})}$$
$$\vec{b}_0|_{[L]-I}=b^{(j_1)}b^{(j_2)}\cdots b^{(j_{L-|I|})},\vec{x}_{\vec{b}_0|[L]-I}=x^{(j_1)}_{b^{(j_1)}}x^{(j_2)}_{b^{(j_2)}}\cdots x^{(j_{L-|I|})}_{b^{(j_{L-|I|})}}$$
Then \eqref{eq:84bn}\eqref{eq:85bn} could be expressed as follows. 
\begin{equation}\label{eq:216tt}\eqref{eq:84bn}=(\vec{x}_{\vec{b}_0|I},\vec{x}_{X(I)\vec{b}_0|I}),\eqref{eq:85bn}=\vec{x}_{\vec{b}_0|[L]-I}.\end{equation}
Then $\ket{\varphi^2}$ will go to the standard basis test with $1/2$ probability, where the server is suppose to measure the combined keys and return a response in $K^{(combined)}$. 
 By \eqref{eq:bnpass} the adversary could fail the standard basis test with probability $\leq 2\epsilon$, by Theorem \ref{thm:sbt} we know there exists an efficient server-side isometry $\tilde\fAdv_{ST}$ such that, define
 $$\ket{\tilde\varphi^2}:=\tilde\fAdv_{ST}\ket{\varphi^2}$$ 
 then  \begin{equation}\label{eq:236id}\text{$\ket{\tilde\varphi^2}$ is $1.5 \sqrt{\epsilon}$-basis-honest for $\bK^{(combined)}$.}
  \end{equation}

  Then $\Pi_{\basishonest(\bK^{(combined)})}\ket{\tilde\varphi^2}$ has the following form:
 $$\underbrace{\sum_{I\in \Domain(\bI)}\ket{I}\otimes \sum_{\tilde K\in \Domain(\tilde\bK)}\ket{\tilde K}\otimes \sum_{\vec{b}_0\in \{0,1\}^L:\text{the $i_1$-th bit is $0$, where $i_1$ is the first bit of $I$}}\otimes \underbrace{\ket{(\vec{x}_{\vec{b}_0|I},\vec{x}_{X(I)\vec{b}_0|I})}}_{\bK^{(combined)}}}_{\text{client}}$$
 $$\otimes\underbrace{\ket{\vec{x}_{\vec{b}_0|[L]-I}}}_{\text{server's response \eqref{eq:85bn}}} \otimes\sum_{\vec{b}\in \{\vec{b}_0,X(I)\vec{b}_0\}}\underbrace{\ket{\vec{x}_{\vec{b}|I}}}_{\tilde\bS_{bsh}}\otimes\ket{\varphi_{\tilde K,I,\vec{b}}}$$
%
where we use the notations in \eqref{eq:216tt}.\par
Define $c^\prime_{I,\vec{b}}$ as the norm of the $\vec{\bx}_{\vec{b}}$ branch of $\Pi_{\basishonest(\bK^{(combined)})}\ket{\tilde\varphi^2}$ when $\bI$ has value $I$,
 which is, \begin{equation}\label{eq:302}|\ket{I}\otimes \sum_{\tilde K\in \Domain(\tilde\bK)}\ket{\tilde K}\otimes\underbrace{\ket{(\vec{x}_{\vec{b}_0|I},\vec{x}_{X(I)\vec{b}_0|I})}}_{\bK^{(combined)}}\otimes \ket{\vec{x}_{\vec{b}|[L]-I}}\otimes\ket{\vec{x}_{\vec{b}|I}}\otimes\ket{\varphi_{\tilde K,I,\vec{b}}}|.\end{equation}
 where in \eqref{eq:302} $\vec{b}_0$ is defined to be $\vec{b}$ if the $i_1$-th bit of $\vec{b}$ is $0$, and $X(I)\vec{b}$ otherwise.\par
We will first prove $(c_{\vec{b}})_{\vec{b}\in \{0,1\}^L}$ and $(c^\prime_{I,\vec{b}})_{\vec{b}\in \{0,1\}^L}$ are close to each other. These two vectors are not directly comparable since we need to define $c_{I,\vec{b}}$ for different values of register $\bI$. That is, starting from \eqref{eq:195t}, after the client samples random values for the $\bI$ register, the state on the basis-honest space becomes:
$$\sum_{I\in \Domain(\bI)}\ket{I}\otimes\sum_{\tilde K\in \Domain(\tilde\bK)}\ket{\tilde K}\otimes \sum_{\vec{b}\in \{0,1\}^{L}}\underbrace{\ket{\vec{x}_{\vec{b}}}}_{\tilde\bS_{bsh}}\otimes \frac{1}{\sqrt{|\Domain(\bI)|}}\ket{\varphi_{\tilde K,\vec{b}}}$$
Then we can define $c_{I,\vec{b}}$ as the norm of the $\vec{\bx}_{\vec{b}}$ branch when the value of $\bI$ register is $I$:
\begin{equation}\label{eq:127}c_{I,\vec{b}}=|\ket{I}\otimes\sum_{\tilde K\in \Domain(\tilde \bK)}\ket{\tilde K}\otimes \ket{\vec{x}_{\vec{b}}}\otimes \frac{1}{\sqrt{|\Domain(\bI)|}}\ket{\varphi_{\tilde K,\vec{b}}}|\quad (\text{which is $=\frac{1}{\sqrt{2^L}}c_{\vec{b}}$})\end{equation}
 By making use of Lemma \ref{lem:multisbt} we can prove 
 \begin{equation}\label{eq:bn1}\sum_{I\in \Domain(\bI),\vec{b}\in \{0,1\}^{L}}|c_{I,\vec{b}}-c_{I,\vec{b}}^\prime|^2\leq 3\sqrt{\epsilon}+\fneg(\kappa)\end{equation}
  The proof details are put in the box below for continuity of proof stream.
 \begin{mdframed}

 Proof of \eqref{eq:bn1}:\par
 Consider the mapping between $\ket{\varphi}$ and $\ket{\tilde\varphi^{2}}$.\par
  The client-side message in the first round of $\fBNTest$ could be denoted as $\llbracket\fCombine(\tilde\bK^{(i)},\bI;1^\kappa)\rrbracket$. Let $D=\tilde\fAdv_{ST}\circ\fResponse\circ\fAdv_1$ where $\fAdv_1$ is the adversary's operation in the first round, and $\fResponse$ is the adversary's measure-and-response operation in the first and second round. By the definition of $\ket{\tilde\varphi^2}$:
  $$\ket{\tilde\varphi^2}=\fCalc(\bK^{(combined)})\circ D(\ket{\varphi}\odot \llbracket\fCombine(\tilde\bK,\bI;1^\kappa)\rrbracket)$$
  where we use $\fCalc(\bK^{(combined)})$ to denote the client-side operation that calculates $\bK^{(combined)}$ from the values of server's response $r$.\par
  By the basis-honest property of $\ket{\varphi}$ and $\ket{\tilde\varphi^2}$ we have 
 $$\Pi_{\basishonest(\bK^{(combined)})}\fCalc^{-1}\ket{\tilde\varphi^2}\approx_{\sqrt{2\epsilon}}\Pi_{\basishonest(\bK^{(combined)})}D\Pi_{\basishonest(\tilde\bK)}(\ket{\varphi}\odot \llbracket\fCombine(\tilde\bK^{(i)},\bI;1^\kappa)\rrbracket)$$
  $$\Pi_{\basishonest(\tilde\bK)}\ket{\varphi}\odot \llbracket\fCombine(\tilde\bK,\bI;1^\kappa)\rrbracket\approx_{\sqrt{2\epsilon}}\Pi_{\basishonest(\tilde\bK)}D^{-1}\Pi_{\basishonest(\bK^{(combined)})}\fCalc^{-1}\ket{\tilde\varphi^2}$$
  Apply Lemma \ref{lem:multisbt} to the right hand sides of both equations,  we get
  \begin{equation}\label{eq:136}\Pi_{\basishonest(\bK^{(combined)})}\ket{\tilde\varphi^2}\approx_{\sqrt{2\epsilon}+\fneg(\kappa)}\sum_{\vec{b}\in\{0,1\}^L}\Pi_{\vec{\bx}_{\vec{b}}}\Pi_{\basishonest(\bK^{(combined)})}D\Pi_{\vec{\bx}_{\vec{b}}}(\ket{\varphi}\odot \llbracket\fCombine(\tilde\bK,\bI;1^\kappa)\rrbracket)\end{equation}
  And
  \begin{equation}\label{eq:137}\Pi_{\basishonest(\tilde\bK)}\ket{\varphi}\odot \llbracket\fCombine(\tilde\bK,\bI;1^\kappa)\rrbracket\approx_{\sqrt{2\epsilon}+\fneg(\kappa)}\sum_{\vec{b}\in\{0,1\}^L}\Pi_{\vec{\bx}_{\vec{b}}}D^{-1}\Pi_{\vec{\bx}_{\vec{b}}}\Pi_{\basishonest(\bK^{(combined)})}\ket{\tilde\varphi^2}\end{equation}
 Compare the amplitudes on both sides.\begin{itemize}\item In the left hand side of \eqref{eq:136} the amplitude of $\vec{\bx}_{\vec{b}}$ branch when the $\bI$ register has value $I$ is $c_{I,\vec{b}}$, while on the right hand side the amplitude of $\vec{\bx}_{\vec{b}}$ branch when the $\bI$ register has value $I$ is no more than $c^\prime_{I,\vec{b}}$. \item Similarly for \eqref{eq:137} in the left hand side the amplitude of $\vec{\bx}_{\vec{b}}$ branch when the $\bI$ register has value $I$ is $c_{I,\vec{b}}$, while on the right hand side the amplitude of $\vec{\bx}_{\vec{b}}$ branch when the $\bI$ register has value $I$ is no more than $c^\prime_{I,\vec{b}}$.
	
\end{itemize}
Applying Fact \ref{fact:bntesttool1} completes the proof.
 \end{mdframed}

 And from $\ket{\varphi^2}$ with $1/2$ probability the state comes into a Hadamard test. By \eqref{eq:bnpass} we know \begin{equation}\label{eq:67nn}|\Pi_{\ffail}\fHadamardTest^{\fAdv_{HT}\tilde\fAdv_{ST}^\dagger}(\bK^{(combined)};1^\kappa)\ket{\tilde\varphi^2}|^2\leq 2\epsilon\end{equation} 
 where we use $\fAdv_{HT}$ to denote the part of $\fAdv$'s operation in the Hadamard test branch of the third step of $\fBNTest$. Then there is (whose proof is put into a box below for the continuity of the proof stream)
 \begin{equation}\label{eq:bn2}
 \sum_{I\in \Domain(\bI),\vec{b}\in \{0,1\}^{L}}|c^\prime_{I,\vec{b}}-c^\prime_{I,X{(I)}\vec{b}}|^2\leq 170\sqrt{\epsilon}+\fneg(\kappa)
 \end{equation}

 \begin{mdframed}
 Proof of \eqref{eq:bn2}:\par
 Define $\ket{\tilde\varphi^{2}_0}$, 	$\ket{\tilde\varphi^{2}_1}$ as correspondingly the $\bx^{(combined)}_0$, $\bx^{(combined)}_1$ branches of $\Pi_{\basishonest(\bK^{(combined)})} \ket{\tilde\varphi^{2}}$. That is,
$$\ket{\tilde\varphi_{0}^{2}}=\sum_{I\in \Domain(\bI)}\ket{I}\otimes \sum_{\tilde K\in \Domain(\tilde \bK)}\ket{\tilde K}\otimes \sum_{\vec{b}_0\in \{0,1\}^L:\text{the $i_1$-th bit is $0$, where $i_1$ is the first bit of $I$}}\otimes \underbrace{\ket{(\vec{x}_{\vec{b}_0|I},\vec{x}_{X(I)\vec{b}_0|I})}}_{\bK^{(combined)}}$$
$$\otimes\underbrace{\ket{\vec{x}_{\vec{b}_0|[L]-I}}}_{\text{server's response \eqref{eq:85bn}}}\otimes\ket{\vec{x}_{\vec{b}_0|I}}\otimes \ket{\varphi_{I,\tilde K,\vec{b}_0,0}}$$
%
$$\ket{\tilde\varphi_{1}^{2}}=\sum_{I\in \Domain(\bI)}\ket{I}\otimes \sum_{\tilde K\in \Domain(\tilde \bK)}\ket{\tilde K}\otimes \sum_{\vec{b}_0\in \{0,1\}^L:\text{the $i_1$-th bit is $0$, where $i_1$ is the first bit of $I$}}\otimes \underbrace{\ket{(\vec{x}_{\vec{b}_0|I},\vec{x}_{X(I)\vec{b}_0|I})}}_{\bK^{(combined)}}$$
$$\otimes\underbrace{\ket{\vec{x}_{\vec{b}_0|[L]-I}}}_{\text{server's response \eqref{eq:85bn}}}\otimes\ket{\vec{x}_{X(I)\vec{b}_0|I}}\otimes \ket{\varphi_{I,\tilde K,\vec{b}_0,1}}$$
 We will show \eqref{eq:67nn} implies an efficient server-side operator $O$ that \begin{equation}\label{eq:76}O(\ket{\tilde\varphi_0^2}\odot \llbracket \fHadamardTest\rrbracket\odot \bK^{(combined)})\approx_{9\epsilon^{1/4}+\fneg(\kappa)} \ket{\tilde\varphi_1^2}\odot \llbracket \fHadamardTest\rrbracket\odot \bK^{(combined)}\end{equation} 
 (where $\llbracket \fHadamardTest\rrbracket$ contains the random padding stored in register $\bpad$). 
 $O$ is constructed as follows:
  \begin{enumerate}
  \item Apply $\fAdv_{HT}\tilde\fAdv_{ST}^\dagger$ (but do not do $\fResponse$). Recall in the real execution of the protocol $\fHadamardTest$ after the server's local operation a $\fResponse$ operator copies the value of some server-side register to the transcript register $\bd$. Denote this server-side register that stores the (unsent) response as $\tilde\bd$.
  \item Do a control phase operation on the $\tilde\bd$ register that adds a $(-1)$ phase on the basis value $d\in \Domain(\tilde\bd)$ such that $$d\cdot (x_0^{(combined)}||H(pad||x_0^{(combined)}))+d\cdot (x_1^{(combined)}||H(pad||x_1^{(combined)}))=1$$ (that is, the space $\tspan(\Pi^{\tilde\bd}_{\eqref{eq:htt}=1})$).
  \item Reverse $\fAdv_{HT}\tilde\fAdv_{ST}^\dagger$.
  \end{enumerate}
Let's show \eqref{eq:76}. This is because, under the condition of \eqref{eq:67nn}, applying Corollary \ref{cor:prepht} we get
  \begin{align}&\Pi^{\tilde\bd}_{\eqref{eq:htt}=0}\Pi_{\neq 0}^{\text{last $\kappa$ bits of }\tilde\bd}\fHadamardTest^{\fAdv_{HT}\fAdv_{ST}^\dagger}(\bK^{(combined)};1^\kappa)\ket{\tilde\varphi_0^2}\\\approx_{3\epsilon^{1/4}+\fneg(\kappa)}& \Pi^{\tilde\bd}_{\eqref{eq:htt}=0}\Pi_{\neq 0}^{\text{last $\kappa$ bits of }\tilde\bd}\fHadamardTest^{\fAdv_{HT}\fAdv_{ST}^\dagger}(\bK^{(combined)};1^\kappa)\ket{\tilde\varphi_1^2}\end{align}
  	\begin{equation}\forall b\in \{0,1\},\Pi_{=0}^{\text{last $\kappa$ bits of }\tilde\bd}\fHadamardTest^\fAdv(\bK^{(combined)};1^\kappa)\ket{\tilde\varphi_b}\approx_{2\epsilon^{1/4}+\fneg(\kappa)}0\end{equation}
  \begin{align}&\Pi^{\tilde\bd}_{\eqref{eq:htt}=1}\Pi_{\neq 0}^{\text{last $\kappa$ bits of }\tilde\bd}\fHadamardTest^{\fAdv_{HT}\fAdv_{ST}^\dagger}(\bK^{(combined)};1^\kappa)\ket{\tilde\varphi_0^2}\\\approx_{2\epsilon^{1/4}+\fneg(\kappa)}& -\Pi^{\tilde\bd}_{\eqref{eq:htt}=1}\Pi_{\neq 0}^{\text{last $\kappa$ bits of }\tilde\bd}\fHadamardTest^{\fAdv_{HT}\fAdv_{ST}^\dagger}(\bK^{(combined)};1^\kappa)\ket{\tilde\varphi_1^2}\end{align}
  Combining these relations with the construction of $O$ completes the proof of \eqref{eq:76}.\par
Now return to the proof of \eqref{eq:bn2}. We  compare the coefficients on both sides of \eqref{eq:76}. On the right hand side of \eqref{eq:76} the norm of components where the $\bI$ register has value $I$, $\bK^{(combined)}$ register has value $(\vec{x}_{\vec{b}_0|I},\vec{x}_{X(I)\vec{b}_0|I})$, and the server-response register has value $\vec{x}_{\vec{b}_0|[L]-I}$, is $c^\prime_{I,X(I)\vec{b}_0}$. On the left hand side this norm is $c^\prime_{I,\vec{b}_0}$ (note that $O$ only has read-only access to the registers above thus it does not change the norm). Thus \eqref{eq:76} implies
 \begin{equation}
 \sqrt{\sum_{I\in \Domain(\bI),\vec{b}_0\in \{0,1\}^{L}:\text{the $i_1$-th bit of $\vec{b}_0$ is $0$, where $i_1$ is the first bit of $I$}}|c^\prime_{I,\vec{b}_0}-c^\prime_{I,X{(I)}\vec{b}_0}|^2}\leq 9\epsilon^{1/4}+\fneg(\kappa)
 \end{equation}
 which by a change of variable implies
  \begin{equation}
 \sqrt{\sum_{I\in \Domain(\bI),\vec{b}_1\in \{0,1\}^{L}:\text{the $i_1$-th bit of $\vec{b}_1$ is $1$, where $i_1$ is the first bit of $I$}}|c^\prime_{I,X{(I)}\vec{b}_1}-c^\prime_{I,\vec{b}_1}|^2}\leq 9\epsilon^{1/4}+\fneg(\kappa)
 \end{equation}
 which together implies \eqref{eq:bn2}.
 \end{mdframed}

 Combining \eqref{eq:bn1}\eqref{eq:bn2} we get
 \begin{equation}\label{eq:bn3} \sum_{I\in \Domain(\bI),\vec{b}\in \{0,1\}^{L}}|c_{I,\vec{b}}-c_{I,X{(I)}\vec{b}}|^2\leq 180\sqrt{\epsilon}+\fneg(\kappa)\end{equation}
 then substitute \eqref{eq:127} we know
 \eqref{eq:bn3} implies
 $$\frac{1}{2^L}\sum_{\vec{b}_0\in \{0,1\}^{L}}\sum_{\vec{b}_1\in \{0,1\}^{L}}|c_{\vec{b}_0}-c_{\vec{b}_1}|^2\leq 180\sqrt{\epsilon}+\fneg(\kappa)$$
 which by Fact \ref{fact:bntesttool} implies
 $$\sum_{\vec{b}\in \{0,1\}^{L}} |c_{\vec{b}}-\frac{1}{\sqrt{2^{L}}}c|^2\leq 800\sqrt{\epsilon}+\fneg(\kappa)$$
\end{proof}

\section{Putting All Together}\label{sec:12}
In this section we put the analysis of each subprotocol together. We prove the optimality of $\OPT$ in $\fpreRSPV$ in Section \ref{sec:12.1} and prove the verifiability property of $\fpreRSPV$ in Section \ref{sec:12.2}. (Recall these statements are formalized in Section \ref{sec:5.4r}). Recall that \begin{equation}\label{eq:143}\OPT=\frac{1}{3}\cos^2(\pi/8),p_{\fquiz}=\frac{1}{10},p_{\fcomp}=\frac{1}{10}\end{equation} 
\subsection{Proof of the Optimality of $\OPT$}\label{sec:12.1}
To prove the optimality of $\OPT$, we need to analyze each step of Protocol \ref{prtl:prerspv}, and make use of the property of $\fInPhTest$. 
\begin{proof}[Proof of Theorem \ref{thm:opt}]
	Suppose the initial purified joint state is $\ket{\varphi^0}=O\ket{0}$ as described in Definition \ref{defn:opt}. Suppose an efficient adversary $\fAdv$ satisfies \begin{equation}\label{eq:100}|\Pi_{\fpass}\fpreRSPV^{\fAdv}(1^L,1^\kappa)\ket{\varphi^0}|> 1-10^{-2000}\end{equation}
	\begin{enumerate}
	\item (Analysis of key-pair-superposition generation) First by Theorem \ref{thm:gp} we could assume the state after the first step of the protocol is within Setup \ref{setup:1}. Denote the purified joint state as $\ket{\varphi^1}$.\par

	\item (Analysis of the first $\fStdBTest$) Since the standard basis test is executed on $\ket{\varphi^1}$ with probability $1/2$, by \eqref{eq:100} the server has to be able to make the client output $\fpass$ in this test with probability $\geq 1-2\times 10^{-2000}$. That allows us to apply Lemma \ref{thm:sbt} and get
	\begin{equation}\label{eq:81}\exists \text{ efficient server-side operator } \tilde\fAdv_{2.sbt}:\tilde\fAdv_{2.sbt}\ket{\varphi^{1}}\text{ is $1.5\times 10^{-1000}$-basis-honest for $(\bK^{(\switch)}, \bK)$}\end{equation}
	 Denote \begin{equation}\ket{\tilde\varphi^{1}}:=\tilde\fAdv_{2.sbt}\ket{\varphi^{1}}\end{equation}
	 	\item (Analysis of the switch gadget technique) Denote the output state of $\fAddPhaseWithswitch$ as $\ket{\varphi^{2.a}}$:
	 	\begin{align}\ket{\varphi^{2.a}}:=&\fAddPhaseWithswitch^{\fAdv_{2.a}}((\bK^{(\switch)},\bK),\bTheta;1^\kappa)\ket{\varphi^1}\\=&\fAddPhaseWithswitch^{\fAdv_{2.a}\tilde\fAdv_{2.sbt}^{-1}}((\bK^{(\switch)},\bK),\bTheta;1^\kappa)\ket{\tilde\varphi^1}\label{eq:139}\end{align}
	 	where $\fAdv_{2.a}$ is the $\fAdv$'s operation in step 2.a of the verifiable state preparation case in Protocol \ref{prtl:prerspv}.\par
	 	Since the verifiable state preparation case is reached with probability $\frac{1}{2}$, together with \eqref{eq:100} we know 
	 	$$|\Pi_{\fpass}\ket{\varphi^{2.a}}|^2\geq 1-2\times 10^{-2000}$$
	 	 That allows us to apply Theorem \ref{thm:7.3}. As Setup \ref{setup:3}, define $H^\prime$ as the blinded version of $H$ where entries $\{0,1\}^\kappa||K^{(\switch)}||\cdots$ are blinded. Define $\fAdv^{blind}_{\geq 2.b}$ as the blinded version of $\fAdv_{\geq 2.b}$ (the part of $\fAdv$ starting from the step $b$ of the verifiable state preparation) where each random oracle query is replaced by a query to $H^\prime$. Then we have
	 	 \begin{equation}\label{eq:202}\fpreRSPV_{\geq 2.b}^{\fAdv_{\geq 2.b}}\ket{\varphi^{2.a}}\approx_{15\times 10^{-1000}+\fneg(\kappa)}\fpreRSPV_{\geq 2.b}^{\fAdv_{\geq 2.b}^{blind}}\ket{\varphi^{2.a}}\end{equation}
	 	where $\fpreRSPV_{\geq 2.b}$ denotes the execution of Protocol \ref{prtl:prerspv} starting from the 2.b step.\par
	 	We will analyze the right hand side of \eqref{eq:202} to understand the left hand side. 
	 	\item (Analysis of the second $\fStdBTest$) First by \eqref{eq:202}\eqref{eq:100} we know
	 	\begin{equation}\label{eq:239ix}|\Pi_{\fpass}\fpreRSPV_{\geq 2.b}^{\fAdv_{\geq 2.b}^{blind}}\ket{\varphi^{2.a}}|^2\geq 1-16\times 10^{-1000}-\fneg(\kappa)\end{equation}
	 	Recall $\cF_{blind}$ is defined as the set of adversaries that only query the blinded oracle $H^\prime$. First in $\fpreRSPV_{\geq 2.b}$ a standard basis test is executed with probability $\frac{1}{5}$. This together with \eqref{eq:239ix} by Theorem \ref{thm:sbt} implies
	 	\begin{equation}\exists \text{ efficient server-side operator } \tilde\fAdv_{2.b.sbt}^{blind}\in \cF_{blind}:\tilde\fAdv_{2.b.sbt}^{blind}\ket{\varphi^{2.a}}\text{ is $(9\times 10^{-500}+\fneg(\kappa))$-basis-honest for $\bK$}\end{equation}
	 	\item (Analysis of $\fAddPhaseWithswitch$) Define 
	 	\begin{equation}\label{eq:327}\ket{\tilde\varphi^{2.a}}:=\tilde\fAdv_{2.b.sbt}^{blind}\ket{\varphi^{2.a}}\end{equation}
	 	Now $\ket{\tilde\varphi^{2.a}}$ is in Setup \ref{setup:3}. 
	 	  Applying Theorem \ref{thm:indofr1} and Fact \ref{fact:8} we have
	 	\begin{equation}\label{eq:142}\cR_1(\ket{\$_1}\otimes\fReviseRO\ket{\tilde\varphi^{2.a}})\approx_{\fneg(\kappa)}\ket{\$_1}\otimes \fReviseRO\ket{\tilde\varphi^{2.a}}\approx^{ind:\cF_{cq\land blind}}\ket{\tilde\varphi^{2.a}}\end{equation}
\item (Analysis of $\fInPhTest$)	Now the left hand side of \eqref{eq:142} is in Setup \ref{setup:4}. 
	From \eqref{eq:239ix}\eqref{eq:327}\eqref{eq:142} and the fact that $\fInPhTest$ is executed with probability $\frac{1}{5}$ we know
	\begin{equation}|\Pi_{\fpass}\fInPhTest^{\fAdv^{blind}_{\fInPhTest}\circ (\tilde\fAdv_{ 2.b.sbt}^{blind})^{-1}}(\bK, \bTheta;1^\kappa)(\cR_1(\ket{\$_1}\otimes\fReviseRO\ket{\tilde\varphi^{2.a}}))|^2\geq 1-80\times 10^{-1000}-\fneg(\kappa)\end{equation}
	where $\fAdv_{\fInPhTest}^{blind}$ is the part of $\fAdv^{blind}_{\geq 2.b}$ in the $\fInPhTest$ step. Apply Theorem \ref{thm:9.2} and we get
	\begin{equation}\label{eq:211j}|\Pi_{\fwin}\fInPhTest^{\fAdv^{blind}_{\fInPhTest}\circ (\tilde\fAdv_{ 2.b.sbt}^{blind})^{-1}}(\bK, \bTheta;1^\kappa)(\cR_1(\ket{\$_1}\otimes\fReviseRO\ket{\tilde\varphi^{2.a}}))|^2\leq (\OPT+10^{-240}+\fneg(\kappa))\end{equation}
	
	\end{enumerate}
	Now we start from \eqref{eq:211j} and substituting these analysis back step-by-step. 
	 Substituting \eqref{eq:142},  we know
	 \begin{equation}|\Pi_{\fwin}\fInPhTest^{\fAdv^{blind}_{\fInPhTest}}(\bK, \bTheta;1^\kappa)\ket{\varphi^{2.a}})|^2\leq (\OPT+10^{-240}+\fneg(\kappa))\end{equation}
	 Since $\fInPhTest$ is executed with $\frac{1}{5}$ probability in $\fpreRSPV_{\geq 2.b}$, this implies
	 \begin{equation}|\Pi_{\fwin}\fpreRSPV_{\geq 2.b}^{\fAdv^{blind}_{\geq 2.b}}(\bK, \bTheta;1^\kappa)\ket{\varphi^{2.a}})|^2\leq (\OPT+10^{-240}+\fneg(\kappa))\cdot\frac{1}{5}\end{equation}
	 Substituting \eqref{eq:202}:
	 \begin{equation}|\Pi_{\fwin}\fpreRSPV_{\geq 2.b}^{\fAdv_{\geq 2.b}}(\bK, \bTheta;1^\kappa)\ket{\varphi^{2.a}})|^2\leq (\OPT+1.01\times 10^{-240}+\fneg(\kappa))\cdot\frac{1}{5}\end{equation}	 
	 Since $\fpreRSPV_{\geq 2.b}$ is executed with $\frac{1}{2}$ probability in $\fpreRSPV$, we get
	 \begin{equation}|\Pi_{\fwin}\fpreRSPV^{\fAdv}(1^L,1^\kappa)\ket{\varphi^0}|\leq ( \OPT+1.01\times 10^{-240}+\fneg(\kappa))\cdot\frac{1}{10}\end{equation}
	  which completes the proof.
\end{proof}

\subsection{A Proof of Theorem \ref{thm:prerspvv}}\label{sec:12.2}
In this section we prove Theorem \ref{thm:prerspvv}, the verifiability property of $\fpreRSPV$. The proof is divided into three parts: (1) First we analyze each step of the protocol; (2) then we formally construct the isometry; (3) finally we prove the isometry satisfies Definition \ref{defn:rspvv}.
%
%
\begin{proof}[Proof of Theorem \ref{thm:prerspvv}]
Suppose the initial purified joint state is $\ket{\varphi^0}=O\ket{0}$ as described in Definition \ref{defn:opt}. Consider an efficient adversary $\fAdv$ whose corresponding final output state is $$\ket{\varphi^\prime}:=\fpreRSPV^\fAdv(1^L,1^\kappa)\ket{\varphi^0}$$ 
Suppose the passing probability and the winning probability are big:
\begin{equation}\label{eq:78}|\Pi_{\fpass}\ket{\varphi^\prime}|^2> 1-10^{-2000}\end{equation}
\begin{equation}\label{eq:226h}|\Pi_{\fwin}\ket{\varphi^\prime}|^2> (\OPT-10^{-200})\cdot p_{\fquiz}=(\OPT-10^{-200})\cdot \frac{1}{10}\end{equation}
Our goal is to give an isometry $\fSim^{\fAdv}$ that satisfies \eqref{eq:61q}.\par
\paragraph{Part I: analysis of the protocol} To achieve this goal, we first need to analyze the protocol step-by-step to understand $\ket{\varphi^\prime}$.\par
This part of the proof is as below. The first five steps are the same as the proof of Theorem \ref{thm:opt}; we give a summary for the important conclusion in each step.
\begin{enumerate}
\item[1-5.]\begin{equation}\ket{\tilde\varphi^{1}}:=\tilde\fAdv_{2.sbt}\ket{\varphi^{1}}\text{is $1.5\times 10^{-1000}$-basis-honest for $(\bK^{(\switch)}, \bK)$.}\end{equation}
\begin{equation}
	\ket{\varphi^{2.a}}:=\fAddPhaseWithswitch^{\fAdv_{2.a}}((\bK^{(\switch)},\bK),\bTheta;1^\kappa)\ket{\varphi^1}
\end{equation}
	 	$$|\Pi_{\fpass}\ket{\varphi^{2.a}}|^2\geq 1-2\times 10^{-2000}$$
	 	 \begin{equation}\label{eq:220j}\fpreRSPV_{\geq 2.b}^{\fAdv_{\geq 2.b}}\ket{\varphi^{2.a}}\approx_{1.5\times 10^{-500}+\fneg(\kappa)}\fpreRSPV_{\geq 2.b}^{\fAdv_{\geq 2.b}^{blind}}\ket{\varphi^{2.a}}\end{equation}
	 	 \begin{equation}\label{eq:253ig}|\Pi_{\fpass}\fpreRSPV_{\geq 2.b}^{\fAdv_{\geq 2.b}^{blind}}\ket{\varphi^{2.a}}|^2\geq 1-3.1\times 10^{-500}-\fneg(\kappa)\end{equation}
	 	\begin{equation}\label{eq:254k}\exists \text{ efficient server-side operator } \tilde \fAdv_{2.b.sbt}^{blind}\in \cF_{blind}:\tilde \fAdv_{2.b.sbt}^{blind}\ket{\varphi^{2.a}}\text{ is $(9\times 10^{-500}+\fneg(\kappa))$-basis-honest for $\bK$}\end{equation}
	 	$$\ket{\tilde\varphi^{2.a}}:=\tilde \fAdv_{2.b.sbt}^{blind}\ket{\varphi^{2.a}}$$
  \begin{equation}\label{eq:161}\cR_1(\ket{\$_1}\otimes \fReviseRO\ket{\tilde\varphi^{2.a}})\approx_{\fneg(\kappa)}\ket{\$_1}\otimes \fReviseRO\ket{\tilde\varphi^{2.a}}\approx^{ind:\cF_{cq\land blind}}\ket{\tilde\varphi^{2.a}}\end{equation}
\item[6.] (Analysis of $\fCoPhTest$)  By \eqref{eq:253ig} and the fact that $\fCoPhTest$ is executed with probability $\frac{1}{5}$ in $\fpreRSPV_{\geq 2.b}$ we know 
  $$|\Pi_{\fpass}\fCoPhTest^{\fAdv_{\fCoPhTest}^{blind}}(\bK,\bTheta;1^\kappa)\ket{\varphi^{2.a}}|^2 \geq 1-16\times 10^{-500}-\fneg(\kappa)$$
  Thus by \eqref{eq:161}
  $$|\Pi_{\fpass}\fCoPhTest^{\fAdv_{\fCoPhTest}^{blind}\circ (\tilde\fAdv_{2.b.sbt}^{blind})^{-1}}(\bK,\bTheta;1^\kappa)\cR_1(\ket{\$_1}\otimes \fReviseRO\ket{\tilde\varphi^{2.a}})|^2 \geq 1-16\times 10^{-500}-\fneg(\kappa)$$
  where $\cR_1(\ket{\$_1}\otimes \fReviseRO\ket{\tilde\varphi^{2.a}})$ is in Setup \ref{setup:4}. We can further randomize $\ket{\tilde\varphi^{2.a}}$ under $\cR_2$ by Theorem \ref{thm:8.1}:
     \begin{equation}\label{eq:149}\cR_2(\ket{\$_2}\otimes (\cR_1(\ket{\$_1}\otimes \fReviseRO\ket{\tilde\varphi^{2.a}})))\approx_{10^{-123}+\fneg(\kappa)}\ket{\$_2}\otimes\cR_1(\ket{\$_1}\otimes \fReviseRO\ket{\tilde\varphi^{2.a}})\end{equation}
\item[7.] (Analysis of $\fInPhTest$) Finally we analyze the $\fInPhTest$. As the previous step, we know $\cR_1(\ket{\$_1}\otimes \fReviseRO\ket{\tilde\varphi^{2.a}})$ is in Setup \ref{setup:4}. By \eqref{eq:253ig}\eqref{eq:161} and the fact that $\fInPhTest$ is executed with probability $\frac{1}{5}$ in $\fpreRSPV_{\geq 2.b}$ we know the following about the passing probability and winning probability: 
     $$|\Pi_{\fpass}\fInPhTest^{\fAdv_{\fInPhTest}^{blind}\circ (\tilde\fAdv_{2.b.sbt}^{blind})^{-1}}(\bK,\bTheta;1^\kappa)\cR_1(\ket{\$_1}\otimes \fReviseRO\ket{\tilde\varphi^{2.a}})|^2> 1-16\times 10^{-500}-\fneg(\kappa)$$
     $$|\Pi_{\fwin}\fInPhTest^{\fAdv_{\fInPhTest}^{blind}\circ (\tilde\fAdv_{2.b.sbt}^{blind})^{-1}}(\bK,\bTheta;1^\kappa)\cR_1(\ket{\$_1}\otimes \fReviseRO\ket{\tilde\varphi^{2.a}})|^2> \OPT- 16\times 10^{-500}-\fneg(\kappa)$$
      By Theorem \ref{thm:9.8} we can further randomize the state under $\cP$:
      \begin{equation}\label{eq:223}\cP^\dagger\cP\cR_1(\ket{\$_1}\otimes \fReviseRO\ket{\tilde\varphi^{2.a}})))\approx_{10^{-31}+\fneg(\kappa)} \cR_1(\ket{\$_1}\otimes \fReviseRO\ket{\tilde\varphi^{2.a}})\end{equation}

   \item[8.] (A temporary summary and preparation)  \eqref{eq:161}\eqref{eq:149}\eqref{eq:223} implies
      
      \begin{equation}\label{eq:92}\cP^\dagger\cP\cR_2(\ket{\$_2}\otimes (\cR_1(\ket{\$_1}\otimes \fReviseRO\ket{\tilde\varphi^{2.a}})))\approx_{1.1\times 10^{-31}+\fneg(\kappa)}  \ket{\$_2}\otimes \ket{\$_1}\otimes\fReviseRO\ket{\tilde\varphi^{2.a}}\approx^{ind:\cF_{cq\land blind}}\ket{\tilde\varphi^{2.a}}\end{equation}
     We aim at giving a server-side isomorphism between the left hand side of \eqref{eq:92} and the target state \eqref{eq:target}.\par
     Let's first analyze the state in the left hand side of \eqref{eq:92}. By Theorem \ref{thm:9.pre} $\cR_2(\ket{\$_2}\otimes (\cR_1(\ket{\$_1}\otimes \fReviseRO\circ \Pi_{\basishonest(\bK)}\ket{\tilde\varphi^{2.a}})))$ is a pre-phase-honest form. Recall the indicator qubit registers used by $\cP$ are denoted by $\bindic_+,\bindic_-$. The operation of $\cP$ extracts the information of whether the state phases are under complex conjugation to these indicator registers, and $\cP^\dagger\cP$ erases the values in the indicator registers. Below we want to make this information explicit thus we first work under $\cP$ instead of $\cP^\dagger\cP$. Besides that, we also want to treat the $\bK^{(0)}$ key pair separately.\par 
      By Lemma \ref{lem:10.6} we have, there exist states $\ket{\phi_{0,+}}$, $\ket{\phi_{0,-}}$, $\ket{\phi_{1,+}}$, $\ket{\phi_{1,-}}$ such that:
     \begin{align}\forall b^{(0)}\in \{0,1\}:\quad &\cP\cR_2(\ket{\$_2}\otimes (\cR_1(\ket{\$_1}\otimes \Pi^{\bS^{(0)}_{bsh}}_{\bx^{(0)}_{b^{(0)}}}\fReviseRO\circ \Pi_{\basishonest(\bK)}\ket{\tilde \varphi_{2.a}})))\nonumber\\=& \underbrace{\ket{\ftrue}}_{\bindic_+}\underbrace{\ket{\ffalse}}_{\bindic_-}\ket{\phi_{b^{(0)},+}}+\ket{\ffalse}\ket{\ftrue}\ket{\phi_{b^{(0)},-}}\label{eq:93}\end{align}
     and for each $b^{(0)}\in \{0,1\}$, $\ket{\phi_{b^{(0)},+}}$, $\ket{\phi_{b^{(0)},-}}$ have the following form, for some state family\\ $(\ket{\phi_{K,b^{(0)},\vec{b},+}},\ket{\phi_{K,b^{(0)},\vec{b},-}})_{K\in \Domain(\bK),b^{(0)}\in \{0,1\},\vec{b}\in \{0,1\}^L}$:
     \begin{equation}\label{eq:167r}\ket{\phi_{b^{(0)},+}}:=\underbrace{\sum_{K\in \Domain(\bK)}\ket{K}\otimes \sum_{\Theta\in \Domain(\bTheta)}\ket{\Theta}}_{\text{client}}\otimes\underbrace{\ket{x_{b^{(0)}}^{(0)}}\otimes \sum_{\vec{b}\in \{0,1\}^L}e^{\tSUM(\vec{\Theta}_{\vec{b}})\mi\pi/4}\ket{\vec{x}_{\vec{b}}}}_{\bS_{bsh}}\ket{\phi_{K,b^{(0)},\vec{b},+}}\end{equation}
        \begin{equation}\label{eq:170r}\ket{\phi_{b^{(0)},-}}:=\sum_{K\in \Domain(\bK)}\ket{K}\otimes \sum_{\Theta\in \Domain(\bTheta)}\ket{\Theta}\otimes\ket{x_{b^{(0)}}^{(0)}}\otimes \sum_{\vec{b}\in \{0,1\}^L}e^{-\tSUM(\vec{\Theta}_{\vec{b}})\mi\pi/4}\ket{\vec{x}_{\vec{b}}}\ket{\phi_{K,b^{(0)},\vec{b},-}}\end{equation}
        For each $b^{(0)}\in \{0,1\}$, $\vec{b}\in \{0,1\}^L$, denote the overall norms of the $\vec{\bx}_{\vec{b}}$-branch of these states as $c_{b^{(0)},\vec{b},+},c_{b^{(0)},\vec{b},-}$:
      \begin{equation}\label{eq:262uy}c_{b^{(0)},\vec{b},+}=|\sum_{K\in \Domain(\bK)}\ket{K}\otimes \sum_{\Theta\in \Domain(\bTheta)}\ket{\Theta}\otimes\ket{x_{b^{(0)}}^{(0)}}\otimes e^{\tSUM(\vec{\Theta}_{\vec{b}})\mi\pi/4}\ket{\vec{x}_{\vec{b}}}\ket{\phi_{K,b^{(0)},\vec{b},+}}|\end{equation}
      \begin{equation}\label{eq:263uy}c_{b^{(0)},\vec{b},-}=|\sum_{K\in \Domain(\bK)}\ket{K}\otimes \sum_{\Theta\in \Domain(\bTheta)}\ket{\Theta}\otimes\ket{x_{b^{(0)}}^{(0)}}\otimes e^{-\tSUM(\vec{\Theta}_{\vec{b}})\mi\pi/4}\ket{\vec{x}_{\vec{b}}}\ket{\phi_{K,b^{(0)},\vec{b},-}}|\end{equation}
      Now the remaining thing is to derive properties of these norms by analyzing $\fBNTest$. 
      \item[9.] (Analysis of $\fBNTest$) By \eqref{eq:253ig} and the fact that $\fBNTest$ is executed with probability $\frac{1}{5}$, we know
      $$|\Pi_{\fpass}\fBNTest^{\fAdv_{\fBNTest}^{blind}(\tilde\fAdv_{2.b.sbt}^{blind})^{-1}}(\tilde\bK;1^\kappa)\ket{\tilde\varphi^{2.a}}|^2>1-16\times 10^{-500}-\fneg(\kappa)$$
            $$\Rightarrow|\Pi_{\ffail}\fBNTest^{\fAdv_{\fBNTest}^{blind}(\tilde\fAdv_{2.b.sbt}^{blind})^{-1}}(\tilde\bK;1^\kappa)\Pi_{\basishonest(\bK)}\ket{\tilde\varphi^{2.a}}|^2<40\times 10^{-500}+\fneg(\kappa)$$
      Apply the collapsing property (Lemma \ref{lem:collapse}) on $\bS_{bsh}^{(0)}$, use $\Pi_{\bx^{(0)}_0}^{\bS_{bsh}^{(0)}}$, $\Pi_{\bx^{(0)}_1}^{\bS_{bsh}^{(0)}}$ to denote the projection onto the $\bx_0^{(0)}$-branch and $\bx_1^{(0)}$-branch, we have
                  \begin{equation}\label{eq:353}|\Pi_{\ffail}\fBNTest^{\fAdv_{\fBNTest}^{blind}(\tilde\fAdv_{2.b.sbt}^{blind})^{-1}}(\tilde\bK;1^\kappa)\Pi_{\bx^{(0)}_0}^{\bS_{bsh}^{(0)}}\Pi_{\basishonest(\bK)}\ket{\tilde\varphi^{2.a}}|^2<40\times 10^{-500}+\fneg(\kappa)\end{equation}
             \begin{equation}\label{eq:3532}|\Pi_{\ffail}\fBNTest^{\fAdv_{\fBNTest}^{blind}(\tilde\fAdv_{2.b.sbt}^{blind})^{-1}}(\tilde\bK;1^\kappa)\Pi_{\bx^{(0)}_1}^{\bS_{bsh}^{(0)}}\Pi_{\basishonest(\bK)}\ket{\tilde\varphi^{2.a}}|^2<40\times 10^{-500}+\fneg(\kappa)\end{equation}
             Now we can analyze each of \eqref{eq:353}\eqref{eq:3532} to derive basis norm relations for $b^{(0)}=0,1$. But to get what we want on \eqref{eq:262uy}\eqref{eq:263uy}, we need to separate the honest phase part and the complex-conjugated phase part, as in \eqref{eq:93}.\par
             Without loss of generality let's first consider \eqref{eq:353} and use it to analyze \eqref{eq:262uy}\eqref{eq:263uy} for $b^{(0)}=0$. To achieve it, we first substitute \eqref{eq:92} to \eqref{eq:353} and get
             \begin{equation*}|\Pi_{\ffail}\fBNTest^{\fAdv_{\fBNTest}^{blind}(\tilde\fAdv_{2.b.sbt}^{blind})^{-1}}(\tilde\bK;1^\kappa)\end{equation*}\begin{equation}\label{eq:167news}\Pi_{\bx^{(0)}_0}^{\bS_{bsh}^{(0)}}\cP^\dagger\cP (\cR_2(\ket{\$_2}\otimes(\cR_1(\ket{\$_1}\otimes \fReviseRO\circ\Pi_{\basishonest(\bK)}\ket{\tilde\varphi^{2.a}})))))|^2<1.2\times 10^{-31}+\fneg(\kappa)\end{equation}
             Recall the definition of $\cP$ in Definition \ref{defn:insub}, there is: $$\cP\Pi^{\bS_{bsh}^{(0)}}_{\bx_0^{(0)}}= (\Pi^{\bindic_{+}}_{\ftrue}\Pi^{\bindic_{-}}_{\ffalse}+ \Pi^{\bindic_{+}}_{\ffalse}\Pi^{\bindic_{-}}_{\ftrue})\cP_{0,-}\cP_{0,+}\Pi^{\bS_{bsh}^{(0)}}_{\bx_0^{(0)}}$$
             Note that with only the $\bx_0^{(0)}$-branch, $\cP_{0,+},\cP_{0,-}$ are equivalent to local operators that only apply on the $\btheta_0^{(0)}$, $\bindic_{+},\bindic_{-}$ registers. But $\fBNTest$ does not take these registers as inputs. And $\Pi_{\ffail}$ and $\cP_{0,+},\cP_{0,-}$ also commute, which implies \eqref{eq:167news} could be further simplified by replacing the $\cP^\dagger\cP$ in it by $\cP$:
             \begin{equation*}|\Pi_{\ffail}\fBNTest^{\fAdv_{\fBNTest}^{blind}(\tilde\fAdv_{2.b.sbt}^{blind})^{-1}}(\tilde\bK;1^\kappa)\end{equation*}\begin{equation}\label{eq:167news2}\Pi_{\bx^{(0)}_0}^{\bS_{bsh}^{(0)}}\cP (\cR_2(\ket{\$_2}\otimes(\cR_1(\ket{\$_1}\otimes \fReviseRO\circ\Pi_{\basishonest(\bK)}\ket{\tilde\varphi^{2.a}})))))|^2<1.2\times 10^{-31}+\fneg(\kappa)\end{equation}
             which could be further decomposed into
 \begin{equation*}|\Pi_{\ffail}\fBNTest^{\fAdv_{\fBNTest}^{blind}(\tilde\fAdv_{2.b.sbt}^{blind})^{-1}}(\tilde\bK;1^\kappa)\end{equation*}\begin{equation}\label{eq:167news3}\Pi^{\bindic_{+}}_{\ftrue}\Pi^{\bindic_{-}}_{\ffalse}\cP_{0,-}\cP_{0,+}\Pi^{\bS_{bsh}^{(0)}}_{\bx_0^{(0)}} (\cR_2(\ket{\$_2}\otimes(\cR_1(\ket{\$_1}\otimes \fReviseRO\circ\Pi_{\basishonest(\bK)}\ket{\tilde\varphi^{2.a}})))))|^2<1.2\times 10^{-31}+\fneg(\kappa)\end{equation}
              \begin{equation*}|\Pi_{\ffail}\fBNTest^{\fAdv_{\fBNTest}^{blind}(\tilde\fAdv_{2.b.sbt}^{blind})^{-1}}(\tilde\bK;1^\kappa)\end{equation*}\begin{equation}\label{eq:167news4}\Pi^{\bindic_{+}}_{\ffalse}\Pi^{\bindic_{-}}_{\ftrue}\cP_{0,-}\cP_{0,+}\Pi^{\bS_{bsh}^{(0)}}_{\bx_0^{(0)}} (\cR_2(\ket{\$_2}\otimes(\cR_1(\ket{\$_1}\otimes \fReviseRO\circ\Pi_{\basishonest(\bK)}\ket{\tilde\varphi^{2.a}})))))|^2<1.2\times 10^{-31}+\fneg(\kappa)\end{equation}
              where the input states are the same as the two terms in \eqref{eq:93}. Thus applying the properties of $\fBNTest$ (Theorem \ref{thm:bntest}) implies
              \begin{equation}\label{eq:3590}\sum_{\vec{b}\in \{0,1\}^L}|c_{0,\vec{b},+}-\frac{1}{\sqrt{2^L}}c_{0,+}|^2\leq 10^{-12}+\fneg(\kappa)\text{ where }c_{0,+}:=\sqrt{\sum_{\vec{b}\in \{0,1\}^L}c_{0,\vec{b},+}^2}\end{equation}
              \begin{equation}\label{eq:3591}\sum_{\vec{b}\in \{0,1\}^L}|c_{0,\vec{b},-}-\frac{1}{\sqrt{2^L}}c_{0,-}|^2\leq 10^{-12}+\fneg(\kappa)\text{ where }c_{0,-}:=\sqrt{\sum_{\vec{b}\in \{0,1\}^L}c_{0,\vec{b},-}^2}\end{equation}
              Similarly for $b^{(0)}=1$ we have
      \begin{equation}\label{eq:3592}\sum_{\vec{b}\in \{0,1\}^L}|c_{1,\vec{b},+}-\frac{1}{\sqrt{2^L}}c_{1,+}|^2\leq 10^{-12}+\fneg(\kappa)\text{ where }c_{1,+}:=\sqrt{\sum_{\vec{b}\in \{0,1\}^L}c_{1,\vec{b},+}^2}\end{equation}
              \begin{equation}\label{eq:3593}\sum_{\vec{b}\in \{0,1\}^L}|c_{1,\vec{b},-}-\frac{1}{\sqrt{2^L}}c_{1,-}|^2\leq 10^{-12}+\fneg(\kappa)\text{ where }c_{1,-}:=\sqrt{\sum_{\vec{b}\in \{0,1\}^L}c_{1,\vec{b},-}^2}\end{equation}
            \end{enumerate}
   \paragraph{Part II: construction of $\fSim$}   With these analysis in hands we can start to construct the isometry $\fSim$ that satisfies \eqref{eq:61q}. At a high level:
   \begin{enumerate}\item In step 0-3 below we first aim at giving a server-side isomorphism between the left hand side of \eqref{eq:92} and the target state \eqref{eq:324rd}.\footnote{Note that although we call \eqref{eq:324rd} the ``target state'', what we are going to do is to apply $\fSim$ on \eqref{eq:324rd} and simulate a state that resembles the left hand side of \eqref{eq:92}.} 
   \item Then the simulation of the left hand side of \eqref{eq:61q} is constructed based on the simulation of the left hand side of \eqref{eq:92}.   	
   \end{enumerate}
       
        $\fSim$ is constructed as follows.
        \begin{enumerate}
        \item[0.]  
         As the preparation, recall the target state \eqref{eq:target} could be expanded into (below we use $\bS_{subs}$ to denote the server-side registers appeared in \eqref{eq:target}):
                 \begin{equation}\label{eq:324rd}\sum_{\theta^{(1)}\theta^{(2)}\cdots \theta^{(L)}\in \{0,1\cdots 7\}^L}\frac{1}{\sqrt{8^L}}\underbrace{\ket{\theta^{(1)}\theta^{(2)}\cdots \theta^{(L)}}}_{\text{client}}\otimes \sum_{\vec{b}\in \{0,1\}^L}e^{(\sum_{i\in [L]} \theta^{(i)}b^{(i)})\mi\pi/4}\underbrace{\ket{b^{(1)}b^{(2)}\cdots b^{(L)}}}_{\bS_{subs}}\cdot \frac{1}{\sqrt{2^L}}
        \end{equation}
       where we use $b^{(i)}$ to denote the $i$-th bit of $\vec{b}$.
        \item  The isometry $\fSim$ first introduces the following registers in its workspace:
        \begin{itemize}\item  $\tilde\bindic_+, \tilde\bindic_-$ registers that hold $\ftrue/\ffalse$ values, which are used to take the place of $\bindic_+, \bindic_-$ that appeared in \eqref{eq:93} of Part I.
        \item Single-bit register $\tilde\bindic_{col}$, which is used to simulate the value of $b^{(0)}$ (thus simulate the collapsing measurement $\Pi^{\bS^{(0)}_{bsh}}_{\bx^{(0)}_{b^{(0)}}}$) in Part I.
        \item $\tilde\bindic_{garbage}$ register, which indicates whether the state is what we aim to construct (when $\tilde\bindic_{garbage}=0$) or it's some garbage state (when $\tilde\bindic_{garbage}\neq 0$). The reason is we aim to construct $\fSim$ as a unitary while the state \eqref{eq:92} is not even necessarily normalized. How could we prepare a sub-normalized state through a unitary? Here we only aim at simulating the left hand side of \eqref{eq:92} in the $\tilde\bindic_{garbage}=0$ part, and arguing later that the norm of the $\tilde\bindic_{garbage}\neq 0$ is small.
        \end{itemize}
In this step we define a server-side operation $O_1$ that operates on the server-side of \eqref{eq:324rd}, and indicator registers introduced above.  This step assigns right values for these indicator registers, and make their norms match what we want. Recall the norms for each possible value of $b^{(0)}$ and $+/-$ are given in \eqref{eq:3590}-\eqref{eq:3593}. First, the operator in this step distribute norms based on values of indicator registers and maps \eqref{eq:324rd} to:
\begin{align} \eqref{eq:324rd}\otimes &(\underbrace{ \ket{0}}_{\tilde\bindic_{garbage}}\otimes(\underbrace{\ket{\ftrue}}_{\tilde\bindic_+}\underbrace{\ket{\ffalse}}_{\tilde\bindic_-}((\underbrace{\ket{0}}_{\tilde\bindic_{col}}c_{0,+}+\ket{1}c_{1,+})\\
&\qquad\qquad\qquad+\ket{\ffalse}\ket{\ftrue}(\ket{0}c_{0,-}+\ket{1}c_{1,-}))\\&+\underbrace{\ket{1}}_{\tilde\bindic_{garbage}}\ket{\cdots})\end{align}
Now recall $\tilde\bindic_+,\tilde\bindic_-$ indicates whether the phases are complex-conjugated, and as discussed in Section \ref{sec:2.1}, the complex-conjugated case is isometric to the honest case. To simulate this part, do a control-X on all the bits of $\bS_{subs}$ conditioned on $\tilde\bindic_-=\ftrue$. Then the final output state of this step is
	\begin{align}&\quad \sum_{\theta^{(1)}\theta^{(2)}\cdots \theta^{(L)}\in \{0,1\cdots 7\}^L}\frac{1}{\sqrt{8^L}}\underbrace{\ket{\theta^{(1)}\theta^{(2)}\cdots \theta^{(L)}}}_{\text{client}}\otimes(\underbrace{\ket{0}}_{\tilde\bindic_{garbage}}\otimes(\\
        	&\underbrace{\ket{\ftrue}\ket{\ffalse}}_{\tilde\bindic_\pm}\sum_{b^{(0)}\in \{0,1\}}\sum_{\vec{b}\in \{0,1\}^L}e^{(\sum_{i\in [L]} \theta^{(i)}b^{(i)})\mi\pi/4}\underbrace{\ket{b^{(0)}}}_{\tilde\bindic_{col}}\underbrace{\ket{\vec{b}}}_{\bS_{subs}}\frac{1}{\sqrt{2^L}}c_{b^{(0)},+}\label{eq:3689q}\\
        	+&\ket{\ffalse}\ket{\ftrue}e^{(\sum_{i\in [L]}\theta^{(i)})\mi\pi/4}\sum_{\vec{b}\in \{0,1\}^L}\sum_{\vec{b}\in \{0,1\}^L}e^{-(\sum_{i\in [L]} \theta^{(i)}b^{(i)})\mi\pi/4}\ket{b^{(0)}}\ket{\vec{b}}\frac{1}{\sqrt{2^L}}c_{b^{(0)},-})\label{eq:3699q}\\
        	+&\underbrace{\ket{1}}_{\tilde\bindic_{garbage}}\ket{\cdots})
        	\end{align}
        	\item The next step is to create the state that simulates each branch of \eqref{eq:167r}\eqref{eq:170r} excluding the phase information. We first need to formally define the state that resembles each branch of \eqref{eq:167r}\eqref{eq:170r}. This is defined step-by-step as follows:
        	\begin{enumerate}
        		\item Define $\ket{\chi_{b^{(0)},\vec{b},\pm}^0}$ as part of the states in \eqref{eq:167r}\eqref{eq:170r} excluding the phase information:
       \begin{equation}\label{eq:3700p}\ket{\chi_{b^{(0)},\vec{b},+}^0}:=\underbrace{\sum_{K\in \Domain(\bK)}\ket{K}}_{\text{client}}\otimes\underbrace{\ket{x_{b^{(0)}}^{(0)}}\otimes \sum_{\vec{b}\in \{0,1\}^L}\ket{\vec{x}_{\vec{b}}}}_{\bS_{bsh}}\ket{\phi_{K,b^{(0)},\vec{b},+}}\end{equation}
              \begin{equation}\label{eq:3710p}\ket{\chi_{b^{(0)},\vec{b},-}^0}:=\underbrace{\sum_{K\in \Domain(\bK)}\ket{K}}_{\text{client}}\otimes\underbrace{\ket{x_{b^{(0)}}^{(0)}}\otimes \sum_{\vec{b}\in \{0,1\}^L}\ket{\vec{x}_{\vec{b}}}}_{\bS_{bsh}}\ket{\phi_{K,b^{(0)},\vec{b},-}}\end{equation}
              \item \eqref{eq:3700p}\eqref{eq:3710p} contains many client-side registers, but what $\fSim$ can simulate should be server-side. To address this problem, define $\ket{\chi_{b^{(0)},\vec{b},\pm}^1}$ as the following states: corresponding to $\bK$ in $\ket{\chi_{b^{(0)},\vec{b},\pm}^0}$, initialize $\bS_{\bK}$ with the same size as $\bK$; and corresponding to other client-side registers implicit in \eqref{eq:3700p}\eqref{eq:3710p}, initialize the corresponding server-side registers with the same size (for example, initialize $\bS_{\bK^{(\switch)}}$ corresponding to $\bK^{(\switch)}$). Denote the collection of them as $\bS_{\text{c}}$. Then define  
                     \begin{equation}\label{eq:3720p}\ket{\chi_{b^{(0)},\vec{b},+}^1}:=\underbrace{\sum_{K\in \Domain(\bK)}\ket{K}}_{\bS_{\bK}}\otimes\underbrace{\ket{x_{b^{(0)}}^{(0)}}\otimes \sum_{\vec{b}\in \{0,1\}^L}\ket{\vec{x}_{\vec{b}}}}_{\bS_{bsh}}\ket{\tilde\phi_{K,b^{(0)},\vec{b},+}}\end{equation}
              \begin{equation}\label{eq:3730p}\ket{\chi_{b^{(0)},\vec{b},-}^1}:=\underbrace{\sum_{K\in \Domain(\bK)}\ket{K}}_{\bS_{\bK}}\otimes\underbrace{\ket{x_{b^{(0)}}^{(0)}}\otimes \sum_{\vec{b}\in \{0,1\}^L}\ket{\vec{x}_{\vec{b}}}}_{\bS_{bsh}}\ket{\tilde\phi_{K,b^{(0)},\vec{b},-}}\end{equation}
              where $\ket{\tilde\phi_{K,b^{(0)},\vec{b},\pm}}$ is defined to be the state coming from swapping the client-side registers and $\bS_c$ (initialized to be empty)  on $\ket{\phi_{K,b^{(0)},\vec{b},\pm}}$.

        \item 
        This step is a preparation for dealing with the blinded oracle later. Recall that when we analyze the final outcome state $\ket{\varphi^\prime}$, one important step is to use \eqref{eq:220j} to replace the adversary querying the normal oracle by an adversary querying a blinded oracle. For the simulated state, we also need a way to connect the adversary that queries the normal oracle to an adversary that query the blinded oracle. This problem is addressed in this step, by artificially swapping the values of normal oracle and the blinded oracle in $\ket{\chi_{b^{(0)},\vec{b},\pm}^1}$, which is the definition of $\ket{\chi_{b^{(0)},\vec{b},\pm}^2}$.\par
        Define $\ket{\chi_{b^{(0)},\vec{b},+}^2},\ket{\chi_{b^{(0)},\vec{b},-}^2}$ as follows: start from states $\ket{\chi_{b^{(0)},\vec{b},+}^1},\ket{\chi_{b^{(0)},\vec{b},-}^1}$ (\eqref{eq:3720p}\eqref{eq:3730p}), 
       swap the values of $\bH(\{0,1\}^\kappa||K^{(\switch)}||\cdots)$ and $\bH^\prime(\{0,1\}^\kappa||K^{(\switch)}||\cdots)$ (recall $\bH^\prime$ is the blinded oracle); 
       Then as in the last step, we use an additional step to remove registers that are not server-side simulatable. Initialize a server-side empty register $\bS_{\bH^\prime,blind}$ that has the same size as the blinded part of the oracle, and swap the $\bH^\prime$ registers that originally hold the blinded part, with $\bS_{\bH^\prime,blind}$.\par
       What we get is the following. Suppose $\cU$ is a server-side operation that queries $H$, and $\cU^{blind}$ is its blinded version (that is, replacing all the queries to $H$ by queries to the blinded one). Both operations output a value to some register, and use $|\Pi_0\ket{\cdot}|^2$ to denote the probability of getting $0$. Then
       \begin{equation}\label{eq:375tmasd}|\Pi_0\cU\ket{\chi_{b^{(0)},\vec{b},\pm}^2}|^2=|\Pi_0\cU^{blind}\ket{\chi_{b^{(0)},\vec{b},\pm}^1}|^2\text{ for all $\pm \in \{+,-\},b^{(0)}\in \{0,1\},\vec{b}\in \{0,1\}^L$}\end{equation}
        \item We want to define an operation that prepares state $\ket{\chi_{b^{(0)},\vec{b},+}^2},\ket{\chi_{b^{(0)},\vec{b},-}^2}$ conditioned on the subscript registers and indicator registers in \eqref{eq:3689q}\eqref{eq:3699q}. But there is an additional  restriction from the norms of each branch in \eqref{eq:3689q}\eqref{eq:3699q}. For each value of $\bS_{sub}$ and indicator registers in \eqref{eq:3689q}\eqref{eq:3699q}, the norm of this branch is $\frac{1}{\sqrt{2^L}}c_{b^{(0)},\pm}$; the state we want to prepare in this state has norm $c_{b^{(0)},\vec{b},\pm}$, which is not compatible, even if they are close on average by \eqref{eq:3590}-\eqref{eq:3593}. We need to rescale the state $\ket{\chi_{b^{(0)},\vec{b},+}^2},\ket{\chi_{b^{(0)},\vec{b},-}^2}$. We need to be careful: we do not want to up-scale them since we can't prepare a state in the random oracle model that is not valid (Definition \ref{defn:validityro}), but down-scaling is fine.\footnote{As an example, if there is an operation that prepares a state $\ket{\varphi}$, it's possible to prepare $\sqrt{1-a^2}\ket{0}\ket{\varphi}+a\ket{1}\ket{garbage}$ for $a\in (0,1)$, where the first bit is in the $\tilde\bindic_{garbage}$ register. But it's not possible to prepare $1/\sqrt{1-a^2}\ket{\varphi}$.}\par
For each $b^{(0)}\in \{0,1\},\vec{b}\in \{0,1\}^L,\pm\in \{+,-\}$, define $\ket{\chi_{b^{(0)},\vec{b},\pm}^3}$ as:
\begin{equation}\label{eq:374q}\ket{\chi_{b^{(0)},\vec{b},\pm}^3}:=\begin{cases}
	\ket{\chi_{b^{(0)},\vec{b},\pm}^2}&\text{ if $c_{b^{(0)},\vec{b},\pm}\leq \frac{1}{\sqrt{2^L}}c_{b^{(0)},\pm}$}\\
	\frac{\frac{1}{\sqrt{2^L}}c_{b^{(0)},\pm}}{c_{b^{(0)},\vec{b},\pm}}\ket{\chi_{b^{(0)},\vec{b},\pm}^2}&\text{ if $c_{b^{(0)},\vec{b},\pm}> \frac{1}{\sqrt{2^L}}c_{b^{(0)},\pm}$}
\end{cases}\end{equation}
        \end{enumerate}
        Then the overall operation of this step is defined as follows. Controlled by the indicator and $\bS_{subs}$ registers, server-side isometry $O_2$ implements the mapping:
        \begin{align}&\underbrace{\ket{0}}_{\tilde\bindic_{garbage}}\underbrace{\ket{\ftrue}}_{\bindic_+}\underbrace{\ket{\ffalse}}_{\bindic_-}\underbrace{\ket{b^{(0)}}}_{\tilde\bindic_{col}}\underbrace{\ket{\vec{b}}}_{\bS_{subs}}\frac{1}{\sqrt{2^L}}c_{b^{(0)},+}\\ \rightarrow&\ket{\ftrue}\ket{\ffalse}\ket{b^{(0)}}\ket{\vec{b}}\begin{cases}
	\underbrace{\ket{0}}_{\tilde\bindic_{garbage}}\ket{\chi_{b^{(0)},\vec{b},+}^3}+\ket{2}\ket{\cdots}&\text{ if $c_{b^{(0)},\vec{b},+}\leq \frac{1}{\sqrt{2^L}}c_{b^{(0)},+}$}\\
	\ket{0}\ket{\chi_{b^{(0)},\vec{b},+}^3}&\text{ if $c_{b^{(0)},\vec{b},+}> \frac{1}{\sqrt{2^L}}c_{b^{(0)},+}$}
\end{cases}
\end{align}
        \begin{align}&\ket{0}\ket{\ffalse}\ket{\ftrue}\ket{b^{(0)}}\ket{\vec{b}}\frac{1}{\sqrt{2^L}}c_{b^{(0)},-}\\\rightarrow& \begin{cases}
	\underbrace{\ket{0}}_{\tilde\bindic_{garbage}}\ket{\chi_{b^{(0)},\vec{b},-}^3}+\ket{2}\ket{\cdots}&\text{ if $c_{b^{(0)},\vec{b},-}\leq \frac{1}{\sqrt{2^L}}c_{b^{(0)},-}$}\\
	\ket{0}\ket{\chi_{b^{(0)},\vec{b},-}^3}&\text{ if $c_{b^{(0)},\vec{b},-}> \frac{1}{\sqrt{2^L}}c_{b^{(0)},-}$}
\end{cases}\end{align}
        \item Define $O_3$ as the operator that erases the subscript vector register $\bS_{subs}$ based on the values of $\tilde\bS_{bsh}$ (where $\tilde\bS_{bsh}$ is defined to be the registers that holds $\vec{x}_{\vec{b}}$ in \eqref{eq:3700p}\eqref{eq:3710p}):
        $$O_3:\underbrace{\ket{K}}_{\bS_{\bK}}\underbrace{\ket{\vec{x}_{\vec{b}}}}_{\bS_{bsh}}\underbrace{\ket{\vec{b}}}_{\bS_{subs}}\rightarrow \ket{K}\ket{\vec{x}_{\vec{b}}}\ket{0}$$
        \item The remaining thing to do is to append the operators that connects a simulation of \eqref{eq:93} given above to the simulation of real execution outcome $\Pi_{\fcomp}\ket{\varphi^\prime}$.\par
        The overall operation of $\fSim$, operating on $\ket{\eqref{eq:target}}$, is defined as $$ \fSet(\btype\rightarrow\ket{\fcomp})\circ\fAdv_{\fcomp}\circ(\tilde\fAdv_{2.b.sbt})^{-1}\circ \fDisgard(\tilde\bindic,\bS_{\bK},\bS_{c},\bS_{\bH^\prime,blind})\circ\fCOPY(\bS_{\tilde\bK}\rightarrow \btrans_{\tilde\bK}) \circ$$
        $$ O_3\circ O_2\circ O_1$$
        where:
        \begin{itemize}
        \item Recall in the $\fcomp$ round of Protocol \ref{prtl:prerspv}, the client needs to send out all the keys in $\tilde K$ to the server in the end. Denote the transcript register that holds this information as $\btrans_{\tilde\bK}$, and $\fCOPY(\bS_{\tilde\bK}\rightarrow \btrans_{\tilde\bK})$ that bitwise-CNOT the corresponding keys in the corresponding simulated registers\footnote{Recall $\bS_{\bK}$ holds the simulation of $\bK$, here $\bS_{\tilde\bK}$ is the part of $\bS_{\bK}$ that simulates $\tilde\bK$.} to $\btrans_{\tilde\bK}$.
        \item $\fDisgard$ operator disgards $\tilde\bindic_{garbage},\tilde\bindic_{+},\tilde\bindic_{-},\tilde\bindic_{col}$ and $\bS_{\bK},\bS_{c},\bS_{\bH^\prime,blind}$ to the environment; these registers are used in our construction but are not accessible by the distinguisher in \eqref{eq:61q}.
        \item $\fAdv_{\fcomp}$ is the operation of $\fAdv$ in the $\fcomp$ round;  $\tilde\fAdv_{2.b.sbt}$ is the non-blinded version of $\tilde\fAdv_{2.b.sbt}$ defined in \eqref{eq:254k}.
        \item $\fSet(\btype\rightarrow\ket{\fcomp})$ means setting register $\btype$ (recall in execution of Protocol \ref{prtl:prerspv} there is a round type register) in the transcript to $\fcomp$.
        \end{itemize}

        \end{enumerate}
       \paragraph{Part III: proof of \eqref{eq:61q}} Let's prove $\fSim$ constructed above satisfies \eqref{eq:61q}. Suppose the efficient distinguisher is $D$, and use $\Pi_0$ to denote the projector that the distinguisher outputs $0$. Thus \eqref{eq:61q} translates to:
       \begin{align}
       	&|\Pi_0D\Pi_{\fcomp}\ket{\varphi^\prime}|\label{eq:379q}\\
       	\approx_{0.1\sqrt{p_{\fcomp}}+\fneg(\kappa)} &|\Pi_0D\sqrt{p_{\fcomp}}\fSim\ket{\eqref{eq:target}}|\label{eq:380q}
       \end{align}
       We move towards it by analyzing both sides. First define $D^{blind}$ as the operation that queries the blinded oracle instead of the original oracle $H$. First by \eqref{eq:220j} we can calculate the inner state of \eqref{eq:379q}:
       \begin{align}
       	&D\Pi_{\fcomp}\ket{\varphi^\prime}\\
       	=&{\scriptsize\sqrt{p_{\fcomp}}D\fSet(\btype\rightarrow\ket{\fcomp})\fDisgard(\text{client-side registers except $\btheta^{(1)}\btheta^{(2)}\cdots \btheta^{(L)}$})\fCalcRel(\btheta^{(1)}\btheta^{(2)}\cdots \btheta^{(L)})}\nonumber\\
       	&\circ\fAdv_{\fcomp}(\ket{\varphi^{2.a}}\odot \tilde\bK)\\
       	\approx_{10^{-100}+\fneg(\kappa)}&{\scriptsize\sqrt{p_{\fcomp}}D^{blind}\fSet(\btype\rightarrow\ket{\fcomp})\fDisgard(\text{client-side registers except $\btheta^{(1)}\btheta^{(2)}\cdots \btheta^{(L)}$})\fCalcRel(\btheta^{(1)}\btheta^{(2)}\cdots \btheta^{(L)})}\label{eq:tq}\\
       	&\circ\fAdv_{\fcomp}^{blind}(\ket{\varphi^{2.a}}\odot \tilde\bK)\label{eq:3857}
       \end{align}
        where $\fCalcRel(\btheta^{(1)}\btheta^{(2)}\cdots \btheta^{(L)})$ is the client-side operation that calculates the relative phase $\btheta^{(1)}\btheta^{(2)}\cdots \btheta^{(L)}$ from $\bTheta$.\par
        Now we analyze \eqref{eq:380q} and move towards \eqref{eq:tq}\eqref{eq:3857}.
        \begin{enumerate}\item By \eqref{eq:3590}-\eqref{eq:3593}, we can change the basis norms of the outputs of $O_3\circ O_2\circ O_1\ket{\eqref{eq:target}}$ to the same norms as $\ket{\phi_{\cdots}}$ in \eqref{eq:167r}\eqref{eq:170r} (that is, remove the re-scaling in \eqref{eq:374q}):
        $$O_3\circ O_2\circ O_1 \circ\eqref{eq:target}$$
        $$\approx_{0.01+\fneg(\kappa)}\sum_{\theta^{(1)}\theta^{(2)}\cdots \theta^{(L)}\in \{0,1\cdots 7\}^L}\frac{1}{\sqrt{8^L}}\ket{\theta^{(1)}\theta^{(2)}\cdots \theta^{(L)}}\otimes ( \underbrace{\ket{0}}_{\tilde\bindic_{garbage}}\otimes($$
        $$\ket{\ftrue}\ket{\ffalse}\sum_{b^{(0)}\in \{0,1\}}\sum_{\vec{b}\in \{0,1\}^L}\underbrace{\ket{b^{(0)}}}_{\tilde\bindic_{col}}\ket{\chi^2_{b^{(0)},\vec{b},+}} $$
        \begin{equation}\label{eq:298ks} +e^{(\sum_{i\in [L]}\theta^{(i)})\mi\pi/4}\ket{\ffalse}\ket{\ftrue}\sum_{b^{(0)}\in \{0,1\}}\sum_{\vec{b}\in \{0,1\}^L}{\ket{b^{(0)}}}\ket{\chi^2_{b^{(0)},\vec{b},-}})+\ket{1}\ket{\cdots})\end{equation}
        \item This step aims at replacing $\ket{\chi^2}$ in \eqref{eq:298ks} by the corresponding $\ket{\chi^1}$. The difference here is we swap the content of the original oracle and the blinded oracle; as discussed in \eqref{eq:375tmasd}, if we also swap the operators that query the original oracle to the blinded oracle, the final probability of outputting $0$ in \eqref{eq:380q} will not change. Explicitly, it is
        \begin{align}
        	&{\scriptsize|\Pi_0D\fSet(\btype\rightarrow\ket{\fcomp})\circ\fAdv_{\fcomp}\circ(\tilde\fAdv_{2.b.sbt})^{-1}\circ \fDisgard(\tilde\bindic,\bS_{\bK},\bS_{c},\bS_{\bH^\prime,blind})\circ\fCOPY(\bS_{\tilde\bK}\rightarrow \btrans_{\tilde\bK}) \circ\eqref{eq:298ks}|}\\
        	=&{\scriptsize|\Pi_0D^{blind}\fSet(\btype\rightarrow\ket{\fcomp})\circ\fAdv_{\fcomp}^{blind}\circ(\tilde\fAdv_{2.b.sbt}^{blind})^{-1}\circ \fDisgard(\tilde\bindic,\bS_{\bK},\bS_{c})\circ\fCOPY(\bS_{\tilde\bK}\rightarrow \btrans_{\tilde\bK})}(\\
        	&\sum_{\theta^{(1)}\theta^{(2)}\cdots \theta^{(L)}\in \{0,1\cdots 7\}^L}\frac{1}{\sqrt{8^L}}\ket{\theta^{(1)}\theta^{(2)}\cdots \theta^{(L)}}\otimes ( \underbrace{\ket{0}}_{\tilde\bindic_{garbage}}\otimes(\label{eq:3895}\\
        &\ket{\ftrue}\ket{\ffalse}\sum_{b^{(0)}\in \{0,1\}}\sum_{\vec{b}\in \{0,1\}^L}\underbrace{\ket{b^{(0)}}}_{\tilde\bindic_{col}}\ket{\chi^1_{b^{(0)},\vec{b},+}} \label{eq:3905}\\
        &+e^{(\sum_{i\in [L]}\theta^{(i)})\mi\pi/4}\ket{\ffalse}\ket{\ftrue}\sum_{b^{(0)}\in \{0,1\}}\sum_{\vec{b}\in \{0,1\}^L}{\ket{b^{(0)}}}\ket{\chi^1_{b^{(0)},\vec{b},-}})+\ket{1}\ket{\cdots})\qquad |\label{eq:391}\\
        \end{align}
        \item This step aims at replacing $\ket{\chi^1}$ in \eqref{eq:391} by the corresponding $\ket{\chi^0}$. The two states are the same up to positions of some registers ($\bK$ corresponds to $\bS_{\bK}$, etc). Thus
        \begin{align}
        &\fDisgard(\tilde\bindic,\bS_{\bK},\bS_{c},)\circ\fCOPY(\bS_{\tilde\bK}\rightarrow \btrans_{\tilde\bK})\eqref{eq:3895}\eqref{eq:3905}\eqref{eq:391}\\
        =&	\fDisgard(\bindic,\bK,\text{client-side registers of $\ket{\chi}$})\circ(\\
        &\sum_{\theta^{(1)}\theta^{(2)}\cdots \theta^{(L)}\in \{0,1\cdots 7\}^L}\frac{1}{\sqrt{8^L}}\ket{\theta^{(1)}\theta^{(2)}\cdots \theta^{(L)}}\otimes ( \underbrace{\ket{0}}_{\tilde\bindic_{garbage}}\otimes(\label{eq:38959}\\
        &\ket{\ftrue}\ket{\ffalse}\sum_{b^{(0)}\in \{0,1\}}\sum_{\vec{b}\in \{0,1\}^L}\underbrace{\ket{b^{(0)}}}_{\tilde\bindic_{col}}\ket{\chi^0_{b^{(0)},\vec{b},+}} \label{eq:39059}\\
        &+e^{(\sum_{i\in [L]}\theta^{(i)})\mi\pi/4}\ket{\ffalse}\ket{\ftrue}\sum_{b^{(0)}\in \{0,1\}}\sum_{\vec{b}\in \{0,1\}^L}{\ket{b^{(0)}}}\ket{\chi^0_{b^{(0)},\vec{b},-}})+\ket{1}\ket{\cdots})\label{eq:3919}\\
        &)\odot \tilde\bK))\label{eq:eq}
        \end{align}
Note when we replace $\bS_{\tilde\bK}$ by $\tilde\bK$ $\fCOPY(\bS_{\tilde\bK}\rightarrow \btrans_{\tilde\bK})$ becomes $\odot \tilde\bK$.
      \item  Note that we have not simulated the $\bTheta$ register yet, and the phases in \eqref{eq:298ks}\eqref{eq:391}\eqref{eq:3919} solely come from \eqref{eq:target}. Introduce registers $\bS_{\bTheta}$ that has the same size as $\bTheta$ (which could hold $(1+L)$ phase pairs). Consider state
    $$\sum_{\theta^{(1)}\theta^{(2)}\cdots \theta^{(L)}\in \{0,1\cdots 7\}^L}\frac{1}{\sqrt{8^L}}(\ket{\theta^{(1)}\theta^{(2)}\cdots \theta^{(L)}}\otimes$$
    $$ \sum_{\substack{\Theta\in \Domain(\bTheta)\text{ such that}\\\text{ the relative phase of $\Theta^{(1)}\Theta^{(2)}\cdots \Theta^{(L)}$ is }\theta^{(1)}\theta^{(2)}\cdots \theta^{(L)}}}\frac{1}{\sqrt{8^L}}\underbrace{\ket{\Theta}}_{{\bTheta}}\otimes$$
        \begin{equation*} e^{\tSUM(\vec{\Theta}_{b^{(0)}\vec{b}})\mi\pi/4}\ket{\ftrue}\ket{\ffalse}\sum_{b^{(0)}\in \{0,1\}}\sum_{\vec{b}\in \{0,1\}^L}\underbrace{\ket{b^{(0)}}}_{\tilde\bindic_{col}}\ket{\chi^0_{b^{(0)},\vec{b},+}}\end{equation*}\begin{equation}\label{eq:301ks}+e^{-\tSUM(\vec{\Theta}_{b^{(0)}\vec{b}})\mi\pi/4}\ket{\ffalse}\ket{\ftrue}\sum_{b^{(0)}\in \{0,1\}}\sum_{\vec{b}\in \{0,1\}^L}{\ket{b^{(0)}}}\ket{\chi^0_{b^{(0)},\vec{b},-}}+\ket{1}\ket{\cdots})\odot \tilde\bK))\end{equation}
Compare \eqref{eq:38959}\eqref{eq:39059}\eqref{eq:3919}\eqref{eq:eq} and \eqref{eq:301ks} two states differ by a global phase on each component corresponding to each value of $\bindic_{col}, \bindic_{garbage},{\bTheta}$. By Fact \ref{fact:overalltool} there is
\begin{equation}\label{eq:312is}\fDisgard({\bTheta})\eqref{eq:301ks}\approx^{ind:\cF_{cq}} \eqref{eq:38959}\eqref{eq:39059}\eqref{eq:3919}\end{equation}
      where $\cF_{cq}$ is the set of operators that operates on the transcript, client-side registers and the $\bindic_{\pm}$ register in a read-only way.\par
Note 
      \begin{align}&\Pi^{\bindic_{garbage}}_0\eqref{eq:301ks}\\
      	=&\fCalcRel(\btheta^{(1)}\btheta^{(2)}\cdots \btheta^{(L)})(\nonumber\\
      	&\qquad\qquad \fCollapse(\bS_{bsh}^{(0)})\cP\cR_2(\ket{\$_2}\otimes (\cR_1(\ket{\$_1}\otimes \fReviseRO\circ\Pi_{\basishonest(\bK)}\ket{\tilde\varphi^{2.a}})))\odot \tilde\bK)\label{eq:341q}
      \end{align}
      where $\fCollapse(\bS_{bsh}^{(0)})$ is the operation that calculate the subscript of the keys in $\bS_{bsh}^{(0)}$ to register $\bindic_{col}$.
      \item Recall on each branch corresponding to the $\bK^{(0)}$, $\cP^\dagger$ could be seen as an operation that operates only on $\bindic_{\pm},\bTheta^{(0)}$. Thus
    \begin{align}
      	&\fCollapse(\bS_{bsh}^{(0)})\cP\cR_2(\ket{\$_2}\otimes (\cR_1(\ket{\$_1}\otimes \fReviseRO\circ\Pi_{\basishonest(\bK)}\ket{\tilde\varphi^{2.a}})))\odot \tilde\bK)\label{eq:eq408}\\
      	\approx^{ind}_0&\fCollapse(\bS_{bsh}^{(0)})\cP^\dagger\cP\cR_2(\ket{\$_2}\otimes (\cR_1(\ket{\$_1}\otimes \fReviseRO\circ\Pi_{\basishonest(\bK)}\ket{\tilde\varphi^{2.a}})))\odot \tilde\bK)\label{eq:eq409}\\
      	\approx^{ind}_{0.01+\fneg(\kappa)}&\fCollapse(\bS_{bsh}^{(0)})\Pi_{\basishonest(\bK)}\ket{\tilde\varphi^{2.a}})))\odot \tilde\bK)\label{eq:eq410}\\
      	\approx^{ind}_{\fneg(\kappa)}& \Pi_{\basishonest(\bK)}\ket{\tilde\varphi^{2.a}}\odot \tilde\bK\label{eq:eq411}\\
      	\approx_{0.0001}& \ket{\tilde\varphi^{2.a}}\odot \tilde\bK\label{eq:eq412}
      \end{align}
      where $ind$ represents efficient operators in $\cF_{cq\land blind}$, defined to be the set of operators that, for the client-side access, it could operate on $\bTheta^{(1)}\cdots \bTheta^{(L)}$ in a read-only way, and only query the blinded oracle. Each step above comes from:
      \begin{itemize}\item \eqref{eq:eq408}\eqref{eq:eq409} comes from the fact that the distinguisher does not operate on the indicator registers and $\bTheta^{(0)}$; \item \eqref{eq:eq409}\eqref{eq:eq410} is from \eqref{eq:92}; \item \eqref{eq:eq410}\eqref{eq:eq411} is by the collapsing property (where the distinguisher has no access to $\bK^{(0)}$), \eqref{eq:eq411}\eqref{eq:eq412} is by \eqref{eq:254k}.\item  Also note that we omit registers that are not used and remain in product state with other parts (for example, $\ket{\$_1},\ket{\$_2}$ etc).\end{itemize}

       \end{enumerate}
       Combining all these steps we have
       \begin{align}
       &|\Pi_0D\fSim\ket{\eqref{eq:target}}|\\
       \approx_{0.011+\fneg(\kappa)}	&|\Pi_0D^{blind}\fSet(\btype\rightarrow\ket{\fcomp})\fDisgard(\text{client-side registers except $\btheta^{(1)}\btheta^{(2)}\cdots \btheta^{(L)}$})\fCalcRel(\btheta^{(1)}\btheta^{(2)}\cdots \btheta^{(L)})\\
       	&\circ\fAdv_{\fcomp}^{blind}(\fAdv_{2.b.sbt}^{blind})^{-1}(\ket{\tilde\varphi^{2.a}}\odot \tilde\bK)|
       \end{align}

       which compared with \eqref{eq:3857} completes the proof.

\end{proof}

\section{From Remote State Preparation to Quantum Computation Verification}\label{sec:13}
In this section we complete the construction of our CVQC protocol thus complete the proof of Theorem \ref{thm:main}.
\subsection{From Pre-RSPV to RSPV}\label{sec:13.1}
In this subsection we will give an RSPV protocol (as defined in Definition \ref{defn:rspv}, \ref{defn:rspvv}) from the pre-RSPV protocol (as defined in Definition \ref{defn:prerspv}, \ref{defn:prerspvv}, constructed in Protocol \ref{prtl:prerspv}).\par
Comparing the definition of RSPV to pre-RSPV, the differences are: 
\begin{itemize}
	\item In Definition \ref{defn:prerspvv} there is an additional case in the conclusion that the winning probability is bounded away from $\OPT$. This means in $\fpreRSPV$ the adversary could possibly cheat by making the winning probability small.\par
	In addition, in Definition \ref{defn:rspvv} the bound on the passing probability in the first case is much smaller than Definition \ref{defn:prerspvv}. This means in $\fRSPV$ the adversary's freedom of cheating without being caught is much smaller.
	\item In $\fpreRSPV$ the output state is only generated in the $\fcomp$ round, which appears with probability $p_{\fcomp}$; in $\fRSPV$ the state should be generated with high probability.
\end{itemize}
We do the security amplification from $\fpreRSPV$ to $\fRSPV$ in two steps, as follows.
\begin{enumerate}
	\item $\fpreRSPV$ to $\fpreRSPVTemp$:
	\begin{enumerate}
	\item In Step 1, both parties do a repetition of $\fpreRSPV$ protocol. And the client calculates the number of winning cases. If it's significantly fewer than the optimal expectation value, the client outputs $\ffail$. (The client also outputs $\ffail$ if any call to the subprotocols $\ffail$.)\par
	In this way we resolve the problems discussed in the first bullet above. For the second bullet we put it into the second step below, and for the honest behavior (correctness property) of this step we use a simple solution as follows:\par
	\item The client chooses a random index $i$ and reveals it. If the $i$-th round is a $\fcomp$ round, then the server gets the state and the client stores $\fcomp$ as the overall flag. Otherwise the client stores $\perp$ as the overall round type.
	\end{enumerate}
	Thus $\fpreRSPVTemp$ will have the following outputs: type $\in \{\fcomp,\perp\}$, flag $\in \{\fpass,\ffail\}$, and the client-side keys and server-side states (only if type$=\fcomp$).
	\item $\fpreRSPVTemp$ to $\fRSPV$:\par
	Both parties repeat the $\fpreRSPVTemp$ protocol for many times to ensure an output state is generated with high probability.
\end{enumerate}
\subsubsection{Step 1: a fully verifiable protocol that does not necessarily generate an output state}
We first construct $\fpreRSPVTemp$ from $\fpreRSPV$. 
\begin{mdframed}
\begin{prtl}[$\fpreRSPVTemp$]\label{prtl:rspv}
Suppose the security parameter is $\kappa$ and the output number is $L$. 
\begin{enumerate}
	\item Take the round number $N=10^{2500}$.\par
	For $i$ in $[N]$:\begin{enumerate}
	\item Both parties run protocol $\fpreRSPV(1^L,1^\kappa)$. \end{enumerate}
	Note for each $i$ there is a round type sampled from $\{\ftest,\fquiz,\fcomp\}$. In addition to the $\fpass/\ffail$ flag, the client will output a score in $\{\fwin,\flose,\perp\}$, where $\fwin/\flose$ only appear in the $\fquiz$ round, and honestly, $score=\fwin$ with probability $\OPT$ (conditioned on $\fquiz$ round).
	\item If any round $\ffail$, the client outputs $\ffail$. \par
	If the total number of $\fwin$ is $\leq N\cdot p_{\fquiz}\cdot(\OPT-10^{-210})$ (recall $\OPT=\frac{1}{3}\cos^2(\pi/8)=0.28451779686\cdots$, $p_{\fquiz}=\frac{1}{10}$), the client outputs $\ffail$.
	\item The client picks a random round $i$  and tells $i$ to the server.
	\begin{itemize}\item If the round type of the $i$-th round is $\fcomp$, the server picks up the corresponding gadgets and disgards the others. The client keeps the keys, outputs $\fcomp$ in the overall round type register and disgards the other systems.
	\item If the round type of the $i$-th round is not $\fcomp$, the client stores $\perp$ in the overall round type register. Both parties disgard everything else.
	\end{itemize}

\end{enumerate}
\end{prtl}
	
\end{mdframed}
\paragraph{Correctness} In the honest setting conditioned on the overall round type is $\fcomp$, with probability $\geq 1-10^{-5}-\fneg(\kappa)$ the joint output state of the client and the server is \eqref{eq:cqtarget}.
\begin{proof}[Proof of correctness]
By Section \ref{sec:5.4r} in the honest setting each call to $\fpreRSPV$ could $\ffail$ only with negligible probability. Thus except with negligible probability, the only case where the client will output $\ffail$ against an honest server is the statistical testing of $\fquiz$ scores (number of $\fwin$) in the second step. 
By properties of the $\fpreRSPV$ protocol the honest server could generate a $\fwin$ score with probability $\geq \OPT-\fneg(\kappa)$ in each $\fquiz$ round. Thus the expectation of $\fwin$ for each round $i\in [N]$ is $\geq p_{\fquiz}\cdot \OPT-\fneg(\kappa)$. Then by Chernoff's bound $$\Pr[\text{ the number of $\fwin$ is $\leq N\cdot p_{\fquiz}\cdot( \OPT-10^{-210})$}]\leq  10^{-6}+\fneg(\kappa)$$
 This completes the proof.
\end{proof}
Now we prove Protocol \ref{prtl:rspv} satisfies a verifiability statement as Definition \ref{defn:rspvv} (verifiability of  RSPV protocol) for target state \eqref{eq:target} with output number $L$. 
\begin{thm}[Verifiability of $\fpreRSPVTemp$]\label{thm:rspvprime}
	For any polynomial time adversary $\fAdv$, any efficiently-preparable initial state $\ket{\varphi^0}=O\ket{0}$, 
there exists a server-side operation $\fSim^{\fAdv,O}$ such that
\begin{align}\label{eq:finalsim}&\Pi_{\fcomp}\Pi_{\fpass}\fpreRSPVTemp^{\fAdv}(1^L,1^\kappa) \ket{\varphi^0}\\\approx^{ind}_{0.11+\fneg(\kappa)}&\sqrt{p_{\fcomp}}\Pi_{\fpass}\fSim^{\fAdv,O} \ket{\text{Equation \eqref{eq:target}}}\end{align}
\end{thm}
\begin{proof}[Proof of Theorem \ref{thm:rspvprime}]
	Consider an efficient adversary $\fAdv$. Let's analyze the output state of the first step of Protocol \ref{prtl:rspv}. The first step is an iteration of $\fpreRSPV$ for each $i\in [N]$. For each round counter $i\in [N]$, suppose the history of round types and output flags and scores of previous tests by the beginning of the $i$-th round are recorded as 
	\begin{equation}\label{eq:59n}rec^{<i}=((\ttype_1\cdots \ttype_{i-1}), (\tflag_1,\cdots \tflag_{i-1}), (\score_1\cdots \score_{i-1}))\end{equation}
	 where $\ttype\in \{\ftest,\fquiz,\fcomp \}$, $\tflag\in \{\fpass,\ffail\}$,  $\score\in \{\fwin,\flose,\bot\}$. 
	 We use bold fonts for their corresponding registers. Then
	$$\Domain(\brec^{<i}):=\text{ The set of all the valid records (equation \eqref{eq:59n}) by the end of round $i-1$}$$
	 Suppose the server's state in the end of round $i-1$ is denoted as $\ket{\varphi^{i-1}}$ and the component (Definition \ref{nota:compo}) when the history record is $rec^{<i}$ is denoted as $\ket{\varphi^{i-1}_{rec^{<i}}}$. Thus
	\begin{equation}\label{eq:286is}\ket{\varphi^{i-1}}=\sum_{rec^{<i}\in \Domain(\brec^{<i})} \underbrace{\ket{rec^{<i}}}_{\brec^{<i}}\otimes\ket{\varphi^{i-1}_{rec^{<i}}}\end{equation}
	When the first step of Protocol \ref{prtl:rspv} completes, the final state could be decomposed as
$$\ket{\varphi^N}=\sum_{rec\in \Domain(\brec^{\leq N})}\underbrace{\ket{rec}}_{\brec^{\leq N}}\otimes\ket{\varphi^N_{rec}}$$
where $\brec^{\leq N}$ are registers for all these records in all the $N$ rounds.\par
	 Denote the part of $\fAdv$'s operation at round $i$  on record $rec^{<i}$ as $\fAdv_{rec^{<i},i}$. For each round $i\in [N]$, history record $rec^{<i}\in \Domain(i-1)$, apply the verifiability of $\fpreRSPV$ (Theorem \ref{thm:prerspvv}) on initial state
	 $$\Pi^{\bflag_{<i}}_{\fpass}\ket{\varphi^{i-1}_{rec^{<i}}}$$
	where $\Pi^{\bflag_{<i}}_{\fpass}$ is the projection onto the space that the registers $\bflag_1\cdots \bflag_{i-1}$ all have value $\fpass$.   Then we know there exists a server-side simulator $\fSim_{rec^{<i},i}^{\fAdv_{rec^{<i},i}}$ such that at least one of the following three cases is true: 
	 \footnote{Note that Theorem \ref{thm:prerspvv} is on normalized initial state, thus we need to normalize $\Pi^{\bflag_{<i}}_{\fpass}\ket{\varphi^{i-1}_{rec^{<i}}}$ before applying it. This is possible as long as $\Pi^{\bflag_{<i}}_{\fpass}\ket{\varphi^{i-1}_{rec^{<i}}}$ is non-negligible norm, and the efficiently-preparable property still preserve by Fact \ref{fact:simfa}. If the norms are negligible it could be merged with the third case below.}
	\begin{itemize}
	
	\item Limited passing probability:
	\begin{equation}\label{eq:61new}|\Pi_{\fpass}^{\bflag_i}\fpreRSPV^{\fAdv_{rec^{<i},i}}(\ket{rec^{<i}}\otimes\Pi^{\bflag_{<i}}_{\fpass}\ket{\varphi^{i-1}_{rec^{<i}}})|^2\leq (1-10^{-2000})|\Pi^{\bflag_{<i}}_{\fpass}\ket{\varphi^{i-1}_{rec^{<i}}}|^2\end{equation}
	(We omit the arguments of the protocol call for simplicity.)
	\item The probability of $\fwin$ is less than expected: 
	\begin{equation}\label{eq:casetwo}|\Pi_{\fwin}^{\bscore_i}\fpreRSPV^{\fAdv_{rec^{<i},i}}(\ket{rec^{<i}}\otimes\Pi^{\bflag_{<i}}_{\fpass}\ket{\varphi^{i-1}_{rec^{<i}}})|^2\leq p_{\fquiz}\cdot(\OPT-10^{-200})|\Pi^{\bflag_{<i}}_{\fpass}\ket{\varphi^{i-1}_{rec^{<i}}}|^2\end{equation}
	\item The output state has verifiability:
	\begin{align}&\Pi_{\fcomp}^{\btype_i}\fpreRSPV^{\fAdv_{rec^{<i},i}}(\ket{rec^{<i}}\otimes\Pi^{\bflag_{<i}}_{\fpass}\ket{\varphi^{i-1}_{rec^{<i}}})\nonumber\\\approx^{ind}_{0.1\sqrt{p_{\fcomp}}|\Pi^{\bflag_{<i}}_{\fpass}\ket{\varphi^{i-1}_{rec^{<i}}}|+\fneg(\kappa)}& \sqrt{p_{\fcomp}}\cdot (\ket{rec^{<i}}\otimes\fSim^{\fAdv_{rec^{<i},i},\fAdv_{rec^{<i},<i}}_{rec^{<i},i}(\ket{\text{equation } \eqref{eq:target}})\label{eq:61}\end{align}
	where $\fAdv_{rec^{<i},<i}$ is the operation of $\fAdv$ by the beginning of the $i$-th round when the record is $rec^{<i}$.
	\end{itemize}
%

We want to show \eqref{eq:61} is true with sufficiently high probability in the passing space for randomly chosen $i$. Formally, define
\begin{equation}\label{eq:296ih}\text{$T:=$ the set of $(rec^{<i},i)$, $i\in [N],rec^{<i}\in \Domain(\brec^{<i})$ that \eqref{eq:61} is true.}	
\end{equation}
We want to show
\begin{equation}\label{eq:251b}|(\bbI-\Pi^{rec^{<i},i}_{\in T})\Pi_{\fpass}(\sum_{i\in [N]}\frac{1}{\sqrt{N}}\ket{i}\otimes\sum_{rec\in \Domain(\brec^{\leq N})}\ket{rec}\otimes\ket{\varphi^N_{rec}})|^2<10^{-4}+\fneg(\kappa)\end{equation}
where $\Pi^{rec^{<i},i}_{\in T}$ is the projection onto the space that the round counter $i$, and the history record by the beginning of time step $i$, is in $T$. We put the proof into a box for continuity of proof stream.
\begin{mdframed}
Consider state $\ket{\varphi^N}$, which is the state when the first step of $\fpreRSPVTemp$ completes. Use $rec\leftarrow \Pi_{\fpass}\ket{\varphi^N}$ to denote the (subnormalized) probability distribution coming from measuring the $rec$ transcript register of $\Pi_{\fpass}\ket{\varphi^N}$. Recall we use $rec^{<i}$ to denote the first $i-1$ records of $rec$. Define events
$$E_1:=\text{ there exists $i\in [N]$, $\bflag_i=\ffail$}$$
$$E_2:=\text{number of $\fwin$}\leq N\cdot p_{\fquiz}\cdot (\OPT-10^{-210})$$
For simplicity define
$$\epsilon =10^{-2000}$$
And define
$$A_{rec}:=\{i\in [N]:(rec^{<i},i)\text{ makes \eqref{eq:61new} hold}\}$$
$$B_{rec}:=\{i\in [N]:(rec^{<i},i)\text{ makes \eqref{eq:casetwo} hold}\}$$
$$W_{rec}:=\{i\in [N]:\bscore^i=\fwin\text{ in $\ket{\varphi^N}$}\}$$
Then
	\begin{align}&\Pr_{i\leftarrow [N]}\Pr_{rec\leftarrow \Pi_{\fpass}\ket{\varphi^N}}[(\eqref{eq:casetwo}\text{ holds }\lor \eqref{eq:61new}\text{ holds })\land \lnot (E_1\lor E_2)]\nonumber\\
	\leq& \Pr_{rec\leftarrow \Pi_{\fpass}\ket{\varphi^N}}[|A_{rec}|\geq 30/\epsilon\land \lnot E_1]\nonumber\\
	&+\Pr_{rec\leftarrow \Pi_{\fpass}\ket{\varphi^N}}[|B_{rec}-A_{rec}|\geq 10^{-5}N\land |A_{rec}|\leq 30/\epsilon\land \lnot E_2]\nonumber\\
	&+60/(N\epsilon)+10^{-5}\label{eq:310}
	\end{align}
	where the first term in \eqref{eq:310} is $\leq (1-\epsilon)^{30/\epsilon}\leq 10^{-5}$.\par
	To bound the second term above, recall the definition of $E_2$ above. We want to bound the probability that the number of $\fwin$ is big with the other events given in the second term. First note
	$$[N]=A_{rec}\cup (B_{rec}-A_{rec})\cup ([N]-B_{rec}-A_{rec})$$
	 We will bound the number of $\fwin$ for $i\in B_{rec}-A_{rec}$ and $i\in ([N]-B_{rec}-A_{rec})$ separately, as follows:
	\begin{equation}\label{eq:320}\Pr_{rec\leftarrow \Pi_{\fpass}\ket{\varphi^N}}[|(B_{rec}-A_{rec})\cap W_{rec}|\geq |B_{rec}-A_{rec}|\cdot p_{\fquiz}\cdot (\OPT-10^{-203})]\leq 10^{-5}\end{equation}
\begin{equation}\label{eq:321}\Pr_{rec\leftarrow \Pi_{\fpass}\ket{\varphi^N}}[|([N]-B_{rec}-A_{rec})\cap W_{rec}|\geq |[N]-B_{rec}-A_{rec}|\cdot p_{\fquiz}\cdot (\OPT+10^{-215})]\leq 10^{-5}+\fneg(\kappa)\end{equation}

Both comes from Corollary \ref{cor:chernoff}. \footnote{The probability upper-bound of each sample in \eqref{eq:321} come from \eqref{eq:61}, while the   probability upper-bound of each sample in \eqref{eq:320} comes from Theorem \ref{thm:opt}:
\begin{align}
&|\Pi_{\fwin}^{\bscore_i}\fpreRSPV^{\fAdv_{rec^{<i},i}}(\ket{rec^{<i}}\otimes\Pi^{\bflag_{<i}}_{\fpass}\ket{\varphi_{rec^{<i}}^{i-1}})|^2\leq p_{\fquiz}\cdot (\OPT+10^{-220})|\Pi^{\bflag_{<i}}_{\fpass}\ket{\varphi_{rec^{<i}}^{i-1}}|^2+\fneg(\kappa)\label{eq:193}
\end{align}}
Combining them we have
$$\Pr_{rec\leftarrow \Pi_{\fpass}\ket{\varphi^N}}[|B_{rec}-A_{rec}|\geq 10^{-5}N\land |A_{rec}|\leq 30/\epsilon\land \lnot E_2]\leq 5\times 10^{-5}+\fneg(\kappa)$$
Summing them up completes the proof.
\end{mdframed}
When the whole protocol completes, for each $rec\in \Domain(\brec^{\leq N})$, suppose $rec$ appears with probability $p_{rec}$. Now consider the simulator $\fSim^\fAdv$ defined as follows. (Since we mix the necessary definitions with the actual operations we use underlines to remark the actual operations of $\fSim^{\fAdv}$.)
\begin{enumerate}
\item   \underline{Pick a random $i\in [N]$}.\par \underline{Sample $rec\in \Domain(\brec^{\leq N})$ with probability $p_{rec}$, and keep the terms with index in $[i-1]$}. Denote it as $rec^{<i}$ ($\in \Domain(\brec^{<i})$) and the appearance probability is $p_{rec^{<i}}$.\par
 The purified overall state by the end of this step is
\begin{equation}\label{eq:269uu}\sqrt{p_{rec^{<i}}}\sum_{i\in [N]}\frac{1}{\sqrt{N}}\ket{i}\otimes\sum_{rec^{<i}\in \Domain(\brec^{<i})}\ket{rec^{<i}}\otimes\ket{0}\end{equation}
\item  
As implicitly used in \eqref{eq:61}, $\ket{\varphi^{i-1}_{rec^{<i}}}$ is efficiently preparable by a unitary followed by a re-scaling by Fact \ref{fact:simfa}. (That is, there exists a polynomial time simulator $\fSim^{\fAdv_{rec^{<i},<i}}_{rec^{<i},<i}$  such that
\begin{equation}\label{eq:298in}\ket{\varphi^{i-1}_{rec^{<i}}}\approx_{\fneg(\kappa)}\sqrt{p_{rec^{<i}}}\fSim^{\fAdv_{rec^{<i},<i}}_{rec^{<i},<i}\ket{0}\end{equation}
)\par
 Suppose $\fAdv_{rec^{<i},i}$ is the operation of $\fAdv$ during the $i$-th round when the previous history record is $rec^{<i}$. Recall $(rec^{<i},i)\in T$, there exists a simulator $\fSim^{\fAdv_{rec^{<i},i},\fAdv_{rec^{<i},<i}}_{rec^{<i},i}$ that satisfies \eqref{eq:61}.\footnote{Following the convention on the superscript of $\fSim$, where $\fSim^{\fAdv,O}$ stands for the simulator corresponding to the protocol execution against adversary $\fAdv$ and initial state $O\ket{0}$, the second part of the superscript here should be $\fSim^{\fAdv_{rec^{<i},<i}}_{rec^{<i},<i}$. Here we simply use $\fAdv_{rec^{<i},<i}$ for simplicity.}\par
\ul{Controlled by the values of $rec^{<i}$ and $i$, if $(rec^{<i},i)\not\in T$, set all the flags to $\ffail$ and stop the construction of $\fSim$. Otherwise apply $\fSim^{\fAdv_{rec^{<i},i},\fAdv_{rec^{<i},<i}}_{rec^{<i},i}$ on the server side of} $\ket{\text{Equation \eqref{eq:target}}}$. 
  Denote the post-execution state
\begin{equation}\label{eq:299in}\ket{\tilde\varphi^{i}_{rec^{<i}}}:=\sqrt{p_{rec^{<i}}}\fSim^{\fAdv_{rec^{<i},i},\fAdv_{rec^{<i},i}}_{rec^{<i},i}\ket{\text{Equation \eqref{eq:target}}}\end{equation}

The overall state by the end of this step when the round counter is $i$, is
\begin{equation}\label{eq:269u}\ket{\tilde\varphi^{i}}:=\sum_{rec^{<i}\in \Domain(\brec^{<i})}\ket{rec^{<i}}\otimes\ket{\tilde\varphi^{i+1}_{rec^{<i}}}\end{equation}
\item \ul{Controlled by the values of $rec^{<i}$ and $i$, apply} \begin{equation}\label{eq:316}\fSim^{\fAdv_{rec^{<i},>i}}_{rec^{<i},>i}=\fDisgard(\bS_{\bbC})\fServersim(\fpreRSPVTemp^{\fAdv_{rec^{<i},> i}}_{> i})\end{equation}
\ul{on $\eqref{eq:269u}$, where}:
\begin{itemize}
\item  $\fAdv_{rec^{<i},>i}$ is the operation of $\fAdv$ starting from the $i+1$-th round when the previous history record by the end of the $i-1$ round is $rec^{<i}$.\footnote{Here the record in the $i$-th round is implicit in $\ket{\tilde\varphi^{i+1}_{rec^{<i}}}$ in \eqref{eq:269u}, and the adversary has access to it; but we make the $rec^{<i}$ part explicit in \eqref{eq:269u} thus we also need to make it explicit in the subscript of the adversary's operation.}
\item $\fServersim(\fPrtl)$ is a server-side operation that (1) first initialize a register $\bS_{\bbC}$ that has the same size as the client-side registers initialized in $\fPrtl$; (2) do all the operations between the server-side of the initial state and $\bS_{\bbC}$ instead of the real client. This transformation transforms an interactive protocol to a server-side operations that simulates the server-side view of the original protocol.
\end{itemize}

 Denote the final state as $$\ket{\tilde\varphi^\prime}=\sum_{rec^{<i}\in \Domain(\brec^{<i})}\ket{rec^{<i}}\otimes\fSim^{\fAdv_{rec^{<i},>i}}_{rec^{<i},>i}\ket{\tilde\varphi^i_{rec^{<i}}}$$
\end{enumerate}
We will prove $\fSim$ satisfies \eqref{eq:finalsim}.\par
 By \eqref{eq:61}\eqref{eq:298in}\eqref{eq:299in}, for any $(rec^{<i},i)\in T$:
\begin{equation}\label{eq:270}\ket{rec^{<i}}\otimes\Pi_{\fcomp}^{\btype_i}\Pi^{\bflag_{<i}}_{\fpass}\ket{\varphi^i_{rec^{<i}}}\approx^{ind}_{(1/10+\fneg(\kappa))\sqrt{p_{\fcomp}}\sqrt{p_{rec^{<i}}}} \ket{rec^{<i}}\otimes\sqrt{p_{\fcomp}}\Pi^{\bflag_{<i}}_{\fpass}\ket{\tilde\varphi^i_{rec^{<i}}}\end{equation}
Thus applying \eqref{eq:270} for each $(rec^{<i},i)\in T$ and summing up the errors we get:
	\begin{align}
	&\Pi_{(rec^{<i},i)\in T}\sum_{i\in [N]}\frac{1}{\sqrt{N}}\ket{i}\otimes\sum_{rec^{<i}\in \Domain(\brec^{<i})}\ket{rec^{<i}}\otimes\Pi_{\fcomp}^{\btype_i}\Pi^{\bflag_{<i}}_{\fpass}\ket{\varphi^i_{rec^{<i}}}\label{eq:302tc}\\
\approx_{(1/10+\fneg(\kappa))\sqrt{p_{\fcomp}}}^{ind} & \Pi_{(rec^{<i},i)\in T}\sum_{i\in [N]}\frac{1}{\sqrt{N}}\ket{i}\otimes\sum_{rec^{<i}\in \Domain(\brec^{<i})}\ket{rec^{<i}}\otimes\sqrt{p_{\fcomp}}\Pi^{\bflag_{<i}}_{\fpass}\ket{\tilde\varphi^i_{rec^{<i}}}\label{eq:303tc}
\end{align}
Then let's consider the application of $\fSim_{rec^{<i},>i}^{\fAdv_{rec^{<i},>i}}$. For any $rec^{<i}$, by construction \eqref{eq:316}:
\begin{equation}\label{eq:274}\fSim_{rec^{<i},>i}^{\fAdv_{rec^{<i},>i}}\ket{\varphi^i_{rec^{<i}}}= \fpreRSPVTemp^{\fAdv_{rec^{<i},> i}}_{> i}\ket{\varphi^i_{rec^{<i}}}\end{equation}
 Combining it with \eqref{eq:302tc}\eqref{eq:303tc} we get 
\begin{align}
	&{\scriptsize\Pi_{(rec^{<i},i)\in T}\sum_{i\in [N]}\frac{1}{\sqrt{N}}\ket{i}\otimes\sum_{rec^{<i}\in \Domain(\brec^{<i})}\ket{rec^{<i}}\otimes\fpreRSPVTemp^{\fAdv_{rec^{<i},> i}}_{>i}(\Pi_{\fcomp}^{\btype_i}\Pi^{\bflag_{<i}}_{\fpass}\ket{\varphi^i_{rec^{<i}}})}\\
	\approx_{(1/10+\fneg(\kappa))\sqrt{p_{\fcomp}}}^{ind} & \Pi_{(rec^{<i},i)\in T}\sum_{i\in [N]}\frac{1}{\sqrt{N}}\ket{i}\otimes\sum_{rec^{<i}\in \Domain(\brec^{<i})}\ket{rec^{<i}}\otimes\sqrt{p_{\fcomp}}\fSim_{rec^{<i},>i}^{\fAdv_{rec^{<i},>i}}\Pi^{\bflag_{<i}}_{\fpass}\ket{\tilde\varphi^i_{rec^{<i}}}
\end{align}
Which implies both sides are close if we focus on the  space that all the flag registers are in value $\fpass$ (denoted as $\Pi_\fpass$):
\begin{align}
	&\Pi_{(rec^{<i},i)\in T}\Pi_{\fpass}\sum_{i\in [N]}\frac{1}{\sqrt{N}}\ket{i}\otimes\sum_{rec^{<i}\in \Domain(\brec^{<i})}\ket{rec^{<i}}\otimes\fpreRSPVTemp^{\fAdv_{rec^{<i},> i}}_{> i}(\Pi_{\fcomp}^{\btype_i}\ket{\varphi^i_{rec^{<i}}})\label{eq:308ym}\\
	\approx_{(1/10+\fneg(\kappa))\sqrt{p_{\fcomp}}}^{ind} & \Pi_{(rec^{<i},i)\in T}\Pi_{\fpass}\sum_{i\in [N]}\frac{1}{\sqrt{N}}\ket{i}\otimes\sum_{rec^{<i}\in \Domain(\brec^{<i})}\ket{rec^{<i}}\otimes\sqrt{p_{\fcomp}}\fSim_{rec^{<i},>i}^{\fAdv_{rec^{<i},>i}}\ket{\tilde\varphi^i_{rec^{<i}}}\label{eq:307ym}
\end{align}
Now we could apply \eqref{eq:251b} to \eqref{eq:308ym} and get
\begin{equation}
	\eqref{eq:308ym}\approx_{0.01+\fneg(\kappa)}\Pi_{\fpass}\sum_{i\in [N]}\frac{1}{\sqrt{N}}\ket{i}\otimes\sum_{rec\in \Domain(\brec^{\leq N})}\ket{rec}\otimes\Pi_{\fcomp}^{\btype_i}\ket{\varphi^N_{rec}}
\end{equation}
And by the construction of $\fSim$ we have the simulated state on the space of $(\bbI-\Pi_{(rec^{<i},i)\in T_i})\Pi_{\fpass}$ has norm $0$. Thus
\begin{equation}
	\eqref{eq:307ym}=\Pi_{\fpass}\sum_{i\in [N]}\frac{1}{\sqrt{N}}\ket{i}\otimes\sum_{rec\in \Domain(\brec^{\leq N})}\ket{rec}\otimes\sqrt{p_{\fcomp}}\ket{\tilde\varphi^N_{rec}})
\end{equation}

which together implies 
$$\Pi_{\fcomp}^{\btype_i}\Pi_{\fpass}\sum_{i\in [N]}\frac{1}{\sqrt{N}}\ket{i}\otimes\ket{\varphi^\prime}\approx^{ind}_{0.11+\fneg(\kappa)}\sqrt{p_{\fcomp}}\Pi_{\fpass}\sum_{i\in [N]}\frac{1}{\sqrt{N}}\ket{i}\otimes\ket{\tilde\varphi^\prime}$$
This completes the proof.
\end{proof}
\subsubsection{Step 2: handing the case where $\fcomp$ round is not reached}
In this section we construct $\fRSPV$ protocol from the $\fpreRSPVTemp$ protocol in the last section.
\begin{mdframed}
\begin{prtl}[$\fRSPV$]\label{prtl:rspvr}Inputs: security parameter $\kappa$, output number $L$.\par
Take $N=10/p_{\fcomp}=100$.
\begin{enumerate}
	\item For $i\in [N]$:\begin{enumerate}
	\item Both parties run $\fpreRSPVTemp$.\begin{itemize}
	\item If this subprotocol call outputs $\fcomp$ in its round type register, break out of the loop.
	\item Otherwise both the client and the server disgard all outputs (client side phases and server-side states)  of this round.
	\end{itemize}
	
	\end{enumerate}
	\item If the loop in the first step does not terminate by the breaking out command, the client outputs $\ffail$ in the flag register. If any subprotocol call returns $\ffail$ as its flag, the client outputs $\ffail$ too.\par
	Otherwise, suppose the first step breaks out when the round counter is $i$. The client and the server use the keys and states generated in the $i$-th round as the output keys and the output state.
\end{enumerate}	
\end{prtl}
	
\end{mdframed}
\paragraph{Correctness}If the server is honest, the protocol succeeds with probability $\geq 0.98-\fneg(\kappa)$ and the joint state of the client and the server is \eqref{eq:cqtarget} (up to negligible distance) in the end.
\begin{proof}The cases where an honest server could result in $\ffail$ are:
\begin{itemize}
	\item In the underlying $\fpreRSPVTemp$ protocol the client could possibly output $\ffail$ against an honest server. This happens with probability $\leq N\cdot (10^{-4}+\fneg(\kappa))\leq 0.01+\fneg(\kappa)$.
	\item The probability that all the calls of $\fpreRSPVTemp$ have round type $\perp$ is $\leq (1-p_{\fcomp})^N\leq 0.01$.
\end{itemize}
If these two cases do not happen, the client and the honest server get the output keys and states \eqref{eq:cqtarget} by the correctness of $\fpreRSPVTemp$. This completes the proof.\end{proof}
Then we could prove Protocol \ref{prtl:rspvr} satisfies Definition \ref{defn:rspvv}, the verifiability requirement for $\fRSPV$ protocol:
\begin{thm}[Verifiability of Protocol \ref{prtl:rspvr}, repeat of Definition \ref{defn:rspvv}]\label{thm:rspvv}
	For any efficient adversary $\fAdv$ 
	there exists a server-side operator $\fSim^{\fAdv}$ such that:
	\begin{equation}\label{eq:211p}\Pi_{\fpass}\fRSPV^\fAdv(1^L,1^\kappa)\ket{0}\approx^{ind}_{0.15+\fneg(\kappa)}\Pi_{\fpass} \fSim^\fAdv \ket{\text{Equation \eqref{eq:target}}}\end{equation}
\end{thm}
\begin{proof}

Expand the left hand side of  \eqref{eq:211p} by the value of the register that stores the round counter $i$ when the protocol terminates:
\begin{align}
	&\Pi_{\fpass}\fRSPV^\fAdv(1^L,1^\kappa) \ket{0}\label{eq:258}\\
	=&\sum_{i\in [N]}\underbrace{\ket{i}}_{\text{round counter}}\otimes\Pi_{\fpass}^{\bflag_i}\Pi_{\fcomp}^{\btype_i}\fpreRSPVTemp^{\fAdv_i} (\Pi_{\fpass}^{\bflag_{<i}}\Pi_{\perp}^{\btype_{<i}}\fRSPV_{<i}^{\fAdv_{<i}}\ket{0})\\
	&+\ket{N}\otimes\Pi_{\fpass}^{\bflag_i}\Pi_{\perp}^{\btype_i}\fpreRSPVTemp^{\fAdv_i} (\Pi_{\fpass}^{\bflag_{<i}}\Pi_{\perp}^{\btype_{<i}}\fRSPV_{<i}^{\fAdv_{<i}}\ket{0})\label{eq:260}
	\end{align}
	where $\Pi_{\fpass}^{\bflag_{<i}}$ is the projection onto the space that the flag registers for round $1$ to $i-1$ all have value $\fpass$, $\Pi_{\perp}^{\btype_{<i}}$ is the projection onto the space that the round type registers for round $1$ to $i-1$ all have value $\perp$, $\fRSPV_{<i}$ is the protocol from round $1$ to $i-1$, $\fAdv_{<i}$ is the part of $\fAdv$ by the end of the $i-1$-th round. And $\Pi_{\fpass}^{\bflag_i}$, $\Pi_{\perp}^{\btype_i}$, $\fAdv_i$ are defined similarly.\par
	The norm of term \eqref{eq:260} is upper bounded by $(1-p_{\fcomp})^N\leq 0.01$. Thus
	\begin{align}
		\eqref{eq:258}\approx_{0.01}&\sum_{i\in [N]}\ket{i}\otimes\Pi_{\fpass}^{\bflag_i}\Pi_{\fcomp}^{\btype_i}\fpreRSPVTemp^{\fAdv_i} (\Pi_{\fpass}^{\bflag_{<i}}\Pi_{\perp}^{\btype_{<i}}\fRSPV_{<i}^{\fAdv_{<i}}\ket{0})\label{eq:257}
	\end{align}
	Then we could apply Theorem \ref{thm:rspvprime} to each term of \eqref{eq:257} which leads to an isometry $\fSim^{i,\fAdv}$ for each $i\in [N]$:
	\begin{align}\forall i\in [N],\quad & \ket{i}\otimes\Pi_{\fpass}^{\bflag_i}\Pi_{\fcomp}^{\btype_i}\fpreRSPVTemp^{\fAdv_i} (\Pi_{\fpass}^{\bflag_{<i}}\Pi_{\perp}^{\btype_{<i}}\fRSPV_{<i}^{\fAdv_{<i}}\ket{0})\\
	\approx^{ind}_{0.11\sqrt{p_i}+\fneg(\kappa)}&\ket{i}\otimes \sqrt{p_i}\sqrt{p_{\fcomp}}\Pi_{\fpass}^{\bflag_i}\fSim^{i,\fAdv} \ket{\text{Equation \eqref{eq:target}}}\end{align}
	where 
	$$p_i:=|\Pi_{\fpass}^{\bflag_{<i}}\Pi_{\perp}^{\btype_{<i}}\fRSPV_{<i}^{\fAdv_{<i}}\ket{0}|^2$$
	 Then summing them up and summing up the error terms we could continue from \eqref{eq:257}:
	\begin{align}
	&\text{Right hand side of }\eqref{eq:257}\label{eq:335}\\
	\approx^{ind}_{0.11+\fneg(\kappa)} & \sum_{i\in [N]}\sqrt{p_i}\sqrt{p_{\fcomp}}\ket{i}\otimes \fSim^{i,\fAdv} \ket{\text{Equation \eqref{eq:target}}}\label{eq:214o}
\end{align}	
where \eqref{eq:214o} defines a simulator $\fSim$ that samples $i$ with probability $p_ip_{\fcomp}$ and runs the corresponding $\fSim^{i,\fAdv}$. Combining \eqref{eq:257}\eqref{eq:335}\eqref{eq:214o} completes the proof.
\end{proof}

\subsection{From RSPV to CVQC}
Now we will construct a CVQC protocol from the RSPV protocol in the last subsection. 
\begin{mdframed}
\begin{prtl}[CVQC]\label{prtl:cvqc}
	Input: circuit $C$ to be verified, which determines the gadget number $L=O(|C|)$ under Theorem \ref{thm:gauvbqc}. Suppose the security parameter is $\kappa$.
	\begin{enumerate}
		\item  Both parties run $\fRSPV(1^L,1^\kappa)$.
		\item Both parties run the gadget-assisted verification protocol (Theorem \ref{thm:gauvbqc}) with the gadgets above. 
	\end{enumerate}
\end{prtl}
	
\end{mdframed}
We have the following for the protocol, which proves Theorem \ref{thm:main}.
\paragraph{Completeness} The protocol has completeness $\frac{2}{3}$.
\begin{proof}
There are two cases where the client in the protocol will output $\ffail$ on a yes instance against an honest server:
\begin{itemize}
	\item The call to the $\fRSPV$ protocol in the first step might fail with probability $\frac{1}{10}+\fneg(\kappa)$, as defined in Definition \ref{defn:rspv}.
	\item The underlying gadget-assisted verification protocol in the third step is allowed to have a failing probability of $\frac{1}{10}$.
\end{itemize}  The total probability is $\leq \frac{1}{3}+\fneg(\kappa)$.	
\end{proof}

\paragraph{Soundness} The protocol has soundness $\frac{1}{3}$ in QROM against BQP adversaries.
\begin{proof}
After the application of $\fRSPV$ against adversary $\fAdv$, define $\ket{\varphi^\prime}$ as the output state:
$$\ket{\varphi^\prime}=\fRSPV^\fAdv(1^L,1^\kappa)\ket{0}$$
 By Theorem \ref{thm:rspvv}, there exists a server-side isometry $\fSim$ such that
\begin{equation}\label{eq:cvqc1}\Pi_{\fpass}\ket{\varphi^\prime}\approx_{0.15+\fneg(\kappa)}^{ind}\Pi_{\fpass} \fSim^\fAdv\ket{\text{Equation \eqref{eq:target}}}\end{equation}
By the soundness property of the underlying gadget-assisted protocol (by Theorem \ref{thm:gauvbqc}, denote as $\fGAUVBQC$) we know for any $\fAdv^\prime$ which denotes the adversary for the gadget-assisted verification step:
$$\forall C,o, \Pr[C\ket{0}=o]\leq \frac{1}{100}:|\Pi_{\fpass}\fGAUVBQC^{\fAdv^\prime}(C,o,\Pi_{\fpass}\fSim^\fAdv\ket{\text{Equation \eqref{eq:target}}})|^2\leq \frac{1}{50}$$
Substitute \eqref{eq:cvqc1} we get
$$\forall C,o, \Pr[C\ket{0}=o]\leq \frac{1}{100}:|\Pi_{\fpass}\fGAUVBQC^{\fAdv^\prime}(C,o,\ket{\varphi^\prime})|^2\leq 0.33+\fneg(\kappa)$$
which completes the proof.
\end{proof}
\paragraph{Efficiency} The protocol runs in time $O(\fpoly(\kappa)|C|)$.

By this time we complete the proof of Theorem \ref{thm:main}.

\appendix
\section{Missing Proofs By Section \ref{sec:4}}\label{app:a}
\begin{proof}[Proof of Fact \ref{fact:3}]This is implied by
	$$|\ket{\varphi}-\ket{\phi}|^2=2(|\ket{\varphi}|^2+|\ket{\phi}|^2)-|\ket{\varphi}+\ket{\phi}|^2$$
\end{proof}
\begin{proof}[Proof of Fact \ref{fact:otpro}]
	This is because state
	$$\frac{1}{\sqrt{8}}\sum_{\theta\in \{0,1\cdots 7\}}\ket{\theta}$$
	is invariant under the operator.
\end{proof}
\begin{proof}[Proof of Fact \ref{fact:cophtool}]This is because for any $sum\in \{0,1\cdots 7\}$, state
	$$\sum_{c_0c_1c_2\cdots c_N\in \cC,\tSUM(c_0c_1c_2\cdots c_N)=sum}\ket{c_0c_1c_2\cdots c_N}$$
	is invariant under the operator.
\end{proof}
\begin{proof}[Proof of Fact \ref{fact:bntesttool1}]
	Suppose $\vec{c}=(c_1,c_2\cdots c_D)$, $\vec{d}=(d_1,d_2\cdots d_D)$. 
	Denote $S$ as the set of index $i$ such that $c_i\leq d_i$. 
	Then
	$$\vec{c^\prime}\approx_{\epsilon_1} \vec{d}\Rightarrow \sum_{i\in S}|c_i-d_i|^2\leq \epsilon_1^2$$
	$$\vec{c}\approx_{\epsilon_2}\vec{d^\prime}\Rightarrow \sum_{i\not\in S}|c_i-d_i|^2\leq \epsilon_2^2$$
	Summing them up completes the proof.
\end{proof}
\begin{proof}[Proof of Fact \ref{fact:bntesttool}]
	$$\frac{1}{{|D|}}\sum_{d_1\in D}\sum_{d_2\in D}|c_{d_1}-c_{d_2}|^2=2\sum_{d\in D}(1+\frac{1}{D})c_d^2-2(\frac{1}{\sqrt{|D|}}\sum_{d\in D}c_d)^2\approx_\epsilon 0$$
	$$\Rightarrow \frac{1}{\sqrt{|D|}}\sum_{d\in D}c_d\approx_{\frac{3}{2}\epsilon/c+\frac{1}{D}} c$$
	$$\Rightarrow \sum_{d\in D} |c_{d}-\frac{1}{\sqrt{D}}c|^2=c^2-2(\frac{1}{\sqrt{|D|}}\sum_{d\in D}c_d)c+c^2\approx_{4\epsilon/c+\frac{2}{D}} 0$$
\end{proof}
\begin{proof}[Proof of Fact \ref{fact:simfa}]
	Define $\tilde\fSim$ as the operation that runs $O$ and measures register $\bbC$ until $c$ appears. Cuting-off the repetition in $\tilde\fSim$ by $\kappa$ and purifying the operation completes the construction.
\end{proof}
Lemma \ref{lem:chernoff} is the usual Chernoff's bound.
\begin{proof}[Proof of Corollary \ref{cor:chernoff}]
We only need to prove
$$\Pr[(\text{$\forall i$, sample history by time $i$ is not in $S_i$})\land (|\{i:s_i=1\}|\geq (1+\delta)pN)]\leq e^{-\delta^2N/4}$$
	which comes from Lemma \ref{lem:chernoff}.
\end{proof}

A proof of Lemma \ref{lem:padro} is given in \cite{jiayu20} using standard techniques. We give a proof here for self-containment.
\begin{proof}[Proof of Lemma \ref{lem:padro}]
	Define $\Pi$ as the projection onto strings with prefix in $\bpads$ in the query input register. Then consider $H(\bbI-\Pi)$ which is the operation that projecting out the space that have a prefix in $\bpads$ from the query.\par
	Define $O_k$ as the operation that the last $k$ queries of $O$ to the random oracle are replaced by $H(\bbI-\Pi)$.\par
	We can prove, for all $k\in [|O|]$ ($|O|$ denotes the number of queries) :
	\begin{equation}\label{eq:app1}O_k\sum_{pads\in \Domain(\bpads)}\frac{1}{\sqrt{|\Domain(\bpads)|}}\ket{pads}\otimes\ket{\varphi}\approx_{\fneg(\kappa)} O_{k-1}\sum_{pads\in \Domain(\bpads)}\frac{1}{\sqrt{|\Domain(\bpads)|}}\ket{pads}\otimes\ket{\varphi}.\end{equation}
	The reason is, \eqref{eq:app1} could be reduced to
	\begin{equation}\label{eq:222}|\Pi\sum_{pads\in \Domain(\bpads)}\frac{1}{\sqrt{|\Domain(\bpads)|}}\ket{pads}\otimes U_{k}\ket{\varphi}|\approx_{\fneg(\kappa)}0\end{equation}
	where $U$ is the part of $O$ just before the $k$-th query counting from the last to the first. Note by the time of the $k$-th query (counting from the last to the first) $pads$ still contains unused freshly new randomness. Thus \eqref{eq:222} is true by direct calculation.\par
	Summing up \eqref{eq:app1} for each $k\in [|O|]$. Then $\ket{\tilde\varphi}:=O_{|O|}\ket{\varphi}$ gives the state we want.
\end{proof}
\section{Missing Proofs in Section \ref{sec:5}}\label{app:mp6}
\begin{proof}[Proof of Lemma \ref{lem:collapse}]
	Substitute $\ket{\varphi}=\ket{\varphi_0}+\ket{\varphi_1}$ and expand we get
	$$|\Pi_0^{\bS}O\ket{\varphi}|^2=|\Pi_0^{\bS}O\ket{\varphi_0}|^2+|\Pi_0^{\bS}O\ket{\varphi_1}|^2+\bra{\varphi_1}O^\dagger\Pi_0^{\bS}O\ket{\varphi_0}+\bra{\varphi_0}O^\dagger\Pi_0^{\bS}O\ket{\varphi_1}$$
	Since in $\ket{\varphi_0}$ $x^{(0)}_0$ is held by the server classically and in $\ket{\varphi_1}$ $x^{(0)}_1$ is held  classically, by the claw-free property the last two terms are both negligible. This completes the proof.
\end{proof}
\begin{proof}[Proof of Lemma \ref{lem:5.2a}]
	Define $H^\prime$ as the blinded oracle where $\{0,1\}^\kappa||\bx_{1-b}^{(i)}||\cdots$ and $\{0,1\}^{2\kappa}||\bx_{1-b}^{(i)}||\cdots$ are blinded. Define $\fAdv^t$ as the operation where the first $t$ queries of $\fAdv$ are replaced by queries to $H^\prime$. First we could prove
	\begin{equation}\label{eq:235}
\forall t\in [|\fAdv|],\fAdv^t(\ket{\varphi}\odot \llbracket \fAuxInf\rrbracket)\approx_{\fneg(\kappa)}\fAdv^{t-1}(\ket{\varphi}\odot \llbracket \fAuxInf\rrbracket)\end{equation}
\eqref{eq:235} is proved as follows. Define $U^{t-1}$ as the operation in $\fAdv^{t-1}$ by the time of the $t$-th query. Define $\bS_{roq}$ to be the register used to hold random oracle queries. Then \eqref{eq:235} is reduced to proving
\begin{equation}\label{eq:263}|\Pi^{\bS_{roq}}_{\in \{0,1\}^\kappa||\bx_{1-b}^{(i)}||\cdots \cup \{0,1\}^{2\kappa}||\bx_{1-b}^{(i)}||\cdots}U^{t-1}(\ket{\varphi}\odot \llbracket \fAuxInf\rrbracket)|\leq \fneg(\kappa)\end{equation}
Use $\bpads$ to denote the set of registers that holds the random pads used in $\fAuxInf$. Applying Lemma \ref{lem:padro} we get $\ket{\tilde\varphi}$ that does not  depend on $H(\bpads||\cdots)$ such that
$$\ket{\tilde\varphi}\approx_{\fneg(\kappa)}\sum_{pads\in \Domain(\bpads)}\frac{1}{\sqrt{|\Domain(\bpads)|}}\ket{pads}\otimes\ket{\varphi}$$
This implies \eqref{eq:263} is reduced to proving
\begin{equation}\label{eq:397}|\Pi^{\bS_{roq}}_{\in \{0,1\}^\kappa||\bx_{1-b}^{(i)}||\cdots \cup \{0,1\}^{2\kappa}}U^{t-1}(\ket{\tilde\varphi}\odot \llbracket \fAuxInf\rrbracket)|\leq \fneg(\kappa)\end{equation}
since entries of $H$ in the form of $\bpads||p_{pre}^{(t)}||\bx_{1-b}^{(i)},\bpads||\bx_{1-b}^{(i)}||p_{post}^{(t)}$ are never queried by the preparation of $\ket{\tilde\varphi}$ and $U^{t-1}$, the left hand side is equal to
\begin{equation}\label{eq:264}|\Pi^{\bS_{roq}}_{\in \{0,1\}^\kappa||\bx_{1-b}^{(i)}||\cdots \cup \{0,1\}^{2\kappa}}U^{t-1}(\ket{\tilde\varphi}\odot \llbracket \$\rrbracket)|\end{equation}
where $\llbracket \$\rrbracket$ denotes random strings of the same size of $\llbracket \fAuxInf\rrbracket)$. Then \eqref{eq:264} is $\leq \fneg(\kappa)$ by claw-freeness. Now \eqref{eq:235} is proved.\par
Now summing up \eqref{eq:235} for each $t\in[|\fAdv|]$ we have
	\begin{equation}\label{eq:399}
\fAdv(\ket{\varphi}\odot \llbracket \fAuxInf\rrbracket)\approx_{\fneg(\kappa)}\fAdv^\prime(\ket{\varphi}\odot \llbracket \fAuxInf\rrbracket)\end{equation}
where $\fAdv^\prime$ is the operation where each query in $\fAdv$ is replaced by query to $H^\prime$. Now by the same reason as \eqref{eq:397} to \eqref{eq:264} we know
$$|\Pi_{\bx_{1-b}^{(i)}}\fAdv^\prime(\ket{\varphi}\odot \llbracket \fAuxInf\rrbracket)|\leq \fneg(\kappa)$$
which together with \eqref{eq:399} completes the proof.
\end{proof}
\begin{proof}[Proof of Lemma \ref{lem:5.2b}]
	Define $\fAdv^t$ as the operation where the first $t$ queries of $\fAdv$ are replaced by queries to $H^\prime$. The problem is reduced to
	$$\fAdv^t\ket{\varphi}\approx_{\fneg(\kappa)}\fAdv^{t-1}\ket{\varphi}$$
	as the proof of Lemma \ref{lem:5.2a}, this is further reduced to
	$$|\Pi^{\bS_{roq}}_{\cdots ||\bx_{1-b}^{(i)}||\cdots}U^{t-1}\ket{\varphi}|\leq \fneg(\kappa)$$
	which holds by claw-freeness.
\end{proof}
\section{Missing Proofs in Section \ref{sec:6}}\label{app:b}
To prove Lemma \ref{lem:prepht}, we first prove the following lemma. 
\begin{lem}\label{lem:C1}
Suppose the client holds a key pair in register $\bK=\{\bx_0,\bx_1\}$. Below we use $|\bx|$ to denote the length of keys in $\bK$. Suppose the purified joint state $$\ket{\tilde\varphi}=\sum_{pad\in \{0,1\}^\kappa}\frac{1}{\sqrt{2^\kappa}}\underbrace{\ket{pad}}_{\bpad\text{ in transcript}}\otimes\ket{\tilde\varphi_{pad}}$$ does not depend on $H(\bpad||\{0,1\}^{|\bx|})$. Suppose an adversary $\fAdv^{blind}$ only queries the blinded oracle $H^\prime$ where entries $\{0,1\}^{\kappa}||\bx_b$ are blinded, for some $b\in \{0,1\}$. Then \begin{equation}\label{eq:293r}|\Pi_{\eqref{eq:htt}=0}^\bd\Pi_{\neq 0}^{\text{last $\kappa$ bits of }\bd}\fHadamardTest_{\geq 2}^{\fAdv^{blind}}(\bK;1^\kappa) \ket{\tilde\varphi}|= \frac{1}{\sqrt{2}}|\Pi_{\neq 0}^{\text{last $\kappa$ bits of }\bd}\fHadamardTest_{\geq 2}^{\fAdv^{blind}}(\bK;1^\kappa)\ket{\tilde\varphi}|\end{equation}
	where the subscript ``$\geq 2$'' means the first step of $\fHadamardTest$ (sampling a random pad) has already been done and thus skipped. 
\end{lem}
The proof is by a direct calculation similar to the proof of Lemma A.7.1 in \cite{jiayu20}.
\begin{proof}
	Without loss of generality assume $b=0$. Since $\ket{\tilde\varphi}$ does not depend on $H(\bpad||\{0,1\}^{|\bx|})$ and $\fAdv^{blind}$ does not query $H(\{0,1\}^{\kappa}||\bx_b)$, $$\text{$\fHadamardTest_{\geq 2}^{\fAdv^{blind}}(\bK;1^\kappa) \ket{\tilde\varphi}$ does not depend on $\bH(\bpad||\bx_0)$}.$$ Suppose the server's response in $\fHadamardTest$ is written into register $\bd=(\bd_1,\bd_2)$ where $\bd_2$ corresponds to the last $\kappa$ bits. Then we can write the post-execution state as
	\begin{equation}\label{eq:396ys}\fHadamardTest_{\geq 2}^{\fAdv^{blind}}(\bK;1^\kappa) \ket{\tilde\varphi}=\sum_{(d_1,d_2)\in \Domain(\bd)}\ket{d_1}\ket{d_2}\ket{\chi_{d_1,d_2}}\end{equation} 
	Then we can expand the left hand side of \eqref{eq:293r} as follows:
	\begin{align}     & |\Pi_{\eqref{eq:htt}=0}^\bd\Pi_{\neq 0}^{\text{last $\kappa$ bits of }\bd}\fHadamardTest_{\geq 2}^{\fAdv^{blind}}(\bK;1^\kappa) \ket{\tilde\varphi}|                                                                                                       \\
	&\text{(substitute \eqref{eq:396ys} and use the condition that $\ket{\chi_{d_1,d_2}}$ does not depend on $\bH(\bpad||\bx_0)$)}\\
		=    & \sqrt{\bE_{H(\bpad||\bx_0)}|\Pi_{d_1\cdot \bx_0+d_2\cdot H(\bpad||\bx_0)=d_1\cdot \bx_1+d_2\cdot \bH(\bpad||\bx_1)}\Pi_{d_2\neq 0}\sum_{d_1d_2} \ket{d_1}\otimes \ket{d_2}\otimes\ket{\chi_{d_1,d_2}}|^2} \\
		= & \sqrt{\frac{1}{2}|\Pi_{d_2\neq 0}\sum_{d_1,d_2} \ket{d_1}\otimes \ket{d_2}\otimes\ket{\chi_{d_1,d_2}}|^2}\\
		=& \frac{1}{\sqrt{2}}|\Pi_{\neq 0}^{\text{last $\kappa$ bits of }\bd}\fHadamardTest_{\geq 2}^{\fAdv^{blind}}(\bK;1^\kappa)\ket{\tilde\varphi}|\end{align}
\end{proof}
Then we could prove Lemma \ref{lem:prepht}.
\begin{proof}[Proof of Lemma \ref{lem:prepht}]
	We will move towards Lemma \ref{lem:C1} step by step.
	\begin{enumerate}
		\item Suppose the random pad used in the Hadamard test is stored in register $\bpad$. By Lemma \ref{lem:padro} there exists $\ket{\tilde\varphi}$ such that $$\ket{\tilde\varphi}\approx_{\fneg(\kappa)}\frac{1}{\sqrt{2^\kappa}}\sum_{pad\in \{0,1\}^\kappa}\underbrace{\ket{pad}}_{\bpad}\otimes\ket{\varphi}$$ and $\ket{\tilde\varphi}$ does not depend on the value of $\bH(\bpad||\cdots)$.
		\item By Lemma \ref{lem:5.2b} we have
		$$\fHadamardTest^\fAdv(\bK;1^\kappa)\ket{\varphi}\approx_{\fneg(\kappa)} \fHadamardTest^{\fAdv^{blind}}(\bK;1^\kappa)\ket{\varphi}$$
		where $\fAdv^{blind}$ come from replacing all the queries in $\fAdv$ by queries to $\bH^{\prime}$, where $\bH^{\prime}$ is the blinded version of $\bH$ where entries $\{0,1\}^\kappa||\bK$ are blinded.
		\item We can apply Lemma \ref{lem:C1} to get
		$$|\Pi_{\eqref{eq:htt}=0}^\bd\Pi_{\neq 0}^{\text{last $\kappa$ bits of }\bd}\fHadamardTest^{\fAdv^{blind}}_{\geq 2}(\bK;1^\kappa)\ket{\tilde\varphi}|= \frac{1}{\sqrt{2}}|\Pi_{\neq 0}^{\text{last $\kappa$ bits of }\bd}\fHadamardTest^{\fAdv^{blind}}_{\geq 2}(\bK;1^\kappa)\ket{\tilde\varphi}|$$
				$$|\Pi_{\eqref{eq:htt}=1}^\bd\Pi_{\neq 0}^{\text{last $\kappa$ bits of }\bd}\fHadamardTest^{\fAdv^{blind}}_{\geq 2}(\bK;1^\kappa)\ket{\tilde\varphi}|= \frac{1}{\sqrt{2}}|\Pi_{\neq 0}^{\text{last $\kappa$ bits of }\bd}\fHadamardTest^{\fAdv^{blind}}_{\geq 2}(\bK;1^\kappa)\ket{\tilde\varphi}|$$

	\end{enumerate}
			Combining all these steps completes the proof.
\end{proof}
Then we could prove Corollary \ref{cor:prepht} from this lemma.
\begin{proof}[Proof of Corollary \ref{cor:prepht}]
For each $b\in \{0,1\}$, define
$$\ket{\varphi_{b,0}^\prime}:=\Pi_{\eqref{eq:htt}=0}^{\bd}\Pi_{\neq 0}^{\text{last $\kappa$ bits of }\bd}\fHadamardTest^\fAdv(\bK;1^\kappa)\ket{\varphi_b}$$
$$\ket{\varphi_{b,1}^\prime}:=\Pi_{\eqref{eq:htt}=1}^{\bd}\Pi_{\neq 0}^{\text{last $\kappa$ bits of }\bd}\fHadamardTest^\fAdv(\bK;1^\kappa)\ket{\varphi_b}$$
$$\ket{\varphi_{b,-}^\prime}:=\Pi_{= 0}^{\text{last $\kappa$ bits of }\bd}\fHadamardTest^\fAdv(\bK;1^\kappa)\ket{\varphi_b}$$
Then by Lemma \ref{lem:prepht} we have
\begin{equation}\label{eq:266}\forall b\in \{0,1\}, |\ket{\varphi_{b,0}^\prime}|\approx_{\fneg(\kappa)} \frac{1}{\sqrt{2}}|\ket{\varphi_{b,0}^\prime}+\ket{\varphi^\prime_{b,1}}|,|\ket{\varphi_{b,1}^\prime}|\approx_{\fneg(\kappa)} \frac{1}{\sqrt{2}}|\ket{\varphi_{b,0}^\prime}+\ket{\varphi^\prime_{b,1}}|\end{equation}
From the condition of Corollary \ref{cor:prepht} we get
\begin{equation}\label{eq:267}|\ket{\varphi_{0,0}^\prime}+\ket{\varphi_{1,0}^\prime}|\geq \sqrt{1-p}-\epsilon\end{equation}
On the other hand \begin{align}|\ket{\varphi_{0,0}^\prime}+\ket{\varphi_{1,0}^\prime}|\leq& \sqrt{2}\sqrt{|\ket{\varphi_{0,0}^\prime}|^2+|\ket{\varphi_{1,0}^\prime}|^2}\\\text{(By \eqref{eq:266})}\leq&\sqrt{|\ket{\varphi_{0,0}^\prime}+\ket{\varphi_{0,1}^\prime}|^2+|\ket{\varphi_{1,0}^\prime}+\ket{\varphi_{1,1}^\prime}|^2}+\fneg(\kappa)\\
\leq &\sqrt{1-\epsilon^2-|\ket{\varphi_{0,-}^\prime}|^2-|\ket{\varphi_{1,-}^\prime}|^2}+\fneg(\kappa)
\end{align}
 Comparing it with \eqref{eq:267} we get
 $$|\ket{\varphi_{0,-}^\prime}|^2+|\ket{\varphi_{1,-}^\prime}|^2\leq p+2\epsilon+\fneg(\kappa)$$
$$\Rightarrow \forall b\in \{0,1\}, |\ket{\varphi_{b,-}^\prime}|\leq \sqrt{p+2\epsilon}+\fneg(\kappa)$$
This completes the proof of \eqref{eq:78ol}.\par
For \eqref{eq:98y}, we can first bound
\begin{equation}
	 |\ket{\varphi_{0,0}^\prime}|^2+|\ket{\varphi_{1,0}^\prime}|^2\leq \frac{1}{2}(|\ket{\varphi_0}|^2+|\ket{\varphi_1}|^2)\leq \frac{1}{2}(1-\epsilon^2)
\end{equation}
%

which together with \eqref{eq:267} allows us to  apply Fact \ref{fact:3} and get
$$\ket{\varphi^{\prime}_{0,0}}\approx_{\sqrt{p+\epsilon}} \ket{\varphi^\prime_{1,0}}$$
which completes the proof of \eqref{eq:98y}.\par
Finally by the condition of Corollary \ref{cor:prepht} again we get, the norm of the failing space of the output state is $\leq \sqrt{p}$, thus
$$|\ket{\varphi^\prime_{0,1}}+\ket{\varphi^\prime_{1,1}}|\leq \sqrt{p}+\sqrt{\epsilon}$$
which completes the proof.
\end{proof}
\section{Missing Proofs in Section \ref{sec:8}}\label{sec:addms}
\begin{proof}[Proof of Lemma \ref{lem:multisbtre}]
	The proof is similar to the proof of Lemma \ref{lem:multisbt}. As the proof of Lemma \ref{lem:multisbt}, Lemma \ref{lem:multisbtre} is reduced to
	\begin{equation}\label{eq:103res}\sum_{b,b^\prime\in \{0,1\}^2,b\neq b^\prime}\Pi_{\bx^{(i)}_{b^\prime}}^{\bS_{bsh}}\circ\cP\circ \Pi_{\bx^{(i)}_b}^{\bS_{bsh}}\Pi_{\basishonest(\bK)}\ket{\varphi}\approx_{\fneg(\kappa)}0\end{equation}
	where $\cP$ is an efficient operation in $\cF_{blind}$. Then \eqref{eq:103res} follows by the claw-free property of $\ket{\varphi}$.
\end{proof}
\begin{proof}[Proof of Lemma \ref{lem:prephtre}]
	Suppose $\ket{\varphi}=\cR_1(\ket{\$_1}\otimes \fReviseRO\ket{\varphi^0})$. Define $\ket{\tilde\varphi}=\fPrtl^{\fAdv_0}\ket{\varphi^0}$. By Lemma \ref{thm:indofr1} we know the passing probability when the initial state is $\ket{\tilde\varphi}$, is $\geq 1-p-\fneg(\kappa)$. Applying Lemma \ref{lem:prepht} proves
		$$|\Pi_{\eqref{eq:htt}=0}^{\bd}\Pi_{\neq 0}^{\text{last $\kappa$ bits of }\bd}\fHadamardTest^\fAdv(\bK^\prime;1^\kappa)\ket{\tilde\varphi}|\approx_{\fneg(\kappa)}\frac{1}{\sqrt{2}}|\Pi_{\neq 0}^{\text{last $\kappa$ bits of }\bd}\fHadamardTest^\fAdv(\bK^\prime;1^\kappa)\ket{\tilde\varphi}|$$
	$$|\Pi_{\eqref{eq:htt}=1}^{\bd}\Pi_{\neq 0}^{\text{last $\kappa$ bits of }\bd}\fHadamardTest^\fAdv(\bK^\prime;1^\kappa)\ket{\tilde\varphi}|\approx_{\fneg(\kappa)} \frac{1}{\sqrt{2}}|\Pi_{\neq 0}^{\text{last $\kappa$ bits of }\bd}\fHadamardTest^\fAdv(\bK^\prime;1^\kappa)\ket{\tilde\varphi}|$$
	Applying Lemma \ref{thm:indofr1} again, together with the fact that $\fReviseRO$ commutes with all the operators in $\fHadamardTest^\fAdv$ and $\fPrtl^{\fAdv_0}$, completes the proof.
\end{proof}
\section{Missing Proofs in Section \ref{sec:10.1}}\label{sec:missing10}
We will prove the lemmas in Section \ref{sec:10.1} step by step. 
\subsection{Basic Inequalities}
Before going to the proofs, we give the following basic lemmas. These lemmas could be proved via basic linear algebra or calculus.
\begin{lem}\label{lem:cauchy}
	Suppose $\vec{a}=(a_1,a_2\cdots a_n)$, $\vec{b}=(b_1,b_2\cdots b_n)$. If $\frac{\vec{a}\cdot \vec{b}}{|\vec{a}||\vec{b}|}\geq 1-\epsilon$, there is
	$$\vec{a}/|\vec{a}|\approx_{\sqrt{2\epsilon}} \vec{b}/|\vec{b}|$$
\end{lem}
\begin{lem}\label{lem:calculus}
If $\cos^2(\frac{1}{4}\arccos \lambda)+	\cos^2(\frac{1}{4}\arccos (-\lambda))\approx_\epsilon 2\cos^2(\pi/8)$, $\epsilon<0.1$, there is $\lambda\approx_{4\sqrt{\epsilon}} 0$.
\end{lem}

\subsection{$3$-states Lemmas}
Below we use $\nRe$ to denote the real part of a complex number.
\begin{lem}\label{lem:d1}
	Suppose vectors $\ket{\varphi}$, $\ket{\psi}$ both have norm $\leq 1$. Then for any $\ket{\chi}$ with norm $\leq 1$,
	\begin{equation}\label{eq:207}|\ket{\varphi}+\ket{\chi}|^2+|\ket{\psi}+\ket{\chi}|^2\leq 8\cos^2(\frac{1}{4}\arccos \nRe(\frac{\braket{\varphi|\psi}}{|\ket{\varphi}|\cdot|\ket{\psi}|}))\end{equation}
	If $\ket{\psi}+\ket{\varphi}\neq 0$ the equality holds iff $|\ket{\varphi}|=|\ket{\psi}|=1$ and
	\begin{equation}\label{eq:286k}\ket{\chi}= (\ket{\psi}+\ket{\varphi})/|\ket{\psi}+\ket{\varphi}|\end{equation}
\end{lem}
\begin{proof}
	The left hand side of \eqref{eq:207} is less than  or equal to
	\begin{equation}\label{eq:286}4+2\nRe(\braket{\varphi|\chi})+2\nRe(\braket{\chi|\psi})\end{equation}
	Up to an isometry, suppose $\ket{\varphi}=h\ket{0}$ and $\ket{\psi}=(a+b\mi)\ket{0}+c\ket{1}$, $h,a,b,c\in \bbR_{\geq 0}$, $h\leq 1$, $a^2+b^2+c^2\leq 1$. Assume $\ket{\chi}=\alpha\ket{0}+\beta\ket{1}+\gamma\ket{2}$, $\alpha,\beta,\gamma\in \bC$, $|\alpha|^2+|\beta|^2+|\gamma|^2\leq 1$. Then we could calculate \eqref{eq:286} directly by Cauchy's inequality:
	\begin{align}&\eqref{eq:286}\label{eq:309}
\\	
=&4+2((h+a)\nRe(\alpha)+b \nIm(\alpha)+c \nRe(\beta))\\
\leq& 4+2\sqrt{(h+a)^2+b^2+c^2}\sqrt{(\nRe(\alpha))^2+(\nIm(\alpha))^2+(\nRe(\beta))^2}\\
	\leq& 4+2\sqrt{(h+a)^2+b^2+c^2}\\
	\leq& 4+2\sqrt{h^2+2ah+1}\\
	\leq& 4+2\sqrt{2+2a}\\
	=&8\cos^2(\frac{1}{4}\arccos \nRe(\frac{\braket{\varphi|\psi}}{|\ket{\varphi}|}))\\
	\leq &8\cos^2(\frac{1}{4}\arccos \nRe(\frac{\braket{\varphi|\psi}}{|\ket{\varphi}|\cdot |\ket{\psi}|}))\label{eq:314}
	\end{align}

	where the equality holds when  $$h=1,a^2+b^2+c^2=1,(\nRe(\alpha))^2+(\nIm(\alpha))^2+(\nRe(\beta))^2=1$$ $$\nRe(\alpha)/(h+a)=\nIm(\alpha)/b=\nRe(\beta)/c$$
	 This implies \eqref{eq:286k}. 
\end{proof}

Then we have the following lemma, which says the $\ket{\chi}$ should be close to the optimal state if \eqref{eq:207} holds approximately:
\begin{lem}\label{lem:d2}
	Suppose vectors $\ket{\varphi}$, $\ket{\psi}$ both have norm $\leq 1$. Assume\footnote{This is to rule out the border case where $\ket{\psi}+\ket{\varphi}\approx 0$. In this case \eqref{eq:207} still holds, but $\ket{\chi}$ is not uniquely determined.} $|\ket{\varphi}+\ket{\psi}|\geq 0.5$. If vector $\ket{\chi}$, $|\ket{\chi}|\leq 1$, satisfies
	\begin{equation}\label{eq:288m}|\ket{\varphi}+\ket{\chi}|^2+|\ket{\psi}+\ket{\chi}|^2\geq 8\cos^2(\frac{1}{4}\arccos \nRe((\frac{\braket{\varphi|\psi}}{|\ket{\varphi}|\cdot|\ket{\psi}|}))-\epsilon\end{equation}
	where $\epsilon<0.1$. Then there is
	\begin{equation}\label{eq:289}\ket{\chi}\approx_{3\sqrt{\epsilon}} (\ket{\psi}+\ket{\varphi})/|\ket{\psi}+\ket{\varphi}|\end{equation}
\end{lem}
\begin{proof}
Similar to the proof of Lemma \ref{lem:d1} we have
$$|\ket{\varphi}+\ket{\chi}|^2+|\ket{\psi}+\ket{\chi}|^2=4+2((h+a) \nRe(\alpha)+b \nIm(\alpha)+c \nRe(\beta))$$
where $h,\alpha, \beta, a,b,c$ are defined in the same way.\par
Then \eqref{eq:309}-\eqref{eq:314} together with \eqref{eq:288m} imply all the inequalities in \eqref{eq:309}-\eqref{eq:314} are equality up to an error $\epsilon$, which means
\begin{equation}\label{eq:426}4+2((h+a)\nRe(\alpha)+b\nIm(\alpha)+c\nRe(\beta))\geq 4+2\sqrt{(h+a)^2+b^2+c^2}\sqrt{(\nRe(\alpha))^2+(\nIm(\alpha))^2+(\nRe(\beta))^2}-\epsilon\end{equation}
\begin{equation}\label{eq:427}4+2\sqrt{(h+a)^2+b^2+c^2}\sqrt{(\nRe(\alpha))^2+(\nIm(\alpha))^2+(\nRe(\beta))^2}
	\geq 4+2\sqrt{(h+a)^2+b^2+c^2}-\epsilon\end{equation}
	\begin{equation}\label{eq:351}4+2\sqrt{(h+a)^2+b^2+c^2}
	\geq 4+2\sqrt{h^2+2ah+1}-\epsilon\end{equation}
	\begin{equation}\label{eq:352}4+2\sqrt{h^2+2ah+1}
	\geq 4+2\sqrt{2+2a}-\epsilon\end{equation}
By Lemma \ref{lem:cauchy}, define
$$\vec{u}:=(\nRe(\alpha),\nIm(\alpha),\nRe(\beta)), \vec{v}:= (h+a,b,c),$$
the inequalities can be translated to
$$\text{(By condition) }|\vec{v}|=|\ket{\varphi}+\ket{\psi}|\geq 0.5$$
$$\text{(By \eqref{eq:426}) }\vec{u}\cdot \vec{v}\geq |\vec{u}|\cdot |\vec{v}|-\epsilon/2$$
\begin{equation}\label{eq:354}\text{(By \eqref{eq:427}) }|\vec{u}|\geq 1-\epsilon/(2|\vec{v}|)\geq 1-\epsilon\end{equation}
thus
$$\frac{\vec{u}\cdot \vec{v}}{|\vec{u}|\cdot |\vec{v}|}\geq 1-\frac{\epsilon}{(1-\epsilon)}$$
By Lemma \ref{lem:cauchy} there is\begin{equation}\label{eq:431} \vec{u}/|\vec{u}|\approx_{\sqrt{2\epsilon/(1-\epsilon)}} \vec{v}/|\vec{v}|\end{equation}
Note that by \eqref{eq:354}, and $|\ket{\chi}|\leq 1$ we have
\begin{equation}\label{eq:432}\ket{\chi}\approx_{\sqrt{2\epsilon}} \alpha\ket{0}+\nRe(\beta)\ket{1}\end{equation}
\begin{equation}\label{eq:433}\vec{u}\approx_{1-1/(1-\epsilon)} \vec{u}/|\vec{u}|\end{equation}
Combining \eqref{eq:431}\eqref{eq:432}\eqref{eq:433} implies \eqref{eq:289}.
\end{proof}

\subsection{$5$-states Lemmas, with Approximate Normalization}
\begin{cor}\label{cor:d1}
Suppose subnormalized vectors $\ket{\phi_0}$, $\ket{\phi_4}$ satisfy $|\ket{\phi_0}|\approx_\epsilon 1$, $|\ket{\phi_4}|\approx_\epsilon 1$, $\ket{\phi_0}\approx_\epsilon -\ket{\phi_4}$, $\epsilon<0.1$. Then for subnormalized vectors $\ket{\phi_1}$, $\ket{\phi_2}$, $\ket{\phi_3}$ there is
$$|\ket{\phi_0}+\ket{\phi_1}|^2+|\ket{\phi_1}+\ket{\phi_2}|^2+|\ket{\phi_2}+\ket{\phi_3}|^2+|\ket{\phi_3}+\ket{\phi_4}|^2\leq 16\cos^2(\pi/8)+6.6\epsilon$$
\end{cor}
\begin{proof}
By Lemma \ref{lem:d1}:
$$|\ket{\phi_0}+\ket{\phi_1}|^2+|\ket{\phi_1}+\ket{\phi_2}|^2\leq 8\cos^2(\frac{1}{4}\arccos\nRe(\frac{\braket{\phi_0|\phi_2}}{|\ket{\phi_0}|\cdot|\ket{\phi_2}|}))$$
$$|\ket{\phi_2}+\ket{\phi_3}|^2+|\ket{\phi_3}+\ket{\phi_4}|^2\leq 8\cos^2(\frac{1}{4}\arccos\nRe(\frac{\braket{\phi_2|\phi_4}}{|\ket{\phi_2}|\cdot|\ket{\phi_4}|}))$$
By basic calculus, from $\ket{\phi_0}/|\ket{\phi_0}|\approx_{3.3\epsilon}-\ket{\phi_4}/|\ket{\phi_4}|$ we know
\begin{align}
	&\cos^2(\frac{1}{4}\arccos\nRe(\frac{\braket{\phi_0|\phi_2}}{|\ket{\phi_0}|\cdot|\ket{\phi_2}|}))+\cos^2(\frac{1}{4}\arccos\nRe(\frac{\braket{\phi_2|\phi_4}}{|\ket{\phi_2}|\cdot|\ket{\phi_4}|}))\\
	\leq &\cos^2(\frac{1}{4}\arccos\nRe(\frac{\braket{\phi_0|\phi_2}}{|\ket{\phi_0}|\cdot|\ket{\phi_2}|}))+\cos^2(\frac{1}{4}\arccos\nRe(-\frac{\braket{\phi_0|\phi_2}}{|\ket{\phi_0}|\cdot|\ket{\phi_2}|}))+3.3\epsilon/4\\
	\leq & 2\cos^2(\pi/8)+3.3\epsilon/4
\end{align}
which completes the proof.
\end{proof}

\begin{cor}\label{cor:d4r}
	Suppose subnormalized vectors $\ket{\phi_0}$, $\ket{\phi_1}$, $\ket{\phi_2}$, $\ket{\phi_3}$, $\ket{\phi_4}$ satisfy $|\ket{\phi_0}|\approx_\epsilon 1$, $|\ket{\phi_4}|\approx_\epsilon 1$, $\ket{\phi_0}\approx_\epsilon -\ket{\phi_4}$, $\epsilon<0.01$. If
	\begin{equation}\label{eq:290o}|\ket{\phi_0}+\ket{\phi_1}|^2+|\ket{\phi_1}+\ket{\phi_2}|^2+|\ket{\phi_2}+\ket{\phi_3}|^2+|\ket{\phi_3}+\ket{\phi_4}|^2\geq 16\cos^2(\pi/8)-\epsilon\end{equation}
	Then $\ket{\phi_0}$, $\ket{\phi_1}$, $\ket{\phi_2}$, $\ket{\phi_3}$, $\ket{\phi_4}$ satisfy:
	\begin{equation}\label{eq:359i}\nRe(\braket{\phi_0|\phi_2})\approx_{4\sqrt{\epsilon}} 0\text{ (that is, $|\ket{\phi_0}+\ket{\phi_2}|\approx_{8\sqrt{\epsilon}} \sqrt{2}$)}\end{equation}
	\begin{equation}\label{eq:360i}\nRe(\braket{\phi_2|\phi_4})\approx_{4\sqrt{\epsilon}} 0\text{ (that is, $|\ket{\phi_2}+\ket{\phi_4}|\approx_{8\sqrt{\epsilon}} \sqrt{2}$)}\end{equation}
	\begin{equation}\label{eq:361i}\ket{\phi_1}\approx_{15\sqrt{\epsilon}} \frac{1}{\sqrt{2}}(\ket{\phi_0}+\ket{\phi_2})\end{equation}
	\begin{equation}\label{eq:362i}\ket{\phi_3}\approx_{15\sqrt{\epsilon}} \frac{1}{\sqrt{2}}(\ket{\phi_2}+\ket{\phi_4})\end{equation}
\end{cor}
\begin{proof}
Notice that by the conditions we get $\frac{\ket{\phi_0}}{|\ket{\phi_0}|}\approx_{1.12\epsilon} \ket{\phi_0}$, $\frac{\ket{\phi_4}}{|\ket{\phi_4}|}\approx_{1.12\epsilon} \ket{\phi_4}$. And 
by the same argument as Corollary \ref{cor:d1} there is
	\begin{align}
		&|\ket{\phi_0}+\ket{\phi_1}|^2+|\ket{\phi_1}+\ket{\phi_2}|^2+|\ket{\phi_2}+\ket{\phi_3}|^2+|\ket{\phi_3}+\ket{\phi_4}|^2\label{eq:359}\\
		\leq &8\cos^2(\frac{1}{4}\arccos\nRe(\frac{\braket{\phi_0|\phi_2}}{|\ket{\phi_0}|\cdot|\ket{\phi_2}|}))+8\cos^2(\frac{1}{4}\arccos\nRe(\frac{\braket{\phi_2|\phi_4}}{|\ket{\phi_2}|\cdot|\ket{\phi_4}|}))\label{eq:360o}\\
	\leq &8\cos^2(\frac{1}{4}\arccos\nRe(\frac{\braket{\phi_0|\phi_2}}{|\ket{\phi_0}|\cdot|\ket{\phi_2}|}))+8\cos^2(\frac{1}{4}\arccos\nRe(-\frac{\braket{\phi_0|\phi_2}}{|\ket{\phi_0}|\cdot|\ket{\phi_2}|}))+6.6\epsilon\label{eq:360}\\
		\leq &16\cos^2(\pi/8)+6.6\epsilon\label{eq:361}
	\end{align}
which together with \eqref{eq:290o} implies each of \eqref{eq:359}\eqref{eq:360o}\eqref{eq:360}\eqref{eq:361} are approximately $16\cos^2(\pi/8)$ with error in $[-\epsilon,6.6\epsilon]$. First by \eqref{eq:360}:
	$$8\cos^2(\frac{1}{4}\arccos\nRe(\frac{\braket{\phi_0|\phi_2}}{|\ket{\phi_0}|\cdot|\ket{\phi_2}|}))+8\cos^2(\frac{1}{4}\arccos\nRe(-\frac{\braket{\phi_0|\phi_2}}{|\ket{\phi_0}|\cdot|\ket{\phi_2}|}))\approx_{7.6\epsilon} 16\cos^2(\pi/8)$$
	which by Lemma \ref{lem:calculus} implies
	$$\nRe(\frac{\braket{\phi_0|\phi_2}}{|\ket{\phi_0}|\cdot|\ket{\phi_2}|})\approx_{4\sqrt{\epsilon}}0$$
	which proves \eqref{eq:359i}; and by symmetry \eqref{eq:360i} is proved too.\par
And by the approximate equality of \eqref{eq:359}\eqref{eq:360o} there is
	$$|\ket{\phi_0}+\ket{\phi_1}|^2+|\ket{\phi_1}+\ket{\phi_2}|^2\geq 8\cos^2(\frac{1}{4}\arccos\nRe(\frac{\braket{\phi_0|\phi_2}}{|\ket{\phi_0}|\cdot|\ket{\phi_2}|}))-7.6\epsilon$$
	$$|\ket{\phi_2}+\ket{\phi_3}|^2+|\ket{\phi_3}+\ket{\phi_4}|^2\geq 8\cos^2(\frac{1}{4}\arccos\nRe(\frac{\braket{\phi_2|\phi_4}}{|\ket{\phi_2}|\cdot|\ket{\phi_4}|}))-7.6\epsilon,$$
	apply Lemma \ref{lem:d2} we get
	$$\ket{\phi_1}\approx_{9\sqrt{\epsilon}} (\ket{\phi_0}+\ket{\phi_2})/|\ket{\phi_0}+\ket{\phi_2}|$$
	$$\ket{\phi_3}\approx_{9\sqrt{\epsilon}} (\ket{\phi_2}+\ket{\phi_4})/|\ket{\phi_2}+\ket{\phi_4}|$$
	then substituting \eqref{eq:359i}\eqref{eq:360i}  implies \eqref{eq:361i}\eqref{eq:362i}.
\end{proof}

\subsection{Proofs of Lemmas in Section \ref{sec:10.1}}
Now we are going to prove Lemma \ref{lem:optla}, which is basically an application of Corollary \ref{cor:d1} above.
\begin{proof}[Proof of Lemma \ref{lem:optla}]
	From the conditions we know $$\forall i\in \{0,1\cdots 7\},\ket{\phi_{0,i}}\approx_{\sqrt{\epsilon}}-\ket{\phi_{1,i+4}}\approx_{\sqrt{\epsilon}}-\ket{\phi_{0,i+4}}$$
	Without loss of generality assume $A_0\leq A_1$. Apply Corollary \ref{cor:d1} we have
	$$|\ket{\phi_{0,0}}+\ket{\phi_{0,1}}|^2+|\ket{\phi_{0,1}}+\ket{\phi_{0,2}}|^2+|\ket{\phi_{0,2}}+\ket{\phi_{0,3}}|^2+|\ket{\phi_{0,3}}+\ket{\phi_{0,4}}|^2\leq 2\cos^2(\frac{\pi}{8})A_0+5\sqrt{\epsilon}$$
	$$|\ket{\phi_{0,4}}+\ket{\phi_{0,5}}|^2+|\ket{\phi_{0,5}}+\ket{\phi_{0,6}}|^2+|\ket{\phi_{0,6}}+\ket{\phi_{0,7}}|^2+|\ket{\phi_{0,7}}+\ket{\phi_{0,0}}|^2\leq 2\cos^2(\frac{\pi}{8})A_0+5\sqrt{\epsilon}$$
	which implies
	$$\sum_{i\in \{0,1\cdots 7\}}|\ket{\phi_{0,i-1}}+\ket{\phi_{0,i}}|^2\leq 4\cos^2(\pi/8)A_0+10\sqrt{\epsilon}$$
	which together with $\sum_{i\in \{0,1\cdots 7\}}|\ket{\phi_{0,i}}-\ket{\phi_{1,i}}|^2\leq \epsilon$ completes the proof.
\end{proof}
Now we prove Lemma \ref{lem:10.2rr}.
\begin{proof}
Without loss of generality assume $A_0\leq A_1$. From the conditions we know \begin{equation}\label{eq:372w}\forall i\in \{0,1\cdots 7\},\ket{\phi_{0,i}}\approx_{\sqrt{\epsilon}}-\ket{\phi_{1,i+4}}\approx_{\sqrt{\epsilon}}-\ket{\phi_{0,i+4}}\end{equation}
Substitute \eqref{eq:372w} to the fourth condition we get
\begin{equation}|\ket{\phi_{0,0}}+\ket{\phi_{0,1}}|^2+|\ket{\phi_{0,1}}+\ket{\phi_{0,2}}|^2+|\ket{\phi_{0,2}}+\ket{\phi_{0,3}}|^2+|\ket{\phi_{0,3}}+\ket{\phi_{0,4}}|^2\geq \frac{1}{2}\cos^2(\frac{\pi}{8})-17\epsilon\end{equation}
	\begin{equation}|\ket{\phi_{0,4}}+\ket{\phi_{0,5}}|^2+|\ket{\phi_{0,5}}+\ket{\phi_{0,6}}|^2+|\ket{\phi_{0,6}}+\ket{\phi_{0,7}}|^2+|\ket{\phi_{0,7}}+\ket{\phi_{0,0}}|^2\geq \frac{1}{2}\cos^2(\frac{\pi}{8})-17\epsilon\end{equation}
	Apply Corollary \ref{cor:d4r} we have
	\begin{equation}\label{eq:292}\ket{\phi_{0,1}}\approx_{11\epsilon^{1/4}} \frac{1}{\sqrt{2}}(\ket{\phi_{0,0}}+\ket{\phi_{0,2}})\end{equation}
	\begin{equation}\label{eq:293}\ket{\phi_{0,3}}\approx_{11\epsilon^{1/4}} \frac{1}{\sqrt{2}}(\ket{\phi_{0,2}}+\ket{\phi_{0,4}})\approx_{\sqrt{\epsilon}} \frac{1}{\sqrt{2}}(\ket{\phi_{0,2}}-\ket{\phi_{0,0}})\end{equation}
	Define
	\begin{equation}\label{eq:213}\ket{\tilde\phi_{0,+}}=\frac{1}{\sqrt{2}}(\ket{\phi_{0,0}}-\mi\ket{\phi_{0,2}})\end{equation}
	\begin{equation}\label{eq:214}\ket{\tilde\phi_{0,-}}=\frac{1}{\sqrt{2}}(\ket{\phi_{0,0}}+\mi\ket{\phi_{0,2}})\end{equation}
	Then through a direct calculation we can verify
	\begin{equation}\label{eq:328}
		\forall i\in \{0,1,2,3,4\},\ket{\phi_{0,i}}\approx_{12\epsilon^{1/4}} e^{-i\mi\pi/4}\ket{\tilde\phi_{0,+}}+e^{i\mi\pi/4}\ket{\tilde\phi_{0,-}}
	\end{equation}
	Which together with \eqref{eq:372w} we get
		\begin{equation}\label{eq:380}
		\forall i\in \{0,1,\cdots 7\},\ket{\phi_{0,i}}\approx_{13\epsilon^{1/4}} e^{-i\mi\pi/4}\ket{\tilde\phi_{0,+}}+e^{i\mi\pi/4}\ket{\tilde\phi_{0,-}}
	\end{equation}
	Thus for $\ket{\phi_{0,+}}$, $\ket{\phi_{0,-}}$ defined in \eqref{eq:163s}\eqref{eq:164s}, we have 
	\begin{equation}\label{eq:296}\ket{\phi_{0,+}}:= \frac{1}{8}\sum_{i\in \{0,1,\cdots 7\}}e^{-i\mi\pi/4}\ket{\phi_{0,i}}\approx_{13\epsilon^{1/4}}\ket{\tilde\phi_{0,+}}\end{equation}
		\begin{equation}\label{eq:297}\ket{\phi_{0,-}}:= \frac{1}{4}\sum_{i\in \{0,1,\cdots 7\}}e^{i\mi\pi/4}\ket{\phi_{0,i}}\approx_{13\epsilon^{1/4}} \ket{\tilde\phi_{0,-}}\end{equation}
		Combining it with \eqref{eq:380} we have 
		$$\forall i\in \{0,1,\cdots 7\},\ket{\phi_{0,i}}\approx_{38\epsilon^{1/4}} e^{-i\mi\pi/4}\ket{\phi_{0,+}}+e^{i\mi\pi/4}\ket{\phi_{0,-}}$$
		Which together with \eqref{eq:372w} completes the proof of Lemma \ref{lem:10.2rr}.
\end{proof}

\bibliographystyle{plain}
\bibliography{main_Jiayu}
\end{document}